\documentclass[12pt]{article}
\usepackage{graphicx} 

\usepackage[]{natbib}
\usepackage{bibunits}
\defaultbibliography{KL_references.bib}
\defaultbibliographystyle{apa}

\usepackage{hyperref}
\usepackage{amsmath}
\usepackage{amssymb}
\usepackage{amsthm} 
\usepackage{arydshln} 
\usepackage{setspace}
\usepackage{caption}
\usepackage{verbatim}
\usepackage{bm}
\usepackage{relsize}
\usepackage{tikz}
\usetikzlibrary{arrows.meta}
\usepackage{float} 

\newcommand*{\Scale}[2][4]{\scalebox{#1}{$#2$}}%

\newtheorem{theorem}{Theorem}

\newtheorem{lemma}{Lemma}

\newtheorem{apptheorem}{Theorem}
\newtheorem{appsubtheorem}{Theorem} 
\newtheorem{applemma}{Lemma}
\newtheorem{appremark}{Remark}
\newtheorem{appcorollary}{Corollary}

\usepackage{enumitem}

\usepackage{iftex}
\ifPDFTeX
  \usepackage[T1]{fontenc}
  \usepackage[utf8]{inputenc}
  \usepackage{textcomp} 
\else 
  \usepackage{unicode-math}
  \defaultfontfeatures{Scale=MatchLowercase}
  \defaultfontfeatures[\rmfamily]{Ligatures=TeX,Scale=1}
\fi
\usepackage{lmodern}
\ifPDFTeX\else  
\fi
\IfFileExists{upquote.sty}{\usepackage{upquote}}{}
\IfFileExists{microtype.sty}{
  \usepackage[]{microtype}
  \UseMicrotypeSet[protrusion]{basicmath} 
}{}
\makeatletter
\@ifundefined{KOMAClassName}{
  \IfFileExists{parskip.sty}{%
    \usepackage{parskip}
  }{
    \setlength{\parindent}{0pt}
    \setlength{\parskip}{6pt plus 2pt minus 1pt}}
}{
  \KOMAoptions{parskip=half}}
\makeatother
\usepackage{xcolor}
\makeatletter
\ifx\paragraph\undefined\else
  \let\oldparagraph\paragraph
  \renewcommand{\paragraph}{
    \@ifstar
      \xxxParagraphStar
      \xxxParagraphNoStar
  }
  \newcommand{\xxxParagraphStar}[1]{\oldparagraph*{#1}\mbox{}}
  \newcommand{\xxxParagraphNoStar}[1]{\oldparagraph{#1}\mbox{}}
\fi
\ifx\subparagraph\undefined\else
  \let\oldsubparagraph\subparagraph
  \renewcommand{\subparagraph}{
    \@ifstar
      \xxxSubParagraphStar
      \xxxSubParagraphNoStar
  }
  \newcommand{\xxxSubParagraphStar}[1]{\oldsubparagraph*{#1}\mbox{}}
  \newcommand{\xxxSubParagraphNoStar}[1]{\oldsubparagraph{#1}\mbox{}}
\fi
\makeatother

\usepackage{longtable,booktabs,array}
\usepackage{calc} 
\usepackage{etoolbox}
\makeatletter
\patchcmd\longtable{\par}{\if@noskipsec\mbox{}\fi\par}{}{}
\makeatother
\IfFileExists{footnotehyper.sty}{\usepackage{footnotehyper}}{\usepackage{footnote}}
\makesavenoteenv{longtable}
\usepackage{graphicx}
\makeatletter
\def\maxwidth{\ifdim\Gin@nat@width>\linewidth\linewidth\else\Gin@nat@width\fi}
\def\maxheight{\ifdim\Gin@nat@height>\textheight\textheight\else\Gin@nat@height\fi}
\makeatother
\setkeys{Gin}{width=\maxwidth,height=\maxheight,keepaspectratio}
\makeatletter
\def\fps@figure{htbp}
\makeatother

\makeatletter
\@ifpackageloaded{caption}{}{\usepackage{caption}}
\AtBeginDocument{%
\ifdefined\contentsname
  \renewcommand*\contentsname{Table of contents}
\else
  \newcommand\contentsname{Table of contents}
\fi
\ifdefined\listfigurename
  \renewcommand*\listfigurename{List of Figures}
\else
  \newcommand\listfigurename{List of Figures}
\fi
\ifdefined\listtablename
  \renewcommand*\listtablename{List of Tables}
\else
  \newcommand\listtablename{List of Tables}
\fi
\ifdefined\figurename
  \renewcommand*\figurename{Figure}
\else
  \newcommand\figurename{Figure}
\fi
\ifdefined\tablename
  \renewcommand*\tablename{Table}
\else
  \newcommand\tablename{Table}
\fi
}
\@ifpackageloaded{float}{}{\usepackage{float}}
\floatstyle{ruled}
\@ifundefined{c@chapter}{\newfloat{codelisting}{h}{lop}}{\newfloat{codelisting}{h}{lop}[chapter]}
\floatname{codelisting}{Listing}

\makeatother
\makeatletter
\@ifpackageloaded{caption}{}{\usepackage{caption}}
\@ifpackageloaded{subcaption}{}{\usepackage{subcaption}}
\makeatother

\ifLuaTeX
  \usepackage{selnolig}  
\fi
\usepackage{bookmark}

\newcommand{\anon}{1}

\graphicspath{Images/}

\begin{document}

\def\spacingset#1{\renewcommand{\baselinestretch}%
{#1}\small\normalsize} \spacingset{1}

\if1\anon
{
  \title{\bf Fixed-Effects Models for Causal Inference in Longitudinal Cluster Randomized and Quasi-Experimental Trials}
  \author{Kenneth M. Lee$^{1}$\thanks{Corresponding author. Center for Clinical Trials Innovation, University of Pennsylvania School of Medicine, 3600 Civic Center Boulevard, Philadelphia, PA 19104}\hspace{.2cm} and Fan Li$^{2}$\\
    $^1$Department of Biostatistics, Epidemiology and Informatics,\\University of Pennsylvania\\
    $^2$Department of Biostatistics, Yale School of Public Health
    }
  \maketitle
} \fi

\if0\anon
{
  \title{\bf Fixed-Effects Models for Causal Inference in Longitudinal Cluster Randomized and Quasi-Experimental Trials}
  \maketitle
} \fi

\bigskip
\begin{abstract}
This article investigates the model-robustness of fixed-effects models for analyzing a broad class of longitudinal cluster trials (CTs) such as stepped-wedge, parallel-with-baseline and crossover designs, encompassing both randomized (CRTs) and quasi-experimental (CQTs) designs. 
We clarify a longstanding misconception in biostatistics, demonstrating that fixed-effects models, traditionally perceived as targeting only finite-sample conditional estimands, can effectively target super-population marginal estimands through an M-estimation framework. 
We comprehensively prove that linear and log-link fixed-effects models with correctly specified treatment effect structures can broadly yield consistent and asymptotically normal estimators for nonparametrically defined treatment effect estimands in longitudinal CRTs, even under arbitrary misspecification of other model components.
We identify that the constant treatment effect estimator generally targets the period-average treatment effect for the overlap population (P-ATO);
accordingly, some CRT designs don't even require correct specification of the treatment effect structure for model-robustness.
We further characterize conditions where fixed-effects models can maintain consistency by adjusting for both cluster-level and individual-level time-invariant confounding in longitudinal CQTs. 
Altogether, supported by simulation and a case study re-analysis, we establish fixed-effects models as a robust and potentially preferable alternative to mixed-effects models for longitudinal CT analysis.
\end{abstract}

\noindent%
\vfill

\newpage
\spacingset{1.8} 

\begin{bibunit}
\section{Introduction}

\subsection{Longitudinal cluster trial designs}
Cluster randomized trials (CRTs), where clusters of individuals are randomized to receive the intervention, are useful trial designs for evaluating causal relationships between interventions and outcomes in settings where individual randomization may not be feasible \citep{hayes_cluster_2017}.
However, even cluster randomization may not always be feasible in practice, resulting in cluster \emph{quasi-experimental} trial (CQT) designs.
Among this broad category of cluster trials (CTs), some CT designs are inherently \emph{longitudinal}: having both between and within-cluster variation in treatment status.
Popular examples of longitudinal CT designs include the stepped-wedge (SW), parallel-with-baseline (PB), and crossover (XO) designs.
Altogether, this article will cover analysis methods for a broad class of longitudinal CT designs, including both CRTs and CQTs, focusing specifically on the SW-CT (SW-CRT, SW-CQT), PB-CT (PB-CRT, PB-CQT), and CXO (CRXO, CQXO) designs.
For simplicity of presentation, we only refer to the complete and standard versions of these designs with an equal number of clusters assigned to each treatment sequence.

\subsection{Fixed-effects versus Mixed-effects regression}

In health science applications, the mixed-effects model has traditionally been the primarily adopted analysis method for the analysis of longitudinal CRTs and fits cluster intercepts as random effects; for example, see a review of mixed-effects models for SW-CRTs in \citet{li_mixed-effects_2021}.
Practically, the mixed-effects model relies on randomization to control for known and unknown confounders; however, the benefits of randomization may be lost when the number of clusters is limited \citep{taljaard_substantial_2016}, potentially yielding inflated Type I error rates and under-coverage \citep{barker_minimum_2017,lee_fixed-effects_2024}.
In contrast to the mixed-effects model, the fixed-effects model is perceived as automatically adjusting for all known or unknown time-invariant covariates and confounding, and can yield valid type I error rates and coverage probabilities in scenarios with a limited number of clusters \citep{lee_fixed-effects_2024}.
Previous SW-CQT studies have implemented the fixed-effects model, citing difficulties that arise from a small number of clusters \citep{kelly_study_2015}, practical and logistical issues that prevented randomization \citep{groshaus_use_2012}, and concerns over confounding \citep{craine_stepped_2015}.
However, in contrast to their mixed-effects counterpart, the fixed-effects model
remains a relatively underutilized, underdeveloped, and largely misunderstood analysis method in the analysis of longitudinal CTs and the broader statistical literature.

Historically, the statistical literature has misinterpreted mixed-effects models as innately estimating coefficients over the entire superpopulation of all possible clusters, and fixed-effects models, over the given finite population of sampled clusters \citep{donner_pitfalls_2004}. 
This may stem from the dominant model-based perspective in the analysis of clustered data, in which mixed-effects models are viewed as directly representing the assumed underlying data-generating process with clusters treated as random samples.
In contrast to the model-based perspective, we focus on estimand-aligned inference where estimands are defined separately from the regression model parameters. Previous efforts have noted that both the linear mixed-effects and fixed-effects models can be applied under either a superpopulation or finite-sample framework \citep{lee_fixed-effects_2024,lee_how_2025,lee_what_2025}.
As we will rigorously demonstrate in this article, standard linear and log-link fixed-effects models do not have an innately pre-determined inference space and can be implemented within a superpopulation framework using estimating equations and M-estimation.

This general misunderstanding in the applied statistics literature is contrasted by the econometrics literature which has widely adopted fixed-effects models due to the typically quasi-experimental and observational nature of their study designs \citep{wooldridge_econometric_2010}.
Nearly 50 years ago, \citet{mundlak_pooling_1978} and \citet{hausman_specification_1978} characterized fixed-effects models as more generalized forms of mixed-effects models with fewer model assumptions governing the associations between unit intercepts and other model covariates.
Alongside its regular application, fixed-effects methods have also been advanced within the econometrics literature in recent years, particularly in scenarios with time-varying treatment effects under causal dynamic difference-in-differences \citep{goodman-bacon_difference--differences_2021,callaway_difference--differences_2021}.

Despite these efforts, the model-robust properties of fixed-effects estimators for targeting nonparametric treatment effect estimands in longitudinal CTs remains unclear. Here, model-robustness refers to valid statistical inference in large samples when the working models may be misspecified.
Crucially, \citet{wang_how_2024} demonstrated that commonly used marginal generalized estimating equations (GEEs) and corresponding mixed-effects model estimators are consistent in SW-CRTs regardless of arbitrary model-misspecification, provided the potentially time-varying treatment effect structure is correctly specified. However, the merit of fixed-effects regression in such contexts has remained unclear and even misrepresented. 
In light of the aforementioned existing literature and standard practice, the contributions of this article are several fold.
First, we demonstrate that linear and log-link fixed-effects models can avoid the incidental parameters problem and prove these fixed-effects models are model-robust across a broad class of longitudinal CRTs under standard assumptions, yielding consistent M-estimators in a super-population framework for nonparametrically defined treatment effect estimands.
The demonstrated model-robustness may require the potentially time-varying treatment effect structure to be correctly specified, but generally does not require correct specification of the remaining aspects of the working model (functional form of covariates and error distribution).
Second, we identify that the fixed-effects constant treatment effect estimator generally targets the period-average treatment effect for the overlap population (P-ATO) estimand in longitudinal CRTs. Accordingly, CRT designs where the P-ATO coincides with general time-averaged treatment effect estimands do not even require correct specification of the treatment effect structure for model-robustness.
Third, we use potential outcomes to rigorously characterize the conditions under which fixed-effects model estimators, unlike mixed-effects model estimators, are consistent in longitudinal CQTs.
We highlight that the fixed-effects model doesn't just automatically control for all cluster-level time-invariant confounding in CTs, but all individual-level as well.
Altogether, these new results provide crucial statistical insights into the underutilized fixed-effects model approach and highlight the conditions under which it is not only comparable to the mixed-effects model, but even preferable for the analysis of longitudinal CTs.

\section{Notation and Assumptions for Longitudinal CTs}
\label{sect:Notation}

Consider a longitudinal CT with $m$ clusters, where each cluster $i \in \{1,...,m\}$ contains $N_i$ individuals in its source population.
Data from each cluster are collected across calendar time as $J$ discrete, equally-spaced periods indexed by $j=1,...,J$.
We accordingly define $Z_i=j$ if cluster $i$ starts receiving treatment beginning in period $j \in \{1,...,J\}$, with $Z_i=0$ indicating that cluster $i$ never receives the treatment.
For each individual $k \in \{1,...,N_i\}$ in cluster $i$, we define $Y_{ijk}$ as their outcome in period $j$ and $\bm{X}_{ik}$ as their vector of baseline covariates. We assume $N_i$ takes values in a bounded subset of positive integers and can vary by clusters, but remains constant across calendar time $j$ in the study.
In practice, not all individuals in the cluster source population are included in the study. 
We define the enrollment indicator $S_{ijk}=1$ if individual $k$ from cluster $i$ is enrolled during period $j$. An observed cluster-period size is then the sum of the enrollment indicators, $N_{ij}=\sum_{k=1}^{N_i} S_{ijk}$.
Under this framework, each individual can appear in one, multiple, or even no periods, thus accommodating cross-sectional, closed-cohort, and open-cohort designs.

Treatment effect estimands are then nonparametrically defined with the potential outcomes framework.
With potential outcomes $Y_{ijk}(\tilde{z})$ where $\tilde{z} \in \{0,z\}$, let $Y_{ijk}(z)$ denote the potential outcome of individual $k$ in period $j$ of cluster $i$ had the cluster been assigned to a sequence where treatment first begins in period $z$ for $1 \leq z \leq J$.
$Y_{ijk}(0)$ denotes the untreated potential outcome, and we assume no anticipation.
In some CT designs, such as standard SW-CTs and PB-CTs, $j=1$ is typically a baseline period where $Y_{i1k} = Y_{i1k}(0) \, \forall \, i, k$. 
Additionally, in SW-CTs, $j=J$ is typically an all-treated period where $Y_{iJk} \neq Y_{iJk}(0) \,\forall\, i,k$. Across different CT designs, the complete, but not fully observed, data vector for each cluster $i$ is denoted as $\bm{W}_i = \{(Y_{ijk}(0), Y_{ijk}(z), S_{ijk}, \bm{X}_{ik}, Z_i, N_i) : k=1,...,N_i, j=1,...,J, 1 \leq z \leq J\}$, and the observed data, $\bm{O}_i=\{Y_{ijk}, \bm{X}_{ik}, Z_i : S_{ijk}=1, j=1,...,J, k=1,...,N_i\}$. 
We outline the following assumptions on $\bm{W}_i$ to enable the subsequent estimand definitions and proofs of the key results.

\textbf{A1.} (\textit{Super-population sampling}) Data vector $\{\bm{W}_i, i=1,...,m\}$  are independent and identically distributed draws from a population distribution $\mathcal{P}$ with finite second moments.
Within each cluster $i$, the data vectors $\{(Y_{i1k}(0), ..., Y_{iJk}(J)), \bm{X}_{ik}\}$ for $k=1,...,N_i$ are identically distributed given the source population size $N_i$.

\textbf{A2.} (\textit{Non-informative enrollment}) The enrollment size $\{N_{ij}: i=1,...,m, j=1,...,J\}$ are independent of $N_i$. Furthermore, the enrollment indicator $\{ S_{ijk} : j=1,...,J, k=1,...,N_i \}$ is independent of all other random variables in $\bm{W}_i$ given $N_i$.

\textbf{A3.} (\textit{Randomization}) Initial treatment period $Z_i$ is independent of all other random variables in $\bm{W}_i$.

\sloppy \noindent \textbf{A4.} \textit{(Mean independence)} i. (\textit{Parallel Trends}) The average within-cluster deviations of the untreated potential outcomes from the enrollment average are mean independent of $Z_i$ and $N_i$, given $S_i$. That is, defining 
\[
    \ddot{Y}_{ijk}(0)=Y_{ijk}(0)- \frac{\sum_{l=1}^{J} \sum_{k'=1}^{N_{i}} S_{ilk'} Y_{ilk'}(0)}{\sum_{l=1}^{J} \sum_{k'=1}^{N_{i}} S_{ilk'}} \,,
\]
then for each $j=1,...,J$ and $k=1,...,N_i$, we have $E\left[\ddot{Y}_{ijk}(0) | Z_i,N_i, S_i\right] = E\left[\ddot{Y}_{ijk}(0) | S_i\right]$. ii. (\textit{Exchangeable treatment effects}) Average treatment effects are mean independent of $Z_i$ and $N_i$, $E\left[Y_{ijk}(z)-Y_{ijk}(0)|Z_i,N_i\right] = E\left[Y_{ijk}(z)-Y_{ijk}(0)\right]$.

For notation convenience, we omit the subscript $i$ when taking expectation with respect to population distribution $\mathcal{P}$ (Assumption A1). For example, $E[f(\bm{O_i})]$, where $f$ is an arbitrary measurable function, is simplified as $E[f(\bm{O})]$, where $\bm{O}$ represents the random variable sampled from $\mathcal{P}$. Assumption A1 describes a superpopulation framework for causal inference in longitudinal CTs. We consider this assumption as a technical device to study the large-sample results and will separately examine the finite-sample performance of the model-robust estimators via simulations in Section \ref{sec:simulation}. 
Assumption A2 describes a random enrollment scheme and assumes away selection bias. 
Furthermore, Assumption A2 serves as a ``restricted informative sizes'' assumption and, alongside Assumption A1, allows the fixed-effects model to target the cluster-average treatment effect estimand.
Assumption A3 holds by design for longitudinal CRTs.
Assumption A4 is a supplementary assumption to establish the consistency of the linear fixed-effects model in the absence of Assumption A3, and will be invoked in longitudinal CQTs (Appendix \ref{app.sect:assumption4}).

More specifically, in the absence of randomization, Assumptions A1, A2, and A4 allow for individual and cluster-level time-invariant confounding in a longitudinal CQT, but imply that there is no additional within-cluster time-varying confounding.
As has rarely been pointed out, ``individual-level time-invariant confounding” in longitudinal CQTs under the stated assumptions operate entirely at the cluster-level due to the clustered implementation of the treatment arm, and any association between individual characteristics and treatment administration arise only through their aggregation into the cluster-level (Appendix \ref{app.sect:assumptions}).
Therefore, fixed-effects methods that adjust for all cluster-level time-invariant confounding will also automatically adjust for all individual-level time-invariant confounding in longitudinal CQT designs.

\section{Treatment effect estimands in longitudinal CTs}
\label{sect:estimands}

We define treatment effect estimands in longitudinal CTs using the potential outcomes framework.
With potential outcomes $Y_{ijk}(\tilde{z}) \in \{Y_{ijk}(0), Y_{ijk}(z)\}$, the duration of the treatment exposure time $d$ can be defined as a function of calendar period $j$ and initial treatment period $z$, where $z=j-d+1$, $z=(2d)/(j-1)$, and $z=j-2d+2$ in a SW-CT, PB-CT, and CXO, respectively (Appendix \ref{app.sect:Theorem_SW_Proof}).
We can then define the marginal cluster-average treatment effect estimand $\Delta_j(d) = E[Y_{.jk}(z)-Y_{.jk}(0)]$ as functions of potential outcomes in longitudinal CT designs (Figure \ref{fig:ct_designs}) while assuming conditional exchangeability within clusters (Assumption A1). 
Here, the expectations in the estimands are taken over the distribution of clusters, with $j$ and $d$ being fixed quantities. That is, $E[f(\bm{W}_i)] = \int f(\bm{w}) d\mathcal{P}(\bm{w})$ for any integrable function $f$. This corresponds to the average causal effect had all clusters been treated for a duration of $d$ exposed periods ($0<d<j$), at calendar period $j$.
The longitudinal CT definitions of $\Delta_j(d)$ are then model-free and fully accommodate time-varying treatment effect heterogeneity.
The following estimands can then be defined for a given longitudinal CT, either randomized or quasi-experimental.

\begin{figure}[ht!]
\setlength{\unitlength}{0.16in} 
\centering 
\begin{picture}(26,34)(0,11) 
\setlength\fboxsep{0pt}
\put(2,11.5){$j=1$}\put(8,11.5){$j=2$}
\put(14,11.5){$j=3$}\put(20,11.5){$j=4$}\put(26,11.5){$j=5$}
\put(1,19){\colorbox{gray!40}{\framebox(4,1.5){$\Delta_1(1)$}}}
\put(7,19){{\framebox(4,1.5){}}}
\put(13,19){\colorbox{gray!40}{\framebox(4,1.5){$\Delta_3(2)$}}}
\put(19,19){{\framebox(4,1.5){}}}
\put(25,19){\colorbox{gray!40}{\framebox(4,1.5){$\Delta_5(3)$}}}
\put(1,17){\colorbox{gray!40}{\framebox(4,1.5){$\Delta_1(1)$}}}
\put(7,17){\framebox(4,1.5){}}
\put(13,17){\colorbox{gray!40}{\framebox(4,1.5){$\Delta_3(2)$}}}
\put(19,17){{\framebox(4,1.5){}}}
\put(25,17){\colorbox{gray!40}{\framebox(4,1.5){$\Delta_5(3)$}}}
\put(1,15){\framebox(4,1.5){}}
\put(7,15){\colorbox{gray!40}{\framebox(4,1.5){$\Delta_2(1)$}}}
\put(13,15){\framebox(4,1.5){}}
\put(19,15){\colorbox{gray!40}{\framebox(4,1.5){$\Delta_4(2)$}}}
\put(25,15){{\framebox(4,1.5){}}}
\put(1,13){\framebox(4,1.5){}}
\put(7,13){\colorbox{gray!40}{\framebox(4,1.5){$\Delta_2(1)$}}}
\put(13,13){\framebox(4,1.5){}}
\put(19,13){\colorbox{gray!40}{\framebox(4,1.5){$\Delta_4(2)$}}}
\put(25,13){{\framebox(4,1.5){}}}
\put(-3,13.5){$\tilde{z}=2$}\put(-3,15.5){$\tilde{z}=2$}\put(-3,17.5){$\tilde{z}=1$}
\put(-3,19.5){$\tilde{z}=1$}
\put(1,22){$(iii)~\textit{CXO design}$}
\put(1,30){\framebox(4,1.5){}}
\put(7,30){\colorbox{gray!40}{\framebox(4,1.5){$\Delta_2(1)$}}}
\put(13,30){\colorbox{gray!40}{\framebox(4,1.5){$\Delta_3(2)$}}}
\put(19,30){\colorbox{gray!40}{\framebox(4,1.5){$\Delta_4(3)$}}}
\put(25,30){\colorbox{gray!40}{\framebox(4,1.5){$\Delta_5(4)$}}}
\put(1,28){\framebox(4,1.5){}}
\put(7,28){\colorbox{gray!40}{\framebox(4,1.5){$\Delta_2(1)$}}}
\put(13,28){\colorbox{gray!40}{\framebox(4,1.5){$\Delta_3(2)$}}}
\put(19,28){\colorbox{gray!40}{\framebox(4,1.5){$\Delta_4(3)$}}}
\put(25,28){\colorbox{gray!40}{\framebox(4,1.5){$\Delta_5(4)$}}}
\put(1,26){\framebox(4,1.5){}}
\put(7,26){\framebox(4,1.5){}}
\put(13,26){\framebox(4,1.5){}}
\put(19,26){{\framebox(4,1.5){}}}
\put(25,26){{\framebox(4,1.5){}}}
\put(1,24){\framebox(4,1.5){}}
\put(7,24){\framebox(4,1.5){}}
\put(13,24){\framebox(4,1.5){}}
\put(19,24){\framebox(4,1.5){}}
\put(25,24){{\framebox(4,1.5){}}}
\put(-3,24.5){$\tilde{z}=0$}\put(-3,26.5){$\tilde{z}=0$}\put(-3,28.5){$\tilde{z}=2$}
\put(-3,30.5){$\tilde{z}=2$}
\put(1,33){$(ii)~\textit{PB-CT design}$}
\put(1,41){\framebox(4,1.5){}}
\put(7,41){\colorbox{gray!40}{\framebox(4,1.5){$\Delta_2(1)$}}}
\put(13,41){\colorbox{gray!40}{\framebox(4,1.5){$\Delta_3(2)$}}}
\put(19,41){\colorbox{gray!40}{\framebox(4,1.5){$\Delta_4(3)$}}}
\put(25,41){\colorbox{gray!40}{\framebox(4,1.5){$\Delta_5(4)$}}}
\put(1,39){\framebox(4,1.5){}}
\put(7,39){\framebox(4,1.5){}}
\put(13,39){\colorbox{gray!40}{\framebox(4,1.5){$\Delta_3(1)$}}}
\put(19,39){\colorbox{gray!40}{\framebox(4,1.5){$\Delta_4(2)$}}}
\put(25,39){\colorbox{gray!40}{\framebox(4,1.5){$\Delta_5(3)$}}}
\put(1,37){\framebox(4,1.5){}}
\put(7,37){\framebox(4,1.5){}}
\put(13,37){\framebox(4,1.5){}}
\put(19,37){\colorbox{gray!40}{\framebox(4,1.5){$\Delta_4(1)$}}}
\put(25,37){\colorbox{gray!40}{\framebox(4,1.5){$\Delta_5(2)$}}}
\put(1,35){\framebox(4,1.5){}}
\put(7,35){\framebox(4,1.5){}}
\put(13,35){\framebox(4,1.5){}}
\put(19,35){\framebox(4,1.5){}}
\put(25,35){\colorbox{gray!40}{\framebox(4,1.5){$\Delta_5(1)$}}}
\put(-3,35.5){$\tilde{z}=5$}\put(-3,37.5){$\tilde{z}=4$}\put(-3,39.5){$\tilde{z}=3$}
\put(-3,41.5){$\tilde{z}=2$}
\put(1,44){$(i)~\textit{SW-CT design}$}
\end{picture}
\caption{
An example of three longitudinal CT designs: i) SW-CT, ii) PB-CT, iii) CXO, with corresponding causal estimands $\Delta_j(d)$ across different treatment sequences $\tilde{z} \in \{0,z\}$ (where $z>0$) and periods $j=1,...,5$.
In this example i) SW-CT, individual $\Delta_5(d)$ is not identifiable by design without assuming an absence of period-specific treatment effects.
}
\label{fig:ct_designs}
\end{figure}
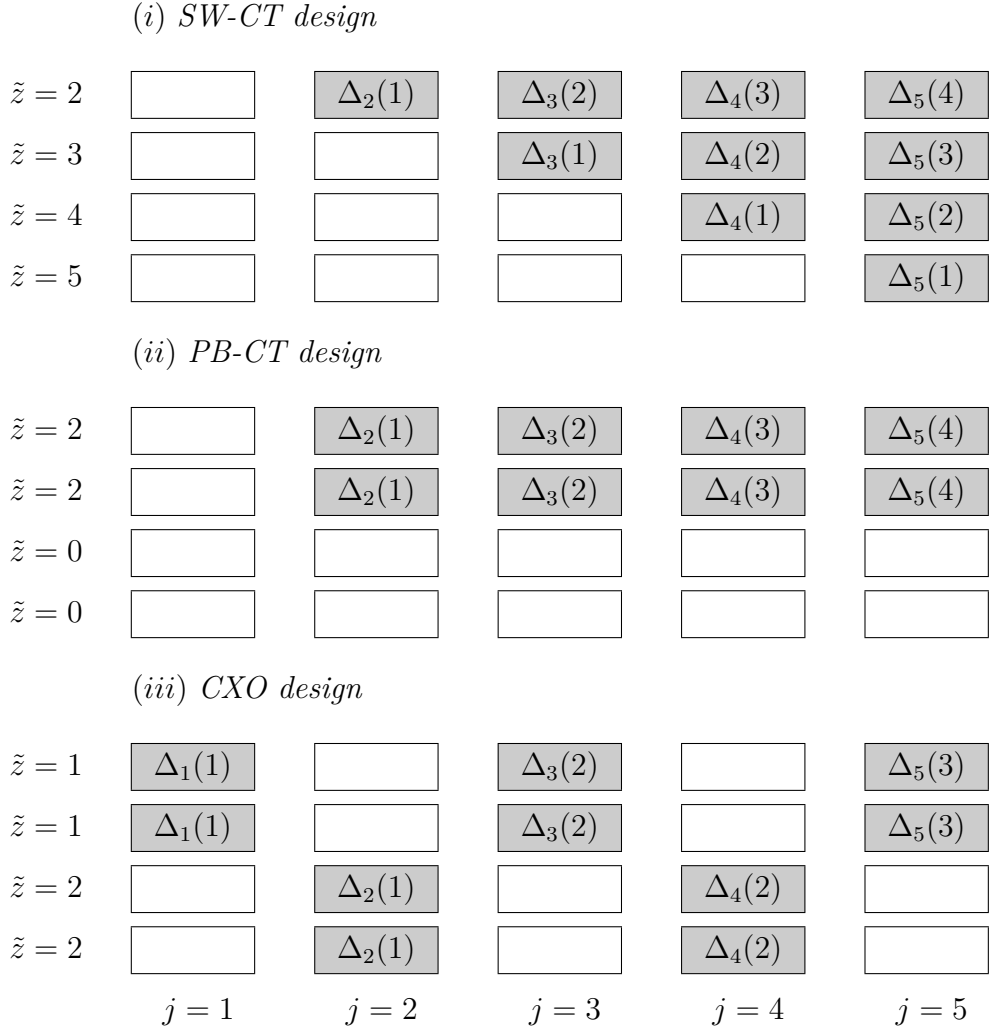

The \textit{constant treatment effect structure} assumes that $\Delta_j(d)$ is invariant across all periods $j$ and duration $d$. Accordingly, $\Delta_j(d) = \Delta$, and leads to a univariate target estimand $\Delta \in \mathbb{R}^1$.
This is a reasonable assumption when the treatment effect is expected to be immediate and sustained, and is typically assumed by default for longitudinal cluster trials \citep{li_mixed-effects_2021}.

The \textit{duration-specific treatment effect structure} considers $\Delta_j(d)$ to be constant across $j$, but vary primarily by $d$, such that $\Delta_j(d)=\Delta(d)$.
The target estimand in this case can be any function of $\bm{\Delta}^D=(\Delta(1),...,\Delta(J-1))^\top \in \mathbb{R}^{J-1}$ in a SW-CT or PB-CT, and $\bm{\Delta}^D=(\Delta(1),...,\Delta([2J+1-(-1)^J]/4))^\top \in \mathbb{R}^{[2J+1-(-1)^J]/4}$ in a CXO.
We can accordingly define the duration-time average treatment effect $\Delta^{D-avg} = {(J-1)^{-1}} \sum_{d=1}^{J-1}\Delta(d)$ in a SW-CT or PB-CT, and $\Delta^{D-avg} = {([2J+1-(-1)^J]/4)^{-1}} \sum_{d=1}^{[2J+1-(-1)^J]/4} \Delta(d)$ in a CXO.
Such a duration-specific treatment effect structure has also been referred to as ``exposure'' time-varying \citep{kenny_analysis_2022,maleyeff_assessing_2023,lee_analysis_2025} or ``time-on-treatment'' \citep{hughes_current_2015,lee_cluster_2024} effects, and has been the subject of recent interest due to potentially severe estimation bias if misspecified with a working constant treatment effect structure \citep{kenny_analysis_2022,lee_analysis_2025}.

The \textit{period-specific treatment effect structure} allows the treatment effect to vary by the calendar period $j$, separate of the treatment duration $d$, such that $\Delta_j(d)=\Delta_j$.
In PB-CTs, period-specific and duration-specific treatment effects coincide by definition $\bm{\Delta}^P = (\Delta_2,...,\Delta_{J})^\top = \bm{\Delta}^D \in \mathbb{R}^{J-1}$, with period-time average treatment effect $\Delta^{P-avg} = (J-1)^{-1}\sum_{j=2}^{J}\Delta_j = \Delta^{D-avg}$.
In CXO's, period-specific treatment effects $\bm{\Delta}^P = (\Delta_1,...,\Delta_{J})^\top \in \mathbb{R}^{J}$ are a more general form of duration-specific treatment effects, with $\Delta^{P-avg} = (J)^{-1}\sum_{j=1}^{J}\Delta_j$.
In a SW-CT, although we can conceptualize $\Delta_J$ as the treatment effect in period $J$, it is not identifiable because $Y_{iJk}(0)$ is never observable by design \citep{chen_model-assisted_2025,wang_how_2024,lee_analysis_2025}. 
Therefore, we will exclude $\Delta_J$ in the period-specific treatment effect setting and the target estimand for SW-CTs can be any function of $\bm{\Delta}^P = (\Delta_{2},...,\Delta_{J-1})^\top \in \mathbb{R}^{J-2}$, such as $\Delta^{P-avg} = (J-2)^{-1} \sum_{j=2}^{J-1} \Delta_j$. Under the period-specific treatment effect structure, the period-average treatment effect for the overlap population (P-ATO) estimand $\Delta^{P-ATO}  = (\sum_{j=1}^J \lambda_{j} \Delta_j)/(\sum_{j=1}^J \lambda_{j})$ is also worth broadly highlighting.
In the P-ATO, weights $\lambda_{j}=\pi_j^s(1-\pi_j^s)$ are the tilting function that generates the overlap propensity weights \citep{li_balancing_2018}, where $\pi_j^s = \sum_{l=1}^j \pi_l$ is the proportion of clusters receiving treatment in period $j$, and $\pi_l = P(Z_i=l)$.
This tilting function $\lambda_j$ helps define the target population for which treatment effects are estimated, emphasizing periods with better treatment overlap. Therefore, any periods lacking treatment positivity, such as baseline and all-exposed periods, do not contribute any information to the $\Delta^{P-ATO}$ estimand.
Notably, $\Delta^{P-ATO} = \Delta^{P-avg}$ in PB-CT and CXO designs.
However, $\Delta^{P-ATO}$ can differ from $\Delta^{P-avg}$ in staggered designs, such as the SW-CT, where the probability of treatment assignment varies depending on the given period.

Finally, the \textit{saturated treatment effect structure} serves as the most general treatment effect structure when the treatment effects $\Delta_j(d)$ are expected to vary by both the calendar time period $j$ and exposure time $d$.
Such a saturated treatment effect structure coincides with duration and period-specific treatment effects in a PB-CT, $\bm{\Delta}^S = (\Delta_2(1),...,\Delta_J(J-1))^\top = \bm{\Delta}^D = \bm{\Delta}^P \in \mathbb{R}^{J-1}$, with saturated-time average treatment effect $\Delta^{S-avg} = (J-1)^{-1}\sum_{j=2}^{J}\Delta_j(j-1) = \Delta^{D-avg} = \Delta^{P-avg}$.
Similarly, the saturated treatment effect structure coincides with period-specific treatment effects in a CXO design, $\bm{\Delta}^S = (\Delta_1(1), \Delta_2(1), \Delta_3(2), \Delta_4(2),..., \Delta_J([2J+1-(-1)^J]/4))^\top = \bm{\Delta}^P \in \mathbb{R}^{J}$, with saturated-time average treatment effect $\Delta^{S-avg} = (J)^{-1}\sum_{j=1}^{J}\Delta_j\left([2j+1-(-1)^j]/4
\right) = \Delta^{P-avg}$.
In contrast, the saturated treatment effect structure is distinct from both duration and period-specific treatment effect structures in a SW-CT.
As with the period-specific treatment effect structure, the satruated treatment effects in period $J$, $(\Delta_J(1),...,\Delta_J(J))$, are not identifiable in a SW-CT and thus excluded.
In this assumption-lean set up, saturated treatment effects in SW-CTs are written as $\bm{\Delta}^S = (\Delta_{2}(1), \Delta_{3}(1), \Delta_{3}(2), ..., \Delta_{J-1}(1), ..., \Delta_{J-1}(J-2))^\top \in \mathbb{R}^{(J-2)(J-1)/2}$ alongside a saturated-time average treatment effect $\Delta^{S-avg} = ((J-2)(J-1)/2)^{-1} \sum_{j=2}^{J-1} \sum_{d=1}^{j-1} \Delta_j(d)$.

\section{Asymptotic properties of the linear fixed-effects regression estimator}

We will broadly refer to models with cluster fixed intercepts as fixed-effects models, even if such models can have within-cluster nested random effects structures. Consider the following linear fixed-effects model
\begin{equation}
\label{eq:FE_model}
    Y_{ijk} = \beta_{0j} + TE_{ij} + \bm{\beta}_X^\top \bm{X}_{ik} + \alpha_i + RE_{ij} + \epsilon_{ijk} \,.
\end{equation}
Here, $\beta_{0j}$ is the secular trend parameter (also referred to as the period effect), $TE_{ij}$ is the treatment effect structure, $\bm{\beta}_X$ is the vector of coefficients for baseline variables, $\alpha_i$ represents the fixed intercept for cluster $i$, $RE_{ij}$ is the within-cluster nested random effects structure, and $\epsilon_{ijk} \sim N(0, \sigma^2)$ is the residual random error.
A covariate unadjusted analysis is a special case of Equation (\ref{eq:FE_model}) where $\bm{X}_{ik}  = \emptyset$. 
Such a linear fixed-effects model with cluster fixed intercepts and cluster-period random interaction terms (Equation \ref{eq:FE_model}) was previously described by \citet{turner_analysis_2007}, and is a general form of the more typical fixed-effects model that excludes $RE_{ij}$, with estimation proceeding via ordinary least squares. 
Importantly, Equation (\ref{eq:FE_model}) is only a working model and can accommodate continuous, binary, or count outcomes.

For estimating the constant treatment effect $\Delta$ in a SW-CT, PB-CT, and CXO, we specify $TE_{ij,SW} = I\{Z_i \leq j\}\beta_Z$, $TE_{ij,PB} = I\{Z_i=2, j>1\}\beta_Z$, and $TE_{ij,XO} = (I\{Z_i = 1, j\in 2n-1\} + I\{Z_i = 2, j\in 2n\})\beta_Z$ where $n \in \{1,...,[2J+1-(-1)^J]/4\}$, with $\beta_Z$ targeting $\Delta$.
For the duration-specific treatment effects $\bm{\Delta}^D$, we set $TE_{ij,SW} = \sum_{d=1}^{j-1} I\{Z_i = j-d+1\}\beta_{Zd}$, $TE_{ij,PB} = \sum_{d=1}^{j-1} I\left\{Z_i=2d/(j-1) \right\}\beta_{Zd}$, and $TE_{ij,XO} = \sum_{d=1}^{[2j+1-(-1)^j]/4} I\{Z_i = j-2d+2\}\beta_{Zd}$,  with $\beta_{Zd}$ targeting $\Delta(d)$. 
For the period-specific treatment effects $\bm{\Delta}^P$, we broadly set $TE_{ij} = I\{Z_i \leq j\} \beta_{jZ}$, with $\beta_{jZ}$ targeting $\Delta_j$.
For the saturated treatment effects $\bm{\Delta}^S$ in SW-CTs, we set $TE_{ij,SW} = \sum_{d=1}^{j-1} I\{Z_i = j-d+1\} \beta_{jZd}$, with $\beta_{jZd}$ targeting $\Delta_j(d)$.

We consider maximum likelihood estimators for Equation (\ref{eq:FE_model}) with different treatment effect structures $TE_{ij}$, based on the observed data $\{\bm{O}_1,...,\bm{O}_m\}$ from clusters $i=1,...,m$.
The linear fixed-effects model can then be framed as estimating equations in an M-estimation framework, with
Lemma \ref{lemma:variance} establishing the consistency of the cluster-robust sandwich variance estimator under certain types of model misspecification \citep{van_der_vaart_asymptotic_1998}.
\begin{lemma}
\label{lemma:variance}
Let $\bm{O}_1,...,\bm{O}_{m}$ be i.i.d. samples from a common distribution on $O$.
Let $\bm{\psi}(\bm{O},\bm{\theta})$ be a known estimating equation with parameters $\bm{\theta} \in \Theta$, a compact set of Euclidean space.
Let $\hat{\bm{\theta}}$ be the solution to $\sum_{i=1}^{m} \bm{\psi}(\bm{O}_i, \bm{\theta})=0$.
We assume that $\bm{\psi}$ satisfies the following regularity conditions:
(1.) there exists a unique solution in the interior of $\Theta$, denoted as $\underline{\bm{\theta}}$, to the equation $E[\bm{\psi}(\bm{O},\bm{\theta})] = 0$;
(2.) the function $\bm{\theta} \mapsto \bm{\psi}(o, \bm{\theta})$, together with its first and second derivatives, is dominated by a square-integrable function for every $o$ in the support of $\bm{O}$;
and (3.) $E\left[\frac{d\bm{\psi}(\bm{O},\bm{\theta})}{d\bm{\theta}^{\top}} \mid_{\bm{\theta}=\underline{\bm{\theta}}}\right]$ is invertible.
\newline We then have
$\hat{\bm{\theta}}\xrightarrow{P} \underline{\bm{\theta}}$
and
$m^{1/2}(\hat{\bm{\theta}} - \underline{\bm{\theta}}) \xrightarrow{d} N(0,\textbf{V})$,
where $\textbf{V}=E[\text{IF}(\bm{O},\underline{\bm{\theta}})\text{IF}(\bm{O},\underline{\bm{\theta}})^\top]$ and $\text{IF}(\bm{O},\underline{\bm{\theta}}) = -\left( E\left[\frac{d\bm{\psi}(\bm{O},\bm{\theta})}{d\bm{\theta}^{\top}} \mid_{\bm{\theta}=\underline{\bm{\theta}}}\right]^{-1} \bm{\psi}(\bm{O},\underline{\bm{\theta}}) \right)$ is the influence function for $\hat{\bm{\theta}}$.
Furthermore, the sandwich variance estimator
is $\hat{\bm{V}} = m^{-1}\sum_{i=1}^{m} \widehat{\text{IF}}(\bm{O}_i,\hat{\bm{\theta}}) \widehat{\text{IF}}(\bm{O}_i,\hat{\bm{\theta}})^\top \xrightarrow{P} \textbf{V}$,
where $\widehat{\text{IF}}(\bm{O}_i,\hat{\bm{\theta}}) = \left( 
\left[
    m^{-1}\sum_{i=1}^{m} \frac{d\bm{\psi}(\bm{O}_i,\bm{\theta})}{d\bm{\theta}^{\top}} \mid_{\bm{\theta}=\hat{\bm{\theta}}}
\right]^{-1}
\bm{\psi}(\bm{O}_i, \hat{\bm{\theta}})
\right)$.
\end{lemma}

Importantly, as the number of clusters $m \rightarrow \infty$, the number of parameters $\alpha_i$ in the linear fixed-effects model (Equation \ref{eq:FE_model}) will also $\rightarrow \infty$. 
Due to this \emph{incidental parameters problem} \citep{neyman_consistent_1948}, the linear fixed-effects model violates standard regularity condition 3 of Lemma \ref{lemma:variance}.
This is solved by first re-framing cluster fixed intercepts as \emph{nuisance parameters} and the remaining fixed covariates as \emph{structural parameters} \citep{neyman_consistent_1948}. We can then demonstrate that Equation (\ref{eq:FE_model}) yields equivalent structural parameter estimators to the linear within-transformed estimator \citep{kiefer_estimation_1980,wooldridge_econometric_2010}, based on the following model
\begin{equation}
\label{eq:FE_model_within}
    \ddot{Y}_{ijk} = \ddot{\beta}_{0j} + \ddot{TE}_{ij} + \bm{\beta}_X^\top \ddot{\bm{X}}_{ik} + \ddot{\alpha}_i + \ddot{RE}_{ij} + \ddot{\epsilon}_{ijk}
\end{equation}
where the observed outcome and model covariate indicators are subtracted by the corresponding observed within-cluster averages (Appendix \ref{app.sect:constant}).
That is,
$\ddot{Y}_{ijk} = Y_{ijk} - (\sum_{l=1}^{J}\sum_{k'=1}^{N_i} S_{ilk'} Y_{ilk'})/(\sum_{l=1}^{J} N_{il})$, 
$\ddot{\beta}_{0j} = \beta_{0j} - (\sum_{l=1}^J \sum_{k'=1}^{N_i} S_{ilk'} I\{l=j\}\beta_{0l})/(\sum_{l=1}^{J} N_{il})$,
$\ddot{TE}_{ij} = TE_{ij} - (\sum_{l=1}^{J}\sum_{k'=1}^{N_i} S_{ilk'} I\{l=j\} TE_{il})/(\sum_{l=1}^{J} N_{il})$,
$\ddot{\bm{X}}_{ik} = \bm{X}_{ik} - (\sum_{l=1}^{J}\sum_{k'=1}^{N_i} S_{ilk'} \bm{X}_{ik'})/(\sum_{l=1}^{J} N_{il})$,
$\ddot{\alpha}_i = \alpha_i - (\sum_{l=1}^{J}\sum_{k'=1}^{N_i} S_{ilk'} \alpha_i)/(\sum_{l=1}^{J} N_{il})$,
$\ddot{RE}_{ij} = RE_{ij} - (\sum_{l=1}^{J}\sum_{k'=1}^{N_i} S_{ilk'} RE_{il})/(\sum_{l=1}^{J} N_{il})$,
and $\ddot{\epsilon}_{ijk} = \epsilon_{ijk} - (\sum_{l=1}^{J}\sum_{k'=1}^{N_i} S_{ilk'} \epsilon_{ilk'})/(\sum_{l=1}^{J} N_{il})$.
Crucially, $\ddot{\alpha}_i = 0$, thus the linear fixed-effects model avoids the incidental parameters problem and satisfies 
standard regularity condition 3 (Lemma \ref{lemma:variance}) by treating cluster intercepts as nuisance parameters that are removed via within-transformation.

For the constant $\Delta$, duration-specific $\bm{\Delta}^D$, period-specific $\bm{\Delta}^P$, and saturated $\bm{\Delta}^S$ treatment effect structures, we can generally define corresponding point estimators $\hat{\beta}_Z$, $\hat{\bm{\beta}}_Z^D$ (being a vector of $\hat{\beta}_{Zd}$), $\hat{\bm{\beta}}_Z^P$ (being a vector of $\hat{\beta}_{jZ}$), and $\hat{\bm{\beta}}_Z^S$ (being a vector of $\hat{\beta}_{jZd}$), respectively.
We then define the finite-sample sandwich variance estimators
$\widehat{Var}\left(\hat{\beta}_Z\right)=\frac{1}{m}\hat{V}_Z$, 
$\widehat{Var}\left(\hat{\bm{\beta}}_Z^D\right) = \frac{1}{m}\hat{\bm{V}}_Z^D$, 
$\widehat{Var}\left(\hat{\bm{\beta}}_Z^P\right) =\frac{1}{m}\hat{\bm{V}}_Z^P$, 
and $\widehat{Var}\left(\hat{\bm{\beta}}_Z^S\right) =\frac{1}{m}\hat{\bm{V}}_Z^S$, corresponding to the constant, duration-specific, period-specific, and saturated treatment effects, respectively.

\begin{theorem}
\label{Theorem:SW_PB_XO}
    Under the standard regularity conditions outlined in Lemma \ref{lemma:variance}, assume either (i.) Assumptions A1-A3 in CRTs, (ii.) Assumptions A1, A2, A4 in CQTs, or (iii.) the mean model (Equation \ref{eq:FE_model}) is correctly specified.
    The following Central Limit Theorems hold for a linear fixed-effects model analysis of a longitudinal cluster trial:
    (a) under a true constant treatment effect structure, $\hat{V}_Z^{-1/2} m^{1/2} (\hat{\beta}_Z - \Delta) \xrightarrow{d} N(0,1)$;
    (b) under a true duration-specific treatment effect structure, $(\hat{\bm{V}}_Z^D)^{-1/2} m^{1/2} (\hat{\bm{\beta}}^D_Z - \bm{\Delta}^D) \xrightarrow{d} N(0,\textbf{I}_{\text{dim}(\bm{\Delta}^D)})$;
    (c) under a true period-specific treatment effect structure, $(\hat{\bm{V}}_Z^P)^{-1/2} m^{1/2} (\hat{\bm{\beta}}^P_Z - \bm{\Delta}^P) \xrightarrow{d} N(0,\textbf{I}_{\text{dim}(\bm{\Delta}^P)})$
    and $\hat{V}_Z^{-1/2} m^{1/2} (\hat{\beta}_Z - \Delta^{P-ATO}) \xrightarrow{d} N(0,1)$;
    and (d) under a true saturated treatment effect structure, $(\hat{\bm{V}}_Z^S)^{-1/2} m^{1/2} (\hat{\bm{\beta}}^S_Z - \bm{\Delta}^S) \xrightarrow{d} N(0,\textbf{I}_{\text{dim}(\bm{\Delta}^S)})$.
\end{theorem}

In Theorem \ref{Theorem:SW_PB_XO}, we prove that despite being typically understood to target a finite-sample conditional treatment effect estimand, the linear fixed-effects model with a correct treatment effect structure in a longitudinal CT is consistent for a super-population marginal treatment effect estimand. This holds in longitudinal CRTs even if other model components (time-invariant baseline covariates, random effects, residual distribution) are arbitrarily misspecified (Appendix \ref{app.sect:Theorem_SW_Proof}-\ref{app.sect:Theorem_PB_XO_Proof}). 
Furthermore, without relying on randomization, the linear fixed-effects model estimator can still provide consistent and asymptotically normal estimators for the average treatment effect in longitudinal CQTs, even in the presence of time-invariant confounding. We can generally extend these conclusions beyond SW-CTs, PB-CTs, and CXOs, to other longitudinal CT designs.

Interestingly, with the stated assumptions, the linear fixed-effects model constant treatment effect estimator is generally consistent for the P-ATO estimand across longitudinal CT designs (Theorem \ref{Theorem:SW_PB_XO}.c).
This resembles similar observations for linear mixed-effects models in SW-CRTs by \citet{lee_analysis_2025} and \citet{wang_how_2024}; although neither of those earlier efforts explicitly made this connection.
Notably, this P-ATO consistency result in the second half of Theorem \ref{Theorem:SW_PB_XO}.c holds for CXO designs, even though the first half of Theorem \ref{Theorem:SW_PB_XO}.c alongside Theorem \ref{Theorem:SW_PB_XO}.d do not apply in this design due to the fixed-effects model collinearity violating standard regularity condition 3 (Lemma \ref{lemma:variance}).
Unlike in staggered designs, the P-ATO is equivalent to the saturated/duration/period-average treatment effect estimands in PB-CT and CXO designs.
Importantly then, model-robust inference of PB-CT and CXO designs with linear fixed-effects models does not even require a correctly specified treatment effect structure due to the structure of these designs.

\section{Beyond the linear fixed-effects regression: the log-link}
\label{sect:GEE-g}

Despite its wide applicability and demonstrated consistency (Theorem \ref{Theorem:SW_PB_XO}), the linear fixed-effects model may not be a natural model choice when the outcome is not continuous. In this section, we study robust inference for marginal causal estimands via the log fixed-effects model (natural for count outcomes and analogous to the modified Poisson model for binary outcomes) and an additional g-computation step, both fit within an M-estimation framework.
Consider a general mean model for the individual-level data as
\begin{equation}
\label{eq:GEE_mean_model}
    \mu_{ijk} = E[Y_{ijk}| Z_i, \bm{X}_{ik}] = g^{-1}(\beta_{0j} + TE_{ij} + \bm{\beta}_X^\top \bm{X}_{ik} + \alpha_i)
\end{equation}
where model coefficients are specified as in the linear model (Equation \ref{eq:FE_model}).
As is typical, we assume that $g$ is a canonical link function. The vector of structural regression coefficients $\bm{\beta}$ (excluding the cluster fixed-effects coefficients) is estimated by $\hat{\bm{\beta}}$, a solution to the following GEE:
\begin{equation}
\label{eq:GEE}
    \sum_{i=1}^{m} \bm{U}_i^\top \bm{\mathcal{Z}}_i^{-1/2} \bm{R}_i^{-1} \bm{Z}_i^{-1/2} (\bm{Y}_i^o - \bm{\mu}_i^o) = 0 \,.
\end{equation}
where $\bm{Y}_i^o= (Y_{ijk})_{(j,k):S_{ijk}=1} \in \mathbb{R}^{\sum_{j=1}^{J}N_{ij}}$
is the observed outcome vector across periods for cluster $i$,
$\bm{\mu}_i^o = \left(g^{-1}(\bm{Q}_{ijk} \bm{\beta} + \alpha_i)\right)_{(j,k):S_{ijk}=1} = (\mu_{ijk})_{(j,k):S_{ijk}=1} \in \mathbb{R}^{\sum_{j=1}^{J}N_{ij}}$ is the mean function vector for all observed individuals in cluster $i$, $\bm{U}_i = \frac{d\bm{\mu}_i^o}{d\bm{\beta}}$ is the derivative matrix,
$\bm{\mathcal{Z}}_i = \text{diag}\{v(Y_{ijk}) : j=1,...,J, k=1,...,N_i, S_{ijk}=1\}$
is the diagonal matrix of the natural variance functions $v(Y_{ijk})$, and $\bm{R}_i$ encodes the working correlation structure for the observed outcomes in cluster $i$.

For Lemma \ref{lemma:variance} to apply, the incidental parameters problem needs to again be avoided in this generalized setting.
Assuming $g$ to be an identity-link with constant working variance $v(Y_{ijk}) = \sigma^2$ leads to the same result as Theorem \ref{Theorem:SW_PB_XO}, where the within-transformation removes the infinite cluster fixed intercepts $\alpha_i$.
Alternatively, assuming $g$ to be the log-link with working variance $v(Y_{ijk}) = \mu_{ijk}$ and working within-cluster independence in the GEE (Equation \ref{eq:GEE}) allows us to condition out the cluster fixed intercepts (Appendix \ref{app.sect:Theorem_SW_G_Proof}). This results from the equivalence between the resulting estimating equation and the score equation of the following full likelihood conditional Poisson model
\[
    L_i(\bm{\beta},\alpha_i) = \prod_{j=1}^J \prod_{k:S_{ijk}=1} \frac{\mu_{ijk}^{Y_{ijk}}}{Y_{ijk}!}e^{-\mu_{ijk}} \,,
\]
assuming $Y_{ijk} \sim Poisson(\mu_{ijk})$ and $\mu_{ijk} = e^{\bm{Q}_{ijk}\bm{\beta} + \alpha_i}$, where $\bm{Q}_{ijk}$ and $\bm{\beta}$ are vectors for the structural parameter indicators and coefficients.
The Poisson likelihood can then be rewritten as
\[
\begin{split}
    L_i(\bm{\beta},\alpha_i) &= \left( \prod_{j=1}^J \prod_{k: S_{ijk}=1} \frac{e^{\bm{Q}_{ijk}\bm{\beta}}}{Y_{ijk}!} \right) 
    \left(e^{\alpha_i \left(\sum_{l=1}^J \sum_{k': S_{ilk'}=1} Y_{ilk'}\right)} e^{-e^{\alpha_i} \left(\sum_{l=1}^J \sum_{k: S_{ilk'}=1} e^{\bm{Q}_{ilk'}\bm{\beta}}\right)} \right) \\
    &= h\left(Y_{i11} ,...,Y_{iJN} ; \bm{\beta}\right) f\left(\sum_{l=1}^{J} \sum_{k': S_{ilk'}=1} Y_{ilk'}; \alpha_i \right)
\end{split}
\]
where $f\left(\sum_{l=1}^J \sum_{k': S_{ilk'}=1} Y_{ilk'}; \alpha_i \right)$ is a function only depending on $Y_{ijk}$ through the cluster-specific summed outcome over enrolled individuals $\sum_{l=1}^J \sum_{k': S_{ilk'}=1} Y_{ilk'}$, and $h(Y_{i11} ,...,Y_{iJN} ; \bm{\beta})$ is a function that does not depend on $\alpha_i$.
By the factorization theorem, $\sum_{l=1}^J \sum_{k': S_{ilk'}=1} Y_{ilk'}$ is a sufficient statistic for cluster intercepts $\alpha_i$.
Then, conditioning on this sufficient statistic removes $\alpha_i$ from the resulting conditional Poisson likelihood
\[
    L_i^{cond}(\bm{\beta}) = \frac{\prod_{j=1}^J \prod_{k:S_{ijk}=1} (e^{\bm{Q}_{ijk} \bm{\beta}})^{Y_{ijk}}/Y_{ijk}!}{\sum_{\{Y'_{ilk'} : \sum_{l=1}^{J}\sum_{k': S_{ilk'}=1} Y'_{ilk'} = \sum_{l=1}^{J}\sum_{k': S_{ilk'}=1} Y_{ilk'} \}}\prod_{l=1}^J \prod_{k': S_{ilk'}=1} (e^{\bm{Q}_{ilk'} \bm{\beta}})^{Y'_{ilk'}}/Y'_{ilk'}!} \,,
\]
invoking a multinomial probability where the denominator represents the sum of all possible combinations of non-negative integers $Y'_{ilk'}$ over given individual $k'$ and period $l$, such that $\sum_{l=1}^{J} \sum_{k': S_{ilk'}=1} Y'_{ilk'} = \sum_{l=1}^{J} \sum_{k': S_{ilk'}=1} Y_{ilk'}$.
Estimators for the cluster intercepts can then be derived with the structural parameter estimators and its sufficient statistic, $\hat{\alpha}_i = ln\left( (\sum_{j=1}^J \sum_{k=1}^{N_i} S_{ijk}Y_{ijk})/(\sum_{j=1}^J \sum_{k=1}^{N_i} S_{ijk}e^{\bm{Q}_{ijk}\hat{\bm{\beta}}}) \right)$.
Altogether, by relying on the equivalence between the working within-cluster independence log-link fixed-effects GEE and the full likelihood conditional Poisson model, cluster fixed intercepts can be treated as nuisance parameters and conditioned out of the estimating equations to avoid the incidental parameters problem (Appendix \ref{app.sect:Theorem_SW_G_Proof}).

We can then construct the g-computation formula using the described log-link fixed-effects GEE with working within-cluster independence
\begin{equation}
\label{eq:g-comp}
    \hat{\mu}_j(b) = 
    \frac{
        \sum_{i=1}^{m} \sum_{l=1}^{J} \sum_{k:S_{ilk}=1} g^{-1}(\hat{\beta}_{0j} + b + \hat{\bm{\beta}}_X^\top \bm{X}_{ik} + \hat{\alpha}_i)
    }
    {
        \sum_{i=1}^{m} \sum_{l=1}^J N_{il}
    } \in \mathbb{R}^1
\end{equation}
where $b$ can be $\hat{\beta}_Z, \hat{\beta}_{Zd}, \hat{\beta}_{jZ}, \hat{\beta}_{jZd}$ or 0 depending on the assumed treatment effect structure.
This g-computation estimator for treatment $b$ and fixed period $j$ is taken as the empirical mean over all enrolled individuals ($S_{ilk}=1$) across all periods $l \in \{1,...,J\}$.
For example, under a constant treatment effect structure, $\hat{\mu}_j(\hat{\beta}_Z)$ and $\hat{\mu}_j(0)$ target the expected treated $E[Y_{ijk}(z)]$ and untreated $E[Y_{ijk}(0)]$ potential outcome, respectively.
Altogether, g-computation can target model-free marginal estimands by standardizing across the covariate distribution in the target population. In contrast to \citet{wang_how_2024} who focused on marginal models and GEE, the prediction and g-computation steps here explicitly incorporate the estimated fixed-effects in constructing the standardized potential outcomes. 

Focusing primarily on difference estimands,
the constant and duration-specific treatment effect estimands can be targeted with a weighted average of  $\hat{\mu}_j(b)-\hat{\mu}_j(0)$ across periods $j$.
The constant treatment effect estimand $\Delta$ can then be estimated with the following g-computation risk difference estimator based on a log-link fixed-effects GEE with a working constant treatment effect structure
\[
\hat{\Delta}_{GEE-g} = \frac{\sum_{i=1}^{m} \sum_{j=1}^J \lambda_{j}(\hat{\mu}_j(\hat{\beta}_Z)-\mu_j(0))}{\sum_{i=1}^{m} \sum_{j=1}^J \lambda_{j}} = \frac{ \sum_{j=1}^J \lambda_{j}(\hat{\mu}_j(\hat{\beta}_Z)-\mu_j(0))}{\sum_{j=1}^J \lambda_{j}} \,.
\]
This estimator notably resembles the P-ATO estimand, with the tilting function weights $\lambda_{j}$ intentionally included in this constant treatment effect estimator to balance the contribution of data from each period, therefore preventing the default precision weighting in GEE to target an ambiguous estimand under model misspecification. 
The duration-specific treatment effect estimands $\bm{\Delta}^D$ can be estimated with
\[
    \hat{\bm{\Delta}}_{GEE-g}^D = \left(\sum_{i=1}^{m} \sum_{d=1}^{dim(\bm{\Delta}^D)} \bm{\lambda}_i(d) \textbf{A}_d \right)^{-1} \left(\sum_{i=1}^{m} \sum_{d=1}^{dim(\bm{\Delta}^D)} \bm{\lambda}_i(d)  \left(\hat{\mu}_j(\hat{\beta}_{Zd}) - \hat{\mu}_j(0) \right)_{j=1,...,J} \right)
\]
As in the constant treatment effect structure, $\bm{\lambda}_i(d)$ is a weighting matrix included to balance the contribution of data from each period under different treatment durations, therefore combating bias due to model misspecification.
$\textbf{A}_d \in \mathbb{R}^{J \times dim(\bm{\Delta}^D)}$ is a matrix with the $d$th column equal to a $J$-length vector of 1's and all other elements 0.
An example of this estimator is illustrated in Appendix (\ref{app.sect:Theorem_SW_G_Proof_ds}).
Provided that any potential collinearity issues in the fixed-effects models are avoided for the given cluster trial design, the robust period-specific $\hat{\bm{\Delta}}^P_{GEE-g}$ and saturated $\hat{\bm{\Delta}}^S_{GEE-g}$ treatment effect estimators targeting their corresponding estimands, $\bm{\Delta}^P$ and $\bm{\Delta}^S$, are considerably simpler to define as the appropriate differences between $\hat{\mu}_j(b)-\hat{\mu}_j(0)$, with no weighting across periods required.

Altogether, these proposed estimators for marginal estimands are solutions to the described log-link fixed-effects models, implemented as GEEs alongside g-computation under an M-estimation framework (Appendix \ref{app.sect:Theorem_SW_G_Proof}).
Accordingly, we can derive consistent sandwich variance estimators for these g-computation difference estimators with Lemma \ref{lemma:variance}, where $\frac{1}{m}\hat{V}_Z$, 
$\frac{1}{m}\hat{\bm{V}}_Z^D$, 
$\frac{1}{m}\hat{\bm{V}}_Z^P$, 
and $\frac{1}{m}\hat{\bm{V}}_Z^S$ correspond to the constant, duration-specific, period-specific, and saturated treatment effect structures, respectively.

\begin{theorem}
\label{Theorem:G}
    Under standard regularity conditions in Lemma \ref{lemma:variance} and given the log-link with working variance $v(Y_{ijk}) = \mu_{ijk}$ and working independence $\rho = 0$, assume Assumptions A1-A3, and in the absence of Assumption A3, assume the mean model (Equation \ref{eq:GEE_mean_model}) is correctly specified.
    Then the following Central Limit Theorems hold for a longitudinal cluster trial:
    (a) under a true constant treatment effect structure, $\hat{V}_{GEE-g}^{-1/2} m^{1/2} (\hat{\Delta}_{GEE-g} - \Delta) \xrightarrow{d} N(0,1)$;
    (b) under a true duration-specific treatment effect structure, $(\hat{\bm{V}}_{GEE-g}^{D})^{-1/2} m^{1/2} (\hat{\bm{\Delta}}_{GEE-g}^D - \bm{\Delta}^D) \xrightarrow{d} N(0,\textbf{I}_{dim(\bm{\Delta}^D)})$;
    (c) under a true period-specific treatment effect structure, $(\hat{\bm{V}}_{GEE-g}^{P})^{-1/2} m^{1/2} (\hat{\bm{\Delta}}_{GEE-g}^P - \bm{\Delta}^P) \xrightarrow{d} N(0,\textbf{I}_{dim(\bm{\Delta}^P)})$ and $\hat{V}_{GEE-g}^{-1/2} m^{1/2} (\hat{\Delta}_{GEE-g} - \Delta^{P-ATO}) \xrightarrow{d} N(0,1)$;
    and (d) under a true saturated treatment effect structure, $(\hat{\bm{V}}_{GEE-g}^{S})^{-1/2} m^{1/2} (\hat{\bm{\Delta}}_{GEE-g}^S - \bm{\Delta}^S) \xrightarrow{d} N(0,\textbf{I}_{dim(\bm{\Delta}^S)})$.
\end{theorem}

As in the linear setting, we prove that despite being typically understood to target a conditional treatment effect estimand, the log-link fixed-effects model in a longitudinal CRT can asymptotically target a marginal treatment effect estimand via m-estimation and g-computation, even if other model components are arbitrarily misspecified.
In longitudinal CQTs without randomization (Assumption A3), consistency of the described log-link fixed effect model g-computation estimator can still be demonstrated by
assuming the mean model is correctly specified.
While assuming the mean model to be correctly specified is seemingly contradictory to our ``model-robustness'' objectives, such a log-link fixed-effects model can still robustly account for all time-invariant multiplicative confounding (Appendix \ref{app.sect:G_CQTs}), in contrast to corresponding marginal log-link GEEs that were previously studied for SW-CRTs \citep{wang_how_2024}.

We primarily target difference estimands by g-computation. However, alternative effect measures, such as ratio estimands, are also common. As we show in Appendix \ref{app.sect:nonlinear_estimands}, the described log-link fixed-effects model g-computation estimators can be rearranged to target these alternative estimands.

\section{Simulation Studies}
\label{sec:simulation}

We describe simulation studies that illustrate our theoretical findings.
We target marginal difference-in-means average treatment effect estimands in 
(1.) a SW-CRT with binary outcomes and period-specific treatment effects and (2.) a PB-CQT with continuous outcomes and saturated/duration/period-specific treatment effects. In the Appendix (\ref{app.sect:additional_sim_scenarios}), we include additional simulation scenarios, exploring (3.) a CRXO with continuous outcomes, saturated/period-specific treatment effects, and a complex correlation structure, covariate structure, plus interactions and (4.) a SW-CRT with count outcomes and period-specific treatment effects. In the described simulation scenarios, we simulate 1000 replicates with $m=6$ clusters and a mean enrolled sample size of $E[N_{ij}]=100$, followed across $J=4$ time-periods in each CT design.
Altogether, our simulations demonstrate the model-robustness of the fixed-effects model in different longitudinal CRT and CQT designs in challenging but common scenarios with a small number of clusters.
Simulation results when $m=100$ clusters are also included in the Appendix (\ref{app.sect:additional_sim_scenarios}).

\subsection{Simulation Scenario 1}

In scenario 1, we simulate a SW-CRT with binary outcomes and a true constant treatment effect structure on the logit-scale to demonstrate the robustness of the linear fixed-effects model and g-computation log-link fixed-effects model estimators under model misspecification.
Notably, the constant treatment effect on a logit scale, results in a period-specific risk difference treatment effect estimand.
We generate potential outcomes for $b=0,1$ (to denote receiving control or treatment) as
\[
    logit(E[Y_{ijk}(b) | \delta_i, X_{1,ijk}, \alpha_i, \gamma_{ij}]) = \mu_0 + \beta_{0j} + \beta_{Z}(1+\delta_i)b 
    + \beta_{X1} X_{1,ijk} 
    + \alpha_i + \gamma_{ij}
\]
with $Y_{ijk}(b) \sim Bernoulli\left(E[Y_{ijk}(b) | \delta_i, X_{1,ijk}, \alpha_i, \gamma_{ij}]\right)$.
We specify $\mu_0 = 1.5$, 
$(\beta_{0j})_{j=1,...,J} = (0.2j)^\top_{j=1,...,J}$.
and $\beta_Z = 0.7$.
Furthermore, the treatment effect randomly varies between clusters $\delta_i \sim N(0, \beta_Z^2/100)$.
An individual-level binary covariate $X_{1,ijk} \sim Bernoulli(X_{1,i})$ is generated where $X_{1,i} \sim Beta(6,4)$ on the cluster-level, and $\beta_{X1} = 1.5$.
Finally, random intercepts are simulated with $\alpha_i \sim N(0, \tau^2)$ where $\tau^2 = 0.176$ and $\gamma_{ij} \sim N(0,\tau^2/4)$.

With such a data-generating process and Assumption A1 satisfied, we can define the marginal period-specific expected potential outcome $E[Y_{ijk}(b)]=v_j(b) \, \forall \, k$ over the joint distribution of the random effects and covariate distributions
\[
\begin{split}
    v_j(b)
    & = E\left[\frac{e^{\mu_0 + \beta_{0j}} e^{\beta_{Z}(1+\delta_i)b} e^{\beta_{X1} X_{1,ijk}} e^{\alpha_i} e^{\gamma_{ij}}}{1+e^{\mu_0 + \beta_{0j}} e^{\beta_{Z}(1+\delta_i)b} e^{\beta_{X1} X_{1,ijk}} e^{\alpha_i} e^{\gamma_{ij}}}\right] \,.
\end{split}
\]
The period-specific treatment effect estimands for periods $j=2,3$ are then computed via numerical integration with $(v_j(1)-v_j(0))_{j=2,3} \approx (0.0729, 0.0630)^\top$ and $\Delta^{P-avg} = \sum_{j=2}^{3} \left(v_j(1)-v_j(0)\right)/2 \approx 0.0679$, which can be targeted by specifying analyses with a period-specific treatment effect structure.
We can additionally define the P-ATO estimand
$\Delta^{P-ATO}  = ( \sum_{j=1}^J \lambda_{j}[v_j(1)-v_j(0)])/(\sum_{j=1}^J \lambda_{j}) \approx 0.0679$,
which can then be targeted by specifying analyses with a constant treatment effect structure.
In this specific $J=4$ period SW-CRT example, $\Delta^{P-avg} = \Delta^{P-ATO}$.
The Appendix (\ref{app.sect:sim_scenario4}) includes an additional simulation scenario 4 with a negative-binomial count data-generating process that yields a more dramatic difference between values of $\Delta^{P-avg}$ and $\Delta^{P-ATO}$.

\subsection{Simulation Scenario 2}
In scenario 2, we simulate a PB-CQT with continuous outcomes and a true time-varying treatment effect structure (where saturated, duration-specific, and period-specific treatment effects coincide by design) with $m=6$ clusters to demonstrate the robustness of linear fixed-effects models under misspecification, including that of the treatment effect structure, and time-invariant confounding.
We generate potential outcomes
\[
    Y_{ijk}(z) = \mu_0 + \beta_{0j} + \beta_{jZ}(1+\delta_i)I\{z=2\} + \beta_{X1} X_{1,ijk} + \beta_{X2}\sqrt{X_{2,ijk}} + \alpha_i + \gamma_{ij} + \epsilon_{ijk} \,. 
\]
The time-varying treatment effects in the PB-CQT are written in terms of
period-time and are simulated as $\beta_{jZ}= \beta_Z\left(1+0.6\left(j-\frac{\sum_{l=2}^{J} l}{J-1}\right)\right) \, \forall \, j > 1$.
Then, $\mu_0$, $\beta_{0j}$, $\beta_Z$, $\delta_i$, $\beta_{X1}$, $X_{1,ijk}$ are generated as in simulation scenario 1.
Additionally, an individual-level count covariate $X_{2,ijk} \sim Poisson(X_{2,i})$ is generated where on the cluster-level $X_{2,i} \sim Gamma(0.5, 200)$, and $\beta_{X2} = 0.02$.
Random intercepts are simulated with $\alpha_i \sim N(0, \tau^2)$ where $\tau^2=0.05/(1-0.05)$.
Within-cluster correlation that decays over time are simulated with $(\gamma_{ij})_{j=1,...,6} \sim MVN(0,\Sigma)$ where the correlation between periods $j$ and $l$ is $\Sigma_{jl} = \kappa^2 e^{-\lambda|j-l|}$ given scale parameter $\kappa=\tau/10$ and decay rate parameter $\lambda=0.5$.
Finally, $\epsilon_{ijk} \sim N(0,1)$.
Importantly, this PB-CQT scenario has cluster-level confounding, where clusters in the treatment sequence ($z=2$) have higher values of $\alpha_i$, $X_1$, and $X_2$ than the control sequence ($z=0$).

\subsection{Simulation results}

All simulation scenarios are fit with a linear fixed-effects model with a working independence within-cluster correlation.
Additionally, a g-computation log-link fixed-effects model estimator is fit in scenarios with binary or count outcomes.
Corresponding linear-mixed effects models are fit for all simulation scenarios, with results reported in the Appendix (\ref{app.sect:sim_table}).
Constant and/or period-specific treatment effect structures are specified accordingly.
With these models, we report the percent relative bias, empirical variance, average sandwich and jackknife variance estimates, and determine the coverage probability of the 95\% confidence intervals using (i.) the sandwich variance with normal approximation and (ii.) the jackknife variance with the $m-2$ degrees of freedom t-distribution.

Despite Lemma \ref{lemma:variance} demonstrating the consistency of the sandwich variance estimator, this variance estimator can be biased with small sample sizes and yield under-coverage of the 95\% confidence intervals \citep{bell_bias_2002}. This is true when using either a normal approximation or t-distribution, even with Satterthwaite degrees of freedom adjustment, and has lead to development for the jackknife and bias-reduced linearization variance estimators for better coverage with a t-distribution \citep{bell_bias_2002}.
Furthermore, the jackknife variance estimator benefits from easily generalizable implementation across different statistical software, 
especially with a $m-2$ degrees of freedom adjustment \citep{ford_maintaining_2020}.

\begin{figure}[htb]
    \centering
    \includegraphics[width=\linewidth]{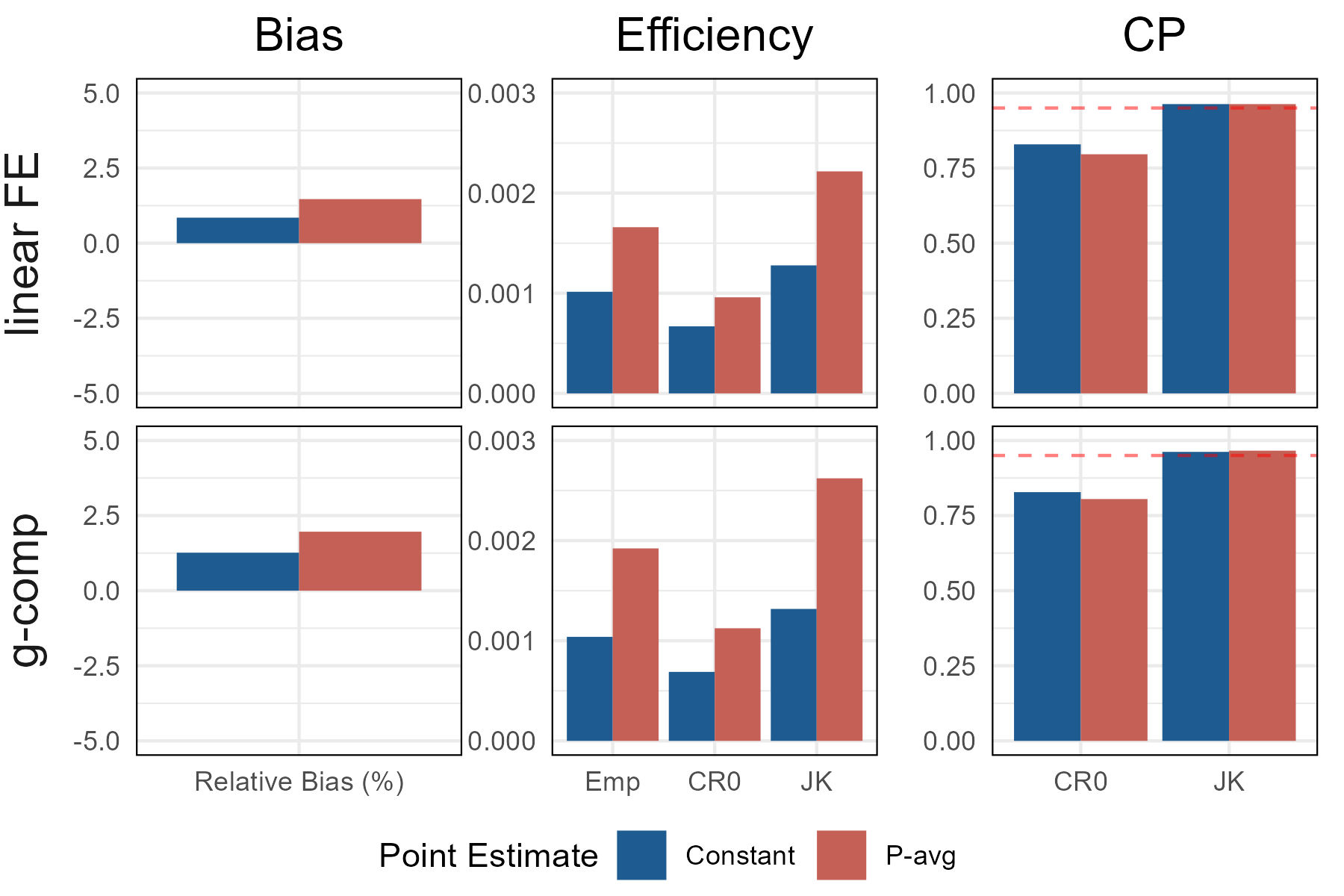}
    \caption{
        Analysis results from simulation scenario 1 of a SW-CRT with binary outcomes and a true period-specific treatment effect structure, using the linear fixed-effects model (``linear FE'') and g-computation log-link fixed-effects model (``g-comp'') with constant and period-average (``P-avg'') treatment effect estimators. 
        Results are reported for (i.) relative bias (\%), (ii.) efficiency in terms of the empirical variance (``Emp'') in comparison to average sandwich (``CR0'') and jackknife (``JK'') variance estimates, and (iii.) coverage probabilities using the sandwich variance with normal approximation (``CR0'') and the jackknife variance with $t(m-2)$
        (``JK'').
        The dashed red line denotes a nominal coverage probability of 95\%.
    }
    \label{fig:scenario_SWCRT_bin}
\end{figure}

The results for the simulation scenario 1 SW-CRT ($m=6$) are illustrated in Figure \ref{fig:scenario_SWCRT_bin}. 
The linear fixed-effects model constant $\hat{\beta}_Z$ and period-average $\hat{\beta}_Z^{P-avg}$ treatment effect estimators are unbiased for their corresponding estimands, $\Delta^{P-ATO}$ and $\Delta^{P-avg}$.
Generally, the constant treatment effect estimators were more efficient than the period-average treatment effect estimators, with smaller empirical variances. This discrepancy is even larger in the g-computation log-link fixed-effects model.
Overall, the sandwich variance estimators slightly underestimate, whereas the jackknife variance estimators slightly overestimate the empirical variances. 
The jackknife variance estimator with $t(m-2)$ had nominal coverage of the 95\% confidence intervals.

The results for simulation scenario 2 PB-CQT ($m=6$) are illustrated in Figure \ref{fig:scenario_PBCQT}.  
In contrast to the substantially biased linear mixed-effects model, the linear fixed-effects model constant $\hat{\beta}_Z$ and period-average $\hat{\beta}_Z^{P-avg}$ treatment effect estimators are unbiased for the time-average treatment effect estimand (referred to here as $\Delta^{P-avg}$; recall $\Delta^{P-avg} = \Delta^{D-avg} = \Delta^{S-avg} = \Delta^{P-ATO}$ in a PB-CT design).
The jackknife variance estimator again yields close to nominal coverage of the 95\% confidence intervals, despite slight overestimation of the empirical variance.

\begin{figure}[htb]
    \centering
    \includegraphics[width=\linewidth]{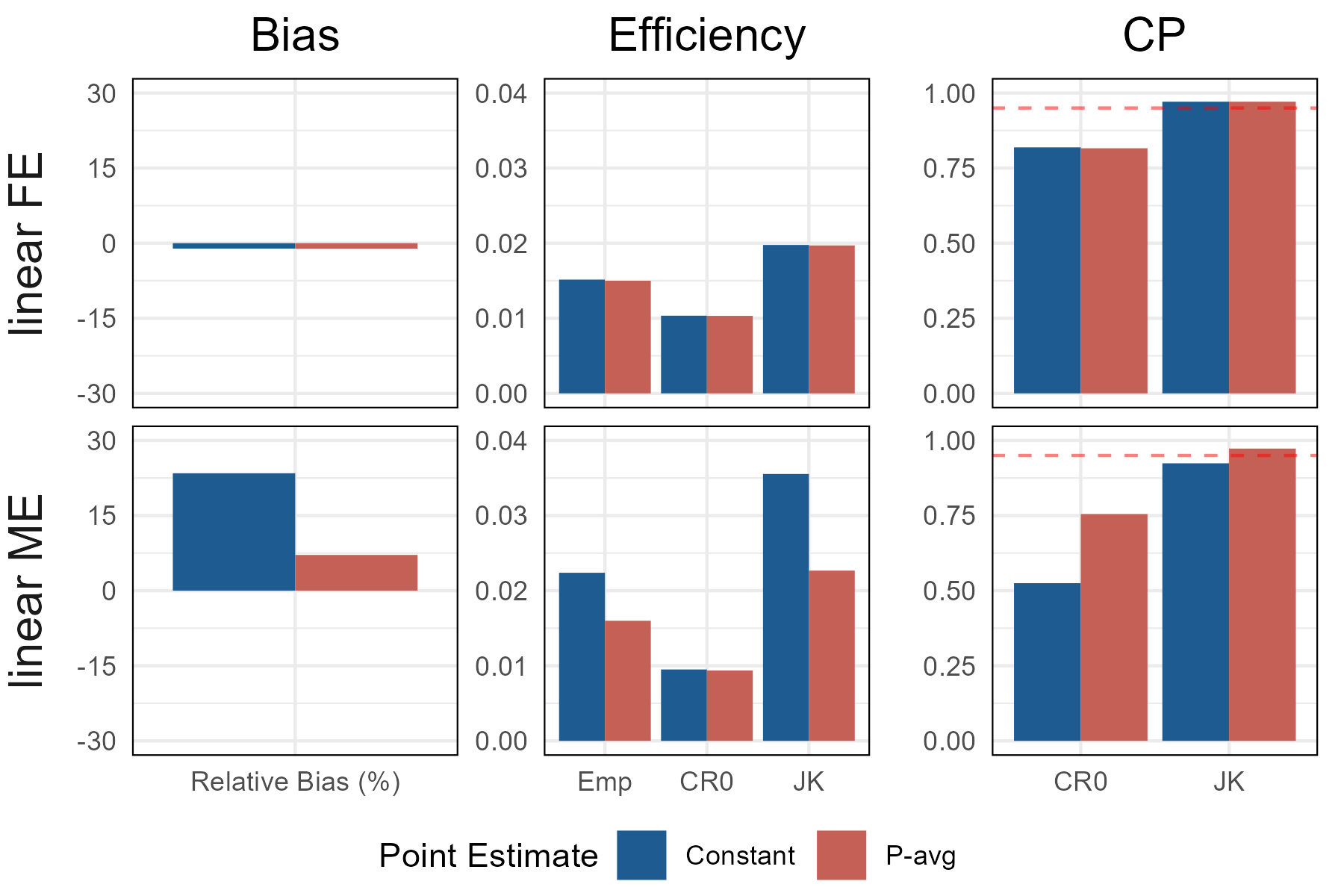}
    \caption{
        Analysis results from simulation scenario 2 of a PB-CQT with continuous outcomes and a true time-varying treatment effect structure, using the linear fixed-effects model (``linear FE'') and nested-exchangeable mixed-effects model (``linear ME'') with constant and period-average (``P-avg'') treatment effect estimators. 
        Results are reported for (i.) relative bias (\%), (ii.) efficiency in terms of the empirical variance (``Emp'') in comparison to average sandwich (``CR0'') and jackknife (``JK'') variance estimates, and (iii.) coverage probabilities using the sandwich variance with normal approximation (``CR0'') and the jackknife variance with $t(m-2)$
        (``JK'').
        The dashed red line denotes a nominal coverage probability of 95\%.
    }
    \label{fig:scenario_PBCQT}
\end{figure}

The additional simulation results are included in the Appendix (\ref{app.sect:additional_sim_scenarios}).
In the simulation scenario 3 CRXO ($m=6$), the linear fixed-effects model constant treatment effect estimator $\hat{\beta}_Z$ is minimally biased for the saturated/period-average treatment effect $\Delta^{S-avg}=\Delta^{P-avg}$, despite model misspecification for the extremely complex covariate structure in the data generating process.
In the simulation scenario 4 SW-CRT ($m=100$), the linear fixed-effects model constant and period-average treatment effects remain unbiased for their corresponding estimands, despite the dramatic differences in the magnitude of $\Delta^{P-ATO}$ and $\Delta^{P-avg}$.
Finally, we additionally simulated scenarios 1-3 with $m=100$ clusters and demonstrate the consistency of and proper coverage with both the sandwich and jackknife variance estimators with a larger number clusters (Appendix \ref{app.sect:sim_scenarios100}).

Across these simulation scenarios with misspecified covariate and correlation structures, the linear fixed-effects models are generally similarly efficient to the linear mixed-effects models with either an exchangeable or nested-exchangeable correlation structure (Appendix \ref{app.sect:sim_table}).
In the scenario 1 SW-CRT ($m=6$), the linear fixed-effects model was slightly less efficient than the linear mixed-effects models, with the linear mixed-effects model with nested-exchangeable correlation being the most efficient.
In the scenario 2 PB-CQT ($m=6$), the linear fixed-effects model was slightly more efficient than the linear mixed-effects models, with the linear mixed-effects model with nested-exchangeable correlation being the least efficient.
In the scenario 3 CRXO ($m=6$), the linear fixed-effects model yields equivalent estimates to the linear mixed-effects model with exchangeable correlation, with both being slightly less efficient than the linear mixed-effects model with nested-exchangeable correlation.
These efficiency results generally extend for the described simulation scenarios with $m=100$ clusters (Appendix \ref{app.sect:sim_table}).

\section{Case study re-analysis}
\label{sect:case_study}

We re-evaluated a SW-CRT with a small number of clusters, evaluating a crowd-sourced HIV testing intervention among men who have sex with men (MSM) in China \citep{tang_crowdsourcing_2018}. 
The trial had a closed-cohort design across 8 cities (clusters) in two Chinese provinces. Cities were paired by province and randomized to receive the intervention sequentially at 3-month intervals, yielding 4 steps over a 12-month period. A total of 1,381 participants were enrolled (averaging approximately 173 participants per cluster), with 1,219 completing at least one follow-up survey. 
The binary primary outcome was self-reported HIV testing in the past 3 months. 

We re-analyzed the binary primary outcome using the described linear fixed-effects model, g-computation log-link fixed-effects model estimator, and a linear mixed-effects model with cluster and individual random intercepts (Figure \ref{fig:case_study}).
The latter mixed-effects model aligns with the closed-cohort nature of the design and is also model-robust in such a SW-CRT \citep{wang_how_2024}.
Models were fit with constant, duration-specific, period-specific, and saturated treatment effect structures to target the corresponding time-average difference in percentage estimands $(\Delta, \Delta^{D-avg}, \Delta^{P-avg}, \Delta^{S-avg})$. 
Confidence intervals were formed with the jackknife variance estimator and a $t(m-2)$ distribution.

\begin{figure}[htb]
    \centering
    \includegraphics[width=\linewidth]{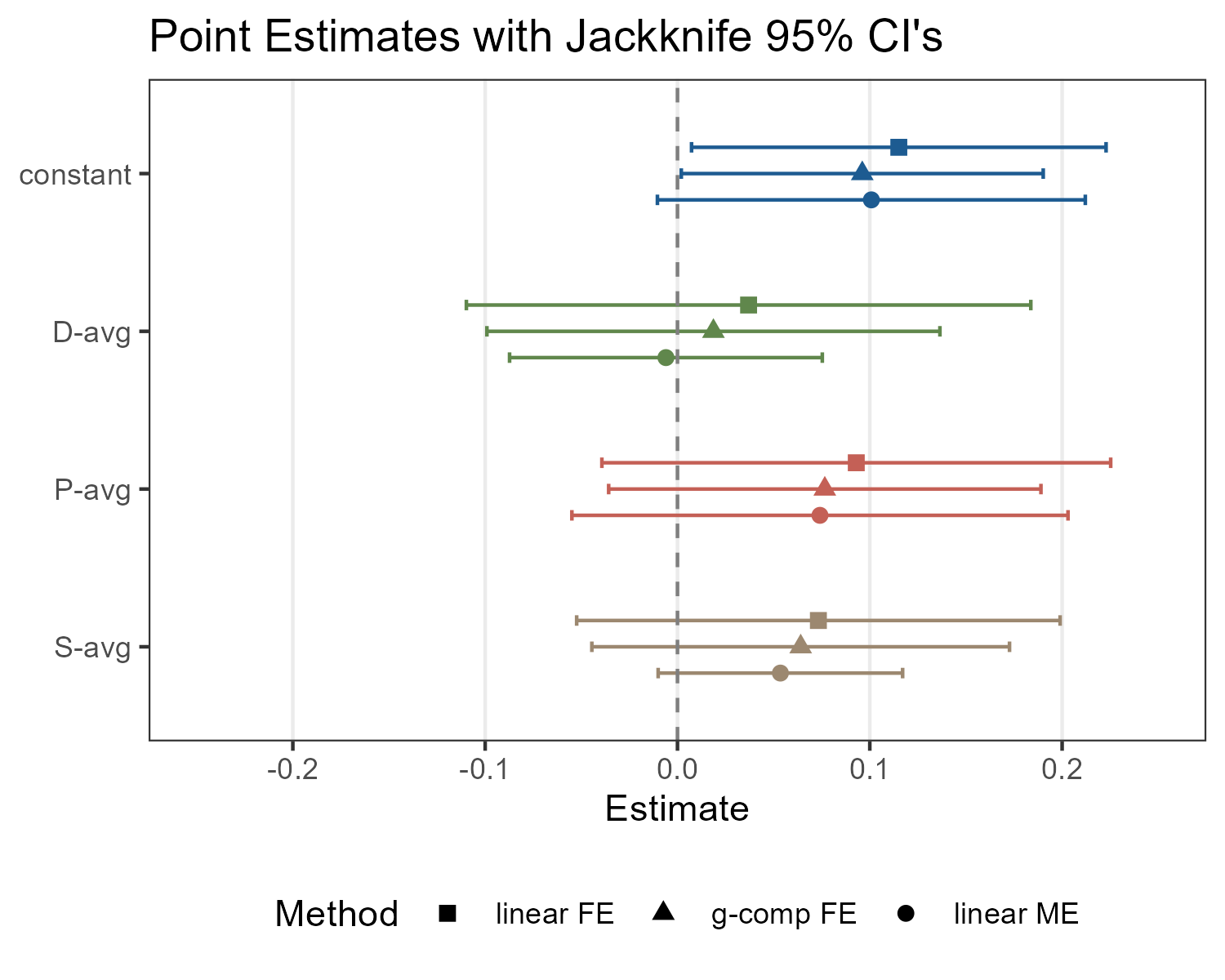}
    \caption{
        Results from the re-analysis of binary outcomes using the linear fixed-effects model (``linear FE''), g-computation log-link fixed-effects model (``g-comp''), and linear mixed-effects model (``linear ME'') with constant, duration-average (``D-avg''), period-average (``P-avg'') and saturated-average (``S-avg'') treatment effect estimators. 
        Confidence intervals are formed using the jackknife variance with 
        $t(m-2)$.
    }
    \label{fig:case_study}
\end{figure}

We generally observe a positive effect of the intervention on self-reported HIV testing, despite the robust confidence intervals generally including 0 (Figure \ref{fig:case_study}).
Overall, the observed magnitude of percentage increase in Figure \ref{fig:case_study} largely corresponds to the results reported in the original trial \citep{tang_crowdsourcing_2018}.
Within the different treatment effect structures, the different employed models all yielded comparable results.
As expected, analyses targeting $\Delta=\Delta^{P-ATO}$ and $\Delta^{P-avg}$ mostly aligned.
In contrast, analyses targeting $\Delta^{D-avg}$ generally produced smaller effects, which predictably led to slightly reduced $\Delta^{S-avg}$ estimates in contrast to $\Delta$ and $\Delta^{P-avg}$.

These re-analysis results may imply the presence of duration-specific treatment effects.
However, as in \citet{lee_analysis_2025}, we recommend that determination of treatment effect structure should be determined \textit{a priori}, with trial domain knowledge.
In practice, the intervention had components that could plausibly strengthen over time \citep{tang_crowdsourcing_2018}. The local story contests were ongoing and designed to build community engagement progressively. Social diffusion effects within the participant networks in each city might amplify as more individuals were exposed to the campaign and shared testing experiences. Furthermore, the self-testing platform could see increased uptake as awareness spread through word-of-mouth.

Importantly, this trial had monotonic participant drop-out, with those patients outcomes assumed to be missing at random.
We include additional analyses following multiple imputation of the missing binary outcomes in the Appendix (\ref{app.sect:case_study}).
The re-analysis results with multiple imputation did not dramatically differ from the results reported in Figure \ref{fig:case_study}.

\section{Discussion}

This article establishes the model-robustness of fixed-effects estimators for nonparametrically defined treatment effect estimands in longitudinal cluster trials (CTs), encompassing both cluster randomized trials (CRTs) and cluster quasi-experimental trials (CQTs) with different trial designs.
Simply, consistency of the fixed-effects model estimators are robustly maintained in longitudinal CRTs via randomization, with robust inference resulting from the use of the sandwich variance estimator (and the jackknife variance estimator in scenarios with a small number of clusters).
Furthermore, consistency of the fixed-effects model estimators can be maintained in longitudinal CQTs via the innate properties of the fixed-effects models.

Several key contributions emerge from this work. We clarify a longstanding misconception in the biostatistics literature regarding the inferential scope of fixed-effects models. 
Contrary to the prevailing view that fixed-effects models inherently target only conditional estimands in the finite population of sampled clusters, we comprehensively demonstrate that these estimators can be implemented within an M-estimation framework and target superpopulation marginal estimands. 

Our theoretical results demonstrate that fixed-effects models, when correctly specifying the treatment effect structure, provide consistent and asymptotically normal estimators for marginal cluster-average treatment effects across stepped-wedge CRTs (SW-CRTs) regardless of misspecification in other model components such as covariate functional forms, random effects structures, or residual distributions.
We further show that the fixed-effects model with a constant treatment effect targets the period-averaged treatment effect for the overlap population (P-ATO) estimand, thereby providing additional robustness when the true treatment effect varies across calendar periods.
This P-ATO estimand corresponds with the saturated-average treatment effect estimand in parallel designs, such as the parallel-with-baseline cluster randomized trial (PB-CRT) and cluster randomized crossover (CRXO) designs. As a result, fixed-effects model analyses do not even require correct specification of the treatment effect structure to yield robust estimators in these CRT designs.

We further establish that fixed-effects models possess a distinct advantage over mixed-effects models in quasi-experimental settings by automatically adjusting for all time-invariant covariates and confounding at both the cluster and individual levels, without requiring explicit covariate specification.
We highlight that such individual-level confounding in longitudinal CQTs operates entirely at the cluster-period level due to the clustered treatment assignment mechanism. 
Consequently, a mean independence assumption (Assumption A4) invoking parallel trends and exchangeable treatment effects is sufficient for consistent estimation in CQTs where randomization (Assumption A3) does not apply.

Several limitations warrant discussion.
First, our consistency results critically rely on the assumption of no time-varying confounding (embedded within Assumptions A2-A4). When unmeasured time-varying confounders exist, fixed-effects methods can yield biased estimates. Sensitivity analyses for violations of this assumption merit future methodological development. 
Second, with the recent interest in informative sizes in clustered settings \citep{lee_what_2025,fang_model-robust_2025}, additional work is necessary to better understand the consistency of fixed-effects model-based estimators for different weighted estimands in longitudinal CT scenarios where the non-informative enrollment assumption (Assumption A2) may not apply.

Fixed-effects models have long been neglected and misinterpreted in the biostatistics literature in favor of mixed-effects models. This is despite fixed-effects models long being understood in the econometrics literature to be more general analogs to corresponding mixed-effects models \citep{hausman_specification_1978,mundlak_pooling_1978}.
In this article, we position fixed-effects models as a viable and potentially preferable alternative to mixed-effects models for the robust analysis of longitudinal cluster trials in both randomized and quasi-experimental contexts.

\if1\anon
{
\section*{Acknowledgement}
Research in this article was supported by the Patient-Centered Outcomes Research Institute\textsuperscript{\textregistered} (PCORI\textsuperscript{\textregistered} Award ME-2022C2-27676). The statements presented in this article are solely the responsibility of the authors and do not necessarily represent the views of PCORI\textsuperscript{\textregistered}, its Board of Governors or Methodology Committee.
} \fi

\section*{Disclosure Statement}
The authors declare no competing interests.

\section*{Data availability statement}
Data sharing is not applicable to this article as no new data were created or analyzed in this study. R code for the simulations are included in the online supplement.

\putbib
\end{bibunit}

\newpage
\appendix
\fontsize{10}{12}\selectfont 

\begin{bibunit}

{\Large \bf{Supplementary Materials for ``Fixed-Effects Models for Causal Inference in Longitudinal Cluster Randomized and Quasi-Experimental Trials''}}

\section{Assumptions}
\label{app.sect:assumptions}

\noindent
We generally follow the convention in the mathematical notation that vectors and matrices are bold and scalars aren't.

Consider a longitudinal cluster trial (CT) with $m$ clusters, and each cluster $i \in \{1,...,m\}$ contains $N_i$ individuals in its source population.
We assume $N_i$ can vary by clusters, and takes values in a bounded subset of positive integers, but remains constant across calendar time in the study.
We assume data from each cluster are collected in $J$ discrete, equally-spaced periods indexed by $j=1,...,J$.
Furthermore we define $Z_i=j$ if cluster $i$ starts receiving treatment in the beginning of period $j \in \{1,...,J\}$, with $Z_i=0$ indicating that cluster $i$ never receives the treatment.
For each individual $k \in \{1,...,N_i\}$ in cluster $i$, we define $Y_{ijk}$ as their outcome in period $j$ and $\bm{X}_{ik}$ as their vector of baseline covariates.

In practice, not all individuals in the cluster source population will be included in the study. We define $S_{ijk}$ as the enrollment indicator such that $S_{ijk}=1$ if individual $k$ from cluster $i$ is included in period $j$; and $N_{ij}=\sum_{k=1}^{N_i} S_{ijk}$ is the observed cluster-period size. Under this framework, we allow each individual to appear in one, multiple, or even no periods, thereby accommodating cross-sectional, closed-cohort, and open-cohort designs \citep{li_mixed-effects_2021}.

We pursue the potential outcomes framework to define treatment effect estimands.
With potential outcomes $Y_{ijk}(\tilde{z})$ where $\tilde{z} \in \{0, z\}$, let $Y_{ijk}(z)$ denote the potential outcome of individual $k$ in cluster $i$ during period $j$ had the cluster been assigned to a sequence where treatment first begins in period $z$ for $1 \leq z \leq j$.
If untreated, $Y_{ijk}(0)$ denotes the untreated potential outcome and we assume no anticipation.
In some CT designs, including a stepped-wedge CT (SW-CT) and parallel-with-baseline CT (PB-CT), $j=1$ denotes a baseline period where $Y_{i1k} = Y_{i1k}(0) \, \forall \, i, k$. 
Additionally, in a SW-CT, $j=J$ denotes an all-treated period in a SW-CT where $Y_{iJk} \neq Y_{iJk}(0) \,\forall\, i,k$.

The observed data for each cluster are denoted as $\bm{O}_i=\{Y_{ijk}, \bm{X}_{ik}, Z_i : S_{ijk}=1, j=1,...,J, k=1,...,N_i\}$; the source population size $N_i$, however, needs not be known.
Across different CT designs, the complete, but not fully observed, data vector for each cluster $i$ is denoted as $\bm{W}_i = \{(Y_{ijk}(0), Y_{ijk}(z), S_{ijk}, \bm{X}_{ik}, Z_i, N_i) : k=1,...,N_i, j=1,...,J, 1 \leq z \leq J\}$. 
To proceed, we outline the following assumptions on $\bm{W}_i$ to enable the subsequent estimand definitions and theorem proofs:

\noindent \textbf{A1.} (\textit{Super-population sampling}) Data vector $\{\bm{W}_i, i=1,...,m\}$  are independent and identically distributed draws from a population distribution $\mathcal{P}$ with finite second moments.
Furthermore, within each cluster $i$, the data vectors $\{(Y_{i1k}(0), ..., Y_{iJk}(J)), \bm{X}_{ik}\}$ for $k=1,...,N_i$ are identically distributed given the source population size $N_i$.

\noindent \textbf{A2.} (\textit{Non-informative enrollment}) 
The enrollment size $\{N_{ij}: i=1,...,m, j=1,...,J\}$ are independent of $N_i$. Furthermore, the enrollment indicator $\{ S_{ijk} : j=1,...,J, k=1,...,N_i \}$ is independent of all other random variables in $\bm{W}_i$ given $N_i$.

\noindent \textbf{A3.} (\textit{Randomization}) Initial treatment period $Z_i$ is independent of all other random variables in $\bm{W}_i$.

\sloppy \noindent \textbf{A4.} \textit{(Mean independence)} i. (\textit{Parallel Trends}) The average within-cluster deviations of the untreated potential outcomes from the enrollment average are mean independent of $Z_i$ and $N_i$, given $S_i$. That is, defining 
$\ddot{Y}_{ijk}(0)=Y_{ijk}(0)- (\sum_{l=1}^{J} \sum_{k'=1}^{N_{i}} S_{ilk'} Y_{ilk'}(0))/(\sum_{l=1}^{J} \sum_{k'=1}^{N_{i}} S_{ilk'})$,
then for each $j=1,...,J$ and $k=1,...,N_i$, we have $E\left[\ddot{Y}_{ijk}(0) | Z_i,N_i, S_i\right] = E\left[\ddot{Y}_{ijk}(0) | S_i\right]$. ii. (\textit{Exchangeable treatment effects}) Average treatment effects are mean independent of $Z_i$ and $N_i$, $E\left[Y_{ijk}(z)-Y_{ijk}(0)|Z_i,N_i\right] = E\left[Y_{ijk}(z)-Y_{ijk}(0)\right]$.

Assumption A1 is typical for causal inference under a super-population framework. Although the dimension of $\bm{W}_i$ varies by $N_i$, the data can be generated via the mixture model $\mathcal{P} = \mathcal{P}^{\bm{W}|N} \times \mathcal{P}^N$, where we first draw cluster size $N_i \sim \mathcal{P}^N$, then draw within-cluster data $\bm{W}_i \sim \mathcal{P}^{\bm{W}|N}$ given $N_i$.
In addition, this assumption requires that the individual-level complete data vector has the same expectation conditional on $N_i$; this allows us to construct estimands based on marginal expectations of individual-level potential outcomes.
Importantly, although A1 assumes between-cluster independence, it allows for arbitrary within-cluster correlation structures among potential outcomes and covariates.
For notation convenience, we omit the subscript $i$ when taking expectation with respect to population distribution $\mathcal{P}$ (Assumption A1). For example, $E[f(\bm{O_i})]$, where $\bm{O}_i$ is the observed data of cluster $i$ and $f$ is an arbitrary measurable function, is simplified as $E[f(\bm{O})]$, where $\bm{O}$ represents the random variable sampled from $\mathcal{P}$.
Assumption A1 justifies this simplification, which assumes the complete data $\bm{W}_i$ of each cluster are independent and identically distributed. Likewise, $(Y_{ijk}, S_{ijk}, N_{ij})$ are denoted as $(Y_{.jk}, S_{.jk}, N_{.j})$ when taking expectations. This simplification is used throughout the manuscript for all notation.

Assumption A2 describes a random enrollment scheme and assumes away selection bias and the resulting time-varying confounding. Under this assumption, the enrollment indicator $S_{ijk}$ in a period is assumed to be irrelevant to the potential outcomes, treatment, or baseline variables, whereas arbitrary correlations among $S_{ijk}$'s are allowed.
For example, a closed-cohort design can thus be accommodated by simply setting $S_{i1k}=S_{i2k}=...=S_{iJk}$ for all $i$ and $k$.
Assumption A2 also allows $N_{ij} = \sum_{k=1}^{N_i} S_{ijk}$ to vary over calendar time either due to randomness or period effects.
By definition, $N_{ij}$ is a function of $S_{ijk}$ and $N_i$.
Importantly, by setting $N_{ij} \perp N_i$ with Assumption A2, which can occur when cluster-period enrollment sizes are set regardless of the underlying cluster size $N_i$, then $N_{ij}$ is solely a function of $S_{ijk}$. 
This implies that $N_{ij}$ is also independent of all other random variables in $\bm{W}_i$ given $N_i$.
Notably, Assumption A1 on its own allows for the presence of informative cluster sizes, where data vectors $\bm{W}_i$ (including potential outcomes $(Y_{ijk}(0),Y_{ijk}(z))$) depend on the source population size $N_i$.
However, Assumption A2 stipulates that observed cluster-period sizes $N_{ij}$ are independent of the source population size $N_i$.
This serves as a ``restricted informative sizes'' assumption, and as we will see, allows the fixed-effects model to target the marginal cluster average treatment effect estimand.

Assumption A3 holds by design for longitudinal cluster randomized trials (CRTs).
Assumption A4 is a supplementary assumption to establish the consistency of the linear fixed-effects model in the absence of Assumption A3, as is the case in longitudinal cluster quasi-experimental trials (CQTs).
Assumption A4 first implies parallel trends, allowing for additive confounding.
Next, the assumption of exchangeable treatment effects ensures that the magnitude of the treatment effect is not confounded by, and is homogeneous across, treatment assignment $Z_i$ and cluster size $N_i$.
In the absence of randomization, Assumptions A1, A2, and A4 allow for individual and cluster-level time-invariant confounding in a longitudinal CQT, but imply that there is no additional within-cluster time-varying confounding. These are crucial assumptions in proving the consistency of fixed-effects methods in quasi-experimental settings where Assumption A3 is not met \citep{wooldridge_econometric_2010}.
That is, unlike the linear mixed-effects model \citep{wang_how_2024}, the linear fixed-effects model-based estimator can still be consistent in longitudinal CQTs by relying on Assumption A4.

To note, the vector of only the enrolled within-cluster deviations of the untreated potential outcomes in Assumption A4 will also be presented in later sections as $\left(\ddot{Y}_{ijk}(0)\right)_{(j,k) : S_{ijk}=1} = \left(Y_{ijk}(0)- (\sum_{l=1}^{J} \sum_{k'=1}^{N_{i}} S_{ilk'} Y_{ilk'}(0))/(\sum_{l=1}^{J} \sum_{k'=1}^{N_{i}} S_{ilk'})\right)_{(j,k) : S_{ijk}=1} = \bm{\mathcal{M}}_i\bm{Y}_i^o(0) = \bm{D}_i^\top \ddot{\bm{Y}}_i(0)=  \ddot{\bm{Y}}_i^o(0)$, with more details in subsequent sections.

As has rarely been pointed out, ``individual-level time-invariant confounding” in longitudinal CQTs under the stated assumptions operate entirely at the cluster-level due to the clustered design.
By design, an individual intuitively will never receive treatment unless their cluster is assigned the treatment, regardless of their individual characteristics.
Treatment is administered exclusively at the cluster-level, and any association between individual characteristics and treatment administration arise only through their aggregation into the cluster-level.
Therefore, all time-invariant confounding is aggregated at the cluster-level, as is depicted in the simple directed acyclic graph with confounding (satisfying Assumption A1)
\begin{center}\begin{tikzpicture}[
    node distance=2cm,
    every node/.style={circle, draw, minimum size=1cm},
    every edge/.style={draw, -Stealth}
]
    \node (N) at (0,1) {$N_i$};
    \node (X) at (0,-1) {$\bm{X}_{ik}$};
    \node (alpha) at (2.5,0) {$\alpha_i$};
    \node (Z) at (5,1) {$Z_i$};
    \node (Y) at (5,-1) {$Y_{ijk}$};
    \draw[->] (N) -- (X);
    \draw[->] (N) -- (alpha);
    \draw[->] (X) -- (alpha);
    \draw[->] (X) -- (Y);
    \draw[->] (alpha) -- (Z);
    \draw[->] (alpha) -- (Y);
    \draw[->] (Z) -- (Y);
\end{tikzpicture}\end{center}
where $\alpha_i$ capture the observed and unobserved cluster-level covariates.
Therefore, with Assumptions A1, A2, and A4, methods that adjust for all cluster-level time-invariant confounding will also automatically adjust for all individual-level time-invariant confounding in a longitudinal CQT.

\subsection{Implication of Assumption A4}
\label{app.sect:assumption4}

As mentioned previously, Assumption A4 (Mean independence) will specifically allow the linear fixed-effects model to yield consistent estimators in longitudinal CQTs lacking Assumption A3 (Randomization).

We quickly summarize key points regarding Assumption A4 below.
Assumption A4.i is akin to a standard parallel trends assumption and assumes that the centered untreated potential outcomes $Y_{ijk}(0)- (\sum_{l=1}^{J} \sum_{k'=1}^{N_{i}} S_{ilk'} Y_{ilk'}(0))/(\sum_{l=1}^{J} \sum_{k'=1}^{N_{i}} S_{ilk'}) = \ddot{Y}_{ijk}(0)$ are mean independent of $Z_i$ and $N_i$, given $S_i$.
Assumption A4.ii further assumes exchangeable treatment effects, that average treatment effects are mean independent of $Z_i$ and $N_i$, $E\left[Y_{ijk}(z)-Y_{ijk}(0)|Z_i,N_i\right] = E\left[Y_{ijk}(z)-Y_{ijk}(0)\right]$.
These mean independence assumptions are weaker than corresponding independence assumptions.

Here, we highlight a simple example that implies Assumption A4.
In practice, suppose the data-generating process (DGP) for the potential outcomes follow a linear structure model
\begin{equation}
\label{app.eq:DGP_ex}
\begin{split}
    Y_{ijk}(z) &= \beta_{0j} + \beta_{jZ}(1+\delta_i) + \bm{\beta_X}^\top \bm{X}_{ik} + \alpha_i + \epsilon_{ijk} \\
    Y_{ijk}(0) &= \beta_{0j} + \bm{\beta_X}^\top \bm{X}_{ik} + \alpha_i + \epsilon_{ijk}
\end{split}
\end{equation}
where $\beta_{0j}$ is the period effect for period $j$, $\beta_{jZ}(1+\delta_i)$ is the period-specific treatment effect structure, $\bm{X}_{ik}$ are the individual time-invariant covariates, $\alpha_i$ are the cluster intercepts, and $\epsilon_{ijk} \sim N(0,\sigma^2)$.
As in the earlier DAG, this DGP allows for (1.) $\bm{X}_{ik}$ and $\alpha_i$ to be arbitrarily correlated, (2.) $\alpha_i$ and treatment status $Z_i$ to be arbitrarily correlated. The latter implies that cluster-specific values of $\alpha_i$ affect whether clusters are more likely to receive treatment and observe the treated potential outcome $Y_{ijk}(z)$, therefore permitting cluster-level time-invariant confounding.

The treatment effect structure $\beta_{jZ}(1+\delta_i)$ includes a period-specific treatment effect structure $\beta_{jZ}$ and cluster-level heterogeneity $\delta_i$. As long as $\delta_i \perp Z_i, \alpha_i$ is assumed (e.g., $\delta_i \sim N(0,\sigma^2_\delta)$), then this DGP will satisfy Assumption A4.ii (exchangeable treatment effects).

We can then demonstrate that Assumption A4.i (parallel trends) holds by design in this DGP.
The within-cluster deviations of the potential outcomes have the following DGP
\begin{equation}
\label{app.eq:DGP_transf_ex}
\begin{split}
    \ddot{Y}_{ijk}(0) &= (\beta_{0j} + \bm{\beta_X}^\top \bm{X}_{ik} + \alpha_i + \epsilon_{ijk}) - \frac{\sum_{l=1}^{J} \sum_{k'=1}^{N_{i}} S_{ilk'} (\beta_{0l} + \bm{\beta_X}^\top \bm{X}_{ik'} + \alpha_i + \epsilon_{ilk'})}{\sum_{l=1}^{J} N_{il}} \,.
\end{split}
\end{equation}
With A1 (Super-population sampling) and A2 (Non-informative enrollment, $N_{ij} \perp N$), then crucially the following terms do not depend on treatment assignment $Z_i$ and cluster size $N_i$
\[
\begin{split}
    E\left[\alpha - \frac{\sum_{l=1}^{J} \sum_{k'=1}^{N} S_{.lk'} \alpha}{\sum_{l=1}^{J} N_{.l}} | Z,N,\bm{S}\right] = E\left[\alpha -\frac{\sum_{l=1}^{J} N_{.l} \alpha }{\sum_{l=1}^{J} N_{.l}} |\bm{S}\right] &= 0 \\
    E\left[\bm{X}_{.k} - \frac{\sum_{l=1}^{J} \sum_{k'=1}^{N} S_{.lk'} \bm{X}_{.k}}{\sum_{l=1}^{J} N_{.l}} | Z,N,\bm{S}\right] = E\left[\bm{X}_{.k} - \frac{\sum_{l=1}^{J} N_{.l} \bm{X}_{.k}}{\sum_{l=1}^{J} N_{.l}} |\bm{S}\right] &= 0 \,.
\end{split}
\]
Finally, we show that with Assumption A1 and A2
\[
\begin{split}
    E\left[\ddot{Y}_{.jk}(0) |N,\bm{S}\right] &= E\left[ \beta_{0j} - \frac{\sum_{l=1}^{J}\sum_{k'=1}^{N} S_{.lk'} I\{l=j\} \beta_{0l}}{\sum_{l=1}^{J} N_{.l}}  | Z,N,\bm{S}\right] \\
    &= E\left[ \beta_{0j} - \frac{\sum_{l=1}^{J} \sum_{k'=1}^{N} S_{.lk'} I\{l=j\} \beta_{0l}}{\sum_{l=1}^{J} N_{.l}} |\bm{S} \right] \\
     &=  E\left[\ddot{Y}_{.jk}(0) | \bm{S}\right] \,,
\end{split}
\]
fulfilling Assumption A4.i with the described DGP.

\newpage

\section{Theorems \& Lemmas}
\label{app.sect:Theorems_Lemmas}

\subsection{Lemma \ref{app.lemma:variance}}
\label{app.sect:Lemma1}

\begin{applemma}
\label{app.lemma:variance}
Let $\bm{O}_1,...,\bm{O}_{m}$ be i.i.d. samples from a common distribution on $O$.
Let $\bm{\psi}(\bm{O},\bm{\theta})$ be a known estimating equation with parameters $\bm{\theta} \in \Theta$, a compact set of Euclidean space.
Let $\hat{\bm{\theta}}$ be the solution to $\sum_{i=1}^{m} \bm{\psi}(\bm{O}_i, \bm{\theta})=0$.
We assume that $\bm{\psi}$ satisfies the following regularity conditions
\begin{enumerate}
    \item There exists a unique solution in the interior of $\Theta$, denoted as $\underline{\bm{\theta}}$, to the equation $E[\bm{\psi}(\bm{O},\bm{\theta})] = 0$.
    \item The function $\bm{\theta} \mapsto \bm{\psi}(o, \bm{\theta})$, together with its first and second derivatives, is dominated by a square-integrable function for every $o$ in the support of $\bm{O}$.
    \item $E\left[\frac{d\bm{\psi}(\bm{O},\bm{\theta})}{d\bm{\theta}^{\top}} \mid_{\bm{\theta}=\underline{\bm{\theta}}}\right]$ is invertible.
\end{enumerate}
Then we have
\[
\begin{split}
    \hat{\bm{\theta}} &\xrightarrow{P} \underline{\bm{\theta}} \\
    m^{1/2}(\hat{\bm{\theta}} - \underline{\bm{\theta}}) &\xrightarrow{d} N(0,\textbf{V})
\end{split}
\]
where $\textbf{V}=E[\text{IF}(\bm{O},\underline{\bm{\theta}})\text{IF}(\bm{O},\underline{\bm{\theta}})^\top]$ and $\text{IF}(\bm{O},\underline{\bm{\theta}}) = -\left( E\left[\frac{d\bm{\psi}(\bm{O},\bm{\theta})}{d\bm{\theta}^{\top}} \mid_{\bm{\theta}=\underline{\bm{\theta}}}\right]^{-1} \bm{\psi}(\bm{O},\underline{\bm{\theta}}) \right)$ is the influence function for $\hat{\bm{\theta}}$.

Furthermore, the sandwich variance estimator
\[
    \hat{\bm{V}} = m^{-1}\sum_{i=1}^{m} \widehat{\text{IF}}(\bm{O}_i,\hat{\bm{\theta}}) \widehat{\text{IF}}(\bm{O}_i,\hat{\bm{\theta}})^\top \xrightarrow{P} \textbf{V}
\]
where $\widehat{\text{IF}}(\bm{O}_i,\hat{\bm{\theta}})=\left( 
\left[
    m^{-1}\sum_{i=1}^{m} \frac{d\bm{\psi}(\bm{O}_i,\bm{\theta})}{d\bm{\theta}^{\top}} \mid_{\bm{\theta}=\hat{\bm{\theta}}}
\right]^{-1}
\bm{\psi}(\bm{O}_i, \hat{\bm{\theta}})
\right)$.
\end{applemma}

\noindent \textit{Proof of Lemma \ref{app.lemma:variance}.}
The above three conditions are standard assumptions in an M-estimation framework \citep{tsiatis_semiparametric_2006,van_der_vaart_asymptotic_1998,ross_m-estimation_2024} to ensure the estimating function $\bm{\psi}$ (equivalent to the likelihood score function) is well-behaved to prove the asymptotic results.
Of note, condition (1) does not imply any component of the working model is correctly specified. Instead, it solely requires the uniqueness of maxima in the maximum likelihood or maximum quasi-likelihood estimation.
This can be achieved by carefully designing $\bm{\psi}$ and restricting the parameter space to rule out degenerative solutions.
As a specific example, if $\bm{\psi}$ is the estimating equation for linear regression (based on ordinary least squares), then condition (1) is equivalent to the invertibility of the covariance matrix of covariates.

We use the notation $o_p(1)$ to denote a sequence of random vectors that converges to zero in probability, and $O_p(1)$ to denote a sequence that is bounded in probability, as per Section 2.2 of \citet{van_der_vaart_asymptotic_1998}.
An asymptotically linear estimator $\hat{\bm{\theta}}$, excluding nuisance parameters (as described in Section 3 of \citet{tsiatis_semiparametric_2006}), can be uniquely characterized by its influence function, as demonstrated in Theorem 3.1 of \citet{tsiatis_semiparametric_2006}. The influence function captures the influence of the $i$th observation on $\hat{\bm{\theta}}$ and is denoted as $\text{IF}(\bm{O}_i,\underline{\bm{\theta}}) = -\left( E\left[\frac{d\bm{\psi}(\bm{O},\bm{\theta})}{d\bm{\theta}^{\top}} \mid_{\bm{\theta}=\underline{\bm{\theta}}}\right]^{-1} \bm{\psi}(\bm{O},\underline{\bm{\theta}}) \right)$, as per Equation 3.6 of \citet{tsiatis_semiparametric_2006}.

The proof proceeds by using ``classical conditions'' for asymptotic normality of M-estimators, as per Section 5.6, and largely follows Theorem 5.41, of \citet{van_der_vaart_asymptotic_1998}.
By condition (2) for $\bm{\psi}$, Example 19.8 of \citet{van_der_vaart_asymptotic_1998} implies that $\{\bm{\psi}(\bm{O},\bm{\theta}) : \bm{\theta} \in \bm{\Theta}\}$ is P-Glivenko-Cantelli (Glivenko-Cantelli in probability).
Then with condition (1), Theorem 5.9 of \citet{van_der_vaart_asymptotic_1998} shows that $\hat{\bm{\theta}} \xrightarrow{P} \underline{\bm{\theta}}$.
Next, we apply Theorem 5.41 of \citet{van_der_vaart_asymptotic_1998} to obtain asymptotic normality, for which our assumptions on $\bm{\psi}$ ensure all conditions needed are satisfied. Then we have
\begin{equation}
    m^{1/2}(\hat{\bm{\theta}} - \underline{\bm{\theta}}) = m^{1/2} \left(\sum_{i=1}^{m} \text{IF}(\bm{O}_i, \underline{\bm{\theta}})\right) + o_p(1) \,,
\end{equation}
as per Theorem 5.41 of \citet{van_der_vaart_asymptotic_1998} and equation 3.1 of \citet{tsiatis_semiparametric_2006}, which implies the desired asymptotic normality by the Central Limit Theorem.
Specifically, by the central limit theorem,  $m^{1/2} \left(\sum_{i=1}^{m} \text{IF}(\bm{O}_i, \underline{\bm{\theta}})\right) \xrightarrow{D} N(0, E[\text{IF}(\bm{O},\underline{\bm{\theta}}) \text{IF}(\bm{O},\underline{\bm{\theta}})^\top])$ and by Slutsky's theorem,
\begin{equation}
    m^{1/2}(\hat{\bm{\theta}} - \underline{\bm{\theta}}) \xrightarrow{D} N\left( 0, E[\text{IF}(\bm{O},\underline{\bm{\theta}}) \text{IF}(\bm{O},\underline{\bm{\theta}})^\top] \right)
\end{equation}
hence $\textbf{V} = E[\text{IF}(\bm{O},\underline{\bm{\theta}}) \text{IF}(\bm{O},\underline{\bm{\theta}})^\top]$, as per Equation 3.7 of \citet{tsiatis_semiparametric_2006}.

We next prove the consistency of the sandwich variance estimator. First, we prove that
\[
    m^{-1}\sum_{i=1}^{m} \frac{d\bm{\psi}(\bm{O}_i,\bm{\theta})}{d\bm{\theta}^{\top}} \mid_{\bm{\theta}=\hat{\bm{\theta}}}
\xrightarrow{P} E\left[\frac{d\bm{\psi}(\bm{O},\bm{\theta})}{d\bm{\theta}^{\top}} \mid_{\bm{\theta}=\underline{\bm{\theta}}}\right] \,.
\]
Denoting $\dot{\bm{\psi}}_{ij}(\hat{\bm{\theta}})$ as the transpose of the $j$th row of $\frac{d\bm{\psi}(\bm{O}_i,\bm{\theta})}{d\bm{\theta}^{\top}} \mid_{\bm{\theta}=\hat{\bm{\theta}}}$, and $\ddot{\bm{\psi}}_{ij}$ being the derivative of $\dot{\bm{\psi}}_{ij}$, we apply the multivariate Taylor expansion to get
\[
    m^{-1} \sum_{i=1}^{m} \dot{\bm{\psi}}_{ij}(\hat{\bm{\theta}}) - m^{-1} \sum_{i=1}^{m} \dot{\bm{\psi}}_{ij}(\underline{\bm{\theta}}) 
    = m^{-1} \sum_{i=1}^{m} \left( \dot{\bm{\psi}}_{ij}(\hat{\bm{\theta}}) - \dot{\bm{\psi}}_{ij}(\underline{\bm{\theta}}) \right)
    = m^{-1} \left(  \sum_{i=1}^{m} \ddot{\bm{\psi}}_{ij}(\tilde{\bm{\theta}}) \right)(\hat{\bm{\theta}}-\underline{\bm{\theta}})
\]
for some $\tilde{\bm{\theta}}$ on the line segment between $\hat{\bm{\theta}}$ and $\underline{\bm{\theta}}$.
By condition 2 and $\hat{\bm{\theta}} \xrightarrow{P} \bm{\theta}$, we have $m^{-1} \sum_{i=1}^{m} \ddot{\bm{\psi}}_{ij}(\tilde{\bm{\theta}}) =O_p(1)$.
As a result, $\hat{\bm{\theta}}-\bm{\theta} = o_p(1)$ implies that $m^{-1} \sum_{i=1}^{m} \dot{\bm{\psi}}_{ij}(\hat{\bm{\theta}}) - m^{-1} \sum_{i=1}^{m} \dot{\bm{\psi}}_{ij}(\underline{\bm{\theta}})  = o_p(1)$.
Then the first step is completed by the fact that
\[
    m^{-1} \sum_{i=1}^{m} \dot{\bm{\psi}}_{ij}(\underline{\bm{\theta}}) = E[\dot{\bm{\psi}}_{ij}(\underline{\bm{\theta}})] + o_p(1)
\]
which results from the Law of Large Numbers and condition 2.
Next we prove
\[
    m^{-1} \sum_{i=1}^{m} \bm{\psi}(\bm{O}_i, \hat{\bm{\theta}}) \bm{\psi}(\bm{O}_i, \hat{\bm{\theta}})^\top 
    \xrightarrow{P} E[\bm{\psi}(\bm{O}_i, \underline{\bm{\theta}}) \bm{\psi}(\bm{O}_i, \underline{\bm{\theta}})^\top ]
\]
following a similar procedure to the prior step. Letting $\bm{\psi}_{ij}(\bm{\theta})$ be the $j$th entry of $\bm{\psi}(\bm{O}_i,\bm{\theta})$, we apply the multivariate Taylor expansion and get
\[
\begin{split}
    m^{-1} &\sum_{i=1}^{m} \bm{\psi}_{ij}(\hat{\bm{\theta}}) \bm{\psi}(\bm{O}_i, \hat{\bm{\theta}}) -  m^{-1} \sum_{i=1}^{m} \bm{\psi}_{ij}(\underline{\bm{\theta}}) \bm{\psi}(\bm{O}_i, \underline{\bm{\theta}}) \\
    &= m^{-1} \left( \sum_{i=1}^{m} \left[ 
    \bm{\psi}(\bm{O}_i,\tilde{\bm{\theta}})\dot{\bm{\psi}}_{ij}(\tilde{\bm{\theta}})^\top + \bm{\psi}_{ij}(\tilde{\bm{\theta}}) \frac{d\bm{\psi}(\bm{O}_i,\bm{\theta})}{d\bm{\theta}^{\top}} \mid_{\bm{\theta}=\tilde{\bm{\theta}}} 
    \right] \right) (\hat{\bm{\theta}}-\underline{\bm{\theta}}) \\
    &= O_p(1)o_p(1)
\end{split}
\]
which, combined with the Law of Large Numbers on $m^{-1} \sum_{i=1}^{m} \bm{\psi}_{ij}(\underline{\bm{\theta}})\bm{\psi}(\bm{O}_i,\underline{\bm{\theta}})$, implies the desired result in this step.
The sandwich variance estimator can be written in terms of influence functions while substituting $\hat{\bm{\theta}}$ for the unknown $\underline{\bm{\theta}}$, as defined in equation 3.10 of \citet{tsiatis_semiparametric_2006}.
\begin{align}
    m^{-1} &\sum_{i=1}^{m} \widehat{\text{IF}}(\bm{O}_i,\hat{\bm{\theta}}) \widehat{\text{IF}}(\bm{O}_i,\hat{\bm{\theta}})^\top \\
    &=\left(
        m^{-1} \sum_{i=1}^{m} \frac{d\bm{\psi}(\bm{O}_i,\bm{\theta})}{d\bm{\theta}^{\top}} \mid_{\bm{\theta}=\hat{\bm{\theta}}}
    \right)^{-1}
    \left(
        m^{-1} \sum_{i=1}^{m} \bm{\psi}(\bm{O}_i,\hat{\bm{\theta}}) \bm{\psi}(\bm{O}_i,\hat{\bm{\theta}})^\top
    \right)
    \left(
        m^{-1} \sum_{i=1}^{m} \frac{d\bm{\psi}(\bm{O}_i,\bm{\theta})}{d\bm{\theta}^{\top}} \mid_{\bm{\theta}=\hat{\bm{\theta}}}
    \right)^{-1\top} \notag \\
\intertext{Finally by the Continuous Mapping Theorem described in Theorem 2.3 of \citet{van_der_vaart_asymptotic_1998}, and the above derivations, we can show the sandwich variance estimator}
    &=\left(
        E\left[\frac{d\bm{\psi}(\bm{O},\bm{\theta})}{d\bm{\theta}^{\top}} \mid_{\bm{\theta}=\underline{\bm{\theta}}}\right] + o_p(1)
    \right)^{-1}
    \left(
        E\left[ \bm{\psi}(\bm{O},\underline{\bm{\theta}}) \bm{\psi}(\bm{O},\underline{\bm{\theta}})^\top \right] + o_p(1)
    \right)
    \left(
        E\left[\frac{d\bm{\psi}(\bm{O},\bm{\theta})}{d\bm{\theta}^{\top}} \mid_{\bm{\theta}=\underline{\bm{\theta}}}\right] + o_p(1)
    \right)^{-1\top} \notag \\
    &= \textbf{V} + o_p(1) \,.
\end{align}
Altogether, the sandwich variance estimator is consistent for the true variance \textbf{V}. $\square$

\subsection{Theorem 1: Asymptotic properties of the linear fixed-effects model estimator}

In Theorems \ref{app.Theorem:SW}-\ref{app.Theorem:XO} below (for SW-CTs, PB-CTs, and CXOs, respectively), we define the finite-sample sandwich variance estimators
$\widehat{Var}\left(\hat{\beta}_Z\right)=\frac{1}{m}\hat{V}_Z$, 
$\widehat{Var}\left(\hat{\bm{\beta}}_Z^D\right) = \frac{1}{m}\hat{\bm{V}}_Z^D$, 
$\widehat{Var}\left(\hat{\bm{\beta}}_Z^P\right) =\frac{1}{m}\hat{\bm{V}}_Z^P$, 
and $\widehat{Var}\left(\hat{\bm{\beta}}_Z^S\right) =\frac{1}{m}\hat{\bm{V}}_Z^S$, corresponding to the constant, duration-specific, period-specific, and saturated treatment effect structures, respectively.

\begin{appsubtheorem}
\label{app.Theorem:SW}
    Under standard regularity conditions (outlined in Lemma \ref{app.lemma:variance}) and with either (i.) Assumptions A1-A3, (ii.) Assumptions A1, A2, A4, 
    or (iii.) the mean model is correctly specified,
    the following Central Limit Theorems hold for a linear fixed-effects model analysis of a stepped-wedge designed cluster trial.
    That is, (a) under a true constant treatment effect structure, $\hat{V}_Z^{-1/2} m^{1/2} (\hat{\beta}_Z - \Delta) \xrightarrow{d} N(0,1)$;
    (b) under a true duration-specific treatment effect structure, $(\hat{\bm{V}}_Z^D)^{-1/2} m^{1/2} (\hat{\bm{\beta}}^D_Z - \bm{\Delta}^D) \xrightarrow{d} N(0,\textbf{I}_{J-1})$;
    (c) under a true period-specific treatment effect structure, $(\hat{\bm{V}}_Z^P)^{-1/2} m^{1/2} (\hat{\bm{\beta}}^P_Z - \bm{\Delta}^P) \xrightarrow{d} N(0,\textbf{I}_{J-2})$
    and $\hat{V}_Z^{-1/2} m^{1/2} (\hat{\beta}_Z - \Delta^{P-ATO}) \xrightarrow{d} N(0,1)$;
    and (d) under a true saturated treatment effect structure, $(\hat{\bm{V}}_Z^S)^{-1/2} m^{1/2} (\hat{\bm{\beta}}^S_Z - \bm{\Delta}^S) \xrightarrow{d} N(0,\textbf{I}_{(J-2)(J-1)/2})$;
    in (a)-(d), $\textbf{I}_q$ is a $q \times q$ identity matrix.
\end{appsubtheorem}

\begin{appsubtheorem}
\label{app.Theorem:PB}
    Under standard regularity conditions (outlined in Lemma \ref{app.lemma:variance}) and with either (i.) Assumptions A1-A3, (ii.) Assumptions A1, A2, A4, 
    or (iii.) the mean model is correctly specified,
    the following Central Limit Theorems hold for a linear fixed-effects model analysis of a parallel-with-baseline designed cluster trial.
    That is, (a) under a true constant treatment effect structure, $\hat{V}_Z^{-1/2} m^{1/2} (\hat{\beta}_Z - \Delta) \xrightarrow{d} N(0,1)$;
    (b) under a true saturated/duration/period-specific treatment effect structure, $(\hat{\bm{V}}_Z^S)^{-1/2} m^{1/2} (\hat{\bm{\beta}}^S_Z - \bm{\Delta}^S) \xrightarrow{d} N(0,\textbf{I}_{J-1})$ and 
    $\hat{V}_Z^{-1/2} m^{1/2} (\hat{\beta}_Z - \Delta^{S-avg}) = \hat{V}_Z^{-1/2} m^{1/2} (\hat{\beta}_Z - \Delta^{P-ATO}) \xrightarrow{d} N(0,1)$;
    in (a)-(b), $\textbf{I}_q$ is a $q \times q$ identity matrix.
\end{appsubtheorem}

\begin{appsubtheorem}
\label{app.Theorem:XO}
    Under standard regularity conditions (outlined in Lemma \ref{app.lemma:variance}) and with either (i.) Assumptions A1-A3, (ii.) Assumptions A1, A2, A4,
    or (iii.) the mean model is correctly specified,
    the following Central Limit Theorems hold for a linear fixed-effects model analysis of a crossover designed cluster trial.
    That is, (a) under a true constant treatment effect structure, $\hat{V}_Z^{-1/2} m^{1/2} (\hat{\beta}_Z - \Delta) \xrightarrow{d} N(0,1)$;
    (b) under a true duration-specific treatment effect structure, $(\hat{\bm{V}}_Z^D)^{-1/2} m^{1/2} (\hat{\bm{\beta}}^D_Z - \bm{\Delta}^D) \xrightarrow{d} N(0,\textbf{I}_{(2J+1-(-1)^J)/4})$. 
    Furthermore, under a true saturated/period-specific treatment effect structure, $\hat{V}_Z^{-1/2} m^{1/2} (\hat{\beta}_Z - \Delta^{S-avg}) = \hat{V}_Z^{-1/2} m^{1/2} (\hat{\beta}_Z - \Delta^{P-ATO}) \xrightarrow{d} N(0,1)$;
    in (a)-(b), $\textbf{I}_q$ is a $q \times q$ identity matrix.
\end{appsubtheorem}

Theorems \ref{app.Theorem:SW}-\ref{app.Theorem:XO} suggest that assuming a correct treatment effect structure, the linear fixed-effects model provides consistent and asymptotically normal estimators for the respective marginal cluster-average treatment effect estimands, even if other model components (time-invariant covariates, random effects, residual distribution) are arbitrarily misspecified in the analysis of (A.) SW-CRTs, (B.) PB-CRTs, and (C.) CRXOs.
We can generally extend these conclusions from these three designs to other longitudinal complete CRT designs that have within-cluster and between-cluster variation in treatment status.

Across these longitudinal CRT designs (Theorem \ref{app.Theorem:SW}-\ref{app.Theorem:XO}), the constant treatment effect estimator is notably still consistent for the period-average treatment effect for the overlap population (P-ATO).
This matches a similar observation in \citet{lee_analysis_2025} for linear mixed-effects models with constant treatment effect structures in SW-CRTs with period-specific treatment effect structures; although \citet{lee_analysis_2025} did not explicitly make the connection between the targeted  estimand and the P-ATO.
Unlike in staggered designs (such as the SW-CT), the P-ATO is also equivalent to the saturated/duration/period-average treatment effect estimands in PB-CT and CXO designs.
This is due to the saturated, period-specific, and duration-specific treatment effect structures coinciding, with the saturated/period-specific treatment effect structures being a more general version of the duration-specific treatment effect structure in a CXO.
That is, regardless of the treatment effect structure, the constant working treatment effect structure still estimates an interpretable time-averaged treatment effect estimand in a PB-CRT or CRXO design.
Importantly then, robust inference of PB-CRT and CRXO designs does not even require a correctly specified treatment effect structure.

Furthermore, without relying on randomization, the linear fixed effects model can still provide consistent and asymptotically normal estimators for the average treatment effect, even in the presence of measured or unmeasured time-invariant confounding or covariate imbalance.
This highlights a central message that robust inference for the treatment effect estimands with such a fixed effects model typically requires a correct treatment effect structure in the working model, particularly in stepped-wedge designs, rather than any other remaining model aspects.

Altogether, we prove that despite being typically understood to target a conditional treatment effect estimand, the linear fixed-effects model in a longitudinal CT can also asymptotically target a marginal treatment effect estimand.
With randomization (Assumption A3), cluster intercepts are uncorrelated with other model covariates.
This, alongside the assumed super-population sampling (Assumption A1), asymptotically produces a \textit{collapsible} and coinciding marginal and conditional difference estimand, that is targeted by the specified linear fixed-effects model.
In the absence of randomization, the detailed supplementary assumption (Assumption A4) can achieve a similar result.
Intuitively, a linear mixed-effects model coincides with a linear GEE, and targets a collapsible estimand that is both marginal and conditional \citep{gardiner_fixed_2009}. Subsequently, mixed-effects models and fixed-effects model estimates coincide under the traditional definition of strict exogeneity, where cluster intercepts are uncorrelated with model fixed covariates, as is the case in randomized trials \citep{hausman_specification_1978,lee_fixed-effects_2024}.
Accordingly, like a mixed-effects model with Assumptions A1-A3 fulfilled \citep{wang_how_2024}, we prove that the linear fixed-effects model yields consistent and asymptotically normal estimators for the marginal cluster-average treatment effect estimators in longitudinal CRT designs.

In the following sections, we will broadly demonstrate the consistency of treatment effect structures using a linear fixed-effects model, where each cluster is included as a separate dummy variable in the model.
Importantly, such a linear model is able to avoid the 
``incidental parameters problem'' described in Neyman \& Scott \citep{neyman_consistent_1948} where the number of model parameters goes to infinity as the sample of clusters goes to infinity. The clusters (or incidental parameters) in these linear fixed-effects models can be treated as \textit{nuisance parameters} and removed from the above fixed-effects likelihood and score function by using a ``within-transformation'' \citep{kiefer_estimation_1980,cameron_microeconometrics_2005,wooldridge_econometric_2010}.
The equivalence between the within-transformed specification (Equation \ref{app.eq:transfY_joint}) and previously defined fixed-effects model (Equation \ref{app.eq:Y_joint}) is due to the \textit{orthogonality} of residuals and model covariates in this linear setting, allowing us to project out the cluster covariates \citep{cameron_microeconometrics_2005,wooldridge_econometric_2010}.
We discuss this in more detail in Section \ref{app.sect:constant}.

\newpage
\section{Proof of Theorem \ref{app.Theorem:SW}}
\label{app.sect:Theorem_SW_Proof}

\textit{Proof of Theorem \ref{app.Theorem:SW}.}

\subsection{Treatment Effect Estimands for Stepped Wedge Designs}
\label{app.sect:estimands_SW}

We pursue the potential outcomes framework to define treatment effect estimands. 
With potential outcomes $Y_{ijk}(\tilde{z})$ where $\tilde{z} \in \{0, z\}$, let $Y_{ijk}(z)$ denote the potential outcome of individual $k$ in cluster $i$ during period $j$ had the cluster been first treated in period $z$ for $2 \leq z \leq j$ in stepped-wedge designs. If $z > j$, we assume no anticipation and $Y_{ijk}(0)$ denotes the untreated potential outcome.
We connect the observed outcome $Y_{ijk}$ and potential outcomes via the following
\[
Y_{ijk} = \sum_{z=2}^j I\{Z_i=z\}Y_{ijk}(z) + I\{Z_i>j\}Y_{ijk}(0) \,, j \in \{1,...,J\} \,.
\] 

We define the marginal cluster-average treatment effect as a function of potential outcomes. Given the treatment adoption time $z$ and in period $j$, we denote $d=j-z+1 >0$ as the duration of treatment or exposure time. Then, $Y_{ijk}(j-d+1)$ is the individual potential outcome in period $j$ had cluster $i$ been treated for $d$ ($0<d<j$) periods (or equivalently had the treatment been first adopted in period $z=j-d+1$). The treatment effect estimand is then defined as
\begin{equation}
\label{app.eq:PO_estimands_SW}
\begin{split}
    \Delta_j(d) &= E\left[ \frac{1}{N} \sum_{k=1}^{N} Y_{.jk}(j-d+1) \right] - E\left[ \frac{1}{N} \sum_{k=1}^{N} Y_{.jk}(0) \right] \\
    & = E\left[ E\left[\frac{1}{N} \sum_{k=1}^{N} Y_{.jk}(j-d+1) \mid N\right] \right] - E\left[ E\left[\frac{1}{N} \sum_{k=1}^{N} Y_{.jk}(0) \mid N\right] \right] \\
    & = E\left[ \frac{1}{N} \sum_{k=1}^{N} E\left[Y_{.jk}(j-d+1) \mid N\right] \right] - E\left[ \frac{1}{N} \sum_{k=1}^{N} E\left[Y_{.jk}(0) \mid N\right] \right] \\
    & = E\left[ E\left[Y_{.jk}(j-d+1) \mid N\right] \right] - E\left[ E\left[Y_{.jk}(0) \mid N\right] \right] \\
    &= E[Y_{.jk}(j-d+1)] - E[Y_{.jk}(0)]
\end{split}
\end{equation}
for $1 \leq d < j \leq J$.
To reiterate, $Y_{i1k}(0)$ is always observed and $Y_{iJk}(0)$ is never observed in a SW-CT.
The third equality is due to the assumed exchangeability within clusters in Assumption A1 (Super-population sampling).
Here, the expectation is taken over the distribution of clusters, with $j$ and $d$ being fixed quantities; that is, $E[f(\bm{W}_i)] = \int f(\bm{w}) d\mathcal{P}(\bm{w})$ for any integrable function $f$.
This corresponds to the average causal effect had all clusters been treated for $d$ ($0<d<j$) periods at calendar period $j$.
Definition \ref{app.eq:PO_estimands_SW} is model-free and fully accommodates the treatment effect heterogeneity due to calendar time and exposure time. However, it is not design free as the number and length of periods can be specific to each study.
One may make the following simplifications of estimands based on content knowledge.

The \textit{constant treatment effect structure} assumes that $\Delta_j(d)$ is invariant across all $j$ and $d$, that is, $\Delta_j(d) = \Delta$, and leads to a univariate target estimand $\Delta$.

The \textit{duration-specific treatment effect structure} considers $\Delta_j(d)$ to be constant across $j$, but vary by $d$, that is, $\Delta_j(d)=\Delta(d)$.
The target estimand in this case can be any function of $\bm{\Delta}^D=(\Delta(1),...,\Delta(J-1))^\top \in \mathbb{R}^{J-1}$.

The \textit{period-specific treatment effect structure} allows the treatment effect to vary by the calendar period, but not the treatment duration, that is, $\Delta_j(d)=\Delta_j$.
Of note, although we can conceptualize $\Delta_J$ as the treatment effect in period $J$, it is not identifiable because $Y_{iJk}(0)$ is never observable by design in a SW-CT \citep{chen_model-assisted_2025}. Therefore, we do not further address $\Delta_J$ in the period-specific treatment effect setting.
The target estimand in this case can be any function of $\bm{\Delta}^P = (\Delta_{2},...,\Delta_{J-1})^\top \in \mathbb{R}^{J-2}$.
Among these different estimands, we will particularly highlight the period-average treatment effect for the overlap population (P-ATO), generally defined as
\[
    \Delta^{P-ATO}  = \frac{ \sum_{j=1}^J \lambda_{j} \Delta_j}{\sum_{j=1}^J \lambda_{j}}
\]
with weights $\lambda_{j}=\pi_j^s(1-\pi_j^s)$ being the tilting function that generates the overlap propensity weights \citep{li_balancing_2018}, where $\pi_j = P(Z_i=j)$ such that $\pi_j^s = \sum_{l=1}^j \pi_l$ is the proportion of clusters receiving treatment in period $j$.
Since the probability of treatment assignment depends on the given period, 
the tilting function helps define the target population for which treatment effects are estimated, emphasizing periods with better treatment overlap \citep{li_balancing_2018,li_propensity_2019}.
Accordingly, the baseline and all-exposed periods lacking treatment positivity in the SW-CRT do not contribute any information to the P-ATO estimand.

Finally, the \textit{saturated treatment effect structure} can be used when the treatment effects are expected to vary by both the calendar time and exposure time.
Similar to the period-specific treatment effect structure, the treatment effects in period $J$, $(\Delta_J(1),...,\Delta_J(J))$, are not identifiable and thus excluded.
In this assumption-lean set up, we write $\bm{\Delta}^S = (\Delta_{2}(1), \Delta_{3}(1), \Delta_{3}(2), ..., \Delta_{J-1}(1), ..., \Delta_{J-1}(J-2))^\top \in \mathbb{R}^{(J-2)(J-1)/2}$.

\subsection{Introducing an equivalent definition of potential outcomes}
\label{app.sect:PO_definition}

Above, we denote $Y_{ijk}(z)$ as the potential outcome of individual $k$ in cluster $i$ during period $j$ had the cluster been first treated in period $z$.
In order to simplify the notation in the proof of our theoretical results, we describe an alternative definition of potential outcomes, as in \citet{wang_how_2024}, and show its equivalence to the above definition of potential outcomes.

Let $\widetilde{Y}_{ijk}(d)$ denote the potential outcome of individual $k$ in cluster $i$ during period $j$ had the cluster first been treated for $d$ ($0<d<j$) periods already. In addition, $\widetilde{Y}_{ijk}(0)$ denotes the untreated potential outcome. We next show these two potential outcome definitions, $\widetilde{Y}_{ijk}(d)$ and $Y_{ijk}(z)$ are mathematically equivalent and can be used interchangeably in proving the subsequent technical results.
First $\widetilde{Y}_{ijk}(0)$ and $Y_{ijk}(0)$ are both potential outcomes in the absence of treatment in period $j$.
When $d>0$, $\widetilde{Y}_{ijk}(d)$ refers to the individual potential outcome in period $j$ had the cluster been treated for $d$ ($0<d<j$) periods, which is equivalent to the individual potential outcome in period $j$ had the treatment started in period $j-d+1$, i.e., $Y_{ijk}(z)$ where $z=j-d+1$. 
Equivalently, when $2 \leq z \leq j, Y_{ijk}(z)$ and $\widetilde{Y}_{ijk}(j-z+1)$ both refer to the individual potential outcome in period $j$ had the treatment started in period $z$.
With the above elaboration, we arrive at $Y_{ijk}(0) = \widetilde{Y}_{ijk}(0)$ and $Y_{ijk}(z)=\widetilde{Y}_{ijk}(j-z+1)$, which establishes a bijection between two potential outcome definitions.

Denoting $d=j-z+1 >0$ and $z=j-d+1$, the causal consistency assumption can be alternatively stated as
\[
\begin{split}
    Y_{ijk} &= \sum_{z=2}^j I\{Z_i = z\}Y_{ijk}(z) + I\{Z_i > j\}Y_{ijk}(0) \\
    &= \sum_{z=2}^j I\{Z_i = j-d+1\}\widetilde{Y}_{ijk}(j-z+1) + I\{Z_i > j\}\widetilde{Y}_{ijk}(0)
\end{split}
\]
In addition, Assumption A1 remains the same under the alternative definition of potential outcomes because the given set $\{Y_{ijk}(z) : k=1 ,..., N_i, 2 \leq z \leq j \leq J\}$ is equal to $\{\widetilde{Y}_{ijk}(d) : k=1 ,..., N_i, 1 \leq d < j \leq J\}$.

In the following proofs, we will use the potential outcome definition $\widetilde{Y}_{ijk}(d)$ in place of $Y_{ijk}(z)$. This is because our treatment effect estimand is defined on $Y_{ijk}(j-d+1)$, which is equal to $\widetilde{Y}_{ijk}(d)$. Therefore, using $\widetilde{Y}_{ijk}(d)$ simplifies the notation.
Finally, we further simplify the notation from $\widetilde{Y}_{ijk}(d)$ to $Y_{ijk}(d)$ in our proofs, which improves readability without loss of clarity.
A similar potential outcomes definition can be extended to other longitudinal CT designs, including PB-CTs and CXOs.

\subsection{Notation and connecting observed to potential outcomes}

We first introduce a few definitions. Recall that, in the prior section (Section \ref{app.sect:PO_definition}), we define $Y_{ijk}(d)$ as the potential outcome of individual $k$ in cluster $i$ during period $j$ had the cluster been treated for $d$ ($0<d<j$) periods already. In addition, $Y_{ijk}(0)$ denotes the untreated potential outcome. Let
\[
\begin{split}
    \bm{Y}_{ij} &= (Y_{ij1},...,Y_{ijN_i})^\top \in \mathbb{R}^{N_i} \,, \\
    \bm{Y}_{ij}(d) &= (Y_{ij1}(d),...,Y_{ijN_i}(d))^\top \in \mathbb{R}^{N_i} \,, \\ 
    \bm{Y}_{i} &= (\bm{Y}_{i1}^\top,...,\bm{Y}_{iJ}^\top)^\top \in \mathbb{R}^{N_iJ} \,, \\
    \bm{Y}_{i}(d) &= (\bm{Y}_{i1}(d)^\top,...,\bm{Y}_{iJ}(d)^\top)^\top \in \mathbb{R}^{N_iJ} \,,\\
    \bm{\Lambda}_{Z_i}^d &= 
    \begin{pmatrix}
        0 & & & & & \\
        & \ddots & & & & \\
        & & 0 & & & \\
        & & & I\{Z_i=2\} & & \\
        & & & & \ddots & \\
        & & & & & I\{Z_i=J-d+1\} \\
    \end{pmatrix} \in \mathbb{R}^{J \times J} \text{ for } d=1,...,J-1 \,, \\
    \bm{\Delta}_{Z_i} &= diag\{I\{Z_i \leq j\}: j=1,...,J\} \in \mathbb{R}^{J \times J}
\end{split}
\]
where $\bm{\Delta}_{Z_i} = \sum_{d=1}^{J} \bm{\Lambda}_{Z_i}^d$.
To clarify, in $\bm{\Lambda}_{Z_i}^d$, the 0's on the diagonal preceding $I\{Z_i = 1\}$ are arbitrarily long such that the diagonal entries of $\bm{\Lambda}_{Z_i}^d$ are a combination of $(0,...,0)^\top \in \mathbb{R}^{d}$ and $(I\{Z_i = 2\} ,..., I\{Z_i = J-d+1\})^\top \in \mathbb{R}^{J-d}$.
We illustrate this below with a 3 cluster, 4 period example
\[
\bm{\Lambda}_{Z_i}^d = \,\,
\begin{matrix}
    & i=1 & i=2 & i=3 \\
    d=1 &
        \begin{pmatrix} 0 & & & \\ & 1 & & \\ & & 0 & \\ & & & 0 \end{pmatrix} \,, & 
        \begin{pmatrix} 0 & & & \\ & 0 & & \\ & & 1 & \\ & & & 0 \end{pmatrix} \,, & 
        \begin{pmatrix} 0 & & & \\ & 0 & & \\ & & 0 & \\ & & & 1 \end{pmatrix} 
    \\
    d=2 & 
        \begin{pmatrix} 0 & & & \\ & 0 & & \\ & & 1 & \\ & & & 0 \end{pmatrix} \,, & 
        \begin{pmatrix} 0 & & & \\ & 0 & & \\ & & 0 & \\ & & & 1 \end{pmatrix} \,, & 
        \begin{pmatrix} 0 & & & \\ & 0 & & \\ & & 0 & \\ & & & 0 \end{pmatrix} 
    \\
    d=3 & 
        \begin{pmatrix} 0 & & & \\ & 0 & & \\ & & 0 & \\ & & & 1 \end{pmatrix} \,, & 
        \begin{pmatrix} 0 & & & \\ & 0 & & \\ & & 0 & \\ & & & 0 \end{pmatrix} \,, & 
        \begin{pmatrix} 0 & & & \\ & 0 & & \\ & & 0 & \\ & & & 0 \end{pmatrix} 
    \\
\end{matrix}
\]
for clusters $i=1,2,3$ corresponding to exposure duration $d=1,2,3$. Accordingly
\[
    \bm{\Delta}_{Z_i} = \left\{
        \begin{pmatrix} 0 & & & \\ & 1 & & \\ & & 1 & \\ & & & 1 \end{pmatrix} \,,
        \begin{pmatrix} 0 & & & \\ & 0 & & \\ & & 1 & \\ & & & 1 \end{pmatrix} \,,
        \begin{pmatrix} 0 & & & \\ & 0 & & \\ & & 0 & \\ & & & 1 \end{pmatrix} 
    \right\} \,,
\]
corresponding to clusters $i=1,2,3$, respectively.

Of note, $\bm{Y}_{ij}(d)$ for $d \geq j$ is not identifiable with the observed data. Here we still define them for notational convenience, but these quantities will not be evaluated throughout the proof. An alternative approach is to simply define $\bm{Y}_{ij}(d)=0$ (or an arbitrary constant quantity) for $d \geq j$.
We further define $\mathbf{I}_q \in \mathbb{R}^{q \times q}$ as the identity matrix for any positive integer $q$, and $\otimes$ as the Kronecker product operator.
With the above definitions, we observe in a SW-CRT that
\[
\begin{split}
    \bm{Y}_{i1} & = \bm{Y}_{i1}(0)\\
    \bm{Y}_{i2} &= I\{Z_i > 2\}\bm{Y}_{i2}(0) + I\{Z_i=2\}\bm{Y}_{i2}(1) \\
    \bm{Y}_{i3} &= I\{Z_i > 3\}\bm{Y}_{i3}(0) + I\{Z_i=3\}\bm{Y}_{i3}(1) + I\{Z_i=2\}\bm{Y}_{i3}(2) \\
    \vdots\\
    \bm{Y}_{iJ} &= I\{Z_i > J\}\bm{Y}_{iJ}(0) + I\{Z_i=J\}\bm{Y}_{iJ}(1) + I\{Z_i=J-1\}\bm{Y}_{iJ}(2) +...+ I\{Z_i=2\}\bm{Y}_{iJ}(J-1) \,.
\end{split}
\]
Altogether
\begin{equation}
\label{app.eq:PO}
\begin{split}
    \bm{Y}_i &= \left[\sum_{d=1}^{J-1}(\bm{\Lambda}_{Z_i}^d \otimes \textbf{I}_{N_i})\bm{Y}_i(d)\right] + \{(\textbf{I}_{J}-\bm{\Delta}_{Z_i}) \otimes \textbf{I}_{N_i}\} \bm{Y}_i(0) \\
    &= \left[\sum_{d=1}^{J-1}(\bm{\Lambda}_{Z_i}^d \otimes \textbf{I}_{N_i})\{\bm{Y}_i(d)-\bm{Y}_i(0)\}\right] + \{\textbf{I}_{J} \otimes \textbf{I}_{N_i}\}\bm{Y}_i(0) \\
    &= \left[\sum_{d=1}^{J-1}(\bm{\Lambda}_{Z_i}^d \otimes \textbf{I}_{N_i})\{\bm{Y}_i(d)-\bm{Y}_i(0)\}\right] + \bm{Y}_i(0)
\end{split}
\end{equation}
which can be interpreted as the observations' potential outcomes under the control ($\bm{Y}_{i}(0)$) added by the corresponding treatment effect ($\bm{Y}_{i}(d)-\bm{Y}_{i}(0))$, depending on the specified treatment effect structure.

For the proof's working linear fixed-effects model, we focus on the case of a within cluster-period exchangeable correlation structure, since the working independence assumption (ordinary least squares estimation) is just a special case with some variance parameters set to zero.

\subsection{Constant Treatment Effect}
\label{app.sect:constant}

We start by proving the results for the constant treatment effect $\Delta$. The linear fixed-effects model can be re-written as
\begin{equation}
    \label{app.eq:Y_constant}
    \bm{Y}_i = (
    \textbf{I}_{J}
    \otimes \bm{1}_{N_i})\bm{\beta}_0 + \left((\bm{\Delta}_{Z_i}\bm{1}_{J}) \otimes \bm{1}_{N_i}\right) \beta_Z + (\bm{1}_{J} \otimes \bm{X}_i) \bm{\beta_X} + \alpha_i \bm{1}_{N_iJ} + \bm{\gamma}_i \otimes \bm{1}_{N_i} + \bm{\epsilon}_i \,,
\end{equation}
where $\bm{\beta}_0 = (\beta_{01},...,\beta_{0J})^{\top}$ for the period effects (with $\beta_{01}=0$ for identifiability), $\bm{1}_q$ is a $q$-dimensional vector of ones, $\bm{X}_i=(\bm{X}_{i1},...,\bm{X}_{1N_i})^\top \in \mathbb{R}^{N_i \times p}$ indicating the $p$ covariates, $\alpha_i$ are the cluster fixed intercepts, $\bm{\gamma}_i=(\gamma_{i1},...,\gamma_{iJ})^\top$ for the cluster-period random effects inducing within-cluster-period correlation, and $\bm{\epsilon}_i = (\epsilon_{i11}, ..., \epsilon_{i1N_i}, ..., \epsilon_{iJ1}, ..., \epsilon_{iJN_i})^\top$ are the residuals.
Accordingly, we can rewrite the above equation as:
\begin{equation}
    \label{app.eq:Y_constant_Q}
    \bm{Y}_i = \bm{Q}_{0,i}\bm{\beta}_0 + \bm{Q}_{Z,i}\beta_Z + \bm{Q}_{X,i}\bm{\beta}_X + \bm{\mathcal{C}}_i + \bm{\gamma}_i \otimes \bm{1}_{N_i} + \bm{\epsilon}_i
\end{equation}
where $\bm{Q}_{0,i} = \textbf{I}_{J} \otimes \bm{1}_{N_i} \in \mathbb{R}^{N_iJ \times J}$, $\bm{Q}_{Z,i} = (\bm{\Delta}_{Z_i} \bm{1}_{J}) \otimes \bm{1}_{N_i} \in \mathbb{R}^{N_iJ \times 1}$, $\bm{Q}_{X,i} = \bm{1}_{J} \otimes \bm{X}_i \in \mathbb{R}^{N_iJ \times p}$, and $\bm{\mathcal{C}}_i=\alpha_i \bm{1}_{N_iJ} \in \mathbb{R}^{N_iJ \times 1}$, are the fixed indicators for period effects, constant treatment effect, time-varying covariates, and cluster intercepts, respectively.
Then, the working model becomes
\begin{equation}
    \label{app.eq:Y_joint}
    \bm{Y}_i|  Z_i, \bm{X}_i, \bm{\mathcal{C}}_i, N_i     \sim N(\bm{Q}_i\bm{\beta} + \bm{\mathcal{C}}_i, \bm{\Sigma}_i) \,,
\end{equation}
where $\bm{Q}_i=(\bm{Q}_{0,i}, \bm{Q}_{Z,i}, \bm{Q}_{X,i}) \in \mathbb{R}^{N_iJ \times (J+1+p)}$, $\bm{\beta}=(\bm{\beta}_0^{\top}, \beta_Z^{\top},\bm{\beta}_X^{\top})^{\top} \in \mathbb{R}^{J+1+p}$, and $\bm{\Sigma}_i=\kappa^2 \textbf{I}_{J} \otimes (\bm{1}_{N_i}\bm{1}_{N_i}^\top) + \sigma^2 \textbf{I}_{N_iJ} \in \mathbb{R}^{N_iJ \times N_iJ}$.
Notably, the above equation assumes there are $N_i$ individuals in each cluster, who can be repeatedly observed over $J$ time periods. Accordingly, the period indicator $\bm{Q}_{0,i}$ is dependent on the total cluster size $N_i$ (with the first $N_i$ entries corresponding to period $j=1$, and so on); therefore, Equation \ref{app.eq:Y_joint} is implicitly dependent on the period indicator.

\sloppy Recall that $N_i$ represents the source population size (may be unknown) and $N_{ij} = \sum_{k=1}^{N_i} S_{ijk}$, where $S_{ijk}=1$ indicates enrolled subjects.
Where $\bm{Y}_{ij} \in \mathbb{R}^{N_i}$, let $\bm{Y}_{ij}^o$ be the observed outcome vector for cluster $i$ in period $j$, and $\bm{D}_{ij} \in \mathbb{R}^{N_i \times N_{ij}}$ be a matrix, such that $\bm{Y}_{ij}^o = \bm{D}_{ij}^\top \bm{Y}_{ij} = (Y_{ijk})_{k:S_{ijk}=1} \in \mathbb{R}^{N_{ij}}$. Specifically, the $k$-th row of $\bm{D}_{ij}$ is a binary row vector with the $\tilde{k}$-th entry being 1 and the rest being 0, where $\tilde{k}$ is the $k$-th smallest index such that $S_{ijk}=1$.
To clarify, consider a simple example with $N_i=4$ and $S_{ijk}=1: \forall \, k \neq 2$, then we have
\[
\bm{Y}_{ij}^o = \bm{D}_{ij}^\top \bm{Y}_{ij}
= \begin{bmatrix}
    1 & 0 & 0 & 0 \\
    0 & 0 & 1 & 0 \\
    0 & 0 & 0 & 1 \\
\end{bmatrix}
\begin{bmatrix}
    Y_{ij1} \\ Y_{ij2} \\ Y_{ij3} \\ Y_{ij4}
\end{bmatrix}
= \begin{bmatrix}
    Y_{ij1} \\ Y_{ij3} \\ Y_{ij4}
\end{bmatrix} \,.
\]
That is, $\bm{D}_{ij}$ is a deterministic function of $(S_{ij1},...,S_{ijN_i}, N_i)$. Let $\mathcal{N}_i =\sum_{j=1}^{J}N_{ij} =\sum_{j=1}^J\sum_{k=1}^{N_i}S_{ijk}$ be the total number of enrolled subjects in cluster $i$, $\bm{S}_i = (S_{i11},...,S_{iJN_i})^{\top} \in \mathbb{R}^{N_iJ}$, and 
$\bm{D}_i = bdiag\{D_{ij} : j=1,...,J\} \in \mathbb{R}^{N_i J \times \mathcal{N}_i}$ denote the block diagonal matrix of $\bm{D}_{ij}$, and $\bm{Y}_i^o = (\bm{Y}_{i1}^o,...,\bm{Y}_{iJ}^o)^{\top} \in \mathbb{R}^{\mathcal{N}_i}$. Then, we have $\bm{Y}_i^o = \bm{D}_i^\top \bm{Y}_i$, $\bm{D}_i^\top \bm{1}_{N_i J}=\bm{1}_{\mathcal{N}_i}$, $\bm{D}_i \bm{1}_{\mathcal{N}_i} = \bm{S}_i$, $\bm{D}_i\bm{D}_i^\top = diag\{\bm{S}_i\}$, and $\bm{D}_i^\top \bm{D}_i = \textbf{I}_{\mathcal{N}_i}$.
Finally, the observed outcome follows
\begin{equation}
    \label{app.eq:Y_joint_enrolled}
    \bm{Y}_i^o | Z_i, \bm{X}_i^o, \bm{\mathcal{C}}_i^o, N_i, \bm{S}_i 
    \sim N(\bm{D}_i^\top \bm{Q}_i\bm{\beta} + \bm{D}_i^\top \bm{\mathcal{C}}_i, \bm{D}_i^\top\bm{\Sigma}_i\bm{D}_i)
\end{equation}
where $\bm{X}_i^o=\bm{D}_i^\top(\bm{1}_{J} \otimes \bm{X}_i)$ is the observed covariate matrix for cluster $i$ across periods.

Importantly, as the number of clusters $m \rightarrow \infty$, the number of parameters $\alpha_i$ (also denoted as $\bm{\mathcal{C}}_i$) in Equations \ref{app.eq:Y_constant_Q} \& \ref{app.eq:Y_joint} will also $\rightarrow \infty$. This is an example of the ``incidental parameters problem'' described in Neyman \& Scott \citep{neyman_consistent_1948} and violates regularity condition 3 in Lemma \ref{app.lemma:variance}.
Importantly, the linear fixed-effects model with cluster dummy variables and potential within-cluster, between-period correlation, yields an equivalent estimator to the ``within-transformed'' estimator in this linear setting \citep{kiefer_estimation_1980,wooldridge_econometric_2010}.
Conveniently, we can then avoid the incidental parameters problem and remove the incidental parameters $\bm{\mathcal{C}}_i$ from the resulting linear fixed-effects likelihood and score function by using this ``within-transformation''.

To derive the ``within-transformed'' estimator assuming full enrollment ($S_{ijk}=1 \, \forall \, i,j,k$), the within-cluster averages can first be derived from Equation (\ref{app.eq:Y_constant_Q}) as
\begin{equation}
\label{app.eq:barY_joint}
    \overline{Y}_i = \overline{\bm{Q}}_{0,i}\beta_0 + \overline{Q}_{Z,i}\beta_Z + \overline{\bm{Q}}_{X,i}\bm{\beta}_{X} + \alpha_i + \overline{\gamma}_i + \overline{\epsilon}_i \,,
\end{equation}
where
$\overline{Y}_i = \frac{1}{N_iJ} \bm{1}_{N_iJ}^\top \bm{Y}_i \in \mathbb{R}^{1}$,
$\overline{\bm{Q}}_{0,i} = \frac{1}{N_iJ} \bm{1}_{N_iJ}^\top \bm{Q}_{0,i} = \frac{1}{J} \bm{1}_J^\top \in \mathbb{R}^{1 \times J}$,
$\overline{Q}_{Z,i} = \frac{1}{N_iJ} \bm{1}_{N_iJ}^\top \bm{Q}_{Z,i} = \frac{1}{J} \bm{1}_{J}^\top \bm{\Delta}_{Z_i} \bm{1}_{J} \in \mathbb{R}^{1}$,
$\overline{\bm{Q}}_{X,i} = \frac{1}{N_iJ} \bm{1}_{N_iJ}^\top \bm{Q}_{X,i} \in \mathbb{R}^{1 \times p}$,
$\overline{\gamma}_i = \frac{1}{J}\bm{1}_{J}^\top \bm{\gamma}_i \in \mathbb{R}^1$,
and $\overline{\epsilon}_i = \frac{1}{N_iJ} \bm{1}_{N_iJ}^\top \bm{\epsilon}_i \in \mathbb{R}^1$.
Multiplying the above equation by $\bm{1}_{N_iJ}$ and subtracting it from Equation (\ref{app.eq:Y_constant_Q}), allows us to rewrite the linear mixed model as the following ``within-transformed'' equation
\begin{equation}
\label{app.eq:ddotY}
    \ddot{\bm{Y}}_i = \ddot{\bm{Q}}_{0,i}\bm{\beta}_0 + \ddot{\bm{Q}}_{Z,i}\beta_Z + \ddot{\bm{Q}}_{X,i}\bm{\beta}_{X} + \ddot{\bm{\gamma}}_i \otimes \bm{1}_{N_i} + \ddot{\bm{\epsilon}}_i \,
\end{equation}
where $\ddot{\bm{Y}}=\bm{Y}_i - \overline{Y}_i \otimes \bm{1}_{N_iJ}$,
$\ddot{\bm{Q}}_{0,i}= \bm{Q}_{0,i} - \overline{\bm{Q}}_{0,i} \otimes \bm{1}_{N_iJ}$, $\ddot{\bm{Q}}_{Z,i} = \bm{Q}_{Z,i} - \overline{Q}_{Z,i} \otimes \bm{1}_{N_iJ}$, $\ddot{\bm{Q}}_{X,i} = \bm{Q}_{X,i} - \overline{\bm{Q}}_{X,i} \otimes \bm{1}_{N_iJ}$,
$\ddot{\bm{\gamma}}_i = \bm{\gamma}_i - \overline{\gamma}_i \otimes \bm{1}_{N_iJ}$,
and $\ddot{\bm{\epsilon}}_i = \bm{\epsilon}_i - \overline{\epsilon}_i \otimes \bm{1}_{N_iJ}$. 
Equation (\ref{app.eq:ddotY}) is commonly referred to as the ``within-transformed Fixed-Effects'' estimator \citep{kiefer_estimation_1980,cameron_microeconometrics_2005,wooldridge_econometric_2010}.
To reiterate, the equivalence between the within-transformed specification (Equation \ref{app.eq:transfY_joint}) and previously defined linear fixed-effects model (Equation \ref{app.eq:Y_joint}) results from the \textit{orthogonality} of residuals $\bm{Y}_i-(\bm{Q}_i\bm{\beta} + \bm{\mathcal{C}}_i)$ and model covariates $(\bm{Q}_i, \bm{\mathcal{C}}_i)$ in this linear setting, such that 
$(\bm{Q}_i, \bm{C}_i)^\top \{\bm{Y}_i-(\bm{Q}_i\bm{\beta} + \bm{\mathcal{C}}_i)\} = 0$
, allowing us to project out the cluster covariates $\bm{\mathcal{C}}_i$ \citep{wooldridge_econometric_2010}. 

In practice, the within-transformation is applied using only the enrolled subjects and can be implemented via matrix multiplication with the following symmetric matrix 
\begin{equation}
\label{app.eq:within-transformation-mat}
    \bm{\mathcal{M}}_i \equiv \textbf{I}_{N_iJ} - \bm{S}_i (\bm{S}_i^\top \bm{S}_i)^{-1} \bm{S}_i^\top = \textbf{I}_{N_iJ} - \frac{1}{\mathcal{N}_i}\bm{S}_i \bm{S}_i^\top
\end{equation}
where recall $\mathcal{N}_i=\sum_{j=1}^JN_{ij} = \sum_{j=1}^J \sum_{k=1}^{N_i}S_{ijk}$ being the total number of enrolled subjects in cluster $i$, and $\bm{S}_i =(S_{i11},...,S_{iJN_i})^\top$.
If all individuals in cluster $i$ are enrolled ($S_{ijk}=1 \, \forall \, i,j,k$), then $\bm{\mathcal{M}}_i \equiv  \textbf{I}_{N_iJ} - \bm{1}_{N_iJ} (\bm{1}_{N_iJ}^\top \bm{1}_{N_iJ})^{-1} \bm{1}_{N_iJ}^\top = \textbf{I}_{N_iJ} - \frac{1}{N_iJ}\bm{1}_{N_iJ}\bm{1}_{N_iJ}^\top$ \citep{kiefer_estimation_1980,wooldridge_econometric_2010}.
The linear fixed-effects ``within-transformed'' working model is then
\begin{equation}
    \label{app.eq:ddotY_joint}
    \ddot{\bm{Y}}_i| Z_i, \bm{X}_i, N_i, \bm{S}_i \sim N(\ddot{\bm{Q}}_i\bm{\beta}, \ddot{\bm{\Sigma}}_i) \,,
\end{equation}
where $\ddot{\bm{Q}}_i=(\ddot{\bm{Q}}_{0,i}, \ddot{\bm{Q}}_{Z,i}, \ddot{\bm{Q}}_{X,i}) = (\bm{\mathcal{M}}_i \bm{Q}_{0,i}, \bm{\mathcal{M}}_i \bm{Q}_{Z,i}, \bm{\mathcal{M}}_i \bm{Q}_{X,i}) \in \mathbb{R}^{N_iJ \times (J+1+p)}$, $\bm{\beta}=(\bm{\beta}_0^\top, \beta_Z^\top,\bm{\beta}_{X}^\top)^\top$, and $\ddot{\bm{\Sigma}}_i= \bm{\mathcal{M}}_i \bm{\Sigma}_i \bm{\mathcal{M}}_i$ \citep{kiefer_estimation_1980,wooldridge_econometric_2010}, or equivalently
\begin{equation}
    \label{app.eq:transfY_joint}
    \bm{\mathcal{M}}_i \bm{Y}_i| Z_i, \bm{X}_i, N_i, \bm{S}_i \sim N(\bm{\mathcal{M}}_i \bm{Q}_i\bm{\beta}, \bm{\mathcal{M}}_i \bm{\Sigma}_i \bm{\mathcal{M}}_i) \,.
\end{equation}
Altogether, the observed fixed-effects within-transformed outcome $\ddot{\bm{Y}}_i^o = \bm{D}_i^\top \bm{\mathcal{M}}_i \bm{Y}_i$ follows 
\begin{equation}
    \label{app.eq:ddotY_joint_enrolled}
    \ddot{\bm{Y}}_i^o | Z_i, \bm{X}_i^o, N_i, \bm{S}_i 
    \sim N(\bm{D}_i^\top \ddot{\bm{Q}}_i\bm{\beta}, \bm{D}_i^\top\ddot{\bm{\Sigma}}_i\bm{D}_i) \,,
\end{equation}
or equivalently
\begin{equation}
    \label{app.eq:transfY_joint_enrolled}
    \bm{D}_i^\top \bm{\mathcal{M}}_i \bm{Y}_i | Z_i, \bm{X}_i^o, N_i, \bm{S}_i  \sim N(\bm{D}_i^\top \bm{\mathcal{M}}_i \bm{Q}_i\bm{\beta}, \bm{D}_i^\top\bm{\mathcal{M}}_i \bm{\Sigma}_i\bm{\mathcal{M}}_i \bm{D}_i) \,.
\end{equation}

Denote $\bm{\theta}=(\bm{\beta}, \kappa^2, \sigma^2)^{\top}$ as the vector of unknown parameters.
Based on the observed data, the log-likelihood function given $\{Z_i, \bm{X}_i^{o}, N_i, \bm{S}_i\}$ is
\[
\begin{split}
    & l(\bm{\theta}; \{\ddot{\bm{Y}}_i^o\}_{i=1}^{m} | \{Z_i, \bm{X}_i^o,  N_i, \bm{S}_i\}_{i=1}^{m}) \\
    & = C - \frac{1}{2} \left( \sum_{i=1}^{m} log(|\bm{D}_i^\top \ddot{\bm{\Sigma}}_i \bm{D}_i|)\right) + (\ddot{\bm{Y}}_i^o - \bm{D}_i^\top \ddot{\bm{Q}}_i\bm{\beta})^\top (\bm{D}_i^\top \ddot{\bm{\Sigma}}_i \bm{D}_i)^{-1} (\ddot{\bm{Y}}_i^o - \bm{D}_i^\top \ddot{\bm{Q}}_i\bm{\beta}) \\
    & = C - \frac{1}{2} \left(\sum_{i=1}^{m} log(|\bm{D}_i^\top \ddot{\bm{\Sigma}}_i \bm{D}_i|)\right) + (\ddot{\bm{Y}}_i - \ddot{\bm{Q}}_i\bm{\beta})^\top \bm{D}_i (\bm{D}_i^\top \ddot{\bm{\Sigma}}_i \bm{D}_i)^{-1} \bm{D}_i^\top (\ddot{\bm{Y}}_i - \ddot{\bm{Q}}_i\bm{\beta})
\end{split}
\]
where $C$ is a constant independent of  parameters $\bm{\theta}$.
The score function, being the derivative of the log-likelihood function, is then
\[
\begin{split}
    & \frac{l(\bm{\theta}; \{\ddot{\bm{Y}}_i^o\}_{i=1}^{m} | \{Z_i, \bm{X}_i^o, N_i, \bm{S}_i\}_{i=1}^{m})}{d\bm{\theta}} \\
    & = -\sum_{i=1}^{m}
        \left(
        \begin{gathered}
            2\ddot{\bm{Q}}_i^\top \ddot{\bm{V}}_i (\ddot{\bm{Y}}_i - \ddot{\bm{Q}}_i\bm{\beta}) \\
            -\bm{1}_{N_i}^\top \ddot{\bm{V}}_i \bm{1}_{N_i} + (\ddot{\bm{Y}}_i - \ddot{\bm{Q}}_i\bm{\beta})^\top \ddot{\bm{V}}_i (\bm{1}_{N_iJ} \bm{1}_{N_iJ}^\top) \ddot{\bm{V}}_i (\ddot{\bm{Y}}_i - \ddot{\bm{Q}}_i\bm{\beta}) \\
            -tr(\ddot{\bm{V}}_i \textbf{I}_J \otimes (\bm{1}_{N_i}\bm{1}_{N_i}^\top)) + (\ddot{\bm{Y}}_i - \ddot{\bm{Q}}_i\bm{\beta})^\top \ddot{\bm{V}}_i (\textbf{I}_J \otimes (\bm{1}_{N_i} \bm{1}_{N_i}^\top)) \ddot{\bm{V}}_i (\ddot{\bm{Y}}_i - \ddot{\bm{Q}}_i\bm{\beta})
        \end{gathered}
        \right)
\end{split}
\]
\citep{wang_how_2024} where $\ddot{\bm{V}}_i = \bm{D}_i (\bm{D}_i^\top \ddot{\bm{\Sigma}}_i \bm{D}_i)^{-1} \bm{D}_i^\top \in \mathbb{R}^{N_iJ \times N_iJ}$ and $tr(\ddot{\bm{V}}_i)$ is the trace of $\ddot{\bm{V}}_i$.
We hence define the estimating function as
\begin{equation}
    \bm{\psi}(\bm{O};\bm{\theta}) =  \left(
        \begin{gathered}
            2\ddot{\bm{Q}}^\top \ddot{\bm{V}} (\ddot{\bm{Y}} - \ddot{\bm{Q}}\bm{\beta}) \\
            -\bm{1}_{N}^\top \ddot{\bm{V}} \bm{1}_{N} + (\ddot{\bm{Y}} - \ddot{\bm{Q}}\bm{\beta})^\top \ddot{\bm{V}} (\bm{1}_{NJ} \bm{1}_{NJ}^\top) \ddot{\bm{V}} (\ddot{\bm{Y}} - \ddot{\bm{Q}}_i\bm{\beta}) \\
            -tr(\ddot{\bm{V}} \textbf{I}_{J} \otimes (\bm{1}_{N}\bm{1}_{N}^\top)) + (\ddot{\bm{Y}} - \ddot{\bm{Q}}\bm{\beta})^\top \ddot{\bm{V}} (\textbf{I}_{J} \otimes (\bm{1}_{N} \bm{1}_{N}^\top)) \ddot{\bm{V}}(\ddot{\bm{Y}} - \ddot{\bm{Q}}\bm{\beta})
        \end{gathered}
        \right)
\end{equation}
for the within-transformed fixed-effects estimator. Recall that subscript $i$ is omitted when taking the expectation with respect to distribution $\mathcal{P}$ (Section \ref{app.sect:Theorems_Lemmas}).
The maximum likelihood estimator for $\bm{\theta}$ is defined as a solution to the estimating equation
\[
    \sum_{i=1}^{m} \bm{\psi}(\bm{O}_i;\bm{\theta})=0
\]

For the estimating equation $\bm{\psi}$, we prove the convergence and asymptotic normality of $\hat{\bm{\theta}}$ by applying Lemma \ref{app.lemma:variance}. In Lemma \ref{app.lemma:variance}, the conditions for $\bm{\psi}$ are assumed in the main paper as regularity conditions, which implies the desired results. In addition, the consistency of the sandwich variance estimators is also implied.
As previously mentioned, the regularity conditions of Lemma \ref{app.lemma:variance} still hold for the fixed-effects model, despite concerns of an incidental parameters problem.

Denoting $\underline{\bm{\theta}}=(\underline{\bm{\beta}}_0, \underline{\beta}_Z, \underline{\bm{\beta}}_{X}^\top, \underline{\kappa}^2, \underline{\sigma^2})^{\top}$
as the solution to $E[\bm{\psi}(\bm{O};\bm{\theta})]=0$, we next prove $\underline{\beta}_Z = \Delta$, which will imply the robustness of $\hat{\beta}_Z$ and completes the proof under the constant treatment effect structure specification.
To proceed, the first two groups of entries of $E[\bm{\psi}(\bm{O};\bm{\theta})]=0$ (corresponding to the period effect $\bm{\beta}_0$ and constant treatment effect $\beta_Z$) are
\begin{equation}
\label{app.eq:period}
    E[\ddot{\bm{Q}}_0^\top \underline{\ddot{\bm{V}}}(\ddot{\bm{Y}} - \ddot{\bm{Q}}\underline{\bm{\beta}})] = 0
\end{equation}
\begin{equation}
\label{app.eq:treatment}
     E[\ddot{\bm{Q}}_Z^\top \underline{\ddot{\bm{V}}}(\ddot{\bm{Y}} - \ddot{\bm{Q}}\underline{\bm{\beta}})] = 0
\end{equation}
corresponding to the score functions for the period and constant treatment effects, respectively, where$\underline{\ddot{\bm{V}}}$ is equal to $\ddot{\bm{V}}$ with $(\kappa^2, \sigma^2)$ replaced by $(\underline{\kappa}^2, \underline{\sigma}^2)$.

\subsubsection{Proof of Fixed-Effects model constant treatment effect consistency}

With Equation (\ref{app.eq:period})
\[
    E[\ddot{\bm{Q}}_0^\top \underline{\ddot{\bm{V}}}(\ddot{\bm{Y}} - \ddot{\bm{Q}}\underline{\bm{\beta}})] =  E[(\bm{\mathcal{M}}\bm{Q}_0)^\top \underline{\ddot{\bm{V}}}(\ddot{\bm{Y}} - \ddot{\bm{Q}}\underline{\bm{\beta}})] = E[\bm{Q}_0^\top \bm{\mathcal{M}} \underline{\ddot{\bm{V}}}(\ddot{\bm{Y}} - \ddot{\bm{Q}}\underline{\bm{\beta}})] = 0
\]
and Equation (\ref{app.eq:treatment})
\[
    E[\ddot{\bm{Q}}_Z^\top \underline{\ddot{\bm{V}}}(\ddot{\bm{Y}} - \ddot{\bm{Q}}\underline{\bm{\beta}})] =  E[(\bm{\mathcal{M}}\bm{Q}_Z)^\top \underline{\ddot{\bm{V}}}(\ddot{\bm{Y}} - \ddot{\bm{Q}}\underline{\bm{\beta}})] = E[\bm{Q}_Z^\top \bm{\mathcal{M}} \underline{\ddot{\bm{V}}}(\ddot{\bm{Y}} - \ddot{\bm{Q}}\underline{\bm{\beta}})]= 0 \,,
\]
respectively.
Define $E[\bm{\Delta}_Z \bm{1}_{J}]=(\pi_1^s,...,\pi_J^s)^\top \in \mathbb{R}^{J}$, where $\pi_j^s=\sum_{j'=1}^{j}\pi_{j'}$, representing the between-cluster, within-period averages of the constant treatment effect indicator.
By left-multiplying Equation (\ref{app.eq:period}) by $-E[\bm{\Delta}_Z \bm{1}_{J}]^\top$ (which would still =0) and adding it to Equation (\ref{app.eq:treatment}), we get
\begin{equation}
\label{app.eq:period_demean}
\begin{split}
    E&[\{\ddot{\bm{Q}}_Z-\ddot{\bm{Q}}_0E[\bm{\Delta}_Z \bm{1}_{J}]\}^\top \underline{\ddot{\bm{V}}} (\ddot{\bm{Y}}-\ddot{\bm{Q}}\underline{\bm{\beta}})] \\
    &=E[\{\bm{\mathcal{M}}(\bm{Q}_Z-\bm{Q}_0E[\bm{\Delta}_Z \bm{1}_{J}])\}^\top \underline{\ddot{\bm{V}}} \bm{\mathcal{M}} (\bm{Y}-\bm{Q}\underline{\bm{\beta}})] \\
    &= E\left[\left\{(\bm{\Delta}_Z \bm{1}_{J} - E[\bm{\Delta}_Z \bm{1}_{J}]) \otimes \bm{1}_N \right\}^\top \bm{\mathcal{M}}  \underline{\ddot{\bm{V}}} \bm{\mathcal{M}} (\bm{Y} - \bm{Q}\underline{\bm{\beta}})\right] \\ 
    &= 0 
\end{split}
\end{equation}
As we will further demonstrate below, doing so will help to remove the contribution of the period effects from the treatment effect score function (Equation \ref{app.eq:treatment}), and isolate the contribution of the only treatment effect.

To clarify, we provide an illustrative example from a 3 cluster, 4 period SW-CRT:
\[
\bm{\Delta}_{Z_i} \bm{1}_J = \left\{ 
    \begin{bmatrix}
        0 \\
        1 \\
        1 \\
        1
    \end{bmatrix},
    \begin{bmatrix}
        0 \\
        0 \\
        1 \\
        1
    \end{bmatrix},    
    \begin{bmatrix}
        0 \\
        0 \\
        0 \\
        1
    \end{bmatrix}
\right\};
\frac{1}{m}\sum_{i=1}^{m}(\bm{\Delta}_{Z_i} \bm{1}_J) =
    \begin{bmatrix}
        0 \\
        1/3 \\
        2/3 \\
        3/3
    \end{bmatrix} ;
\]
\[
\bm{\Delta}_{Z_i} \bm{1}_J - \frac{1}{m}\sum_{i=1}^{m}(\bm{\Delta}_{Z_i} \bm{1}_J) = \left\{ 
    \begin{bmatrix}
        0 \\
        2/3 \\
        1/3 \\
        0
    \end{bmatrix},
    \begin{bmatrix}
        0 \\
        -1/3 \\
        1/3 \\
        0
    \end{bmatrix},    
    \begin{bmatrix}
        0 \\
        -1/3 \\
        -2/3 \\
        0
    \end{bmatrix}
\right\}
\]
with entries corresponding to clusters $i=1,2,3$, respectively.

We can then subtract the following components from Equation \ref{app.eq:period_demean}.
\[
\begin{split}
    E&\left[\left\{(\bm{\Delta}_Z \bm{1}_{J} - E[\bm{\Delta}_Z \bm{1}_{J}]) \otimes \bm{1}_N \right\}^\top \bm{\mathcal{M}}  \underline{\ddot{\bm{V}}} \bm{\mathcal{M}} \bm{Q}_0\right] \\
    &= E\left[(\bm{\Delta}_Z \bm{1}_{J} - E[\bm{\Delta}_Z \bm{1}_{J}])^\top (\textbf{I}_J \otimes \bm{1}_N)^\top \bm{\mathcal{M}}  \underline{\ddot{\bm{V}}} \bm{\mathcal{M}} \bm{Q}_0\right] \\
    &= E\left[E\left[(\bm{\Delta}_Z \bm{1}_{J} - E[\bm{\Delta}_Z \bm{1}_{J}])^\top (\textbf{I}_J \otimes \bm{1}_N)^\top \bm{\mathcal{M}}  \underline{\ddot{\bm{V}}} \bm{\mathcal{M}} \bm{Q}_0 | N\right]\right] \\
\text{(A2} &\text{, Non-informative enrollment)}\\
    &= E\left[E\left[\bm{\Delta}_Z \bm{1}_{J} - E[\bm{\Delta}_Z \bm{1}_{J}] | N\right]^\top E\left[(\textbf{I}_J \otimes \bm{1}_N)^\top \bm{\mathcal{M}}  \underline{\ddot{\bm{V}}} \bm{\mathcal{M}} \bm{Q}_0 | N\right]\right] \\
\text{(A2} &\text{, Non-informative enrollment)}\\
    &= E\left[E\left[\bm{\Delta}_Z \bm{1}_{J} - E[\bm{\Delta}_Z \bm{1}_{J}] | N\right]^\top E\left[(\textbf{I}_J \otimes \bm{1}_N)^\top \bm{\mathcal{M}}  \underline{\ddot{\bm{V}}} \bm{\mathcal{M}} \bm{Q}_0\right]\right] \\
    &= E\left[E\left[\bm{\Delta}_Z \bm{1}_{J} - E[\bm{\Delta}_Z \bm{1}_{J}]\right]^\top E\left[(\textbf{I}_J \otimes \bm{1}_N)^\top \bm{\mathcal{M}}  \underline{\ddot{\bm{V}}} \bm{\mathcal{M}} \bm{Q}_0\right]\right] \\
    &= 0
\end{split}
\]
removing the influence of the period effects.
In general, we can state $\left\{(\bm{\Delta}_Z \bm{1}_{J} - E[\bm{\Delta}_Z \bm{1}_{J}]) \otimes \bm{1}_N \right\}^\top = (\bm{\Delta}_Z \bm{1}_{J} - E[\bm{\Delta}_Z \bm{1}_{J}])^\top (\textbf{I}_J \otimes \bm{1}_N)^\top$, producing the above first equality.
Then, $\bm{\Delta}_Z$ is a function of $Z$ and $(\textbf{I}_J \otimes \bm{1}_N)^\top\bm{\mathcal{M}}  \underline{\ddot{\bm{V}}} \bm{\mathcal{M}} \bm{Q}_0$ is a function of the working variance components and $N_{.j}$ 
Therefore, the third equality results from $N_{.j} \perp Z | N$ (A2), and the fourth equality results from $N_{.j} \perp N$ (A2).

Furthermore
\[
\begin{split}
    E&\left[\left\{(\bm{\Delta}_Z \bm{1}_{J} - E[\bm{\Delta}_Z \bm{1}_{J}]) \otimes \bm{1}_N \right\}^\top \bm{\mathcal{M}}  \underline{\ddot{\bm{V}}} \bm{\mathcal{M}} \bm{Q}_X \right] \\
    &= E\left[(\bm{\Delta}_Z \bm{1}_{J} - E[\bm{\Delta}_Z \bm{1}_{J}])^\top (\textbf{I}_J \otimes \bm{1}_N)^\top \bm{\mathcal{M}}  \bm{D} (\bm{D}^\top \underline{\ddot{\bm{\Sigma}}} \bm{D})^{-1} \bm{D}^\top \ddot{\bm{Q}}_X \right] \\
    &= E\left[E\left[(\bm{\Delta}_Z \bm{1}_{J} - E[\bm{\Delta}_Z \bm{1}_{J}])^\top (\textbf{I}_J \otimes \bm{1}_N)^\top \bm{\mathcal{M}}  \bm{D} (\bm{D}^\top \underline{\ddot{\bm{\Sigma}}} \bm{D})^{-1} \bm{D}^\top \ddot{\bm{Q}}_X |\bm{S},N,Z\right] \right] \\
    &= E\left[(\bm{\Delta}_Z \bm{1}_{J} - E[\bm{\Delta}_Z \bm{1}_{J}])^\top (\textbf{I}_J \otimes \bm{1}_N)^\top \bm{\mathcal{M}}  \bm{D} (\bm{D}^\top \underline{\ddot{\bm{\Sigma}}} \bm{D})^{-1} \bm{D}^\top E\left[ \ddot{\bm{Q}}_X |\bm{S},N,Z\right] \right] \\
\text{(A1} &\text{, Super-population sampling; A2, Non-informative enrollment)} \\
    &= 0
\end{split}
\]
where $(\bm{\Delta}_Z \bm{1}_{J} - E[\bm{\Delta}_Z \bm{1}_{J}])^\top (\textbf{I}_J \otimes \bm{1}_N)^\top \bm{\mathcal{M}}  \bm{D} (\bm{D}^\top \underline{\ddot{\bm{\Sigma}}} \bm{D})^{-1} \bm{D}^\top$ is a function of $\bm{S},N,Z$, and therefore deterministic leading to the third equality.
The fourth equality arises from
\[
\begin{split}
    E\left[\ddot{\bm{Q}}_X |\bm{S},N,Z\right]
     &= E\left[ 
     \left(\bm{X}_{.k} - \frac{\sum_{l=1}^{J} \sum_{k'=1}^{N} S_{.lk'} \bm{X}_{.k}}{\sum_{l=1}^{J} \sum_{k'=1}^{N} S_{.lk'}}
     \right)_{j=1,...,J;k=1,...,N}  | \bm{S},N,Z\right] 
     = \bm{0}_{NJ} \,.
\end{split}
\]
This result can also straightforwardly proceed by relying on randomization (A3) in longitudinal CRTs.

Finally, we can also show
\[
\begin{split}
    E&\left[\left\{(\bm{\Delta}_Z \bm{1}_{J} - E[\bm{\Delta}_Z \bm{1}_{J}]) \otimes \bm{1}_N \right\}^\top \bm{\mathcal{M}}  \underline{\ddot{\bm{V}}} \bm{\mathcal{M}} (\textbf{I}_J \otimes \textbf{I}_N)\bm{Y}(0) \right] \\
    &= E\left[(\bm{\Delta}_Z \bm{1}_{J} - E[\bm{\Delta}_Z \bm{1}_{J}])^\top (\textbf{I}_J \otimes \bm{1}_N)^\top \bm{\mathcal{M}}  \underline{\ddot{\bm{V}}} \bm{\mathcal{M}} \bm{Y}(0) \right] \\
    &= E\left[E\left[(\bm{\Delta}_Z \bm{1}_{J} - E[\bm{\Delta}_Z \bm{1}_{J}])^\top (\textbf{I}_J \otimes \bm{1}_N)^\top \bm{\mathcal{M}}  \underline{\ddot{\bm{V}}} \bm{\mathcal{M}} \bm{Y}(0) |\bm{S},N,Z\right] \right] \\
    &= E\left[(\bm{\Delta}_Z \bm{1}_{J} - E[\bm{\Delta}_Z \bm{1}_{J}])^\top (\textbf{I}_J \otimes \bm{1}_N)^\top \bm{\mathcal{M}}  \underline{\ddot{\bm{V}}} \bm{\mathcal{M}} E\left[ \bm{Y}(0) |\bm{S},N,Z\right] \right] \\
\text{(A1} &\text{, Super-population sampling; A2, Non-informative enrollment)} \\
    &= E\left[(\bm{\Delta}_Z \bm{1}_{J} - E[\bm{\Delta}_Z \bm{1}_{J}])^\top (\textbf{I}_J \otimes \bm{1}_N)^\top \bm{\mathcal{M}}  \underline{\ddot{\bm{V}}} \bm{\mathcal{M}} (\textbf{I}_J \otimes \bm{1}_N)  \bm{v}(0, N, Z) \right]  \\
    &= E\left[(\bm{\Delta}_Z \bm{1}_{J} - E[\bm{\Delta}_Z \bm{1}_{J}])^\top (\textbf{I}_J \otimes \bm{1}_N)^\top \bm{\mathcal{M}}  \underline{\ddot{\bm{V}}} (\textbf{I}_J \otimes \bm{1}_N) \ddot{\bm{v}}(0,\bm{S},N,Z) \right]  \\
\text{(A4} &\text{, Mean independence)} \\
    &= E\left[(\bm{\Delta}_Z \bm{1}_{J} - E[\bm{\Delta}_Z \bm{1}_{J}])^\top (\textbf{I}_J \otimes \bm{1}_N)^\top \bm{\mathcal{M}}  \underline{\ddot{\bm{V}}} (\textbf{I}_J \otimes \bm{1}_N) \ddot{\bm{v}}(0,\bm{S}) \right]  \\
\text{(A2} &\text{, Non-informative enrollment)} \\
    &= E\left[E\left[\bm{\Delta}_Z \bm{1}_{J} - E[\bm{\Delta}_Z \bm{1}_{J}]|N \right]^\top E\left[(\textbf{I}_J \otimes \bm{1}_N)^\top \bm{\mathcal{M}}  \underline{\ddot{\bm{V}}} (\textbf{I}_J \otimes \bm{1}_N) \ddot{\bm{v}}(0,\bm{S}) |N \right] \right]  \\
\text{(A2} &\text{, Non-informative enrollment)} \\
    &= E\left[\bm{\Delta}_Z \bm{1}_{J} - E[\bm{\Delta}_Z \bm{1}_{J}] \right]^\top E\left[(\textbf{I}_J \otimes \bm{1}_N)^\top \bm{\mathcal{M}}  \underline{\ddot{\bm{V}}} (\textbf{I}_J \otimes \bm{1}_N) \ddot{\bm{v}}(0,\bm{S}) \right]  \\
    &= 0 \,.
\end{split}
\]
As previously stated, $(\bm{\Delta}_Z \bm{1}_{J} - E[\bm{\Delta}_Z \bm{1}_{J}])^\top (\textbf{I}_J \otimes \bm{1}_N)^\top \bm{\mathcal{M}}  \underline{\ddot{\bm{V}}} \bm{\mathcal{M}}$ is a function of $\bm{S},N,Z$ and therefore deterministic in the third equality.
With A1, we can rewrite $E\left[\bm{Y}(0) |\bm{S},N,Z\right] = (\textbf{I}_J \otimes \bm{1}_N) \bm{v}(0,\bm{S},N,Z)$ where $\bm{v}(0,\bm{S},N,Z) = (v_j(0))_{j=1,...,J} \in \mathbb{R}^J$ and $E[\bm{Y}_{.j}(0)|\bm{S},N,Z] = v_j(0,\bm{S},N,Z)\bm{1}_N \in \mathbb{R}^N$. Then with A2, $v_j(0,\bm{S},N,Z)=v_j(0,N,Z)$, leading to the fourth equality.
Similarly, with A1 and A2, $E[\ddot{\bm{Y}}_{ij}(0)|\bm{S},N,Z] = \ddot{v}_j(0,\bm{S},N,Z)\bm{1}_N = \left(v_j(0,N,Z) - [\sum_{l=1}^J N_{.l} v_l(0,N,Z)]/[\sum_{l=1}^J N_{.l}]\right)\bm{1}_N \in \mathbb{R}^N$ that depends on $\bm{S}$ through $N_{ij}$, and we can define $\ddot{\bm{v}}(0,\bm{S},N,Z) = (\ddot{v}_j(0,\bm{S},N,Z))_{j=1,...,J} \in \mathbb{R}^J$. Then, with $\underline{\ddot{\bm{V}}}$ projecting onto the enrolled subspace, $\bm{\mathcal{M}}  \underline{\ddot{\bm{V}}} \bm{\mathcal{M}} (\textbf{I}_J \otimes \bm{1}_N) \bm{v}(0,N,Z) = \bm{\mathcal{M}}  \underline{\ddot{\bm{V}}} (\textbf{I}_J \otimes \bm{1}_N) \ddot{\bm{v}}(0,\bm{S},N,Z)$ in the fifth equality.
The sixth equality is a direct application of A4.
Finally, where $(\textbf{I}_J \otimes \bm{1}_N)^\top \bm{\mathcal{M}}  \underline{\ddot{\bm{V}}} (\textbf{I}_J \otimes \bm{1}_N)$ is a function of the working variance components and $N_{.j}$, then the seventh and eighth equalities arise from $\bm{S} \perp Z|N$ (A2) and $N_{.j} \perp N$ (A2), respectively.
Alternatively, the above result can also straightforwardly proceed by relying on randomization (A3) in longitudinal CRTs.

Therefore, using formula (\ref{app.eq:PO}), the above equalities, and the definition of $\bm{Q}\underline{\bm{\beta}}$, we can simplify the constant treatment effect score function (Equations \ref{app.eq:treatment} \& \ref{app.eq:period_demean}) to the following estimating equation
\[
\begin{split}
    E&\left[\left\{(\bm{\Delta}_Z \bm{1}_{J} - E[\bm{\Delta}_Z \bm{1}_{J}]) \otimes \bm{1}_N \right\}^\top \bm{\mathcal{M}} \underline{\ddot{\bm{V}}} \bm{\mathcal{M}} \{\bm{Y}-\bm{Q}_Z\underline{\bm{\beta}}_Z\}\right] \\
    &= E\left[\left\{(\bm{\Delta}_Z \bm{1}_{J} - E[\bm{\Delta}_Z \bm{1}_{J}]) \otimes \bm{1}_N \right\}^\top \bm{\mathcal{M}}\underline{\ddot{\bm{V}}} \bm{\mathcal{M}} \left\{ \sum_{d=1}^{J-1} \left( (\bm{\Lambda}_Z^d \otimes \textbf{I}_N) \{\bm{Y}(d)-\bm{Y}(0)\} \right) - ((\bm{\Delta}_Z \bm{1}_{J})\otimes \bm{1}_N) \underline{\beta}_Z \right\}\right] \\
    &= E\left[\left\{(\bm{\Delta}_Z \bm{1}_{J} - E[\bm{\Delta}_Z \bm{1}_{J}]) \otimes \bm{1}_N \right\}^\top \bm{\mathcal{M}}\underline{\ddot{\bm{V}}} \bm{\mathcal{M}} \left\{ \sum_{d=1}^{J-1} \left( (\bm{\Lambda}_Z^d \otimes \textbf{I}_N) \{\bm{Y}(d)-\bm{Y}(0)\} \right) - \left\{ \bm{\Delta}_Z \otimes \textbf{I}_N\right\} 1_{NJ} \underline{\beta}_Z \right\}\right] \,.
\end{split}
\]
We can then simplify the above formula with
\[
\begin{split}
    \bm{Y}(d)-\bm{Y}(0) &\in \mathbb{R}^{NJ} \\
    \bm{U}_d^\top  = \left\{(\bm{\Delta}_Z \bm{1}_{J} - E[\bm{\Delta}_Z \bm{1}_{J}]) \otimes \bm{1}_N \right\}^\top \bm{\mathcal{M}}\underline{\ddot{\bm{V}}} \bm{\mathcal{M}} \left\{ \bm{\Lambda}_Z^d \otimes \textbf{I}_N\right\} &\in \mathbb{R}^{1 \times NJ} \\
    \sum_{d=1}^{J-1} \bm{U}_d^\top  = \left\{(\bm{\Delta}_Z \bm{1}_{J} - E[\bm{\Delta}_Z \bm{1}_{J}]) \otimes \bm{1}_N \right\}^\top \bm{\mathcal{M}}\underline{\ddot{\bm{V}}} \bm{\mathcal{M}} \left\{ \bm{\Delta}_Z \otimes \textbf{I}_N\right\}  &\in \mathbb{R}^{1 \times NJ} \\
    \sum_{d=1}^{J-1} \bm{U}_d^\top  \bm{1}_{NJ} = \left\{(\bm{\Delta}_Z \bm{1}_{J} - E[\bm{\Delta}_Z \bm{1}_{J}]) \otimes \bm{1}_N \right\}^\top \bm{\mathcal{M}} \underline{\ddot{\bm{V}}} \bm{\mathcal{M}} \left\{ \bm{\Delta}_Z \otimes \textbf{I}_N\right\} \bm{1}_{NJ} &\in \mathbb{R}^{1} \,.
\end{split}
\]
then the constant treatment effect score function (Equations \ref{app.eq:treatment} \& \ref{app.eq:period_demean}) can be further simplified to
\[
    E\left[ \sum_{d=1}^{J-1} \bm{U}_d^\top \left(\{\bm{Y}(d)-\bm{Y}(0)\} - \bm{1}_{NJ}\underline{\beta}_Z\right) \right] =0
\]

Notably, $\bm{U}_d$ is a function of $Z$ and $N_{.j}$. 
Accordingly
\[
\begin{split}
    E\left[ \sum_{d=1}^{J-1} \bm{U}_d^\top  \{\bm{Y}(d)-\bm{Y}(0)\} \right] &= E\left[ \sum_{d=1}^{J-1} \bm{U}_d^\top \bm{1}_{NJ} \right] \underline{\beta}_Z  \\
    E\left[ \sum_{d=1}^{J-1} E\left[\bm{U}_d^\top \{\bm{Y}(d)-\bm{Y}(0)\} | \bm{S}, N, Z \right] \right] &= E\left[ \sum_{d=1}^{J-1} \bm{U}_d^\top \bm{1}_{NJ} \right] \underline{\beta}_Z \\
    E\left[ \sum_{d=1}^{J-1} \bm{U}_d^\top E\left[\bm{Y}(d)-\bm{Y}(0) | \bm{S}, N, Z \right] \right] &= E\left[ \sum_{d=1}^{J-1} \bm{U}_d^\top \bm{1}_{NJ} \right] \underline{\beta}_Z \\
\text{(A2} &\text{, Non-informative enrollment)}\\
    E\left[ \sum_{d=1}^{J-1} \bm{U}_d^\top  E\left[\bm{Y}(d)-\bm{Y}(0) | N, Z \right] \right] &= E\left[ \sum_{d=1}^{J-1} \bm{U}_d^\top \bm{1}_{NJ} \right] \underline{\beta}_Z \\
\text{(A3} &\text{, Randomization)}\\
    E\left[ \sum_{d=1}^{J-1} \bm{U}_d^\top E\left[\bm{Y}(d)-\bm{Y}(0) | N \right] \right] &= E\left[ \sum_{d=1}^{J-1} \bm{U}_d^\top \bm{1}_{NJ} \right] \underline{\beta}_Z \\
\text{(A2, Non-informative enrollment} &\text{; A3, Randomization)}\\
    E\left[ \sum_{d=1}^{J-1} \bm{U}_d^\top \right] E\left[E\left[\bm{Y}(d)-\bm{Y}(0) | N \right]\right] &= E\left[ \sum_{d=1}^{J-1} \bm{U}_d^\top \bm{1}_{NJ} \right] \underline{\beta}_Z \\
\end{split}
\]
The second equality results from the law of total expectation.
The third equality is a result of $\bm{U}_d$ being deterministic given $(\bm{S}, N, Z)$.
The fourth and fifth equalities then naturally extend from Assumptions A2 and A3.
In the final equality, $\bm{U}_d$ is a function of $Z$ and $N_{.j}$, which recall $N_{.j} \perp N$ (A2) and $Z \perp \bm{Y}(d),\bm{Y}(0), N$ (A3).
In the absence of randomization (A3), as is the case in SW-CQTs, the fifth and sixth equalities can still proceed with the mean independence assumption of exchangeable treatment effects $E\left[\bm{Y}(d)-\bm{Y}(0) | N, Z \right] = E\left[\bm{Y}(d)-\bm{Y}(0) \right]$ (A4).

\sloppy Then, alongside Assumption A1, assuming a true duration-invariant and period-invariant constant treatment effect, $E\left[ E[\bm{Y}(d)-\bm{Y}(0)|N] \right] = E\left[ \bm{Y}(d)-\bm{Y}(0) \right] = \bm{1}_{NJ} \Delta \in \mathbb{R}^{NJ}$, and
\[
    E\left[ \sum_{d=1}^{J-1} \bm{U}_d^\top \bm{1}_{NJ} \right] \Delta = E\left[ \sum_{d=1}^{J-1} \bm{U}_d^\top \bm{1}_{NJ} \right] \underline{\beta}_Z \,.
\]
Finally, we have the desired consistency result
\[
\underline{\beta}_Z = \Delta \,.
\]

Altogether, we establish the consistency of the working linear fixed-effects model within-transformed constant treatment effect estimator for $\Delta$ in a SW-CRT by relying on Assumptions A1 (Super-population sampling), A2 (Non-informative enrollment), and  A3 (Randomization), and in a SW-CQT by relying on assumptions A1, A2, and A4 (Mean independence).

\subsection{Duration-specific Treatment Effect}
\label{app.sect:duration}

The proof for the consistency of a correctly specified duration-specific treatment effect extends directly from the proof outlined in Section \ref{app.sect:constant}.
Under the duration-specific treatment effect setting, the linear mixed model becomes
\begin{equation}
    \label{app.eq:Y_DS}
    \bm{Y}_i=(
    \textbf{I}_{J}
    \otimes \bm{1}_{N_i})\bm{\beta}_0 + (\textbf{H}_{Z_i} \otimes \bm{1}_{N_i}) \bm{\beta}_Z^D + (\bm{1}_{J} \otimes \bm{X}_i) \bm{\beta}_X + \alpha_i \bm{1}_{N_iJ} + \bm{\gamma}_i \otimes \bm{1}_{N_i} + \bm{\epsilon}_i,
\end{equation}
resembling Equation \ref{app.eq:Y_constant}, but with the duration-specific treatment effect structure denoted as
\[
\textbf{H}_{Z_i} = \begin{pmatrix}
0 & 0 & \cdots &0\\
I\{Z_i=2\} & & & \\
I\{Z_i=3\} & I\{Z_i=2\} &  & \\
\vdots & & \ddots & \\
I\{Z_i=J\} & I\{Z_i=J-1\} & \cdots & I\{Z_{i}=2\} \\
\end{pmatrix} \in \mathbb{R}^{J \times (J-1)} \,.
\]
To reiterate, $Z_i=j$ if cluster $i$ starts receiving treatment in the beginning of period $j\in \{2,...,J\}$, and $\bm{\beta}_Z^D = (\beta_{Z1}, ... , \beta_{Zd},... ,\beta_{Z,J-1})^{\top} \in \mathbb{R}^{J-1}$ in a SW-CT.
Notice that the estimating equations under this setting are the same as those under the constant treatment effect setting, except that 
$\bm{Q}=(\bm{Q}_0, \bm{Q}_Z^D, \bm{Q}_X)$ where $\bm{Q}_Z^D = \textbf{H}_Z \otimes \bm{1}_N$, and $\bm{\beta} = (\bm{\beta}_0^\top, \bm{\beta}_Z^{D\top}, \bm{\beta}_X^{\top})^{\top}$.
Accordingly, we can rewrite Equation \ref{app.eq:Y_DS} as:
\begin{equation}
    \label{app.eq:Y_DS_Q}
    \bm{Y}_i = \bm{Q}_{0,i}\bm{\beta}_0 + \bm{Q}_{Z,i}^D\bm{\beta}_Z^D + \bm{Q}_{X,i}\bm{\beta}_X + \bm{\mathcal{C}}_i + \bm{\gamma}_i \otimes \bm{1}_{N_i} + \bm{\epsilon}_i
\end{equation}
akin to the notation in Equation \ref{app.eq:Y_constant_Q}, but with $\bm{Q}_Z^D$ denoting the duration-specific treatment effect.
We can then rewrite the linear mixed model as the following ``within-transformed'' equation 
\begin{equation}
\label{app.eq:ddotY_DS}
    \ddot{\bm{Y}}_i = \ddot{\bm{Q}}_{0,i}\bm{\beta}_0 + \ddot{\bm{Q}}_{Z,i}^D\bm{\beta}_Z^D + \ddot{\bm{Q}}_{X,i}\bm{\beta}_{X} + \ddot{\bm{\gamma}}_i \otimes \bm{1}_{N_i} + \ddot{\bm{\epsilon}}_i \,
\end{equation}
akin to the notation in Equation \ref{app.eq:ddotY}, but with $\ddot{\bm{Q}}_{Z,i}^D = \bm{\mathcal{M}}_i \bm{Q}_{Z,i}^D$, where recall $\bm{\mathcal{M}}_i \equiv \textbf{I}_{N_iJ} - \bm{S}_i (\bm{S}_i^\top \bm{S}_i)^{-1} \bm{S}_i^\top = \textbf{I}_{N_iJ} - \frac{1}{\mathcal{N}_i}\bm{S}_i \bm{S}_i^\top$ as the following symmetric matrix that returns the within-transformed, de-meaned matrices and vectors (Equation \ref{app.eq:within-transformation-mat}).
Then the asymptotic normality of $\hat{\bm{\beta}}_Z^D$ and consistency of variance estimators can be proved similarly using Lemma \ref{app.lemma:variance}.
Hence the proof for this setting is completed if we can prove $\underline{\bm{\beta}}_Z^D = \bm{\Delta}^D$, where $\bm{\Delta}^D=(\Delta(1), ..., \Delta(d) , ...,\Delta(J-1))^\top$ with $\Delta(d)$ being the $d$ duration-specific treatment effect.
To this end, we observe that the component of $E[\bm{\psi}(\bm{O};\bm{\theta})]=0$ corresponding to the duration-specific treatment effects $\bm{\beta}_Z^D$ is
\begin{equation}
\label{app.eq:treatment_DS}
    E[\{\ddot{\bm{Q}}_Z^D\}^\top \underline{\ddot{\bm{V}}}(\ddot{\bm{Y}} - \ddot{\bm{Q}}\underline{\bm{\beta}})] = 0
\end{equation}
corresponding to Equation \ref{app.eq:treatment}.
The component of $E[\bm{\psi}(\bm{O};\bm{\theta})]=0$ corresponding to the period fixed-effects is of the same form as Equation \ref{app.eq:period}, while accounting for the duration-specific treatment effect structure.

\subsubsection{Proof of Fixed-Effects model duration-specific treatment effect consistency}

\noindent By left-multiplying the estimating equation corresponding to period effects (Equation \ref{app.eq:period}), but with a duration-specific treatment effect structure) by $-E[\textbf{H}_Z]^{\top}$ and adding it to Equation (\ref{app.eq:treatment_DS}), we get
\begin{equation}
\label{app.eq:period_demean_DS}
\begin{split}
    E&[\{\ddot{\bm{Q}}_Z^D-\ddot{\bm{Q}}_0E[\textbf{H}_Z]\}^\top \underline{\ddot{\bm{V}}} (\ddot{\bm{Y}}-\ddot{\bm{Q}}\underline{\bm{\beta}})] \\
    &=E[\{\bm{\mathcal{M}}(\bm{Q}_Z^D-\bm{Q}_0 E[\textbf{H}_Z])\}^\top \underline{\ddot{\bm{V}}} \bm{\mathcal{M}} (\bm{Y}-\bm{Q}\underline{\bm{\beta}})] \\
    &= E\left[\left\{(\textbf{H}_Z - E[\textbf{H}_Z]) \otimes \bm{1}_N \right\}^\top \bm{\mathcal{M}}  \underline{\ddot{\bm{V}}} \bm{\mathcal{M}} (\bm{Y} - \bm{Q}\underline{\bm{\beta}})\right] \\ 
    &= 0 
\end{split}
\end{equation}
where $\bm{Q}\underline{\bm{\beta}} = (\bm{Q}_0 \underline{\bm{\beta}}_0 + \bm{Q}_Z^D \underline{\bm{\beta}}_Z^D + \bm{Q}_X \underline{\bm{\beta}}_X)$ and 
$\bm{Q}_0=\textbf{I}_{J} \otimes \bm{1}_{N}$,
$\bm{Q}_Z^D = \textbf{H}_Z \otimes \bm{1}_N$, $\bm{Q}_X = \bm{1}_{J} \otimes \bm{X}_i$, corresponding to the indicators for the period, duration-specific treatment, and covariate  indicators.

To provide an illustrative example from a 3 cluster, 4 period SW-CRT:
\[
\textbf{H}_{Z_i} = \left\{
\begin{pmatrix}
    0 & 0 & 0 \\
    1 & 0 & 0 \\
    0 & 1 & 0 \\
    0 & 0 & 1 \\
\end{pmatrix} \,, 
\begin{pmatrix}
    0 & 0 & 0 \\
    0 & 0 & 0 \\
    1 & 0 & 0 \\
    0 & 1 & 0 \\
\end{pmatrix} \,,
\begin{pmatrix}
    0 & 0 & 0 \\
    0 & 0 & 0 \\
    0 & 0 & 0 \\
    1 & 0 & 0 \\
\end{pmatrix}\right\} ;
\frac{1}{m}\sum_{i=1}^{m} \textbf{H}_{Z_i} = \begin{pmatrix}
    0 & 0 & 0 \\
    1/3 & 0 & 0 \\
    1/3 & 1/3 & 0 \\
    1/3 & 1/3 & 1/3 \\
\end{pmatrix}
\]
and
\[
\textbf{H}_Z - \frac{1}{m}\sum_{i=1}^{m} \textbf{H}_{Z_i} = \left\{
    \begin{pmatrix}
        0 & 0 & 0 \\
        2/3 & 0 & 0 \\
        -1/3 & 2/3 & 0 \\
        -1/3 & -1/3 & 2/3 \\
    \end{pmatrix},
    \begin{pmatrix}
        0 & 0 & 0 \\
        -1/3 & 0 & 0 \\
        2/3 & -1/3 & 0 \\
        -1/3 & 2/3 & -1/3 \\
    \end{pmatrix},
    \begin{pmatrix}
        0 & 0 & 0 \\
        -1/3 & 0 & 0 \\
        -1/3 & -1/3 & 0 \\
        2/3 & -1/3 & -1/3 \\
    \end{pmatrix}
\right\}
\]
corresponding to clusters $i=1,2,3$, respectively.

As in the scenario with constant treatment effects, we can simplify the duration-specific treatment effect score function (Equation \ref{app.eq:period_demean_DS}) to the following estimating equation.
\[
\begin{split}
    E&\left[\left\{(\textbf{H}_Z - E[\textbf{H}_Z]) \otimes \bm{1}_N \right\}^\top \bm{\mathcal{M}} \underline{\ddot{\bm{V}}} \bm{\mathcal{M}} \{\bm{Y}-\bm{Q}_Z^D \underline{\bm{\beta}}_Z^D\}\right] \\
    &= E\left[\left\{(\textbf{H}_Z - E[\textbf{H}_Z]) \otimes \bm{1}_N \right\}^\top \bm{\mathcal{M}}\underline{\ddot{\bm{V}}} \bm{\mathcal{M}} \left\{ \sum_{d=1}^{J-1} \left((\bm{\Lambda}_Z^d \otimes \textbf{I}_N) \{\bm{Y}(d)-\bm{Y}(0)\} \right) - (\textbf{H}_Z\otimes \bm{1}_N) \underline{\bm{\beta}}_Z^D \right\}\right] \\
    &= 0 \,.
\end{split}
\]
We can simplify the above equation by setting 
\[
    \bm{U}^\top = \left\{(\textbf{H}_Z - E[\textbf{H}_Z]) \otimes \bm{1}_N \right\}^\top \bm{\mathcal{M}}\underline{\ddot{\bm{V}}}\bm{\mathcal{M}} \in \mathbb{R}^{(J-1) \times NJ}
\]
where $\bm{U}^\top$ and $\bm{U}^\top (\bm{\Lambda}_Z^d \otimes \textbf{I}_N)$ are functions of $\bm{S}$, $N$, and $Z$.
\[
    E\left[\bm{U}^\top \left\{ \sum_{d=1}^{J-1} \left((\bm{\Lambda}_Z^d \otimes \textbf{I}_N) \{\bm{Y}(d)-\bm{Y}(0)\} \right) -  (\textbf{H}_Z\otimes \bm{1}_N) \underline{\bm{\beta}}_Z^D \right\}\right] =0
\]
As in the proof of consistency for the constant treatment effect, with Assumptions A3 (Randomization), A2 (Non-informative enrollment), we can demonstrate that
\[
\begin{split}
    E\left[\bm{U}^\top \left\{ \sum_{d=1}^{J-1} \left((\bm{\Lambda}_Z^d \otimes \textbf{I}_N) \{\bm{Y}(d)-\bm{Y}(0)\} \right)  \right\}\right] &= E\left[\bm{U}^\top (\textbf{H}_Z\otimes \bm{1}_N) \right]\underline{\bm{\beta}}_Z^D \\
    E\left[ \sum_{d=1}^{J-1} \bm{U}^\top (\bm{\Lambda}_Z^d \otimes \textbf{I}_N) \{\bm{Y}(d)-\bm{Y}(0)\}   \right] &= E\left[\bm{U}^\top (\textbf{H}_Z\otimes \bm{1}_N) \right]\underline{\bm{\beta}}_Z^D \\
    E\left[ \sum_{d=1}^{J-1} E\left[ \bm{U}^\top (\bm{\Lambda}_Z^d \otimes \textbf{I}_N) \{\bm{Y}(d)-\bm{Y}(0)\}  | \bm{S}, N, Z \right]\right] &= E\left[\bm{U}^\top (\textbf{H}_Z\otimes \bm{1}_N) \right]\underline{\bm{\beta}}_Z^D \\
    E\left[ \sum_{d=1}^{J-1} \bm{U}^\top (\bm{\Lambda}_Z^d \otimes \textbf{I}_N) E[\bm{Y}(d)-\bm{Y}(0) | \bm{S}, N, Z] \right] &= E\left[\bm{U}^\top (\textbf{H}_Z\otimes \bm{1}_N) \right]\underline{\bm{\beta}}_Z^D \\
\text{(A2} &\text{, Non-informative enrollment)}\\
    E\left[ \sum_{d=1}^{J-1} \bm{U}^\top (\bm{\Lambda}_Z^d \otimes \textbf{I}_N) E[\bm{Y}(d)-\bm{Y}(0) | N, Z] \right] &= E\left[\bm{U}^\top (\textbf{H}_Z\otimes \bm{1}_N) \right]\underline{\bm{\beta}}_Z^D \\
\text{(A3} &\text{, Randomization)}\\
    E\left[ \sum_{d=1}^{J-1} \bm{U}^\top (\bm{\Lambda}_Z^d \otimes \textbf{I}_N) E[\bm{Y}(d)-\bm{Y}(0) | N] \right] &= E\left[\bm{U}^\top (\textbf{H}_Z\otimes \bm{1}_N) \right]\underline{\bm{\beta}}_Z^D \\
\text{(A2, Non-informative enrollment; A3} &\text{, Randomization)}\\
   \sum_{d=1}^{J-1} E[\bm{U}^\top (\bm{\Lambda}_Z^d \otimes \textbf{I}_N) ] E\left[ E[\bm{Y}(d)-\bm{Y}(0) | N] \right] &= E\left[\bm{U}^\top (\textbf{H}_Z\otimes \bm{1}_N) \right]\underline{\bm{\beta}}_Z^D \\
\end{split}
\]
The third equality results from the law of total expectation.
The fourth equality is a result of $\bm{U}^\top (\bm{\Lambda}_Z^d \otimes \textbf{I}_N)$ being deterministic given $(\bm{S}, N, Z)$.
The fifth and sixth equalities then naturally extend from Assumptions A2 and A3.
In the final equality, $\bm{U}^\top (\bm{\Lambda}_Z^d \otimes \textbf{I}_N)$ is a function of $Z$ and $N_{.j}$, which recall $N_{.j} \perp N$ (A2) and $Z \perp \bm{Y}(d), \bm{Y}(0), N$ (A3). 
In the absence of randomization (A3), as is the case in SW-CQTs, the sixth and seventh equalities can still proceed with the mean independence assumption of exchangeable treatment effects (A4).

Then, alongside Assumption A1, assuming a duration-specific treatment effect, $E\left[ E[\bm{Y}(d)-\bm{Y}(0) | N] \right] = E\left[ \bm{Y}(d)-\bm{Y}(0) \right] = \bm{1}_{NJ}\Delta(d) \in \mathbb{R}^{NJ}$.
Altogether
\[
\begin{split}
    &\sum_{d=1}^{J-1} E[\bm{U}^\top (\bm{\Lambda}_Z^d \otimes \textbf{I}_N)] E\left[\bm{Y}(d)-\bm{Y}(0) \right] \\
    &= \sum_{d=1}^{J-1} E[\bm{U}^\top (\bm{\Lambda}_Z^d \otimes \textbf{I}_N) \bm{1}_{NJ}\Delta(d)] \\
    &= E\left[ \bm{U}^\top  \sum_{d=1}^{J-1} (\bm{\Lambda}_Z^d \otimes \textbf{I}_N) \bm{1}_{NJ}\Delta(d) \right]
\end{split}
\]
Furthermore, $\sum_{d=1}^{J-1} (\bm{\Lambda}_Z^d \otimes \textbf{I}_N) \bm{1}_{NJ}\Delta(d) = (\textbf{H}_Z\otimes \bm{1}_N)\bm{\Delta}^D \in \mathbb{R}^{NJ}$.
Altogether
\[
    E\left[\bm{U}^\top (\textbf{H}_Z\otimes \bm{1}_N) \right] \bm{\Delta}^D = E\left[\bm{U}^\top (\textbf{H}_Z\otimes \bm{1}_N) \right]\underline{\bm{\beta}}_Z^D \,.
\]
Finally, we have the desired consistency result 
\[\underline{\bm{\beta}}_Z^D = \bm{\Delta}^D \,.\]

Altogether, we establish the consistency of the working linear fixed-effects model within-transformed duration-specific treatment effect estimator for $\bm{\Delta}^D$ in a SW-CRT by relying on Assumptions A1 (Super-population sampling), A2 (Non-informative enrollment), and  A3 (Randomization), and in a SW-CQT by relying on assumptions A1, A2, and A4 (Mean independence).

\subsection{Period-specific Treatment Effect}
\label{app.sect:period}

For the period-specific treatment effect setting, we drop the period $J$ to avoid over-parameterization.
To simplify the switch from $J$ periods to $J-1$ periods, we use the superscript * to denote all subsequent changes in notation. For example, $J^* = J-1$, $\bm{Y}_i^* = (\bm{Y}_{i1}^{\top},...,\bm{Y}_{i,J-1}^{\top})^{\top} \in \mathbb{R}^{NJ^*}$,
$\bm{\beta}_0^*=(\beta_{01},...,\beta_{0,J-1})^{\top} \in \mathbb{R}^{J^*}$ (where $\beta_{01} = 0$ for identifiability).
Generically, $\textbf{M}^* \in \mathbb{R}^{J^* \times J^*}$ is a matrix consisting of the first $J-1$ columns and rows of a matrix $\textbf{M} \in \mathbb{R}^{J \times J}$.
As a result, the linear mixed model becomes
\begin{equation}
    \label{app.eq:Y_PS}
    \bm{Y}_i^*=(\textbf{I}_{J^*} \otimes \bm{1}_{N_i})\bm{\beta}_0^* + (\bm{\Delta}_{Z_i}^* \otimes \bm{1}_{N_i}) \bm{\beta}_Z^P + (\bm{1}_{J^*} \otimes \bm{X}_i) \bm{\beta}_X + \alpha_i \bm{1}_{N_iJ^*} + \bm{\gamma}_i^* \otimes \bm{1}_N + \bm{\epsilon}_i^*,
\end{equation}
where $\bm{\beta}_Z^P=(\beta_{1Z}, ..., \beta_{jZ},...,\beta_{J^*Z})^{\top} \in \mathbb{R}^{J^*}$ is the treatment effect coefficient vector where $\beta_{1Z}=0$, and $\alpha_i$ are the cluster fixed intercepts.
Notice that the estimating equations under this setting remain the same as those under the constant treatment effect setting, except that
$\bm{Q}^* = (\bm{Q}_0^*, \bm{Q}_Z^P, \bm{Q}_X^*)$
and $\bm{\beta}^* = (\bm{\beta}_0^{*\top}, \bm{\beta}_Z^{P\top}, \bm{\beta}_X^{\top})^{\top}$,
where $\bm{Q}_Z^P=\bm{\Delta}_{Z_i}^* \otimes \bm{1}_{N_i} \in \mathbb{R}^{NJ^* \times J^*}$ denotes the period-specific treatment effect structure.
Accordingly, we can rewrite Equation \ref{app.eq:Y_PS} as:
\begin{equation}
    \label{app.eq:Y_PS_Q}
    \bm{Y}_i^* = \bm{Q}_{0,i}^* \bm{\beta}_0^* + \bm{Q}_{Z,i}^P \bm{\beta}_Z^P + \bm{Q}_{X,i}^* \bm{\beta}_X + \bm{\mathcal{C}}_i^* + \bm{\gamma}_i^* \otimes \bm{1}_{N_i} + \bm{\epsilon}_i^*
\end{equation}
akin to the notation in Equation \ref{app.eq:Y_constant_Q}, but with $\bm{Q}_Z^P$ and $\bm{\beta}_Z^P$ denoting the period-specific treatment indicators and effects.
We can then rewrite the linear mixed model as the following ``within-transformed'' equation 
\begin{equation}
\label{app.eq:ddotY_PS}
    \ddot{\bm{Y}}_i^* = \ddot{\bm{Q}}_{0,i}^* \bm{\beta}_0^* + \ddot{\bm{Q}}_{Z,i}^P \bm{\beta}_Z^P + \ddot{\bm{Q}}_{X,i}^* \bm{\beta}_{X} + \ddot{\bm{\gamma}}_i^* \otimes \bm{1}_{N_i} + \ddot{\bm{\epsilon}}_i^* \,
\end{equation}
akin to the notation in Equation \ref{app.eq:ddotY}, but with $\ddot{\bm{Q}}_{Z,i}^P = \bm{\mathcal{M}}_i^* \bm{Q}_{Z,i}^P$,
where $\bm{\mathcal{M}}_i^* \equiv \textbf{I}_{N_iJ^*} - \bm{S}_i^* (\bm{S}_i^{*\top} \bm{S}_i^*)^{-1} \bm{S}_i^{*\top} = \textbf{I}_{N_iJ^*} - \frac{1}{\mathcal{N}_i^*}\bm{S}_i^* \bm{S}_i^{*\top}$ as the following symmetric matrix that returns the within-transformed, de-meaned matrices and vectors (Equation \ref{app.eq:within-transformation-mat}).

Then the asymptotic normality of $\hat{\bm{\beta}}_Z^P$ and consistency of variance estimators can be proved similarly using Lemma \ref{app.lemma:variance}.
Hence, the proof for this setting is completed if we can prove $\underline{\bm{\beta}}_Z^P = (0,\bm{\Delta}^{P\top})^\top$, where $\bm{\Delta}^P = (\Delta_2,...,\Delta_j...,\Delta_{J^*})^\top \in \mathbb{R}^{J^*-1}$ with $\Delta_j$ being the $j\geq2$ period-specific treatment effect.
This is equivalent to stating that $\Delta_1=0$, but will be presented as in the previous sentence throughout this proof for simplicity.
To this end, we observe that the component of $E[\bm{\psi}(\bm{O}^*; \underline{\bm{\theta}})] = 0$ corresponding to period-specific treatment effects $\bm{\beta}_Z^P$ imply that
\begin{equation}
\label{app.eq:treatment_PS}
    E[\{\ddot{\bm{Q}}_Z^P\}^\top \underline{\ddot{\bm{V}}}^* (\ddot{\bm{Y}}^*-\ddot{\bm{Q}}^*\underline{\bm{\beta}}^*)] = 0
\end{equation}
corresponding to Equation \ref{app.eq:treatment}.
The components of $E[\bm{\psi}(\bm{O}^*;\bm{\theta})]=0$ corresponding to the period and cluster fixed-effects are of the same form as equations \ref{app.eq:period}, while accounting for the period-specific treatment effect structure and the switch from $J$ periods to $J^*=J-1$ periods.

\subsubsection{Proof of Fixed-Effects model period-specific treatment effect consistency}

As in our previous proofs, Equation \ref{app.eq:treatment_PS} implies
\[
    E[\{(\bm{\Delta}_Z^* - E[\bm{\Delta}_Z^*]) \otimes \bm{1}_N\}^\top \bm{\mathcal{M}}^* \underline{\bm{V}}^* \bm{\mathcal{M}}^*  (\bm{Y}^* - \bm{Q}^* \underline{\bm{\beta}}^*)] = 0
\]
To provide an illustrative example from a 4 cluster, 5 ($J^*=4$) period SW-CRT:
\[
\bm{\Delta}_{Zi}^* = \left\{
\begin{pmatrix}
    0 & 0 & 0 & 0 \\
    0 & 1 & 0 & 0 \\
    0 & 0 & 1 & 0 \\
   0 &  0 & 0 & 1
\end{pmatrix} , 
\begin{pmatrix}
    0 & 0 & 0 & 0 \\
    0 & 0 & 0 & 0 \\
    0 & 0 & 1 & 0\\
    0 & 0 & 0 & 1 \\
\end{pmatrix} ,
\begin{pmatrix}
    0 & 0 & 0 & 0 \\
    0 & 0 & 0 & 0 \\
    0 & 0 & 0 & 0 \\
    0 & 0 & 0 & 1 \\
\end{pmatrix} ,
\begin{pmatrix}
    0 & 0 & 0 & 0 \\
    0 & 0 & 0 & 0 \\
    0 & 0 & 0 & 0 \\
    0 & 0 & 0 & 0 \\
\end{pmatrix} \right\} ;
\]
\[
\frac{1}{m}\sum_{i=1}^{m} \bm{\Delta}_{Z_{i}}^* = \begin{pmatrix}
    0 & 0 & 0 & 0 \\
    0 & 1/4 & 0 & 0 \\
    0 & 0 & 2/4 & 0 \\
    0 & 0 & 0 & 3/4 \\
\end{pmatrix} \,.
\]
Following the previous proofs, we have
\[
\begin{split}
    E&\left[\left\{(\bm{\Delta}_Z^* - E[\bm{\Delta}_Z^*]) \otimes \bm{1}_N \right\}^\top \bm{\mathcal{M}}^* \underline{\ddot{\bm{V}}}^* \bm{\mathcal{M}}^* \{\bm{Y}^*-\bm{Q}_Z^P \underline{\bm{\beta}}_Z^P\}\right] \\
    &= E\left[\left\{(\bm{\Delta}_Z^* - E[\bm{\Delta}_Z^*]) \otimes \bm{1}_N \right\}^\top \bm{\mathcal{M}}^* \underline{\ddot{\bm{V}}}^* \bm{\mathcal{M}}^* \left\{ \sum_{d=1}^{J^*-1} \left((\bm{\Lambda}_Z^{d*} \otimes \textbf{I}_N) \{\bm{Y}^*(d)-\bm{Y}^*(0)\} \right) - (\bm{\Delta}_Z^*\otimes \bm{1}_N) \underline{\bm{\beta}}_Z^P \right\}\right] \\
    &= E\left[\left\{(\bm{\Delta}_Z^* - E[\bm{\Delta}_Z^*]) \otimes \bm{1}_N \right\}^\top \bm{\mathcal{M}}^* \underline{\ddot{\bm{V}}}^* \bm{\mathcal{M}}^* \left\{ \sum_{d=1}^{J^*-1} \left((\bm{\Lambda}_Z^{d*} \otimes \textbf{I}_N) \{\bm{Y}^*(d)-\bm{Y}^*(0)\} \right) - \left\{ \bm{\Delta}_Z^* \otimes \textbf{I}_N\right\} (\textbf{I}_{J^*} \otimes \bm{1}_N) \underline{\bm{\beta}}_Z^P \right\}\right] \\
    &= 0 \,.
\end{split}
\]
To further simplify the above formula, we can compute
\[
\begin{split}
    \bm{Y}^*(d)-\bm{Y}^*(0) &\in \mathbb{R}^{NJ^*}\\
    \bm{U}_d^{*\top} = \left\{(\bm{\Delta}_Z^* - E[\bm{\Delta}_Z^*]) \otimes \bm{1}_N \right\}^\top \bm{\mathcal{M}}^* \underline{\ddot{\bm{V}}}^* \bm{\mathcal{M}}^* \left\{ \bm{\Lambda}_Z^{d*} \otimes \textbf{I}_N \right\} &\in \mathbb{R}^{J^* \times NJ^*} \\
    \sum_{d=1}^{J^*-1} \bm{U}_d^{*\top}  = \left\{(\bm{\Delta}_Z^* - E[\bm{\Delta}_Z^*]) \otimes \bm{1}_N \right\}^\top \bm{\mathcal{M}}^* \underline{\ddot{\bm{V}}}^* \bm{\mathcal{M}}^* \left\{ \bm{\Delta}_Z^* \otimes \textbf{I}_N\right\}  &\in \mathbb{R}^{J^* \times NJ^*} \\
    \sum_{d=1}^{J^*-1} \bm{U}_d^{*\top} (\textbf{I}_{J^*} \otimes \bm{1}_N) = \left\{(\bm{\Delta}_Z^* - E[\bm{\Delta}_Z^*]) \otimes \bm{1}_N \right\}^\top \bm{\mathcal{M}}^* \underline{\ddot{\bm{V}}}^* \bm{\mathcal{M}}^* \left\{ \bm{\Delta}_Z^* \otimes \textbf{I}_N\right\} (\textbf{I}_{J^*} \otimes \bm{1}_N) &\in \mathbb{R}^{J^* \times J^*}
\end{split}
\]
and $\bm{U}_d^*$ is a function of $\bm{S}$, $N$, and $Z$.
Then, the period-specific treatment effect score function can be further simplified to
\[
    E\left[ \sum_{d=1}^{J^*-1} \bm{U}_d^{*\top} \left(\{\bm{Y}^*(d)-\bm{Y}^*(0)\} - (\textbf{I}_{J^*} \otimes \bm{1}_N)\underline{\bm{\beta}}_Z^P\right) \right] =0 \,.
\]
As in the earlier proofs, by the law of total probability, Assumptions A3 (Randomization), and A2 (Non-informative enrollment)
\[
\begin{split}
    E\left[ \sum_{d=1}^{J^*-1} \bm{U}_d^{*\top} \{\bm{Y}^*(d)-\bm{Y}^*(0)\} \right] &= E\left[ \sum_{d=1}^{J^*-1} \bm{U}_d^{*\top} (\textbf{I}_{J^*} \otimes \bm{1}_N) \right] \underline{\bm{\beta}}_Z^P \\
    E\left[ \sum_{d=1}^{J^*-1} E\left[ \bm{U}_d^{*\top} \{\bm{Y}^*(d)-\bm{Y}^*(0)\} |\bm{S}, N, Z \right] \right] &= E\left[ \sum_{d=1}^{J^*-1} \bm{U}_d^{*\top} (\textbf{I}_{J^*} \otimes \bm{1}_N) \right] \underline{\bm{\beta}}_Z^P \\
    E\left[ \sum_{d=1}^{J^*-1} \bm{U}_d^{*\top}  E\left[\bm{Y}^*(d)-\bm{Y}^*(0) |\bm{S}, N, Z \right] \right] &= E\left[ \sum_{d=1}^{J^*-1} \bm{U}_d^{*\top} (\textbf{I}_{J^*} \otimes \bm{1}_N) \right] \underline{\bm{\beta}}_Z^P \\
\text{(A2} &\text{, Non-informative enrollment)}\\
    E\left[ \sum_{d=1}^{J^*-1} \bm{U}_d^{*\top} E\left[\bm{Y}^*(d)-\bm{Y}^*(0) |N, Z \right] \right] &= E\left[ \sum_{d=1}^{J^*-1} \bm{U}_d^{*\top} (\textbf{I}_{J^*} \otimes \bm{1}_N) \right] \underline{\bm{\beta}}_Z^P \\
\text{(A3} &\text{, Randomization)}\\
    E\left[ \sum_{d=1}^{J^*-1} \bm{U}_d^{*\top} E\left[\bm{Y}^*(d)-\bm{Y}^*(0) |N \right] \right] &= E\left[ \sum_{d=1}^{J^*-1} \bm{U}_d^{*\top} (\textbf{I}_{J^*} \otimes \bm{1}_N) \right] \underline{\bm{\beta}}_Z^P \\
\text{(A2, Non-informative enrollment; A3} &\text{, Randomization)}\\
    \sum_{d=1}^{J^*-1} E\left[ \bm{U}_d^{*\top} \right] E\left[ E\left[\bm{Y}^*(d)-\bm{Y}^*(0) |N \right] \right] &= E\left[ \sum_{d=1}^{J^*-1} \bm{U}_d^{*\top} (\textbf{I}_{J^*} \otimes \bm{1}_N) \right] \underline{\bm{\beta}}_Z^P \\
\end{split}
\]
where individuals $k$ in the $h$-th set of $N$ entries in $E[\bm{Y}^*(d) | N] \in \mathbb{R}^{NJ^*}$ are $E[Y_{.hk}^*(d) | N]$.
The second equality results from the law of total expectation.
The third equality is a result of $\bm{U}_d^{*\top}$ being deterministic given $(\bm{S}, N, Z)$.
The fourth and fifth equalities then naturally extend from Assumptions A2 and A3.
In the final equality, $\bm{U}_d^{*\top}$ is a function of $Z$ and $N_{.j}$, which recall $N_{.j} \perp N$ (A2) and $Z \perp \bm{Y}^*(d),\bm{Y}^*(0), N$ (A3).
In the absence of randomization (A3), as is the case in SW-CQTs, the fifth and sixth equalities can still proceed with the mean independence assumption of exchangeable treatment effects (A4).

Then, alongside Assumption A1, assuming a period-specific treatment effect,
$E\left[E[\bm{Y}^*(d)-\bm{Y}^*(0) |N] \right] = E\left[\bm{Y}^*(d)-\bm{Y}^*(0) \right] = (0,\bm{\Delta}^{P\top})^\top \otimes 1_{N} =(\textbf{I}_{J^*} \otimes \bm{1}_N) (0,\bm{\Delta}^{P\top})^\top \in \mathbb{R}^{NJ^*}$.
Accordingly
\[
    E\left[ \sum_{d=1}^{J^*-1} \bm{U}_d^{*\top} (\textbf{I}_{J^*} \otimes \bm{1}_N) \right] (0, \bm{\Delta}^{P\top})^\top = E\left[ \sum_{d=1}^{J^*-1} \bm{U}_d^{*\top} (\textbf{I}_{J^*} \otimes \bm{1}_N) \right] \underline{\bm{\beta}}_Z^P \,.
\]
Finally, we have the desired consistency result 
$\underline{\bm{\beta}}_Z^P =(0,\beta_{2Z},...,\beta_{J^*Z})^\top = (0,\bm{\Delta}^{P\top})^\top$
and
\[
\beta_{jZ} = \Delta_j
\]
for $j \geq 2$.

Altogether, we establish the consistency of the working linear fixed-effects model within-transformed period-specific treatment effect estimator for $\bm{\Delta}^P$ in a SW-CRT by relying on Assumptions A1 (Super-population sampling), A2 (Non-informative enrollment), and  A3 (Randomization), and in a SW-CQT by relying on assumptions A1, A2, and A4 (Mean independence).

\subsection{Saturated Treatment Effect}
\label{app.sect:saturated}

Under the saturated treatment effect setting, we also drop the data from period $J$ to avoid over-parameterization. Similar to the period-specific treatment effects, we use the superscript * to denote all subsequent changes in notation.
Then the linear mixed model becomes
\begin{equation}
\label{app.eq:Y_S}
    \bm{Y}_i^*=(\textbf{I}_{J^*} \otimes \bm{1}_{N_i})\bm{\beta}_0^* + \sum_{d=1}^{J^* -1} \left[\left( \tilde{\bm{\Lambda}}_{Z_i}^{d*} \otimes \bm{1}_{N_i} \right) \bm{\beta}_d^S \right] + (\bm{1}_{J^*} \otimes \bm{X}_i) \bm{\beta}_X + \alpha_i \bm{1}_{N_iJ^*} + \bm{\gamma}_i^* \otimes \bm{1}_N + \bm{\epsilon}_i^*,
\end{equation}
where $\tilde{\bm{\Lambda}}_{Z_i}^{d*} \in \mathbb{R}^{J^* \times (J^* - d)}$ contains column $d+1$ to column $J^*$ of $\bm{\Lambda}_{Z_i}^{d*}$ and column vector $\bm{\beta}_d^S = (\beta_{dZd}, ... , \beta_{J^*-1, Zd})^{\top} \in \mathbb{R}^{ (J^* - d)}$, such that $\left( \tilde{\bm{\Lambda}}_{Z_i}^{d*} \otimes \bm{1}_{N_i} \right) \bm{\beta}_d^S  \in \mathbb{R}^{N_iJ^*}$, and $\alpha_i$ are the cluster fixed intercepts.
Notice that the estimating equations under this setting remain the same as those under the period-specific treatment effect setting, except that
$\bm{Q}^* = (\bm{Q}_0^*, \bm{Q}_{Z1}^S,...,\bm{Q}_{Z,J^*-1}^S,\bm{Q}_X^*)$ where $\bm{Q}_{Zd}^S = \tilde{\bm{\Lambda}}_{Z}^{d*} \otimes \bm{1}_{N}$,
and $\bm{\beta}^* = (\bm{\beta}_0^{*\top}, \bm{\beta}_1^{S\top}, ..., \bm{\beta}_{J^*-1}^{S\top}, \bm{\beta}_X^{\top})^{\top}$.

Accordingly, we can rewrite Equation \ref{app.eq:Y_S} as:
\begin{equation}
    \label{app.eq:Y_S_Q}
    \bm{Y}_i^* = \bm{Q}_{0,i}^* \bm{\beta}_0^* + \sum_{d=1}^{J^*-1} \bm{Q}_{Zd,i}^S \bm{\beta}_d^S + \bm{Q}_{X,i}^* \bm{\beta}_X + \bm{\mathcal{C}}_i^* + \bm{\gamma}_i^* \otimes \bm{1}_{N_i} + \bm{\epsilon}_i^*
\end{equation}
akin to the notation in Equation \ref{app.eq:Y_constant_Q}, but with $\bm{Q}_{Zd}^S$ and $\bm{\beta}_d^S$ denoting the saturated treatment indicators and effects during treatment duration $d$.
We can then rewrite the linear mixed model as the following ``within-transformed'' equation 
\begin{equation}
\label{app.eq:ddotY_PS}
\begin{split}
    \ddot{\bm{Y}}_i^* &= \ddot{\bm{Q}}_{0,i}^* \bm{\beta}_0^* + \sum_{d=1}^{J^*-1} \ddot{\bm{Q}}_{Zd,i}^S \bm{\beta}_d^S + \ddot{\bm{Q}}_{X,i}^* \bm{\beta}_{X} + \ddot{\bm{\gamma}}_i^* \otimes \bm{1}_{N_i} + \ddot{\bm{\epsilon}}_i^* \\
\end{split}
\end{equation}
akin to the notation in Equation \ref{app.eq:ddotY}, but with $\ddot{\bm{Q}}_{Zd,i}^S = \bm{\mathcal{M}}_i^* \bm{Q}_{Zd,i}^S$, 
where $\bm{\mathcal{M}}_i^* \equiv \textbf{I}_{N_iJ^*} - \bm{S}_i^* (\bm{S}_i^{*\top} \bm{S}_i^*)^{-1} \bm{S}_i^{*\top} = \textbf{I}_{N_iJ^*} - \frac{1}{\mathcal{N}_i^*}\bm{S}_i^* \bm{S}_i^{*\top}$
 as the following symmetric matrix that returns the within-transformed, de-meaned matrices and vectors (Equation \ref{app.eq:within-transformation-mat}).

Then using Lemma \ref{app.lemma:variance}, the asymptotic normality of $\hat{\bm{\beta}}_1^S, ..., \hat{\bm{\beta}}_d^S,..., \hat{\bm{\beta}}_{J^* -1}^S$ and consistency of variance estimators can be proved.
Hence, the proof for this setting is completed if we can prove $\underline{\beta}_{jZd} = \Delta_j(d)$ for $d < j$, with $\Delta_j(d)$ being the saturated (both duration and period-specific) treatment effect in period $j$ and having been administered for duration $d$.
To this end, we observe that the component of $E[\bm{\psi}(\bm{O}^*; \underline{\bm{\theta}})] = 0$ corresponding to the vector of saturated treatment effect $\bm{\beta}_d^S$ for each $d=1,...,J^*-1$ is
\begin{equation}
\label{app.eq:treatment_S}
    E[\{\ddot{\bm{Q}}_{Zd}^S\}^\top \underline{\ddot{\bm{V}}}^* (\ddot{\bm{Y}}^* -\ddot{\bm{Q}}^*\underline{\bm{\beta}}^*)] = 0 \,.
\end{equation}
The components of $E[\bm{\psi}(\bm{O};\bm{\theta})]=0$ corresponding to the period and cluster fixed-effects are of the same form as equations \ref{app.eq:period}, while accounting for the saturated treatment effect structure and the switch from $J$ periods to $J^* = J-1$ periods.

\subsubsection{Proof of Fixed-Effects model saturated treatment effect consistency}
\label{app.sect:fe_s_rand}

As in our previous proofs, Equation \ref{app.eq:treatment_S} implies, for $d=1,...,J^*-1$, 
\[
E[\{(\tilde{\bm{\Lambda}}_Z^{d*} - E[\tilde{\bm{\Lambda}}_Z^{d*}]) \otimes \bm{1}_N\}^\top \bm{\mathcal{M}}^* \underline{\bm{V}}^* \bm{\mathcal{M}}^* (\bm{Y}^* - \bm{Q}^* \underline{\bm{\beta}}^*)] = 0 \,.
\]
To provide an illustrative example from a 4 cluster, 5 ($J^*=4$) period SW-CRT:
\[
\tilde{\bm{\Lambda}}_{Z_i}^{1*} = \left\{
\begin{pmatrix}
    0 & 0 & 0 \\
    1 & 0 & 0 \\
    0 & 0 & 0 \\
    0 & 0 & 0 \\
\end{pmatrix},
\begin{pmatrix}
    0 & 0 & 0 \\
    0 & 0 & 0 \\
    0 & 1 & 0 \\
    0 & 0 & 0 \\
\end{pmatrix},
\begin{pmatrix}
    0 & 0 & 0 \\
    0 & 0 & 0 \\
    0 & 0 & 0 \\
    0 & 0 & 1 \\
\end{pmatrix},
\begin{pmatrix}
    0 & 0 & 0 \\
    0 & 0 & 0 \\
    0 & 0 & 0 \\
    0 & 0 & 0 \\
\end{pmatrix} \right\} ;
\]
\[
\tilde{\bm{\Lambda}}_{Z_i}^{2*} = \left\{
\begin{pmatrix}
    0 & 0 \\
    0 & 0 \\
    1 & 0 \\
    0 & 0 \\
\end{pmatrix},
\begin{pmatrix}
    0 & 0 \\
    0 & 0 \\
    0 & 0 \\
    0 & 1 \\
\end{pmatrix},
\begin{pmatrix}
    0 & 0 \\
    0 & 0 \\
    0 & 0 \\
    0 & 0 \\
\end{pmatrix},
\begin{pmatrix}
    0 & 0 \\
    0 & 0 \\
    0 & 0 \\
    0 & 0 \\
\end{pmatrix} \right\};
\]
\[
\tilde{\bm{\Lambda}}_{Z_i}^{3*} = \left\{
\begin{pmatrix}
    0 \\
    0 \\
    0 \\
    1 \\
\end{pmatrix},
\begin{pmatrix}
    0 \\
    0 \\
    0 \\
    0 \\
\end{pmatrix},
\begin{pmatrix}
    0 \\
    0 \\
    0 \\
    0 \\
\end{pmatrix},
\begin{pmatrix}
    0 \\
    0 \\
    0 \\
    0 \\
\end{pmatrix} \right\}
\]
with the entries in $\tilde{\bm{\Lambda}}_{Z_i}^{d*}$ corresponding to clusters $i=1,2,3,4$, respectively.
Then
\[
\frac{1}{m}\sum_{i=1}^m \tilde{\bm{\Lambda}}_{Z_i}^{d*} = \left\{
\begin{pmatrix}
    0 & 0 & 0 \\
    1/4 & 0 & 0 \\
    0 & 1/4 & 0 \\
    0 & 0 & 1/4 \\
\end{pmatrix},
\begin{pmatrix}
    0 & 0 \\
    0 & 0 \\
    1/4 & 0 \\
    0 & 1/4 \\
\end{pmatrix},
\begin{pmatrix}
    0 \\
    0 \\
    0 \\
    1/4 \\
\end{pmatrix}
\right\}
\]
with entries corresponding to $d=1,2,3$, respectively

Since $E[\{(\tilde{\bm{\Lambda}}_Z^{d*} - E[\tilde{\bm{\Lambda}}_Z^{d*}]) \otimes \bm{1}_N\}^\top \bm{\mathcal{M}}^* \underline{\bm{V}}^* \bm{\mathcal{M}}^* (\textbf{I}_{J^*} \otimes \bm{1}_N)] = 0$, $E[\{(\tilde{\bm{\Lambda}}_Z^{d*} - E[\tilde{\bm{\Lambda}}_Z^{d*}]) \otimes \bm{1}_N\}^\top \bm{\mathcal{M}}^* \underline{\bm{V}}^* \bm{\mathcal{M}}^* (\textbf{I}_{J^*} \otimes \bm{X})] = 0$, and
$E[\{(\tilde{\bm{\Lambda}}_Z^{d*} - E[\tilde{\bm{\Lambda}}_Z^{d*}]) \otimes \bm{1}_N\}^\top \bm{\mathcal{M}}^* \underline{\bm{V}}^* \bm{\mathcal{M}}^* (\textbf{I}_{J^*} \otimes \textbf{I}_N)\bm{Y}^*(0)] = 0$,
formula (\ref{app.eq:PO}) (with period $J$ dropped) and the expression of $Q^*$ implies
\[
\begin{split}
    E&\left[\{ (\tilde{\bm{\Lambda}}_Z^{d*} - E[\tilde{\bm{\Lambda}}_Z^{d*}]) \otimes \bm{1}_N \}^\top \bm{\mathcal{M}}^* \underline{\bm{V}}^* \bm{\mathcal{M}}^* \left\{\bm{Y}^* - \sum_{d'=1}^{J^*-1} \left[\left( \tilde{\bm{\Lambda}}_{Z}^{d'*} \otimes \bm{1}_{N} \right) \underline{\bm{\beta}}_{d'}^S \right]\right\} \right] \\
    &= E\left[\{ (\tilde{\bm{\Lambda}}_Z^{d*} - E[\tilde{\bm{\Lambda}}_Z^{d*}]) \otimes \bm{1}_N \}^\top \bm{\mathcal{M}}^* \underline{\bm{V}}^* \bm{\mathcal{M}}^* \left\{\sum_{d'=1}^{J^*-1}(\bm{\Lambda}_Z^{d'*} \otimes \textbf{I}_N)[\bm{Y}^*(d')-\bm{Y}^*(0)] - \sum_{d'=1}^{J^*-1} \left[\left( \tilde{\bm{\Lambda}}_{Z}^{d'*} \otimes \bm{1}_{N} \right) \underline{\bm{\beta}}_{d'}^S \right]\right\} \right] \\
    &= E\left[ \sum_{d'=1}^{J^*-1} \{ (\tilde{\bm{\Lambda}}_Z^{d*} - E[\tilde{\bm{\Lambda}}_Z^{d*}]) \otimes \bm{1}_N \}^\top \bm{\mathcal{M}}^* \underline{\bm{V}}^* \bm{\mathcal{M}}^* \left\{(\bm{\Lambda}_Z^{d'*} \otimes \textbf{I}_N)[\bm{Y}^*(d')-\bm{Y}^*(0)] - \left( \tilde{\bm{\Lambda}}_{Z}^{d'*} \otimes \bm{1}_{N} \right) \underline{\bm{\beta}}_{d'}^S \right\} \right] \\
    &= 0 \,.
\end{split}
\]
Note the distinction above between $\bm{\Lambda}_Z^{d'*} \in \mathbb{R}^{J^* \times J^*}$ (defined in Section \ref{app.sect:Theorem_SW_Proof}) and $\tilde{\bm{\Lambda}}_{Z}^{d'*} \in \mathbb{R}^{J^* \times (J^*-d)}$ (defined in Section \ref{app.sect:saturated}).
To further simplify the above formula, we can compute
\[
\begin{split}
    \bm{Y}^*(d)-\bm{Y}^*(0) &\in \mathbb{R}^{NJ^*}\\
    \tilde{\bm{U}}_{d}^{*\top} = \{ (\tilde{\bm{\Lambda}}_Z^{d*} - E[\tilde{\bm{\Lambda}}_Z^{d*}]) \otimes \bm{1}_N \}^\top \bm{\mathcal{M}}^* \underline{\bm{V}}^* \bm{\mathcal{M}}^* &\in \mathbb{R}^{(J^*-d) \times NJ^*} \\
\end{split}
\]
and $\tilde{\bm{U}}_{d}^*$ and $\tilde{\bm{U}}_{d}^{*\top} (\bm{\Lambda}_Z^{d'*} \otimes \textbf{I}_N)$ are functions of $\bm{S}$, $N$, and $Z$.
Then, the saturated treatment effect score function can be further simplified to
\[
    E\left[ \sum_{d'=1}^{J^*-1} \tilde{\bm{U}}_{d}^{*\top} \left((\bm{\Lambda}_Z^{d'*} \otimes \textbf{I}_N)\{\bm{Y}^*(d')-\bm{Y}^*(0)\} - \left( \tilde{\bm{\Lambda}}_{Z}^{d'*} \otimes \bm{1}_{N} \right) \underline{\bm{\beta}}_{d'}^S\right) \right] =0 \,.
\]
As in the earlier proofs, by the law of total probability, Assumptions A3 (Randomization), and A2 (Non-informative enrollment)
\[
\begin{split}
        E\left[\sum_{d'=1}^{J^*-1} \tilde{\bm{U}}_{d}^{*\top} (\bm{\Lambda}_Z^{d'*} \otimes \textbf{I}_N) \{\bm{Y}^*(d')-\bm{Y}^*(0)\} \right] &= E\left[\sum_{d'=1}^{J^*-1} \tilde{\bm{U}}_{d}^{*\top} \left( \tilde{\bm{\Lambda}}_{Z}^{d'*} \otimes \bm{1}_{N} \right) \underline{\bm{\beta}}_{d'}^S \right] \\
        \sum_{d'=1}^{J^*-1} E\left[E\left[ \tilde{\bm{U}}_{d}^{*\top} (\bm{\Lambda}_Z^{d'*} \otimes \textbf{I}_N) \{\bm{Y}^*(d')-\bm{Y}^*(0)\} | \bm{S}, N, Z \right]\right] &= \sum_{d'=1}^{J^*-1} E\left[\tilde{\bm{U}}_{d}^{*\top} \left( \tilde{\bm{\Lambda}}_{Z}^{d'*} \otimes \bm{1}_{N} \right) \right] \underline{\bm{\beta}}_{d'}^S \\
        \sum_{d'=1}^{J^*-1} E\left[\tilde{\bm{U}}_{d}^{*\top} (\bm{\Lambda}_Z^{d'*} \otimes \textbf{I}_N) E[\bm{Y}^*(d')-\bm{Y}^*(0) | \bm{S}, N, Z] \right] &= \sum_{d'=1}^{J^*-1} E\left[\tilde{\bm{U}}_{d}^{*\top} \left( \tilde{\bm{\Lambda}}_{Z}^{d'*} \otimes \bm{1}_{N} \right) \right] \underline{\bm{\beta}}_{d'}^S \\
\text{(A2} &\text{, Non-informative enrollment)}\\
        \sum_{d'=1}^{J^*-1} E\left[ \tilde{\bm{U}}_{d}^{*\top} (\bm{\Lambda}_Z^{d'*} \otimes \textbf{I}_N) E[\bm{Y}^*(d')-\bm{Y}^*(0) | N, Z] \right] &= \sum_{d'=1}^{J^*-1} E\left[\tilde{\bm{U}}_{d}^{*\top} \left( \tilde{\bm{\Lambda}}_{Z}^{d'*} \otimes \bm{1}_{N} \right) \right] \underline{\bm{\beta}}_{d'}^S \\
\text{(A3} &\text{, Randomization)}\\
        \sum_{d'=1}^{J^*-1} E\left[ \tilde{\bm{U}}_{d}^{*\top} (\bm{\Lambda}_Z^{d'*} \otimes \textbf{I}_N) E[\bm{Y}^*(d')-\bm{Y}^*(0) | N] \right] &= \sum_{d'=1}^{J^*-1} E\left[\tilde{\bm{U}}_{d}^{*\top} \left( \tilde{\bm{\Lambda}}_{Z}^{d'*} \otimes \bm{1}_{N} \right) \right] \underline{\bm{\beta}}_{d'}^S \\
\text{(A2, Non-informative enrollment; A3} &\text{, Randomization)}\\
        \sum_{d'=1}^{J^*-1} E[\tilde{\bm{U}}_{d}^{*\top} (\bm{\Lambda}_Z^{d'*} \otimes \textbf{I}_N)] E\left[E[\bm{Y}^*(d')-\bm{Y}^*(0) | N] \right] &= \sum_{d'=1}^{J^*-1} E\left[\tilde{\bm{U}}_{d}^{*\top} \left( \tilde{\bm{\Lambda}}_{Z}^{d'*} \otimes \bm{1}_{N} \right) \right] \underline{\bm{\beta}}_{d'}^S \,.
\end{split}
\]
The above prove proceeds most similarly to the proof of consistency for period-specific treatment effect estimators.
The second equality results from the law of total expectation.
The third equality is a result of $\tilde{\bm{U}}_d^{*\top} (\bm{\Lambda}_Z^{d'*} \otimes \textbf{I}_N)$ being deterministic given $(\bm{S}, N, Z)$.
The fourth and fifth equalities then naturally extend from Assumptions A2 and A3.
In the final equality, $\tilde{\bm{U}}_d^{*\top} (\bm{\Lambda}_Z^{d'*} \otimes \textbf{I}_N)$ is a function of $Z$ and $N_{.j}$, which recall $N_{.j} \perp N$  (A2) and $Z \perp \bm{Y}^*(d'), \bm{Y}^*(0), N$ (A3).
In the absence of randomization (A3), as is the case in SW-CQTs, the fifth and sixth equalities can still proceed with the mean independence assumption of exchangeable treatment effects (A4).

Assuming a saturated treatment effect, alongside Assumption A1, $(\bm{\Lambda}_Z^{d'*} \otimes \textbf{I}_N) E[\bm{Y}^*(d')-\bm{Y}^*(0)] = (\bm{\Lambda}_Z^{d'*} \otimes \textbf{I}_N) E[\bm{Y}^*(d')-\bm{Y}^*(0)] = ( \tilde{\bm{\Lambda}}_{Z}^{d'*} \otimes \bm{1}_{N})\bm{\Delta}_{d'}^S \in \mathbb{R}^{NJ^*}$ where $\bm{\Delta}_d^S = (\Delta_{d}(d), ..., \Delta_{J^*-1}(d))^\top \in \mathbb{R}^{J^*-d}$.
Accordingly
\[
    \sum_{d'=1}^{J^*-1} E\left[\tilde{\bm{U}}_{d}^{*\top} \left( \tilde{\bm{\Lambda}}_{Z}^{d'*} \otimes \bm{1}_{N} \right) \right] \bm{\Delta}_{d'}^S = \sum_{d'=1}^{J^*-1} E\left[\tilde{\bm{U}}_{d}^{*\top} \left( \tilde{\bm{\Lambda}}_{Z}^{d'*} \otimes \bm{1}_{N} \right) \right] \underline{\bm{\beta}}_{d'}^S \,.
\]
Finally, by the regularity condition that the solution is unique, we have the desired consistency result
\[
\underline{\bm{\beta}}_{d}^S = \bm{\Delta}_d^S \,.
\]

Altogether, we establish the consistency of the working linear fixed-effects model within-transformed saturated treatment effect estimator for $\bm{\Delta}_d^S$ in a SW-CRT by relying on Assumptions A1 (Super-population sampling), A2 (Non-informative enrollment), and  A3 (Randomization), and in a SW-CQT by relying on assumptions A1, A2, and A4 (Mean independence).

\subsection{Proof of Fixed-Effects model constant treatment effect consistency for the P-ATO in a SW-CRT}
\label{app.sect:SW_const=P-ATO}

We can further demonstrate, using the same set of assumptions presented above, that the fixed-effects model constant treatment effect estimator is typically consistent the period-average treatment effect for the overlap population (P-ATO), generally defined as
\[
    \Delta^{P-ATO}  = \frac{ \sum_{j=1}^J \lambda_{j} \Delta_j}{\sum_{j=1}^J \lambda_{j}}
\]
with weights $\lambda_{j}=\pi_j^s(1-\pi_j^s)$ being the tilting function that generates the overlap propensity weights \citep{li_balancing_2018}, where $\pi_j^s$ is the proportion of clusters receiving treatment in period $j$.

As presented in Section \ref{app.sect:constant}, modeling a SW-CRT with a constant treatment effect results in the following score function
\[
   \sum_{d=1}^{J-1} E\left[ \bm{U}_{d}^\top \right]  \left( E\left[\bm{Y}(d)-\bm{Y}(0) \right] - \bm{1}_{NJ} \underline{\beta}_Z \right) = 0
\]
with
$\sum_{d=1}^{J-1} \bm{U}_{d}^\top = \left\{(\bm{\Delta}_{Z}\bm{1}_{J} - E[\bm{\Delta}_{Z}\bm{1}_{J}]) \otimes \bm{1}_N \right\}^\top \bm{\mathcal{M}}  \underline{\ddot{\bm{V}}} \bm{\mathcal{M}} \left\{ \bm{\Delta}_{Z} \otimes \textbf{I}_N\right\} \in \mathbb{R}^{1 \times NJ}$.
Assuming a true underlying period-specific treatment effect,
$E\left[\bm{Y}(d)-\bm{Y}(0) \right] = \bm{\Delta}^{P} \otimes 1_{N} =(\textbf{I}_{J} \otimes \bm{1}_N) \bm{\Delta}^{P} \in \mathbb{R}^{NJ}$.
Unlike Section \ref{app.sect:period}, we define the full vector of all period-specific treatment effects $\bm{\Delta}^P = (\Delta_1,...,\Delta_J)^\top = (\Delta_j)_{j=1,...,J}$.
Then, the previous score function can be re-written as
\[
     E\left[\sum_{d=1}^{J-1} \bm{U}_{d}^\top\right]  (\textbf{I}_{J} \otimes \bm{1}_N)  \left( \bm{\Delta}^{P} - \bm{1}_J\underline{\beta}_Z \right) = 0 \,.
\]

Recall, 
$\bm{\Delta}_{Z_i} = diag\{I\{Z_i \leq j\}: j=1,...,J\} \in \mathbb{R}^{J \times J}$,
$E[\bm{\Delta}_{Z}] = diag\{\frac{j-1}{J-1}: j=1,...,J\} \in \mathbb{R}^{J \times J}$,
$\left\{(\bm{\Delta}_{Z}\bm{1}_{J} - E[\bm{\Delta}_{Z}\bm{1}_{J}]) \otimes \bm{1}_N \right\}^\top = (\bm{\Delta}_{Z}\bm{1}_{J} - E[\bm{\Delta}_{Z}\bm{1}_{J}])^\top (\textbf{I}_{J} \otimes \bm{1}_N)^\top$
and,
$\bm{\mathcal{M}}_i \ddot{\bm{V}}_i \bm{\mathcal{M}}_i = \bm{\mathcal{M}}_i \bm{D}_i (\bm{D}_i^\top \ddot{\bm{\Sigma}}_i \bm{D}_i)^{-1} \bm{D}_i^\top \bm{\mathcal{M}}_i \in \mathbb{R}^{N_iJ \times N_iJ}$ is a symmetric matrix.
Define 
$[\bm{\mathcal{M}}  \underline{\ddot{\bm{V}}} \bm{\mathcal{M}}]_{\text{m},\text{n}}$ as the entry in row $\text{m}$ of column $\text{n}$ of the matrix $\bm{\mathcal{M}}  \underline{\ddot{\bm{V}}} \bm{\mathcal{M}}$; then define here
\[
    \mathcal{H}_{a,\text{n}} = \sum_{m=(a-1)N+1}^{aN} [\bm{\mathcal{M}}  \underline{\ddot{\bm{V}}} \bm{\mathcal{M}}]_{\text{m},\text{n}} \, \forall \, a \in \{1,...,J\}, \text{n} \in \{1,...,NJ\} \,, 
\]
summing entries from row $\text{m}=(a-1)N+1$ to $\text{a}N$ (such that we sum from row $1$ to $N$, $N+1$ to $2N$, ..., etc.) for $[\bm{\mathcal{M}}  \underline{\ddot{\bm{V}}} \bm{\mathcal{M}}]_{\text{m},\text{n}}$, given column $\text{n}$. Then
\[
\begin{split}
    \mathcal{H}_A &= \mathcal{H}_{a,\text{n}} \, \forall \, \text{n} \in \{(a-1)N+1,(a-1)N+2,...,aN\}, a \in \{1,...,J\} \,, \\
    \mathcal{H}_B &= \mathcal{H}_{a,\text{n}} \, \forall \, \text{n} \notin \{(a-1)N+1,(a-1)N+2,...,aN\}, a \in \{1,...,J\} \,.
\end{split}
\]
Altogether, this can be formulated as $(\textbf{I}_J \otimes \bm{1}_N)^\top\bm{\mathcal{M}}  \underline{\ddot{\bm{V}}} \bm{\mathcal{M}} = (\textbf{I}_J(\mathcal{H}_A-\mathcal{H}_B) + \bm{1}_J \bm{1}_J^\top \mathcal{H}_B) \otimes \bm{1}_N^\top \in \mathbb{R}^{J \times NJ}$.
Accordingly, we can then derive the following
\[
\begin{split}
    &(\textbf{I}_{J} \otimes \bm{1}_N)^\top \bm{\mathcal{M}}  \underline{\ddot{\bm{V}}} \bm{\mathcal{M}}\left\{ \bm{\Delta}_{Z} \otimes \textbf{I}_N\right\}  (\textbf{I}_{J} \otimes \bm{1}_N) \\
    &= \begin{pmatrix}
        0 & (I\{Z \leq j\} \mathcal{H}_B N)^\top_{j=2,...,J} \\
        \bm{0}_{J-1}  & \left( (I\{Z \leq j\} \mathcal{H}_B N)^\top_{j=2,...,J} \otimes \bm{1}_{J-1} \right) - diag\left(I\{Z \leq j\}(\mathcal{H}_B-\mathcal{H}_A)N \, : \,j=2,...,J\right)      
    \end{pmatrix} 
\end{split} \,.
\]
Altogether,
\[
\begin{split}
    E&\left[\sum_{d=1}^{J-1} \bm{U}_{d}^\top\right]  (\textbf{I}_{J} \otimes \bm{1}_N) \left( \bm{\Delta}^{P} - \bm{1}_J\underline{\beta}_Z \right) \\
    &= E\left[ \left(I\{Z \leq j\} - \frac{j-1}{J-1} \right)^\top_{j=1,...,J}
    (\textbf{I}_{J} \otimes \bm{1}_N)^\top \bm{\mathcal{M}}  \underline{\ddot{\bm{V}}} \bm{\mathcal{M}}\left\{ \bm{\Delta}_{Z} \otimes \textbf{I}_N\right\}  (\textbf{I}_{J} \otimes \bm{1}_N) 
    \left( \bm{\Delta}^{P} - \bm{1}_J\underline{\beta}_Z \right)\right] \\
    &= E\left[ \left(0, \left(I\{z \leq j\} N \left[\mathcal{H}_A\left(1-\frac{j-1}{J-1}\right) + \mathcal{H}_B \sum_{l \neq j }\left(I\{z \leq l\} - \frac{l-1}{J-1}\right) \right] \right)^\top_{j=2,...,J} \right) \left( \bm{\Delta}^{P} - \bm{1}_J\underline{\beta}_Z \right) \right] \\
    &= E\left[ \left(0, \left(I\{z \leq j\} N \mathcal{H}_A\left(\frac{J-j}{J-1}\right) + \mathcal{H}_B \sum_{l \leq j } \left(I\{z \leq l\} - \frac{l-1}{J-1}\right) \right)^\top_{j=2,...,J} \right) \left( \bm{\Delta}^{P} - \bm{1}_J\underline{\beta}_Z \right) \right] \\
    &=  \left(0, \left(\left(\frac{j-1}{J-1}\right) N \mathcal{H}_A \left(\frac{J-j}{J-1}\right) \right)^\top_{j=2,...,J} \right) \left( \bm{\Delta}^{P} - \bm{1}_J\underline{\beta}_Z \right) \\
    &= 0 \,. 
\end{split}
\]
Finally, we have
\[
    \underline{\beta}_Z = \frac{\sum_{j=1}^J \lambda_j \Delta_j}{ \sum_{j=1}^J \lambda_j } = \Delta^{P-ATO}
\]
where $\lambda_j = \pi_j^s(1-\pi_j^s)$ and $\pi_j^s = \frac{j-1}{J-1}$. Note $\lambda_1=\lambda_J=0$ in a SW-CT.

Altogether, we establish the consistency of the working linear fixed-effects model within-transformed constant treatment effect estimator for the period-average treatment effect for the overlap population (P-ATO) when there is a true period-specific treatment effect in a SW-CRT by relying on Assumptions A1 (Super-population sampling), A2 (Non-informative enrollment), and  A3 (Randomization), and in a SW-CQT by relying on assumptions A1, A2, and A4 (Mean independence).

\noindent $\square$

\newpage
\section{Proofs of Theorems \ref{app.Theorem:PB} \& \ref{app.Theorem:XO}}
\label{app.sect:Theorem_PB_XO_Proof}

\subsection{Extension to other longitudinal CRT designs}

In the previous sections, we demonstrate that a linear model with cluster fixed-effects can yield consistent treatment effect estimators for the true treatment effect structure, as long as the treatment effect structure is correctly specified, in a SW-CRT.
The SW-CRT is an example of a longitudinal CRT design, and can potentially have treatment effect structures that can vary over time by duration, period, or both (saturated), as we have previously highlighted.
Other examples of popular longitudinal and complete (all clusters are observed over the same number of periods) CRT designs include the parallel-with-baseline CRT (PB-CRT) and cluster randomized crossover (CRXO) designs.

It's then straightforward to extend the prior results from the SW-CRT with different treatment effect structures to other longitudinal and complete CRT designs.
Notably, these results with fixed-effects models only hold for longitudinal CRT designs that have both within-cluster and between-cluster variation in treatment status (i.e., SW-CRT, PB-CRT, CRXO).
For example, the fixed-effects model with a constant treatment effect structure cannot estimate a treatment effect if $\bm{Q}_{Z,i} = \bar{Q}_{Z,i} 1_J \, \forall \, i$ since $\ddot{\bm{Q}}_{Z,i}=0 \, \forall i$ and a treatment effect $\beta_Z$ is not estimable via a within-transformed fixed-effects estimator.
That is, more formally
\[
\exists \, i : \bm{Q}_{Z,i} \neq \bar{Q}_{Z,i} 1_J
\]
and there exists some $i$ where $\bm{Q}_{Z,i} \neq \bar{Q}_{Z,i}1_J$ and $\ddot{\bm{Q}}_{Z,i} \neq 0$.

We explicitly demonstrate this by outlining proofs of consistency of a fixed-effects model treatment effect estimator in a PB-CRT and CRXO.
We can then define the treatment effect estimands in PB-CRTs and CRXOs as was done previously in SW-CRTs (Equation \ref{app.eq:PO_estimands_SW}).

In a PB-CRT, we denote $d=\frac{z}{2}(j-1)$ as the duration of treatment exposure time, with $\tilde{z} \in \{0, z=2\}$.
Where $z=\frac{2d}{j-1}$, then $Y_{ijk}(z)=Y_{ijk}\left(\frac{2d}{j-1}\right)$.
The treatment effect estimand in a PB-CRT can then be defined as 
\begin{equation}
\label{app.eq:PO_estimands_PB}
\begin{split}
    \Delta_j(d) &= E\left[ \frac{1}{N} \sum_{k=1}^{N} Y_{.jk}\left(\frac{2d}{j-1}\right) \right] - E\left[ \frac{1}{N} \sum_{k=1}^{N} Y_{.jk}(0) \right] \\
    &= E\left[Y_{.jk}\left(\frac{2d}{j-1}\right)\right] - E[Y_{.jk}(0)] \,.
\end{split}
\end{equation}

In a CRXO, we denote $d=\frac{j+2-z}{2}$ as the duration of treatment exposure time, with $z \in \{1,2\}$.
Where $z=j-2d+2$, then $Y_{ijk}(z)=Y_{ijk}\left(j-2d+2\right)$.
The treatment effect estimand in a CRXO can then be defined as 
\begin{equation}
\label{app.eq:PO_estimands_XO}
\begin{split}
    \Delta_j(d) &= E\left[ \frac{1}{N} \sum_{k=1}^{N} Y_{.jk}\left(j-2d+2\right) \right] - E\left[ \frac{1}{N} \sum_{k=1}^{N} Y_{.jk}(0) \right] \\
    &= E\left[Y_{.jk}\left(j-2d+2\right)\right] - E[Y_{.jk}(0)] \,.
\end{split}
\end{equation}

\subsubsection{Proof of Fixed-Effects model constant treatment effect consistency in a PB-CRT}
\label{app.sect:PB_constant}

\textit{Proof of Theorem \ref{app.Theorem:PB}.}

Consider a $J$ period PB-CRT, with period $j=1$ representing the baseline period, and $j=2,...,J$ representing the follow-up periods (where some clusters receive treatment) in the PB-CRT design.
We then define the linear model corresponding to Equation \ref{app.eq:Y_constant}
\[
    \bm{Y}_i = (
    \textbf{I}_{J}
    \otimes \bm{1}_{N_i})\bm{\beta}_0 + \left((\bm{\Delta}_{Z_i,PB}\bm{1}_{J}) \otimes \bm{1}_{N_i}\right) \beta_Z + (\bm{1}_{J} \otimes \bm{X}_i) \bm{\beta}_X + \alpha_i \bm{1}_{N_iJ} + \bm{\gamma}_i \otimes \bm{1}_{N_i} + \bm{\epsilon}_i \,,
\]
where $\bm{Y}_i = (\bm{Y}_{i1}^\top,...,\bm{Y}_{iJ}^\top)^\top \in \mathbb{R}^{N_iJ}$, $\bm{\Delta}_{Z_i,PB} = diag\{(0, I\{Z_i=2\} \bm{1}_{J-1}^\top)^\top\} \in \mathbb{R}^{J \times J}$ as a diagonal matrix with the first diagonal entry being 0 and the remaining $J$ diagonal entries being 1 if $Z_i=2$ (where clusters begin receiving treatment in period $j=2$) or 0 if $Z_i=0$ (where clusters never receive the treatment).
As in a SW-CRT (Equation (\ref{app.eq:PO})), where the observed outcomes are related to the potential outcomes, we can also do so with in a PB-CRT with
\begin{equation}
\label{app.eq:PO_PB}
\begin{split}
    \bm{Y}_i  = \bm{Y}_i(0) + \sum_{d=1}^{J-1}(\bm{\Lambda}_{Z,PB}^d \otimes \textbf{I}_N)[\bm{Y}_i(d) - \bm{Y}_i(0)]
\end{split}
\end{equation}
where $\bm{\Lambda}_{Z,PB}^d = diag\{I\{Z=2, d=j-1 \} : j=1,...,J\} \in \mathbb{R}^{J \times J}$ and $\sum_{d=1}^{J-1} \bm{\Lambda}_{Z,PB}^d = \bm{\Delta}_{Z,PB}$.

As in the prior sections, we can then equivalently define using the fixed-effects model with 
Assumptions A1-A3 or Assumptions A1, A2, and A4.
\[
\begin{split}
    E\left[\left\{(\bm{\Delta}_{Z,PB}\bm{1}_{J} - E[\bm{\Delta}_{Z,PB}\bm{1}_{J}]) \otimes \bm{1}_N \right\}^\top \bm{\mathcal{M}}  \underline{\ddot{\bm{V}}} \bm{\mathcal{M}} \bm{Q}_0\right] &= 0 \,, \\
    E\left[\left\{(\bm{\Delta}_{Z,PB}\bm{1}_{J} - E[\bm{\Delta}_{Z,PB}\bm{1}_{J}]) \otimes \bm{1}_N \right\}^\top \bm{\mathcal{M}}  \underline{\ddot{\bm{V}}} \bm{\mathcal{M}} \bm{Q}_X\right] &= 0 \,, \\
    E\left[\left\{(\bm{\Delta}_{Z,PB}\bm{1}_{J} - E[\bm{\Delta}_{Z,PB}\bm{1}_{J}]) \otimes \bm{1}_N \right\}^\top \bm{\mathcal{M}}  \underline{\ddot{\bm{V}}} \bm{\mathcal{M}} (\textbf{I}_J \otimes \textbf{I}_N)\bm{Y}(0) \right] &= 0 \,,
\end{split}
\]
where
\[
\begin{split}
    \bm{U}_{d,PB}^\top &= \left\{(\bm{\Delta}_{Z,PB}\bm{1}_{J} - E[\bm{\Delta}_{Z,PB}\bm{1}_{J}]) \otimes \bm{1}_N \right\}^\top \bm{\mathcal{M}}  \underline{\ddot{\bm{V}}} \bm{\mathcal{M}} \left\{ \bm{\Lambda}_{Z,PB}^d \otimes \textbf{I}_N\right\} \in \mathbb{R}^{1 \times NJ} \,, \\
    \sum_{d=1}^{J-1} \bm{U}_{d,PB}^\top &= \left\{(\bm{\Delta}_{Z,PB}\bm{1}_{J} - E[\bm{\Delta}_{Z,PB}\bm{1}_{J}]) \otimes \bm{1}_N \right\}^\top \bm{\mathcal{M}}  \underline{\ddot{\bm{V}}} \bm{\mathcal{M}} \left\{ \bm{\Delta}_{Z,PB} \otimes \textbf{I}_N\right\} \in \mathbb{R}^{1 \times NJ} \,. \\
\end{split}
\]
Then
\[
    \sum_{d=1}^{J-1}  E\left[\bm{U}_{d,PB}^\top\right] E\left[\bm{Y}(d)-\bm{Y}(0) \right] = \sum_{d=1}^{J-1} E\left[ \bm{U}_{d,PB}^\top \bm{1}_{NJ} \right] \underline{\beta}_Z \,,
\]
where $E[\bm{Y}(d)-\bm{Y}(0)]=\bm{1}_{NJ}\Delta \in \mathbb{R}^{NJ} \, \forall d$ assuming a constant treatment effect.
Accordingly, we have the desired consistency result for a constant treatment effect in a PB-CRT
\[
\underline{\beta}_Z = \Delta \,.
\]

Altogether, we establish the consistency of the working linear fixed-effects model within-transformed constant treatment effect estimator for $\Delta$ in a PB-CRT by relying on Assumptions A1 (Super-population sampling), A2 (Non-informative enrollment), and  A3 (Randomization), and in a PB-CQT by relying on assumptions A1, A2, and A4 (Mean independence).

\subsubsection{Proof of Fixed-Effects model saturated treatment effect consistency in a PB-CRT}

In a PB-CRT, saturated, duration, and period-specific treatment effects coincide with $\bm{\Delta}^{S} = (\Delta_2(1), ...,\Delta_{J}(J-1))^\top = \bm{\Delta}^D = (\Delta(1), ...,\Delta(J-1))^\top = \bm{\Delta}^{P} = (\Delta_2, ...,\Delta_{J})^\top \in \mathbb{R}^{J-1}$.
In this case, with saturated treatment effects, we have $\Delta_j(d)$ where $d = j-1$.
For this proof, we will generally use the notation for duration-specific treatment effects $\bm{\Delta}^D$ to stand in for saturated treatment effects.

A PB-CRT with a saturated/duration/period-specific treatment effect structure can be specified similarly to the previous section but with the treatment indicator replaced by $(\bm{\Delta}_{Z_i,PB} \otimes \bm{1}_{N_i}) (0, \bm{\beta}_Z^{D\top})^\top$, in the linear model
\[
    \bm{Y}_i = (
    \textbf{I}_{J}
    \otimes \bm{1}_{N_i})\bm{\beta}_0 + \left(\bm{\Delta}_{Z_i,PB} \otimes \bm{1}_{N_i}\right) (0, \bm{\beta}_Z^{D\top})^\top + (\bm{1}_{J} \otimes \bm{X}_i) \bm{\beta}_X + \alpha_i \bm{1}_{N_iJ} + \bm{\gamma_i} \otimes \bm{1}_{N_i} + \bm{\epsilon}_i \,,
\]
where again $\bm{\Delta}_{Z,PB} = I\{Z_i=2\} \left[diag\{(0, \bm{1}_{J-1}^\top)^\top\}\right] \in \mathbb{R}^{J \times J}$ and $\bm{\beta}_Z^D = (\beta_{Z1}^D, ..., \beta_{Z,J-1}^D)^\top \in \mathbb{R}^{J-1}$ in a PB-CRT.
Accordingly,
\[
    \bm{U}_{PB}^{\top} = \left\{(\bm{\Delta}_{Z_i,PB} - E[\bm{\Delta}_{Z_i,PB}]) \otimes \bm{1}_N \right\}^\top \bm{\mathcal{M}}  \underline{\ddot{\bm{V}}} \bm{\mathcal{M}} \in \mathbb{R}^{J \times NJ} \,,
\]
Then, using the fixed-effects model with Assumptions A1-A3 or Assumptions A1, A2, and A4 results in
\[
    \sum_{d=1}^{J-1} E\left[\bm{U}_{PB}^{\top} \left\{ \bm{\Lambda}_{Z,PB}^d \otimes \textbf{I}_N\right\} \right] E\left[\bm{Y}(d)-\bm{Y}(0) \right] = E\left[\bm{U}_{PB}^{\top}  \left(\bm{\Delta}_{Z,PB} \otimes \bm{1}_{N}\right)  \right] (0, \underline{\bm{\beta}}_Z^{D\top})^\top
\]
where $E[\bm{Y}(d)-\bm{Y}(0)] = \bm{1}_{NJ}\Delta(d)$
assuming a saturated/duration/period-specific treatment effect.
Then
\[
\begin{split}
    & \sum_{d=1}^{J-1} E\left[\bm{U}_{PB}^{\top} \left\{ \bm{\Lambda}_{Z,PB}^d \otimes \textbf{I}_N\right\} \right] E\left[\bm{Y}(d)-\bm{Y}(0) \right] \\
    &= \sum_{d=1}^{J-1} E\left[\bm{U}_{PB}^{\top}  \left\{ \bm{\Lambda}_{Z,PB}^d \otimes \textbf{I}_N\right\} \bm{1}_{NJ}\Delta(d)\right] \\
    &= E\left[\bm{U}_{PB}^{\top} \sum_{d=1}^{J-1} \left\{ \bm{\Lambda}_{Z,PB}^d \otimes \textbf{I}_N\right\} \bm{1}_{NJ}\Delta(d)\right]
\end{split}
\]
Finally, we can show $\sum_{d=1}^{J-1} \left\{ \bm{\Lambda}_{Z,PB}^d \otimes \textbf{I}_N\right\} \bm{1}_{NJ}\Delta(d) = \left(\bm{\Delta}_{Z,PB} \otimes \bm{1}_{N}\right)(0,\bm{\Delta}^{D\top})^\top$.
Accordingly, we have the desired consistency result for saturated/duration/period-specific treatment effects in a PB-CRT
\[
    \underline{\bm{\beta}}_Z^D = \bm{\Delta}^D \,.
\]

Altogether, we establish the consistency of the working linear fixed-effects model within-transformed saturated/duration/period-specific treatment effect estimator for $\bm{\Delta}^S = \bm{\Delta}^D = \bm{\Delta}^P$ in a PB-CRT by relying on Assumptions A1 (Super-population sampling), A2 (Non-informative enrollment), and  A3 (Randomization), and in a PB-CQT by relying on assumptions A1, A2, and A4 (Mean independence).

\subsubsection{Proof of Fixed-Effects model constant treatment effect consistency for P-ATO in a PB-CRT}
\label{app.sect:PB_const=time-avg}

We can further demonstrate that the fixed-effects model constant treatment effect estimator is typically consistent for the period-average treatment effect for the overlap population (P-ATO), generally defined as
\[
    \Delta^{P-ATO}  = \frac{ \sum_{j=1}^J \lambda_{j} \Delta_j}{\sum_{j=1}^J \lambda_{j}}
\]
(with weights $\lambda_{j}=\pi_j^s(1-\pi_j^s)$ being the tilting function that generates the overlap propensity weights) which incidentally is equivalent to the saturated/duration/period-time-averaged treatment effects $\Delta^{P-ATO}=\Delta^{S-avg} = \Delta^{D-avg} = \Delta^{P-avg}$ in a PB-CRT, using the same set of assumptions presented above.

As presented in Section \ref{app.sect:PB_constant}, modeling a PB-CRT with a constant treatment effect results in the following score function
\[
   \sum_{d=1}^{J-1} E\left[ \bm{U}_{d,PB}^\top \right]  \left( E\left[\bm{Y}(d)-\bm{Y}(0) \right] - \bm{1}_{NJ} \underline{\beta}_Z \right) = 0
\]
with
$\sum_{d=1}^{J-1} \bm{U}_{d,PB}^\top = \left\{(\bm{\Delta}_{Z,PB}\bm{1}_{J} - E[\bm{\Delta}_{Z,PB}\bm{1}_{J}]) \otimes \bm{1}_N \right\}^\top \bm{\mathcal{M}}  \underline{\ddot{\bm{V}}} \bm{\mathcal{M}}   \left\{ \bm{\Delta}_{Z,PB} \otimes \textbf{I}_N\right\} \in \mathbb{R}^{1 \times NJ}$.
Assuming a true underlying saturated/duration/period-specific treatment effect,
$E[\bm{Y}(d)-\bm{Y}(0)] =  \bm{1}_{NJ} \Delta(d) \ \in \mathbb{R}^{NJ}$
\[
    \sum_{d=1}^{J-1} E\left[ \bm{U}_{d,PB}^\top \right] \bm{1}_{NJ}\left( \Delta(d) - \underline{\beta}_Z \right) = 0
\]
denoted above as duration-specific treatment effects (where in a PB-CRT, saturated, duration, and period-specific treatment effects coincide).

Recall, 
$\bm{\Delta}_{Z_i,PB} = I\{Z_i=2\} 
\left[diag\{(0,  \bm{1}_{J-1}^\top)^\top\} \right]\in \mathbb{R}^{J \times J}$, 
$\bm{\mathcal{M}}_i \ddot{\bm{V}}_i \bm{\mathcal{M}}_i = \bm{\mathcal{M}}_i \bm{D}_i (\bm{D}_i^\top \ddot{\bm{\Sigma}}_i \bm{D}_i)^{-1} \bm{D}_i^\top \bm{\mathcal{M}}_i \in \mathbb{R}^{N_iJ \times N_iJ}$ is a symmetric matrix.
Then $E[\bm{\Delta}_{Z_i,PB}] = 1/2 \left[diag\{(0, \bm{1}_{J-1}^\top)^\top\}\right]$ in a equally allocated PB-CRT (where half of all clusters are randomized to receive treatment) with a single baseline period when $j=1$.
Furthermore, define $[\bm{\mathcal{M}}  \underline{\ddot{\bm{V}}} \bm{\mathcal{M}}]_{m,n}$ as the entry in row $m$ of column $n$ of the matrix $\bm{\mathcal{M}}  \underline{\ddot{\bm{V}}} \bm{\mathcal{M}}$.
Then
\[
\begin{split}
    \left\{(\bm{\Delta}_{Z,PB}\bm{1}_{J} - E[\bm{\Delta}_{Z,PB}\bm{1}_{J}]) \otimes \bm{1}_N \right\}^\top &= \left\{ (0, (I\{Z=2\}-1/2) \bm{1}_{J-1}^\top)^\top  \otimes \bm{1}_N \right\}^\top \in \mathbb{R}^{1 \times NJ} \,, \\
    \left\{(\bm{\Delta}_{Z,PB}\bm{1}_{J} - E[\bm{\Delta}_{Z,PB}\bm{1}_{J}]) \otimes \bm{1}_N \right\}^\top \bm{\mathcal{M}}  \underline{\ddot{\bm{V}}} \bm{\mathcal{M}} &= \left(I\{Z=2\}-\frac{1}{2} \right) \left( \sum_{\text{m}=N+1}^{NJ} [\bm{\mathcal{M}}  \underline{\ddot{\bm{V}}} \bm{\mathcal{M}}]_{\text{m},1} , ..., \sum_{\text{m}=N+1}^{NJ} [\bm{\mathcal{M}}  \underline{\ddot{\bm{V}}} \bm{\mathcal{M}}]_{\text{m},NJ} \right) \\
    &= \left(I\{Z=2\}-\frac{1}{2} \right) (\mathcal{H}_1, ..., \mathcal{H}_{NJ}) \in \mathbb{R}^{1 \times NJ} \,.
\end{split}
\]
We define $\mathcal{H}_{\text{n}}=\sum_{\text{m}=N+1}^{NJ} [\bm{\mathcal{M}}  \underline{\ddot{\bm{V}}} \bm{\mathcal{M}}]_{\text{m,n}} \in \mathbb{R}^1$.
In a PB-CRT, $E[\mathcal{H}_{\text{n}} | \bm{S},N,Z] = E[\mathcal{H}_{\text{n}}]$ with Assumptions A2 (Non-informative enrollment) and A3 (Randomization).
Furthermore, $E[\mathcal{H}_{\text{n}}] = E[\mathcal{H}] \, \forall \text{n} \in \{N+1,...,NJ\}$, and
\[
\begin{split}
    E\left[ \sum_{d=1}^{J-1} \bm{U}_{d,PB}^\top \right] &= E[\left\{( \bm{\Delta}_{Z,PB}\bm{1}_{J} - E[\bm{\Delta}_{Z_i,PB}\bm{1}_{J}]) \otimes \bm{1}_N \right\}^\top \bm{\mathcal{M}}  \underline{\ddot{\bm{V}}} \bm{\mathcal{M}}   \left\{ \bm{\Delta}_{Z,PB} \otimes \textbf{I}_N \right\}] \\
    &= E\left[ \left(I\{Z=2\}-\frac{1}{2} \right) I\{Z=2\} (0 \otimes \bm{1}_N^\top , \mathcal{H}_{N+1} ,..., \mathcal{H}_{NJ}) \right] \\
    &= E\left[\left(\frac{1}{2} \right) I\{Z=2\} (0 \otimes \bm{1}_N^\top , \mathcal{H}_{N+1} ,..., \mathcal{H}_{NJ}) \right] \\
    &=  \left(\frac{1}{4} \right) (0 \otimes \bm{1}_N^\top , E[\mathcal{H}] \otimes \bm{1}_{N(J-1)}) \in \mathbb{R}^{1 \times NJ} \,.
\end{split}
\]
Finally, modeling with a constant treatment effect when there the true treatment effect structure is $\bm{\Delta}^{S} = (\Delta_2(1), ...,\Delta_{J}(J-1))^\top = \bm{\Delta}^D = (\Delta(1), ...,\Delta(J-1))^\top = \bm{\Delta}^{P} = (\Delta_2, ...,\Delta_{J})^\top \in \mathbb{R}^{J-1}$, produces 
\[
\begin{split}
    \sum_{d=1}^{J-1} E\left[ \bm{U}_{d,PB}^\top \right] \bm{1}_{NJ} \Delta(d) &= E\left[\sum_{d=1}^{J-1} \bm{U}_{d,PB}^\top  \right] \bm{1}_{NJ} \underline{\beta}_Z \\
    \sum_{d=1}^{J-1} E\left[   \left(\frac{1}{4} \right) \left( \sum_{\text{n}=Nd}^{N(d+1)} \mathcal{H}_{\text{n}} \right) \right] \Delta(d) &= \left(\frac{1}{4}\right) E[\mathcal{H}] N (J-1)  \underline{\beta}_Z \\
    E[\mathcal{H}] N \left(\sum_{d=1}^{J-1} \Delta(d) \right) &= E[\mathcal{H}] N (J-1)  \underline{\beta}_Z \\
    \sum_{d=1}^{J-1} \Delta(d) &= (J-1)  \underline{\beta}_Z \\
\end{split}
\]
proving that the working constant treatment effect estimator is consistent for the time averaged treatment effect when there is a true saturated/duration/period-specific treatment effect in a PB-CRT
\[
\begin{split}
    \underline{\beta}_Z &= \frac{\sum_{j=2}^{J} \Delta_j(j-1)}{J-1} = \frac{1}{J-1} \bm{1}_{J-1}^\top \bm{\Delta}^S = \Delta^{S-avg} \\
    &= \frac{\sum_{d=1}^{J-1} \Delta(d)}{J-1} = \frac{1}{J-1} \bm{1}_{J-1}^\top \bm{\Delta}^D = \Delta^{D-avg} \\
    & = \frac{\sum_{j=2}^{J} \Delta_j}{J-1} = \frac{1}{J-1} \bm{1}_{J-1}^\top \bm{\Delta}^P  = \Delta^{P-avg} \\
    & = \frac{ \sum_{j=1}^J \lambda_{j} \Delta_j}{\sum_{j=1}^J \lambda_{j}} = \Delta^{P-ATO}
\,.
\end{split}
\]

Altogether, we establish the consistency of the working linear fixed-effects model within-transformed constant treatment effect estimator for the P-ATO estimand, which is equivalent to the saturated/duration/period-time-averaged treatment effect estimand in a PB-CRT by relying on Assumptions A1 (Super-population sampling), A2 (Non-informative enrollment), and  A3 (Randomization), and in a PB-CQT by relying on assumptions A1, A2, and A4 (Mean independence).

\noindent $\square$

\newpage
\subsubsection{Proof of Fixed-Effects model constant treatment effect consistency in a CRXO}
\label{app.sect:XO_constant}

\textit{Proof of Theorem \ref{app.Theorem:XO}.}

In the simplest $J=2$ period CRXO, clusters randomized to $Z=1$ first receive the treatment in period $j=1$ (and control in $j=2$), and clusters randomized to $Z=2$ first receive the treatment in period $j=2$ (and control in $j=1$).
This definition can be extended to a $J$ period CRXO, where clusters randomized to $Z=1$ receive the treatment in odd number periods $(j=2a-1, a \in \mathbb{Z})$, and clusters randomized to $Z=2$ receive the treatment in even number periods $(j=2a, a \in \mathbb{Z})$.
In this article, we will focus on completely balanced CRXO designs with equal allocation of clusters to either $Z=1$ or $2$, with clusters followed for an even $J$ periods.

We then define the linear model corresponding to Equation \ref{app.eq:Y_constant}
\[
    \bm{Y}_i = (
    \textbf{I}_{J}
    \otimes \bm{1}_{N_i})\bm{\beta}_0 + \left((\bm{\Delta}_{Z_i,XO}\bm{1}_{J}) \otimes \bm{1}_{N_i}\right) \beta_Z + (\bm{1}_{J} \otimes \bm{X}_i) \bm{\beta}_X + \alpha_i \bm{1}_{N_i(J)} + \bm{\gamma}_i \otimes \bm{1}_{N_i} + \bm{\epsilon}_i \,,
\]
where provided $J$ is an even number, $\bm{\Delta}_{Z,XO} = \textbf{I}_{J/2} \otimes diag\{(I\{Z_i=1\},I\{Z
_i=2\})^\top\} \in \mathbb{R}^{J \times J}$.
Observed outcomes in a CRXO are related to the potential outcomes with
\begin{equation}
\label{app.eq:PO_XO}
\begin{split}
    \bm{Y}_i  = \bm{Y}_i(0) + \sum_{d=1}^{J/2}(\bm{\Lambda}_{Z,XO}^d \otimes \textbf{I}_N)[\bm{Y}_i(d) - \bm{Y}_i(0)]
\end{split}
\end{equation}
where $\bm{\Lambda}_{Z,XO}^d = diag\{I\{Z_i=1\}I\{d=(j+1)/2 \} + I\{Z_i=2\}I\{d=j/2 \} : j=1,...,J\} \in \mathbb{R}^{J \times J}$ 
and $\sum_{d=1}^{J/2} \bm{\Lambda}_{Z,XO}^d = \bm{\Delta}_{Z,XO}$.

As in the prior sections, we can then equivalently define
\[
\begin{split}
    E\left[\left\{(\bm{\Delta}_{Z_,XO}\bm{1}_{J} - E[\bm{\Delta}_{Z,XO}\bm{1}_{J}]) \otimes \bm{1}_N \right\}^\top \bm{\mathcal{M}}  \underline{\ddot{\bm{V}}} \bm{\mathcal{M}} \bm{Q}_0\right] &= 0 \,, \\
    E\left[\left\{(\bm{\Delta}_{Z,XO}\bm{1}_{J} - E[\bm{\Delta}_{Z,XO}\bm{1}_{J}]) \otimes \bm{1}_N \right\}^\top \bm{\mathcal{M}}  \underline{\ddot{\bm{V}}} \bm{\mathcal{M}} \bm{Q}_X\right] &= 0 \,, \\
    E\left[\left\{(\bm{\Delta}_{Z,XO}\bm{1}_{J} - E[\bm{\Delta}_{Z,XO}\bm{1}_{J}]) \otimes \bm{1}_N \right\}^\top \bm{\mathcal{M}}  \underline{\ddot{\bm{V}}} \bm{\mathcal{M}} (\textbf{I}_J \otimes \textbf{I}_N)\bm{Y}(0) \right] &= 0 \,,
\end{split}
\]
where
\[
\begin{split}
    \bm{U}_{d,XO}^\top &= \left\{(\bm{\Delta}_{Z,XO}\bm{1}_{J} - E[\bm{\Delta}_{Z,XO}\bm{1}_{J}]) \otimes \bm{1}_N \right\}^\top \bm{\mathcal{M}}  \underline{\ddot{\bm{V}}} \bm{\mathcal{M}} \left\{ \bm{\Lambda}_{Z,XO}^d \otimes \textbf{I}_N\right\} \in \mathbb{R}^{1 \times NJ} \,, \\
    \sum_{d=1}^{J/2} \bm{U}_{d,XO}^\top &= \left\{(\bm{\Delta}_{Z,XO}\bm{1}_{J} - E[\bm{\Delta}_{Z,XO}\bm{1}_{J}]) \otimes \bm{1}_N \right\}^\top \bm{\mathcal{M}}  \underline{\ddot{\bm{V}}} \bm{\mathcal{M}} \left\{ \bm{\Delta}_{Z,XO} \otimes \textbf{I}_N\right\} \in \mathbb{R}^{1 \times NJ} \,. \\
\end{split}
\]
Then, using the fixed-effects model with Assumptions A1-A3 or Assumptions A1, A2, and A4 results in
\[
    \sum_{d=1}^{J/2} E\left[\bm{U}_{d,XO}^\top \right] E\left[\bm{Y}(d)-\bm{Y}(0) \right] = E\left[\sum_{d=1}^{J/2} \bm{U}_{d,XO}^\top \bm{1}_{NJ} \right] \underline{\beta}_Z \,,
\]
where $E[\bm{Y}(d)-\bm{Y}(0)]=\bm{1}_{NJ}\Delta \in \mathbb{R}^{NJ} \, \forall d$ assuming a constant treatment effect.
Finally, we have the desired consistency result for a constant treatment effect in a CRXO
\[
\underline{\beta}_Z = \Delta \,.
\]

Altogether, we establish the consistency of the working linear fixed-effects model within-transformed constant treatment effect estimator for $\Delta$ in a CRXO by relying on Assumptions A1 (Super-population sampling), A2 (Non-informative enrollment), and  A3 (Randomization), and in a CQXO by relying on assumptions A1, A2, and A4 (Mean independence).

\subsubsection{Proof of Fixed-Effects model duration-specific treatment effect consistency in a CRXO}

The analysis of a CRXO with a fixed-effects model can be further defined with duration-specific treatment effects.
In contrast, a period-specific treatment effect is unidentifiable in a fixed-effects model analysis of a CRXO, due to perfect collinearity with the period indicators violating standard regularity condition 3 of Lemma \ref{app.lemma:variance}.

Again, we focus on completely balanced CRXO designs with clusters followed for an even $J$ periods.
Duration-specific treatment effects in a CRXO are then defined here with
$\bm{\Delta}^D = (\Delta(1), ...,\Delta(J/2))^\top \in \mathbb{R}^{J/2}$.
A CRXO with a duration-specific treatment effect structure can be specified similarly to the previous section but with the treatment indicator replaced by $(\bm{\Delta}_{Z_i,XO} \otimes \bm{1}_{N_i}) (\bm{\beta}_Z^D \otimes \bm{1}_2)$, in the linear model
\[
    \bm{Y}_i = (
    \textbf{I}_{J}
    \otimes \bm{1}_{N_i})\bm{\beta}_0 + \left(\bm{\Delta}_{Z_i,XO} \otimes \bm{1}_{N_i}\right) \left(\bm{\beta}_Z^D \otimes \bm{1}_2\right) + (\bm{1}_{J} \otimes \bm{X}_i) \bm{\beta}_X + \alpha_i \bm{1}_{N_iJ} + \bm{\gamma}_i \otimes \bm{1}_{N_i} + \bm{\epsilon}_i \,,
\]
where again $\bm{\Delta}_{Z,XO} = \textbf{I}_{J/2} \otimes diag\{(I\{Z_i=1\},I\{Z
_i=2\})^\top\} \in \mathbb{R}^{J \times J}$ and $\bm{\beta}_Z^D = ( \beta_{Z1}^D, ..., \beta_{Z,J/2}^D)^\top \in \mathbb{R}^{J/2}$ in a CRXO.
Accordingly,
\[
    \bm{U}_{XO}^{\top} = \left\{(\bm{\Delta}_{Z,XO} - E[\bm{\Delta}_{Z,XO}]) \otimes \bm{1}_N \right\}^\top \bm{\mathcal{M}}  \underline{\ddot{\bm{V}}} \bm{\mathcal{M}} \in \mathbb{R}^{J \times NJ} \,,
\]
Then, using the fixed-effects model with Assumptions A1-A3 or Assumptions A1, A2, and A4 results in
\[
    \sum_{d=1}^{J/2} E\left[\bm{U}_{XO}^{\top} \left\{ \bm{\Lambda}_{Z,XO}^d \otimes \textbf{I}_N\right\} \right] E\left[\bm{Y}(d)-\bm{Y}(0) \right] =  E\left[ \bm{U}_{XO}^{\top} \left(\bm{\Delta}_{Z,XO} \otimes \bm{1}_{N}\right) \right] \left(\underline{\bm{\beta}}_Z^D \otimes \bm{1}_2\right) \,,
\]
where $E[\bm{Y}(d)-\bm{Y}(0)] = \bm{1}_{NJ}\Delta(d)$
assuming a duration-specific treatment effect.
Then, the left side of the above equation is
\[
\begin{split}
    & \sum_{d=1}^{J-1} E\left[\bm{U}_{XO}^{\top} \left\{ \bm{\Lambda}_{Z,XO}^d \otimes \textbf{I}_N\right\} \right] E\left[\bm{Y}(d)-\bm{Y}(0) \right] \\
    &= \sum_{d=1}^{J-1} E\left[\bm{U}_{XO}^{\top} \left\{ \bm{\Lambda}_{Z,XO}^d \otimes \textbf{I}_N\right\} \bm{1}_{NJ}\Delta(d)\right] \\
    &= E\left[\bm{U}_{XO}^{\top} \sum_{d=1}^{J-1} \left\{ \bm{\Lambda}_{Z,XO}^d \otimes \textbf{I}_N\right\} \bm{1}_{NJ}\Delta(d)\right] \,.
\end{split}
\]
Furthermore, $\sum_{d=1}^{J-1} \left\{ \bm{\Lambda}_{Z,XO}^d \otimes \textbf{I}_N\right\} \bm{1}_{NJ}\Delta(d) = \left(\bm{\Delta}_{Z,XO} \otimes \bm{1}_{N}\right) (\bm{\Delta}^D \otimes \bm{1}_2)$.
Finally, we have the desired consistency result for duration-specific treatment effects in a CRXO
\[
    \underline{\bm{\beta}}_Z^D = \bm{\Delta}^D \,.
\]

Altogether, we establish the consistency of the working linear fixed-effects model within-transformed duration-specific treatment effect estimator for $\bm{\Delta}^D$ in a CRXO by relying on Assumptions A1 (Super-population sampling), A2 (Non-informative enrollment), and  A3 (Randomization), and in a CQXO by relying on assumptions A1, A2, and A4 (Mean independence).

\subsubsection{Proof of Fixed-Effects model constant treatment effect consistency for P-ATO in a CRXO}
\label{app.sect:XO_const=time-avg}

In a CRXO, the saturated and period-specific treatment effect structures coincide, $\bm{\Delta}^S = (\Delta_1(1),\Delta_2(1), \Delta_3(2),\Delta_4(2),...,\Delta_{J-1}(J/2), \Delta_J(J/2))^\top = \bm{\Delta}^P = (\Delta_1, ...,\Delta_{J})^\top \in \mathbb{R}^J$, and are a more general form of the duration-specific treatment effect. 
In this case with saturated treatment effects, we have $\Delta_j(d)$ where $d=\frac{2j+1-(-1)^j}{4}$.
With such a treatment effect structure, we can then define the period-average treatment effect for the overlap population (P-ATO), generally defined as
\[
    \Delta^{P-ATO}  = \frac{ \sum_{j=1}^J \lambda_{j} \Delta_j}{\sum_{j=1}^J \lambda_{j}}
\]
(with weights $\lambda_{j}=\pi_j^s(1-\pi_j^s)$ being the tilting function that generates the overlap propensity weights) which incidentally is equivalent to the saturated-average treatment effect estimand $\Delta^{P-ATO} = \Delta^{S-avg} = \Delta^{P-avg}$.
Despite the fixed-effects model being unidentifiable with a saturated/period-specific treatment effect structure in a CRXO due to collinearity with the period and cluster indicators, we can demonstrate that the fixed-effects model constant treatment effect estimator is still consistent for the saturated-average treatment effect estimand.

As presented in Section \ref{app.sect:XO_constant}, modeling a CRXO with a constant treatment effect results in the following score function
\[
    \sum_{d=1}^{J/2} E\left[\bm{U}_{d,XO}^\top \right] \left( E\left[\bm{Y}(d)-\bm{Y}(0) \right] - \bm{1}_{NJ} \underline{\beta}_Z \right) = 0
\]
with
$\sum_{d=1}^{J/2} \bm{U}_{d,XO}^\top = \left\{(\bm{\Delta}_{Z,XO}\bm{1}_{J} - E[\bm{\Delta}_{Z,XO}\bm{1}_{J}]) \otimes \bm{1}_N \right\}^\top \bm{\mathcal{M}}  \underline{\ddot{\bm{V}}} \bm{\mathcal{M}} \left\{ \bm{\Delta}_{Z,XO} \otimes \textbf{I}_N\right\} \in \mathbb{R}^{1 \times NJ}$. Assuming a true underlying saturated/period-specific treatment effect,
$E[\bm{Y}(d)-\bm{Y}(0)] =  \bm{\Delta}^P \otimes 1_N \in \mathbb{R}^{NJ}$, the above score function can be rewritten as
\[
    E\left[\sum_{d=1}^{J/2} \bm{U}_{d,XO}^\top \right] \left( (\bm{\Delta}^P \otimes 1_N) - \bm{1}_{NJ}\underline{\beta}_Z \right) = 0 \,.
\]

Recall, 
$\bm{\Delta}_{Z,XO} = \textbf{I}_{J/2} \otimes diag\{(I\{Z=1\},I\{Z=2\})^\top\} \in \mathbb{R}^{J \times J}$, 
$\bm{\mathcal{M}} \ddot{\bm{V}} \bm{\mathcal{M}} = \bm{\mathcal{M}} \bm{D} (\bm{D}^\top \ddot{\bm{\Sigma}} \bm{D})^{-1} \bm{D}^\top \bm{\mathcal{M}} \in \mathbb{R}^{NJ \times NJ}$ is a symmetric matrix.
Then $E[\bm{\Delta}_{Z,XO}] = 1/2 \left[diag\{\bm{1}_{J}\}\right]$ in a equally allocated CRXO (where half of all clusters are randomized to receive treatment).
Furthermore, define $[\bm{\mathcal{M}}  \underline{\ddot{\bm{V}}} \bm{\mathcal{M}}]_{m,n}$ as the entry in row $m$ of column $n$ of the matrix $\bm{\mathcal{M}}  \underline{\ddot{\bm{V}}} \bm{\mathcal{M}}$.
Then
\[
\begin{split}
    \left\{(\bm{\Delta}_{Z,XO}\bm{1}_{J} - E[\bm{\Delta}_{Z,XO}\bm{1}_{J}]) \otimes \bm{1}_N \right\}^\top &= \left\{ \left[ \bm{1}_{J/2} \otimes \left( I\{Z=1\}-\frac{1}{2}, I\{Z=2\}-\frac{1}{2} \right)^\top \right] \otimes \bm{1}_N\right\}^\top \in \mathbb{R}^{1 \times NJ} \\
    \left\{(\bm{\Delta}_{Z,XO}\bm{1}_{J} - E[\bm{\Delta}_{Z,XO}\bm{1}_{J}]) \otimes \bm{1}_N \right\}^\top \bm{\mathcal{M}}  \underline{\ddot{\bm{V}}} \bm{\mathcal{M}} &= (\mathcal{H}_1, ..., \mathcal{H}_{\text{n}}, ...,\mathcal{H}_{NJ}) \in \mathbb{R}^{1 \times NJ}
\end{split}
\]
where $\mathcal{H}_{\text{n}} = \sum_{b=1}^{J/2} \left\{ \left( I\{Z=1\}-\frac{1}{2} \right) \sum_{\text{m}=(2b-2)N+1}^{(2b-1)N} [\bm{\mathcal{M}}  \underline{\ddot{\bm{V}}} \bm{\mathcal{M}}]_{\text{m,n}} +\left( I\{Z=2\}-\frac{1}{2} \right) \sum_{\text{m}=(2b-1)N+1}^{2bN} [\bm{\mathcal{M}}  \underline{\ddot{\bm{V}}} \bm{\mathcal{M}}]_{\text{m,n}} \right\}$.

In a CRXO, $E[\mathcal{H}_{\text{n}} | \bm{S},N,Z] = E[\mathcal{H}_{\text{n}}]$ with Assumptions A3 (Randomization) A2 (Non-informative enrollment).
Furthermore, $E[\mathcal{H}_{\text{n}}] = E[\mathcal{H}] \, \forall \text{n}$, and
\[
\begin{split}
    E\left[\sum_{d=1}^{J/2} \bm{U}_{d,XO}^\top \right] &= E[\left\{( \bm{\Delta}_{Z,XO}\bm{1}_{J} - E[\bm{\Delta}_{Z,XO}\bm{1}_{J}]) \otimes \bm{1}_N \right\}^\top \bm{\mathcal{M}}  \underline{\ddot{\bm{V}}} \bm{\mathcal{M}}   \left\{ \bm{\Delta}_{Z,XO} \otimes \textbf{I}_N\right\}] \\
    &= E\left[ (\mathcal{H} \otimes \bm{1}_{NJ})^\top \left\{ \left( \textbf{I}_{J/2} \otimes diag\{(I\{Z=1\},I\{Z=2\})^\top\} \right) \otimes \textbf{I}_N \right\} \right] \\
    &= E\left[ \mathcal{H} \left\{ \bm{1}_{J/2}^\top \otimes (I\{Z=1\}, I\{Z=2\}) \otimes \bm{1}_{N}^\top \right\} \right] \\
    &= E\left[ \mathcal{H} \right] E\left[ \bm{1}_{J/2}^\top \otimes (I\{Z=1\}, I\{Z=2\}) \otimes \bm{1}_{N}^\top \right] \\
    &= E\left[ \mathcal{H} \right] \left( \frac{1}{2} \right) \bm{1}_{NJ}^\top \in \mathbb{R}^{1 \times NJ} \,.
\end{split}
\]
Finally, modeling with a constant treatment effect when there the true treatment effect structure is saturated/period-specific, $\bm{\Delta}^P = ( \Delta_1, ...,\Delta_{J})^\top \in \mathbb{R}^J$, produces
\[
\begin{split}
    E\left[\sum_{d=1}^{J/2} \bm{U}_{d,XO}^\top \right] (\bm{\Delta}^P \otimes 1_N) &=  E\left[ \sum_{d=1}^{J/2} \bm{U}_{d,XO}^\top \right] \bm{1}_{NJ} \underline{\beta}_Z \\
    E[\mathcal{H}] N (1/2) \left( \sum_{j=1}^{J} \Delta_j \right) &= E[\mathcal{H}] N  (J/2) \underline{\beta}_Z \\
    \sum_{j=1}^{J} \Delta_j &= (J)\underline{\beta}_Z \,.
\end{split}
\]
This altogether proves that the constant treatment effect is consistent for the saturated/period-time averaged treatment effect and P-ATO estimands in a CRXO
\[
\begin{split}
    \underline{\beta}_Z &= \frac{\sum_{j=1}^{J} \Delta_j\left(\frac{2j+1-(-1)^j}{4}\right)}{J} = \frac{1}{J} \bm{1}_{J}^\top \bm{\Delta}^S = \Delta^{S-avg} \\
    &= \frac{\sum_{j=1}^{J} \Delta_j}{J} = \frac{1}{J} \bm{1}_{J}^\top \bm{\Delta}^P = \Delta^{P-avg} \\
    & = \frac{ \sum_{j=1}^J \lambda_{j} \Delta_j}{\sum_{j=1}^J \lambda_{j}} = \Delta^{P-ATO} \,.
\end{split}
\]

Altogether, we establish the consistency of the working linear fixed-effects model within-transformed constant treatment effect estimator for the P-ATO estimand, which is equivalent to the saturated/period-time-averaged treatment effect in a CRXO by relying on Assumptions A1 (Super-population sampling), A2 (Non-informative enrollment), and  A3 (Randomization), and in a CQXO by relying on assumptions A1, A2, and A4 (Mean independence).

\noindent $\square$

\newpage
\section{Theorem \ref{app.Theorem:G} \& Lemma \ref{app.lemma:g}}

Despite its wide applicability, the linear fixed-effects model may not be a natural model choice when the outcome is not continuous.
In this section, we study robust inference for our marginal causal estimands via fixed-effects models which can be fit by GEE, to naturally handle non-continuous outcomes via non-identity-link functions.
We consider a general mean model for the individual-level data as
\begin{equation}
\label{app.eq:GEE_mean_model}
    E[Y_{ijk}| Z_i, \bm{X}_{ik}] = g^{-1}(\beta_{0j} + TE_{ij} + \bm{\beta}_X^\top \bm{X}_{ik} + \alpha_i) =g^{-1}(\bm{Q}_{ijk} \bm{\beta} + \alpha_i) = \mu_{ijk} 
\end{equation}
where $g$ is the link function, $\beta_{01}=0$ for identifiability, and $TE_{ij}$ is the treatment effect structure specified in the same way as in previous sections.
Following the most common practice, we assume that $g$ is a canonical link function. The vector of regression coefficients $\bm{\beta}$ (excluding the cluster fixed-effects coefficients) is estimated by $\hat{\bm{\beta}}$, a solution to the following GEE:
\begin{equation}
\label{app.eq:GEE}
    \sum_{i=1}^{m} \bm{U}_i^\top \bm{\mathcal{Z}}_i^{-1/2} \bm{R}_i^{-1} \bm{Z}_i^{-1/2} (\bm{Y}_i^o - \bm{\mu}_i^o) = 0 \,.
\end{equation}
where $\bm{Y}_i^o= (Y_{ijk})_{(j,k):S_{ijk}=1} \in \mathbb{R}^{\mathcal{N}_i}$ is the observed outcome vector across periods for cluster $i$,
$\bm{\mu}_i^o = \left(g^{-1}(\bm{Q}_{ijk} \bm{\beta} + \alpha_i)\right)_{(j,k):S_{ijk}=1} = (\mu_{ijk})_{(j,k):S_{ijk}=1} \in \mathbb{R}^{\mathcal{N}_i}$
is the mean function vector for all observed individuals in cluster $i$, $\bm{U}_i = \frac{d\bm{\mu}_i^o}{d\bm{\beta}}$ is the derivative matrix, $\bm{\mathcal{Z}}_i = \text{diag}\{v(Y_{ijk}) : j=1,...,J, k=1,...,N_i, S_{ijk}=1\}$ is the diagonal matrix of the natural variance functions $v(Y_{ijk})$, and $\bm{R}_i$ encodes the working correlation structure for the observed outcomes in cluster $i$.

\begin{apptheorem}
\label{app.Theorem:G}
    Under standard regularity conditions in Lemma \ref{app.lemma:variance}, assume Assumptions A1-A3, and at least one of the following supplemental conditions:
    (I.) $g$ is the identity-link with constant working variance $v(Y_{ijk}) = \sigma^2$,
    or (II.) $g$ is the log-link with working variance $v(Y_{ijk}) = \mu_{ijk}$ and $\rho = 0$ such that working independence is assumed.
    In the absence of Assumption A3, further assume alongside supplemental condition (II.), (III.) the mean model (Equation \ref{app.eq:GEE_mean_model}) is correctly specified.
    Then the following Central Limit Theorems hold for a SW-CT.
    That is, (a) under a true constant treatment effect structure, $\hat{V}_{GEE-g}^{-1/2} m^{1/2} (\hat{\Delta}_{GEE-g} - \Delta) \xrightarrow{d} N(0,1)$;
    (b) under a true duration-specific treatment effect structure, $(\hat{\bm{V}}_{GEE-g}^{D})^{-1/2} m^{1/2} (\hat{\bm{\Delta}}_{GEE-g}^D - \bm{\Delta}^D) \xrightarrow{d} N(0,\textbf{I}_{J-1})$;
    (c) under a true period-specific treatment effect structure, $(\hat{\bm{V}}_{GEE-g}^{P})^{-1/2} m^{1/2} (\hat{\bm{\Delta}}_{GEE-g}^P - \bm{\Delta}^P) \xrightarrow{d} N(0,\textbf{I}_{J-2})$ and $\hat{V}_{GEE-g}^{-1/2} m^{1/2} (\hat{\Delta}_{GEE-g} - \Delta^{P-ATO}) \xrightarrow{d} N(0,1)$; 
    and (d) under a true saturated treatment effect structure, $(\hat{\bm{V}}_{GEE-g}^{S})^{-1/2} m^{1/2} (\hat{\bm{\Delta}}_{GEE-g}^S - \bm{\Delta}^S) \xrightarrow{d} N(0,\textbf{I}_{(J-2)(J-1)/2})$.
\end{apptheorem}

To reiterate, in the absence of Assumption A3 (Randomization), as long as the mean model (Equation \ref{app.eq:GEE_mean_model}) is correctly specified, the resulting model-based g-computation estimator is consistent.
Unlike a similar generalized linear mixed-effects model where clusters are fit as random intercepts, the fixed-effects model described in equation \ref{app.eq:GEE_mean_model} does not restrict cluster intercepts to be uncorrelated with other model covariates.
Accordingly, the fixed-effects model automatically adjusts for all time-invariant, cluster-level confounding.
With the described fixed-effects models in a generalized  setting, the incidental parameters problem can be avoided with supplemental condition (I) via the within-transformation described in previous sections, or with supplemental condition (II) via a demonstrated equivalence to the conditional Poisson model, as we will describe (Section \ref{app.sect:Theorem_SW_G_Proof}).

Unlike the marginal mean model described in Wang et al. \citep{wang_how_2024}, the fixed-effects model is understood to target a conditional estimand.
However, even when using non-collapsible link functions (as is in supplemental condition (II)), g-computation with a fixed-effects model can still target a marginal estimand over the empirical distribution of the sampled clusters.
Asymptotically, the sample intuitively includes all clusters that comprise the entire population and the resulting g-computation estimator based on the fixed-effects model will target the marginal estimand over the entire population of clusters.
This holds even in non-linear settings.

\begin{applemma}
\label{app.lemma:g}
    Under either of the supplemental conditions (I) or (II) described in Theorem \ref{app.Theorem:G}, $\bm{\mathcal{Z}}_i^{-1/2} \bm{R}_i^{-1} \bm{\mathcal{Z}}_i^{-1/2} = \bm{\mathcal{Z}}_i^{-1}\bm{R}_i^{-1}$
\end{applemma}
\textit{Proof of Lemma \ref{app.lemma:g}.}
It suffices to show $\bm{R}_i^{-1}\bm{\mathcal{Z}}_i^{-1/2} = \bm{\mathcal{Z}}_i^{-1/2}\bm{R}_i^{-1}$.
For supplemental condition (I) that $v(Y_{ijk}) \equiv \sigma^2$, then $\bm{\mathcal{Z}}_i=\sigma^2 \textbf{I}_{\mathcal{N}_i}$ is a diagonal matrix, which yields the desired result.
For supplemental condition (II), with $\rho=0$, we have a working independence correlation $\bm{R}_i = \textbf{I}_{\mathcal{N}_i}$, which directly implies the desired result.
Then $\bm{R}_i^{-1} \bm{\mathcal{Z}}_i^{-1/2} = \bm{\mathcal{Z}}_i^{-1/2} \bm{R}_i^{-1}$, which completes the proof.
Finally, the above derivation also works if we replace $J$ with $J^* = J-1$, i.e., dropping the data from the last period for period-specific or saturated treatment effects. $\square$

Lemma \ref{app.lemma:g} then helps simplify some of the equations in the subsequent proofs.

\subsection{Proof of Theorem \ref{app.Theorem:G}}
\label{app.sect:Theorem_SW_G_Proof}

\noindent \textit{Proof of Theorem \ref{app.Theorem:G}.}

We inherit all notation from the proof of Theorems \ref{app.Theorem:SW}-\ref{app.Theorem:XO}. We first focus on the proof with supplemental condition (I) or (II), where the mean model $E[Y_{ijk}|Z_i, \bm{X}_{ik}] = g^{-1}(\beta_{0j} + TE_{ij} + \bm{\beta}_X^\top \bm{X}_{ik} + \alpha_i) =g^{-1}(\bm{Q}_{ijk} \bm{\beta} + \alpha_i) = \mu_{ijk}$ can be misspecified.
Then, $\bm{\mu}_i = \left(g^{-1}( \beta_{0j} + TE_{ij} + \bm{\beta}_X^\top \bm{X}_{ik} + \alpha_i )\right)_{k=1,...,N_i, j=1,...,J} = \left(g^{-1}(\bm{Q}_{ijk} \bm{\beta} + \alpha_i)\right)_{k=1,...,N_i, j=1,...,J} = \left(\mu_{ijk}\right)_{k=1,...,N_i, j=1,...,J} \in \mathbb{R}^{N_iJ}$.
The proof structure is similar to Theorems \ref{app.Theorem:SW}-\ref{app.Theorem:XO}, except for the additional g-computation step.
We can then construct the g-computation formula
\begin{equation}
\label{app.eq:g-comp}
    \hat{\mu}_j(b) = 
    \frac{
        \sum_{i=1}^{m} \sum_{l=1}^{J} \sum_{k:S_{ilk}=1} g^{-1}(\hat{\beta}_{0j} + b + \hat{\bm{\beta}}_X^\top \bm{X}_{ik} + \hat{\alpha}_i)
    }
    {
        \sum_{i=1}^{m} \sum_{l=1}^J N_{il}
    } \in \mathbb{R}^1
\end{equation}
where $b$ can be $\hat{\beta}_Z, \hat{\beta}_{Zd}, \hat{\beta}_{jZ}, \hat{\beta}_{jZd}$ or 0 depending on the assumed treatment effect structure.
This g-computation estimator for treatment $b$ and fixed period $j$ is taken as the empirical mean over all enrolled individuals ($S_{ilk}=1$) across all periods $l \in \{1,...,J\}$. 
For example, under a constant treatment effect structure, $\hat{\mu}_j(\hat{\beta}_Z)$ and $\hat{\mu}_j(0)$ target $E[Y_{ijk}(z)]$ (for the expected treated potential outcome) and $E[Y_{ijk}(0)]$ (for the expected untreated potential outcome), respectively.
Altogether, g-computation allows one to obtain estimates for model-free marginal estimands by standardizing across the covariate distribution in the target population.

The subsequent subsections will solely focus on proving Theorem \ref{app.Theorem:G} under supplemental condition (II).
Notably, for Lemma \ref{app.lemma:variance} to apply, the incidental parameters problem needs to be avoided, despite the inclusion of infinite cluster fixed intercepts.
Supplemental condition (I) leads to the same proofs as Theorems \ref{app.Theorem:SW}-\ref{app.Theorem:XO}, hence the within-transformation removes the infinite cluster fixed intercepts $\alpha_i$.
With supplemental condition (II), we will demonstrate that such a mean model specification allows us to draw on its equivalence to a full likelihood conditional Poisson model, allowing us to condition out the cluster fixed intercepts.

With supplemental condition (II) in Lemma \ref{app.lemma:g}, the consistency of the sandwich variance estimator under the M-estimation framework \citep{tsiatis_semiparametric_2006,van_der_vaart_asymptotic_1998,ross_m-estimation_2024} relies on the resulting equivalence between the score equations from such a semi-parametric GEE and a full-likelihood conditional Poisson model. This allows us to lean on the likelihood-based assumptions to condition out cluster fixed intercepts $\alpha_i$ and avoid contradicting condition (3) in Lemma \ref{app.lemma:variance}.
Assuming $Y_{ijk} \sim Poisson(\mu_{ijk})$ and $\mu_{ijk} = e^{\bm{Q}_{ijk}\bm{\beta} + \alpha_i}$, the Poisson likelihood is then
\[
    L_i(\bm{\beta},\alpha_i) = \prod_{j=1}^J \prod_{k:S_{ijk}=1} \frac{\mu_{ijk}^{Y_{ijk}}}{Y_{ijk}!}e^{-\mu_{ijk}}
\]
which yields an equivalent score equation to the third line of the above score function corresponding to the GEE, given supplemental condition (II) in Lemma \ref{app.lemma:g} \citep{cameron_microeconometrics_2005}.
The likelihood can be rewritten as
\[
\begin{split}
    L_i(\bm{\beta},\alpha_i) &= \left( \prod_{j=1}^J \prod_{k: S_{ijk}=1} \frac{e^{\bm{Q}_{ijk}\bm{\beta}}}{Y_{ijk}!} \right) 
    \left(e^{\alpha_i \left(\sum_{l=1}^J \sum_{k': S_{ilk'}=1} Y_{ilk'}\right)} e^{-e^{\alpha_i} \left(\sum_{l=1}^J \sum_{k': S_{ilk'}=1} e^{\bm{Q}_{ilk'}\bm{\beta}}\right)} \right) \\
    &= h\left(Y_{i11} ,...,Y_{iJN} ; \bm{\beta}\right) f\left(\sum_{l=1}^J \sum_{k: S_{ilk'}=1} Y_{ilk'}; \alpha_i \right)
\end{split}
\]
where $f\left(\sum_{l=1}^J \sum_{k': S_{ilk'}=1} Y_{ilk'}; \alpha_i \right)$ is a function only depending on $Y_{ijk}$ through the cluster-specific summed outcome over enrolled individuals $\sum_{l=1}^J \sum_{k': S_{ilk'}=1} Y_{ilk'}$, and $h(Y_{i11} ,...,Y_{iJN} ; \bm{\beta})$ is a function that does not depend on $\alpha_i$.
By the factorization theorem, $\sum_{l=1}^J \sum_{k': S_{ilk'}=1} Y_{ilk'}$ is a sufficient statistic for cluster intercepts $\alpha_i$.
Then, conditioning on this sufficient statistic for $\alpha_i$ results in the conditional Poisson likelihood:
\[
    L_i^{cond}(\bm{\beta}) = \frac{\prod_{j=1}^J \prod_{k:S_{ijk}=1} (e^{\bm{Q}_{ijk} \bm{\beta}})^{Y_{ijk}}/Y_{ijk}!}{\sum_{\{Y'_{ilk'} : \sum_{l=1}^J \sum_{k': S_{ilk'}=1} Y'_{ilk'} = \sum_{l=1}^J \sum_{k': S_{ilk'}=1} Y_{ilk'} \}}\prod_{l=1}^J \prod_{k': S_{ilk'}=1} (e^{\bm{Q}_{ilk'} \bm{\beta}})^{Y'_{ilk'}}/Y'_{ilk'}!}
\]
invoking a multinomial probability, where the denominator represents the sum of all possible combinations of non-negative integers $Y'_{ilk'}$ over given individual $k'$ and period $l$ such that $\sum_{l=1}^J \sum_{k': S_{ilk'}=1} Y'_{ilk'} = \sum_{l=1}^J \sum_{k': S_{ilk'}=1} Y_{ilk'}$. 
Cluster fixed intercepts can accordingly be treated as nuisance parameters and conditioned out by relying on the equivalence between the GEE given supplemental condition (II) in Lemma \ref{app.lemma:g} and a conditional Poisson model. 
Altogether, the incidental parameters problem of infinite $\alpha_i$ dummy variables is avoided, ensuring that Lemma \ref{app.lemma:variance}, and subsequently Theorem \ref{app.Theorem:G}, hold.

Alternatively, recall $\mu_{ijk} = g^{-1}(\bm{Q}_{ijk} \bm{\beta} + \alpha_i)$. With the previous assumptions and taking into account enrollment, the cluster fixed intercepts $\alpha_i$ can be solved for as
\[
    \underline{\alpha_i} = ln\left( \frac{\sum_{l=1}^J \sum_{k'=1}^{N_i} S_{ilk'}Y_{ilk'}}{\sum_{l=1}^J \sum_{k'=1}^{N_i} S_{ilk'}e^{\bm{Q}_{ilk'}\underline{\bm{\beta}}}} \right)
\]
\citep{cameron_microeconometrics_2005} (page 805).
Replacing these terms in the score equation yields an equivalent score equation to the conditional Poisson score equation.
Then $\underline{\bm{\mu}}_i$ is dependent solely on $\underline{\bm{\beta}}$ and the sufficient statistic $\sum_{j=1}^J \sum_{k: S_{ijk}=1} Y_{ijk}$
\[
    \underline{\bm{\mu}}_{i} = e^{\bm{Q}_i\underline{\bm{\beta}} + \underline{\alpha}_i\bm{1}_{N_iJ}} 
    = e^{\bm{Q}_i\underline{\bm{\beta}}}e^{\alpha_i} 
    = e^{\bm{Q}_i\underline{\bm{\beta}}} \left( \frac{\sum_{l=1}^J \sum_{k'=1}^{N_i} S_{ilk'}Y_{ilk'}}{\sum_{l=1}^J \sum_{k'=1}^{N_i} S_{ilk'}e^{\bm{Q}_{ilk'}\underline{\bm{\beta}}}} \right) \,,
\]
again demonstrating that the fixed-effects GEE, given supplemental condition (II) in Lemma \ref{app.lemma:g}, avoids the incidental parameter problem \citep{cameron_microeconometrics_2005} (page 805).

\subsubsection{Constant Treatment Effect}
\label{app.sect:Theorem_SW_G_Proof_constant}

We start by proving the results under the constant treatment effect setting.
By Lemma \ref{app.lemma:g}, the estimating equations become 
\[
    \sum_{i=1}^{m} \bm{U}_i^\top \bm{\mathcal{Z}}_i^{-1} \bm{R}_i^{-1}(\bm{Y}_i^o - \bm{\mu}_i^o) =  \sum_{i=1}^{m} \bm{U}_i^\top \bm{\mathcal{Z}}_i^{-1} (\bm{Y}_i^o - \bm{\mu}_i^o) = 0 \,.
\]
Again, supplemental condition (II) specifies $\bm{R}_i = \textbf{I}_{\mathcal{N}_i}$ with a working independence correlation structure.
Modeling with a constant treatment effect, $\bm{\mu}_i = \left(g^{-1}( \beta_{0j} + \beta_{Z}I\{Z_i \leq j\} + \bm{\beta}_X^\top \bm{X}_{ik} + \alpha_i )\right)_{k=1,...,N_i, j=1,...,J} = \left(g^{-1}(\bm{Q}_{ijk} \bm{\beta} + \alpha_i)\right)_{k=1,...,N_i, j=1,...,J} = \left(\mu_{ijk}\right)_{k=1,...,N_i, j=1,...,J}\} \in \mathbb{R}^{N_{i}J}$, then $\bm{\mu}_i^o = \bm{D}_i^\top \bm{\mu}_i \in \mathbb{R}^{\mathcal{N}_i}$.
Recall that $\bm{Q}_i = (\textbf{I}_J \otimes \bm{1}_{N_i}, (\bm{\Delta}_{Z_i} \bm{1}_J) \otimes \bm{1}_{N_i}, \bm{1}_J \otimes \bm{X}_i) \in \mathbb{R}^{N_iJ \times (J+1+p)}$ is the design matrix 
and $\bm{Y}_i^o = \bm{D}_i^\top \bm{Y}_i \in \mathbb{R}^{\mathcal{N}_i}$.
Since a canonical link function is used, then $U_i = \frac{d\bm{\mu}_i^o}{d\bm{\beta} } = \bm{D}_i^\top \frac{d\bm{\mu}_i}{d\bm{\beta}} = \bm{D}_i^\top (\bm{\mathcal{Z}}_i\bm{Q}_i)$, where with supplemental condition (II) $\bm{\mathcal{Z}}_i = \text{diag}\{\bm{\mu}_{i}\} \in \mathbb{R}^{N_iJ \times N_iJ}$.
Altogether, the above estimating equations become
\[
    \sum_{i=1}^{m} \bm{Q}_i^\top \bm{D}_i\bm{D}_i^\top(\bm{Y}_i - \bm{\mu}_i) = 0 \,.
\]
where recall $\bm{D}_i\bm{D}_i^\top = \text{diag}\{\bm{S}_i\} \in \mathbb{R}^{N_iJ \times N_iJ}$.

Denoting $\bm{\theta} = (\Delta_{GEE-g}, \left(\mu_j(b)\right)_{b\ \in \{\beta_Z, 0\}, j=1,...,J}^{\top},\bm{\beta}^{\top})^{\top}$ as the vector of the parameters, our estimator is a solution to the estimating equations $\sum_{i=1}^{m} \bm{\psi}(\bm{O}_i; \bm{\theta}) = 0$, where
\begin{equation}
\label{app.eq:GEE_c}
    \bm{\psi}(\bm{O}_i;\bm{\theta}) =  \left(
        \begin{gathered}
            \sum_{j=1}^J \lambda_{j} \{\Delta_{GEE-g} - (\mu_j(\beta_Z) - \mu_j(0)) \} \\
            \mathcal{N}_i \mu_j(b) - \bm{S}_i^\top \bm{h}_{ij}(b), b \in \{\beta_Z,0\}, j=1,...,J \\
            \bm{Q}_i^\top \bm{D}_i \bm{D}_i^\top (\bm{Y}_i - \bm{\mu}_i) 
        \end{gathered}
        \right)
\end{equation}
with 
$\mathcal{N}_i = \sum_{j=1}^J N_{ij}$, $\mu_j(b) \in \mathbb{R}^1$ being the g-computation estimator defined in Equation (\ref{app.eq:g-comp}), $\bm{h}_{ij}(b) = \left(g^{-1}(\beta_{0j} + b + \bm{\beta}_X^\top \bm{X}_{ik} + \alpha_i)\right)_{k=1,...,N_i} \otimes \bm{1}_J \in \mathbb{R}^{N_iJ}$, and $\lambda_{j} = \pi_j^s (1-\pi_j^s)$ as will be derived later on.
Intuitively, the first line of the above score functions are used to estimate the  weighted mean (GEE-g) of the g-computation estimators ($\mu_j(b)$), such that setting it equal to 0 returns
\[
\Delta_{GEE-g} = \frac{\sum_{i=1}^{m} \sum_{j=1}^J \lambda_{j}(\mu_j(\beta_Z)-\mu_j(0))}{\sum_{i=1}^{m} \sum_{j=1}^J \lambda_{j}} = \frac{\sum_{j=1}^J \lambda_{j}(\mu_j(\beta_Z)-\mu_j(0))}{\sum_{j=1}^J \lambda_{j}}
\]
such that the first line of the score functions does not depend on cluster $i$ and is $=0$.
This estimator draws an immediate and obvious resemblence to the P-ATO estimand.
Then the second line of the score functions connects the g-computation estimators $\mu_j(b)$ to the model-based means $\bm{h}_{ij}$.
The third line of the score function corresponds to the GEE estimating equation, as similarly detailed for the linear estimating equations in previous sections.
Much of the following proof will focus on the score function in the third line, then invoking the first and second lines of Equation (\ref{app.eq:GEE_c}) to finish the proof of consistency.

By Assumption A1 (super-population sampling) and regularity conditions 1-3 of Lemma \ref{app.lemma:variance}, Lemma \ref{app.lemma:variance} alongside the supplemental conditions in Lemma \ref{app.lemma:g} imply that $\bm{\theta}$ converges in probability to $\underline{\bm{\theta}} = (\underline{\Delta}_{GEE-g}, \left(\underline{\mu}_j(b)\right)_{b\ \in \{\beta_Z, 0\}, j=1,...,J}^{\top}, \underline{\bm{\beta}}^{\top})^{\top}$
that solves $E[\bm{\psi}(\bm{O};\bm{\theta})] = 0$, is asymptotically normally distributed, and has a consistent sandwich variance estimator $\hat{W}_{GEE-g}$ (Lemma \ref{app.lemma:variance}).
The estimating equation (Equation \ref{app.eq:GEE_c}) can be simply re-written as 
\[
    \bm{\psi}(\bm{O}_i;\bm{\theta}) =
    \begin{pmatrix}
        \psi_1(O_i,\theta) \\ \bm{\psi}_2(O_i,\theta) \\ \bm{\psi}_3(O_i,\theta)
    \end{pmatrix} \in \mathbb{R}^{2+3J+p} \,.
\]
We can then demonstrate
\[
\begin{split}
    \frac{d\bm{\psi}(\bm{O}_i,\bm{\theta})}{d\bm{\theta}^{\top}} & \mid_{\bm{\theta}=\hat{\bm{\theta}}}\\ 
    & = 
    \begin{pmatrix}
    \frac{d\psi_1(O_i,\theta)}{d \Delta_{GEE-g}} & \left(\frac{d\psi_1(O_i,\theta)}{d \mu_j(\beta_Z)} \right)^\top_{j=1,...,J} , \left(\frac{d\psi_1(O_i,\theta)}{d \mu_j(0)} \right)^\top_{j=1,...,J} & \frac{d\psi_1(O_i,\theta)}{d \bm{\beta}^\top} \\
    \frac{d\bm{\psi}_2(O_i,\theta)}{d \Delta_{GEE-g}} & \left(\frac{d\bm{\psi}_2(O_i,\theta)}{d \mu_j(\beta_Z)} \right)^\top_{j=1,...,J} , \left(\frac{d\bm{\psi}_2(O_i,\theta)}{d \mu_j(0)} \right)^\top_{j=1,...,J} & \frac{d\bm{\psi}_2(O_i,\theta)}{d \bm{\beta}^\top} \\
    \frac{d\bm{\psi}_3(O_i,\theta)}{d \Delta_{GEE-g}} & \left(\frac{d\bm{\psi}_3(O_i,\theta)}{d \mu_j(\beta_Z)} \right)^\top_{j=1,...,J} , \left(\frac{d\bm{\psi}_3(O_i,\theta)}{d \mu_j(0)} \right)^\top_{j=1,...,J} & \frac{d\bm{\psi}_3(O_i,\theta)}{d \bm{\beta}^\top} \\
    \end{pmatrix} \mid_{\bm{\theta}=\hat{\bm{\theta}}}
    \\
    &=
    \begin{pmatrix}
        \sum_{j=1}^J \lambda_j & \left(-\lambda_j\right)_{j=1,...,J}^\top,\left(\lambda_j \right)_{j=1,...,J}^\top & \bm{0}_{J+1+p}^\top \\
        \bm{0}_{2J}  & \mathcal{N}_i \textbf{I}_{2J} & \left( -\bm{S}_i^\top \frac{d \bm{h}_{ij}(b)}{d\bm{\beta}^\top} \right)_{b \in \{\beta_Z,0\}, j=1,...,J} \\
        \bm{0}_{J+1+p} & \bm{0}_{J+1+p} \bm{0}_{2J}^\top & -\bm{Q}_i^\top \bm{D}_i \bm{D}_i^\top \frac{d \bm{\mu}_i}{d\bm{\beta}^\top}
    \end{pmatrix} \\
    &\in \mathbb{R}^{(2+3J+p) \times (2+3J+p)} \,.
\end{split}
\]
Above, $J+1+p$ is the length of the parameters vector $\bm{\beta} \in \mathbb{R}^{J+1+p}$ (excluding cluster intercepts) (Section \ref{app.sect:constant}) and $\bm{0}_n \in \mathbb{R}^n$ is an $n$-dimension vector of 0's.
Building on the equivalence to the conditional Poisson model under the conditions in Lemma \ref{app.lemma:g}.II and Theorem \ref{app.Theorem:G}.II, the following derivatives are
\[
\begin{split}
    \frac{d \bm{\mu}_i}{d\bm{\beta}^\top} &= 
    \left(\left( e^{\bm{Q}_{ijk}\bm{\beta}} \frac{\sum_{l=1}^J \sum_{k'=1}^{N_i} S_{ilk'} Y_{ilk'}}{\sum_{l=1}^J \sum_{k'=1}^{N_i} S_{ilk'} e^{\bm{Q}_{ilk'}\bm{\beta}}} \right) \left(\bm{Q}_{ijk} - \frac{\sum_{l=1}^J \sum_{k'=1}^{N_i} S_{ilk'} e^{\bm{Q}_{ilk'}\bm{\beta}} \bm{Q}_{ilk'}}{\sum_{l=1}^J \sum_{k'=1}^{N_i} S_{ilk'} e^{\bm{Q}_{ilk'}\bm{\beta}}}\right)\right)_{j=1,...,J; k=1,...,N_i} \\
    &= 
    \left(\mu_{ijk} \left(\bm{Q}_{ijk} - \frac{\sum_{l=1}^J \sum_{k'=1}^{N_i} S_{ilk'} e^{\bm{Q}_{ilk'}\bm{\beta}} \bm{Q}_{ilk'}}{\sum_{l=1}^J \sum_{k'=1}^{N_i} S_{ilk'} e^{\bm{Q}_{ilk'}\bm{\beta}}}\right)\right)_{j=1,...,J; k=1,...,N_i}
    \in \mathbb{R}^{N_iJ \times (J+1+p)}
\end{split}
\]
and
\[
\begin{split}
    \frac{d \bm{h}_{ij}(b)}{d\bm{\beta}^\top}  &=
    \left(\left( e^{\bm{Q}_{ijk}(b)\bm{\beta}} \frac{\sum_{l=1}^J \sum_{k'=1}^{N_i} S_{ilk'} Y_{ilk'}}{\sum_{l=1}^J \sum_{k'=1}^{N_i} S_{ilk'} e^{\bm{Q}_{ilk'}\bm{\beta}}} \right) \left(\bm{Q}_{ijk}(b) - \frac{\sum_{l=1}^J \sum_{k'=1}^{N_i} S_{ilk'} e^{\bm{Q}_{ilk'}\bm{\beta}} \bm{Q}_{ilk'}}{\sum_{l=1}^J \sum_{k'=1}^{N_i} S_{ilk'} e^{\bm{Q}_{ilk'}\bm{\beta}}}\right)\right)_{k=1,...,N_i} \\
     &= 
    \left(\mu_{ijk}(b) \left(\bm{Q}_{ijk}(b) - \frac{\sum_{l=1}^J \sum_{k'=1}^{N_i} S_{ilk'} e^{\bm{Q}_{ilk'}\bm{\beta}} \bm{Q}_{ilk'}}{\sum_{l=1}^J \sum_{k'=1}^{N_i} S_{ilk'} e^{\bm{Q}_{ilk'}\bm{\beta}}}\right)\right)_{k=1,...,N_i}
    \in \mathbb{R}^{N_i \times (J+1+p)} \,.
\end{split}
\]
Altogether, these components can be combined with the estimating equations (Equation \ref{app.eq:GEE_c}) to derive the influence functions and sandwich variance estimator (Lemma \ref{app.lemma:variance}) under an M-estimation framework \citep{tsiatis_semiparametric_2006,van_der_vaart_asymptotic_1998,ross_m-estimation_2024}.

The only remaining part that then needs to be addressed is consistency, i.e., $\underline{\Delta}_{GEE-g} = \Delta$.
Using the fact that $E[\bm{\psi}(\bm{O};\underline{\bm{\theta}})]=0$ and $\bm{Q}=(\textbf{I}_J \otimes \bm{1}_N, (\bm{\Delta}_Z \bm{1}_J) \otimes \bm{1}_N, \bm{1}_J \otimes \bm{X})$, we have
\[
\begin{split}
    E&[(\textbf{I}_J \otimes \bm{1}_N)^\top \bm{D}\bm{D}^\top (\bm{Y}-\underline{\bm{\mu}})] = 0 \\
    E&[\{(\bm{\Delta}_Z\bm{1}_J) \otimes \bm{1}_N\}^\top \bm{D}\bm{D}^\top (\bm{Y}-\underline{\bm{\mu}})] = 0 \,,
\end{split}
\]
where $\underline{\bm{\mu}} = \left(g^{-1}(\underline{\beta}_{0j} + \underline{\beta}_Z I\{Z \leq j\} + \underline{\bm{\beta}}_X^\top \bm{X}_{.k} + \underline{\alpha})\right)_{k=1,...,N, j=1,...,J} \in \mathbb{R}^{NJ}$.
Recall $E[\bm{\Delta}_Z \bm{1}_J] = (\pi_1^s,...,\pi_J^s)^\top$.
By left-multiplying the first equation by $-E[\bm{\Delta}_Z \bm{1}_J]^\top$ and adding it to the second equation, we get
\[
    E[\{(\bm{\Delta}_Z\bm{1}_J - E[\bm{\Delta}_Z \bm{1}_J]) \otimes \bm{1}_N\}^\top \bm{D} \bm{D}^\top (\bm{Y}-\underline{\bm{\mu}})] = 0 \,.
\]
By the assumption that the treatment effect is duration invariant, we have $\bm{Y} = (\bm{\Delta}_Z \otimes \textbf{I}_N) \{\bm{Y}(d) - \bm{Y}(0)\} + \bm{Y}(0)$ in the above equation.
Furthermore, we have $\underline{\bm{\mu}} = (\bm{\Delta}_Z \otimes \textbf{I}_N) \{\underline{\bm{\mu}}(\underline{\beta}_Z) - \underline{\bm{\mu}}(0)\} + \underline{\bm{\mu}}(0)$, where $\underline{\bm{\mu}}(b) = \left(g^{-1}(\underline{\beta}_{0j} + b + \underline{\bm{\beta}}_X^\top \bm{X}_{.k} + \underline{\alpha})\right)_{k=1,...,N, j=1,...,J} \in \mathbb{R}^{NJ}$.

Where with Assumption A3 (randomization) or assuming the mean model is correctly specified (as would be necessary for consistency results to hold for longitudinal CQTs)
\[
    E[\{(\bm{\Delta}_Z\bm{1}_J - E[\bm{\Delta}_Z \bm{1}_J]) \otimes \bm{1}_N\}^\top \bm{D}\bm{D}^\top \{\bm{Y}(0) - \underline{\bm{\mu}}(0) \}] = 0 \,.
\]
This altogether implies
\[
    E[\{(\bm{\Delta}_Z\bm{1}_J - E[\bm{\Delta}_Z \bm{1}_J]) \otimes \bm{1}_N\}^\top \bm{D}\bm{D}^\top (\bm{\Delta}_Z \otimes \textbf{I}_N )(\{\bm{Y}(d) - \bm{Y}(0)\} - \{\underline{\bm{\mu}}(\underline{\beta}_Z) - \underline{\bm{\mu}}(0) \})] = 0 \,.
\]
To reiterate
\[
\begin{split}
    \bm{Y}(d)-\bm{Y}(0) &\in \mathbb{R}^{NJ} \\
    \underline{\bm{\mu}}(\underline{\beta}_Z) - \underline{\bm{\mu}}(0) &\in \mathbb{R}^{NJ} \\
    \left\{(\bm{\Delta}_Z \bm{1}_{J} - E[\bm{\Delta}_Z \bm{1}_{J}]) \otimes \bm{1}_N \right\}^\top \bm{D}\bm{D}^\top \left\{ \bm{\Delta}_Z \otimes \textbf{I}_N\right\} &\in \mathbb{R}^{1 \times NJ} \\
    \left\{(\bm{\Delta}_Z \bm{1}_{J} - E[\bm{\Delta}_Z \bm{1}_{J}]) \otimes \bm{1}_N \right\}^\top \bm{D}\bm{D}^\top \left\{ (\bm{\Delta}_Z \bm{1}_{J})\otimes \bm{1}_N\right\} &\in \mathbb{R}^{1} \,.
\end{split}
\]
As in the linear setting (Section \ref{app.sect:constant}), by Assumptions A2 (non-informative enrollment) and A3 (randomization), the above equation becomes
\[
    E[\{(\bm{\Delta}_Z\bm{1}_J - E[\bm{\Delta}_Z \bm{1}_J]) \otimes \bm{1}_N\}^\top \bm{D}\bm{D}^\top (\bm{\Delta}_Z \otimes \textbf{I}_N )(E[\bm{Y}(d) - \bm{Y}(0)|N] - 
    E[\underline{\bm{\mu}}(\underline{\beta}_Z) - \underline{\bm{\mu}}(0)|N]
    )] = 0 \,.
\]
We observe that the above formula is a function of $N_{.j}$ and $Z$, and is independent of potential outcomes and covariates by Assumptions A2 (non-informative enrollment) and A3 (randomization).
Furthermore, the $j$-th component of the above formula has expectation $E[\lambda_{ij}]$.
Under the setting of constant treatment effects, $E[\bm{Y}(d)-\bm{Y}(0)] = \Delta \bm{1}_{NJ}$
and
\[
\begin{split}
    \Delta = E[&\{(\bm{\Delta}_Z\bm{1}_J - E[\bm{\Delta}_Z \bm{1}_J]) \otimes \bm{1}_N\}^\top \bm{D} \bm{D}^\top (\bm{\Delta}_Z \bm{1}_J \otimes \bm{1}_N)]^{-1} \\
    & E[\{(\bm{\Delta}_Z\bm{1}_J - E[\bm{\Delta}_Z \bm{1}_J]) \otimes \bm{1}_N\}^\top \bm{D} \bm{D}^\top (\bm{\Delta}_Z \otimes \textbf{I}_N)] E[\underline{\bm{\mu}}(\underline{\beta}_Z) - \underline{\bm{\mu}}(0)|N]
\end{split}
\]
Since $\underline{\bm{\mu}}(b)$ is a function of $\bm{X}$, Assumption A1 (super-population sampling) implies that $E[\underline{\bm{\mu}}(b) | N] \in \mathbb{R}^{NJ}$ can be written as
$E[\tilde{\underline{\bm{\mu}}}(b) | N] \otimes \bm{1}_N = \tilde{\underline{\bm{\mu}}}(b) \otimes \bm{1}_N$, where 
$\tilde{\underline{\bm{\mu}}}(b) = \left(E[g^{-1}(\underline{\beta}_{0j} + b + \underline{\bm{\beta}}_X^\top \bm{X} + \underline{\alpha})]\right)_{j=1,...,J} \in \mathbb{R}^J$.
Accordingly, the above equation can then be further simplified to
\[
\begin{split}
    \Delta &= E[\{(\bm{\Delta}_Z\bm{1}_J - E[\bm{\Delta}_Z \bm{1}_J]) \otimes \bm{1}_N\}^\top \bm{D} \bm{D}^\top (\bm{\Delta}_Z \bm{1}_J \otimes \bm{1}_N)]^{-1} \\
    & \,\,\,\,\,\,\,\,\,\,\,\, E[\{(\Delta_Z\bm{1}_J - E[\bm{\Delta}_Z \bm{1}_J]) \otimes \bm{1}_N\}^\top \bm{D} \bm{D}^\top (\bm{\Delta}_Z \otimes \bm{1}_N)] 
    \{\tilde{\underline{\bm{\mu}}}(\underline{\beta}_Z) - \tilde{\underline{\bm{\mu}}}(0)\} \\
    &= \left(\sum_{j=1}^J \lambda_{j}\right)^{-1} \sum_{j=1}^J \lambda_{j} E\left[ g^{-1}(\underline{\beta}_{0j} + \underline{\beta}_Z + \underline{\bm{\beta}}_X^\top \bm{X}+ \underline{\alpha}
    ) - g^{-1}(\underline{\beta}_{0j} + \underline{\bm{\beta}}_X^\top \bm{X} + \underline{\alpha}
    ) \right] \\
\end{split}
\]
where $g$ is the log-link, and $\lambda_{j} = \pi_j^s (1-\pi_j^s)$ \citep{wang_how_2024} is the tilting function that generates the overlap propensity weights embedded within the GEE estimator with an independence working correlation and an identity-link function \citep{li_balancing_2018}.
Notably, $\lambda_{1} = \lambda_{J} = 0$, which correspond with the periods that lack treatment positivity, the first period (where all clusters receive the control, $b=0$) and final period (where all clusters receive the treatment, $b=\beta_Z$).
To illustrate, in a simple example of a 3 sequence, 4 period SW-CRT with total enrollment, we would have $(\lambda_{1}, \lambda_{2}, \lambda_{3}, \lambda_{4}) = (0, 2/9, 2/9, 0)$.

To show $\underline{\Delta}_{GEE-g} = \Delta$, we recall that the first two entries of estimating equations (\ref{app.eq:GEE_c}) imply
\[
\begin{split}
    E\left[\sum_{j=1}^J \lambda_{j} \{\Delta_{GEE-g} - (\mu_j(\beta_Z) - \mu_j(0)) \}\right] &=0 \,, \\
    E\left[\mathcal{N} \mu_j(b) - \bm{S}^\top \bm{h}_{.j}(b)\right] &= 0 \,.
\end{split}
\]
The first equation then implies that
\[
\begin{split}
    \underline{\Delta}_{GEE-g} &= \left(\sum_{j=1}^J \lambda_{j}\right)^{-1} \sum_{j=1}^J \lambda_{j} E\left[ \mu_j(\underline{\beta}_Z) - \mu_j(0) \right]  \\
    & = \left(\sum_{j=1}^J \lambda_{j}\right)^{-1} \sum_{j=1}^J \lambda_{j} \left(\mu_j(\underline{\beta}_Z) - \mu_j(0) \right) 
\end{split}
\]
with the second equality arising from the definition of the g-computation formula (Equation \ref{app.eq:g-comp}).
The second entry in the estimating equations (\ref{app.eq:GEE_c}) indicates that $E[\underline{\mu}_j(b)] = \underline{\mu}_j(b) = E\left[\mathcal{N} \right]^{-1}E[\bm{S}^\top \bm{h}_{.j}(b)] = E[g^{-1}(\underline{\beta}_{0j} + b + \underline{\bm{\beta}}_X^\top \bm{X} + \underline{\alpha}
)] \in \mathbb{R}^1$ using Assumption A2 (non-informative enrollment).
\[
\begin{split}
    \underline{\Delta}_{GEE-g} &= \left(\sum_{j=1}^J \lambda_{j}\right)^{-1} \sum_{j=1}^J \lambda_{j} E\left[ g^{-1}(\underline{\beta}_{0j} + \underline{\beta}_Z + \underline{\bm{\beta}}_X^\top \bm{X} + \underline{\alpha}) - g^{-1}(\underline{\beta}_{0j} + \underline{\bm{\beta}}_X^\top \bm{X} + \underline{\alpha}) \right] \\
    &= \Delta    
\end{split}
\]
which completes the proof of consistency.

\subsubsection{Duration-specific Treatment Effect}
\label{app.sect:Theorem_SW_G_Proof_ds}

Similarly, we can straightforwardly prove the results under the duration-specific treatment effect setting. The corresponding estimating equations are
\begin{equation}
\label{app.eq:GEE_ds}
    \bm{\psi}(\bm{O}_i;\bm{\theta}) =  \left(
        \begin{gathered}
            \sum_{d=1}^{J-1} \bm{\lambda}_{i}(d) \left[ \textbf{A}_d \bm{\Delta}_{GEE-g}^D - \{ \left(\mu_1(\beta_{Zd}) - \mu_1(0) \right) ,..., \left(\mu_J(\beta_{Zd}) - \mu_J(0) \right) \}^\top \right] \\
            \mathcal{N}_i \mu_j(b) - \bm{S}_i^\top \bm{h}_{ij}(b), j=1,...,J, b \in \{\beta_{Z1},  ..., \beta_{Z,J-1}, 0\} \\
            \bm{Q}_i^\top \bm{D}_i \bm{D}_i^\top (\bm{Y}_i - \bm{\mu}_i) 
        \end{gathered}
        \right)
\end{equation}
where $\textbf{A}_d \in \mathbb{R}^{J \times (J-1)}$ is a matrix with the $d$th column equal to $\bm{1}_J$ and all other elements 0, $\bm{Q}_i = (\textbf{I}_J \otimes \bm{1}_{N_i}, \textbf{H}_{Z_i} \otimes \bm{1}_{N_i}, \bm{1}_J \otimes \bm{X}_i)$, $\mu_j(b)$ is the g-computation estimator, $\bm{h}_{ij}(b)$ are defined as in the constant treatment effect setting, and
\[
    \bm{\lambda}_i(d) = \{(\textbf{H}_{Z_i} - E[\textbf{H}_{Z}]) \otimes \bm{1}_{N_i}\}^\top \bm{D}_i \bm{D}_i^\top (\bm{\Lambda}_{Z_i}^d \otimes \bm{1}_{N_i}) \in \mathbb{R}^{(J-1) \times J}
\]
as we will demonstrate.
Recall $\textbf{H}_{Z_i} \in \mathbb{R}^{J \times (J-1)}$ is defined in Section \ref{app.sect:duration}.
Importantly, the first column of $\bm{\lambda}_i(d)$ is composed of all 0's, hence only $\mu_1(0)$ is observed due to $j=1$ being the baseline period.
Some algebra shows that $\bm{\lambda}_i(d)$ is a function of $N_{ij}$ and $Z_i$ but not $\bm{S}_i$ nor $N_i$. Then $E[\bm{\psi}(\bm{O}, \underline{\bm{\theta}})] = 0$ implies that
\[
\begin{split}
    E&[(\textbf{I}_J \otimes \bm{1}_N)^\top \bm{D}\bm{D}^\top (\bm{Y}-\underline{\bm{\mu}})] = 0 \\
    E&[(\textbf{H}_Z \otimes \bm{1}_N)^\top \bm{D}\bm{D}^\top (\bm{Y}-\underline{\bm{\mu}})] = 0 \,.
\end{split}
\]
Left multiplying the first equation by $-E[\textbf{H}_Z]^\top$ and adding it to the second equation, we get
\[
    E[\{(\textbf{H}_Z  -E[\textbf{H}_Z]) \otimes \bm{1}_N\}^\top \bm{D}\bm{D}^\top (\bm{Y}-\underline{\bm{\mu}})] = 0 \,.
\]
For $\underline{\bm{\mu}}$, we have
\[
    \underline{\bm{\mu}} = \sum_{d=1}^{J-1} (\bm{\Lambda}_{Z}^d \otimes \textbf{I}_{N}) \{ \underline{\bm{\mu}}(\underline{\beta}_{Zd}) - \underline{\bm{\mu}}(0) \} + \{\textbf{I}_J \otimes \textbf{I}_{N}\} \underline{\bm{\mu}}(0) \,.
\]
Using Equation \ref{app.eq:PO}, we get
\[
    \sum_{d=1}^{J-1} E[ \{ (\textbf{H}_Z  -E[\textbf{H}_Z]) \otimes \bm{1}_N \}^\top \bm{D}\bm{D}^\top (\bm{\Lambda}_{Z_i}^d \otimes \textbf{I}_{N_i}) (\{\bm{Y}(d)-\bm{Y}(0)\} - \{\underline{\bm{\mu}}(\underline{\beta}_{Zd}) - \underline{\bm{\mu}}(0) \}) ] = 0 \,.
\]
The proof here extends as seen in earlier proofs,
where
\[
\begin{split}
    E[\bm{Y}(d)-\bm{Y}(0) | N] &\in \mathbb{R}^{NJ} \\
    \underline{\bm{\mu}}(\underline{\beta}_{Zd}) - \underline{\bm{\mu}}(0) &\in \mathbb{R}^{NJ} \\
    \tilde{\underline{\bm{\mu}}}(\underline{\beta}_{Zd}) - \tilde{\underline{\bm{\mu}}}(0) &\in \mathbb{R}^J \\
    \bm{\Lambda}_Z^d &\in \mathbb{R}^{J \times J} \,.
\end{split}
\]
Recall that $\tilde{\underline{\bm{\mu}}}(b) = \left(E[g^{-1}(\underline{\beta}_{0j} + b + \underline{\bm{\beta}}_X^\top \bm{X} + \underline{\alpha})] \right)_{j=1,...,J} \in \mathbb{R}^J$ and with Assumption A1 (super-population sampling)
\[
    (\bm{\Lambda}_{Z}^d \otimes \textbf{I}_{N})  E[\underline{\bm{\mu}}(\underline{\beta}_{Zd}) - \underline{\bm{\mu}}(0)|N] = (\bm{\Lambda}_{Z}^d \otimes \bm{1}_{N}) \{\tilde{\underline{\bm{\mu}}}(\underline{\beta}_{Zd}) - \tilde{\underline{\bm{\mu}}}(0)\} \,.
\]
The duration-specific treatment effect then implies
\[
\begin{split}
    (\bm{\Lambda}_Z^d \otimes \textbf{I}_N) E[\bm{Y}(d)-\bm{Y}(0) | N] &= (\bm{\Lambda}_Z^d \otimes \textbf{I}_N) \{\Delta(d)\bm{1}_{NJ} \in \mathbb{R}^{NJ}\} \\
    \sum_{d=1}^{J-1} (\bm{\Lambda}_Z^d \otimes \textbf{I}_N) E[\bm{Y}(d)-\bm{Y}(0) | N] &= \sum_{d=1}^{J-1}(\bm{\Lambda}_Z^d \otimes \textbf{I}_N) \{\Delta(d)\bm{1}_{NJ}\} \\
    &= \sum_{d=1}^{J-1} (\bm{\Lambda}_Z^d \otimes \textbf{I}_N) \{(\textbf{A}_d \bm{\Delta}^D) \otimes \bm{1}_N\} \\
    &= \sum_{d=1}^{J-1} \{(\bm{\Lambda}_Z^d \otimes \bm{1}_N) \textbf{A}_d \} \bm{\Delta}^D \in \mathbb{R}^{NJ} \,.
\end{split}
\]
Finally, as in the linear setting (Section \ref{app.sect:duration}) and relying on Assumptions A2 (non-informative enrollment) and A3 (randomization), we have
\[
\begin{split}
    \sum_{d=1}^{J-1} & E[ \{ (\textbf{H}_Z  -E[\textbf{H}_Z]) \otimes \bm{1}_N \}^\top \bm{D}\bm{D}^\top (\bm{\Lambda}_{Z_i}^d \otimes \textbf{I}_{N_i}) (\{\bm{Y}(d)-\bm{Y}(0)\} - \{\underline{\bm{\mu}}(\underline{\beta}_{Zd}) - \underline{\bm{\mu}}(0) \}) ] \\
    &= \sum_{d=1}^{J-1} E[ \{ (\textbf{H}_Z  -E[\textbf{H}_Z]) \otimes \bm{1}_N \}^\top \bm{D}\bm{D}^\top (\bm{\Lambda}_{Z_i}^d \otimes \textbf{I}_{N_i}) (E[\bm{Y}(d)-\bm{Y}(0)|N] - E[\underline{\bm{\mu}}(\underline{\beta}_{Zd}) - \underline{\bm{\mu}}(0)|N]) ] \\ 
   & = 0 \,.
\end{split}
\]
such that
\[
\begin{split}
    \bm{\Delta}^D = E&\left[\sum_{d=1}^{J-1} \{(\textbf{H}_Z  -E[\textbf{H}_Z]) \otimes \bm{1}_N\}^\top \bm{D}\bm{D}^\top \{(\bm{\Lambda}_Z^d \otimes \bm{1}_N) \textbf{A}_d \} \right]^{-1} \\
    & \sum_{d=1}^{J-1}E\left[ \{ (\textbf{H}_Z  -E[\textbf{H}_Z]) \otimes \bm{1}_N \}^\top \bm{D}\bm{D}^\top  (\bm{\Lambda}_{Z}^d \otimes \bm{1}_{N}) \{\tilde{\underline{\bm{\mu}}}(\underline{\beta}_{Zd}) - \tilde{\underline{\bm{\mu}}}(0)\} \right] \\
    = E&\left[ \sum_{d=1}^{J-1} \bm{\lambda}(d) \textbf{A}_d \right]^{-1} E\left[ \sum_{d=1}^{J-1} \bm{\lambda}(d) \left(\tilde{\underline{\bm{\mu}}}(\underline{\beta}_{Zd}) - \tilde{\underline{\bm{\mu}}}(0)\right) \right] \,.
\end{split}
\]

Recall that the first two entries of estimating equations (\ref{app.eq:GEE_ds}) imply that
\[
\begin{split}
    E\left[ \sum_{d=1}^{J-1} \bm{\lambda}(d) \left[ \textbf{A}_d \underline{\bm{\Delta}}_{GEE-g}^D - \{ \left(\mu_1(\beta_{Zd}) - \mu_1(0) \right) ,..., \left(\mu_J(\beta_{Zd}) - \mu_J(0) \right) \}^\top \right] \right] &= 0 \\
    E\left[ \mathcal{N} \mu_j(b) - \bm{S}^\top \bm{h}_{.j}(b) \right] &= 0 \,.
\end{split}
\]
The second entry in the estimating equations (\ref{app.eq:GEE_ds}) indicates that $E[\underline{\mu}_j(b)] = \underline{\mu}_j(b) = E[g^{-1}(\underline{\beta}_{0j} + b + \underline{\bm{\beta}}_X^\top \bm{X} + \underline{\alpha})] \in \mathbb{R}^1$. Hence, the first entry of estimating equation (\ref{app.eq:GEE_ds}) is equivalent to
\[
    E\left[ \sum_{d=1}^{J-1} \bm{\lambda}(d) \left( \textbf{A}_d \underline{\bm{\Delta}}_{GEE-g}^D - \left(\tilde{\underline{\bm{\mu}}}(\underline{\beta}_{Zd}) - \tilde{\underline{\bm{\mu}}}(0)\right)  \right) \right] = 0 \,,
\]
yielding $\underline{\bm{\Delta}}_{GEE-g}^D = \bm{\Delta}^D$, which completes the proof of consistency.

\paragraph{Illustration of the duration-specific treatment effect g-computation estimator in a 4 sequence SW-CT design.} \mbox{}\\

As previously described, the duration-specific treatment effect g-computation estimator can be formulated as
\[
    \bm{\Delta}^D = E\left[ \sum_{d=1}^{J-1} \bm{\lambda}(d) \textbf{A}_d \right]^{-1} E\left[ \sum_{d=1}^{J-1} \bm{\lambda}(d) \left(\tilde{\underline{\bm{\mu}}}(\underline{\beta}_{Zd}) - \tilde{\underline{\bm{\mu}}}(0)\right) \right] \,.
\]
In the example of a 4 sequence SW-CT with $J=5$ periods (as in the main manuscript case study re-analysis of \citet{tang_crowdsourcing_2018}), we can then demonstrate with treatment duration $d = 1,2,3,4$ for clusters $i$ when $Z_i=2$: 
\[\bm{\lambda}_i(d) \left(\tilde{\underline{\bm{\mu}}}(\underline{\beta}_{Zd}) - \tilde{\underline{\bm{\mu}}}(0)\right)=\] 
\[\Scale[0.8]{
    \left\{
    \begin{pmatrix}
        \frac{3}{4}N_{i2} \left(\tilde{\underline{\mu}}_{2}(\underline{\beta}_{Z1}) - \tilde{\underline{\mu}}_{2}(0)\right) \\
        0 \\
        0 \\
        0 \\
    \end{pmatrix},
    \begin{pmatrix}
        - \frac{1}{4}N_{i3} \left(\tilde{\underline{\mu}}_{3}(\underline{\beta}_{Z2}) - \tilde{\underline{\mu}}_{3}(0)\right) \\
        \frac{3}{4}N_{i3} \left(\tilde{\underline{\mu}}_{3}(\underline{\beta}_{Z2}) - \tilde{\underline{\mu}}_{3}(0)\right) \\
        0 \\
        0 \\
    \end{pmatrix},
    \begin{pmatrix}
        - \frac{1}{4}N_{i4} \left(\tilde{\underline{\mu}}_{4}(\underline{\beta}_{Z3}) - \tilde{\underline{\mu}}_{4}(0)\right) \\
        - \frac{1}{4}N_{i4} \left(\tilde{\underline{\mu}}_{4}(\underline{\beta}_{Z3}) - \tilde{\underline{\mu}}_{4}(0)\right) \\
        \frac{3}{4}N_{i4} \left(\tilde{\underline{\mu}}_{4}(\underline{\beta}_{Z3}) - \tilde{\underline{\mu}}_{4}(0)\right) \\
        0 \\
    \end{pmatrix},
    \begin{pmatrix}
        - \frac{1}{4}N_{i5} \left(\tilde{\underline{\mu}}_{5}(\underline{\beta}_{Z4}) - \tilde{\underline{\mu}}_{5}(0)\right) \\
        - \frac{1}{4}N_{i5} \left(\tilde{\underline{\mu}}_{5}(\underline{\beta}_{Z4}) - \tilde{\underline{\mu}}_{5}(0)\right) \\
        - \frac{1}{4}N_{i5} \left(\tilde{\underline{\mu}}_{5}(\underline{\beta}_{Z4}) - \tilde{\underline{\mu}}_{5}(0)\right) \\
        \frac{3}{4}N_{i5} \left(\tilde{\underline{\mu}}_{5}(\underline{\beta}_{Z4}) - \tilde{\underline{\mu}}_{5}(0)\right) \\
    \end{pmatrix}
    \right\}
}^\top \,,\]
for clusters $i$ when $Z_i=3$:
\[\Scale[0.8]{
    \left\{
    \begin{pmatrix}
        \frac{3}{4}N_{i3} \left(\tilde{\underline{\mu}}_{3}(\underline{\beta}_{Z1}) - \tilde{\underline{\mu}}_{3}(0)\right) \\
        - \frac{1}{4}N_{i3} \left(\tilde{\underline{\mu}}_{3}(\underline{\beta}_{Z1}) - \tilde{\underline{\mu}}_{3}(0)\right) \\
        0 \\
        0 \\
    \end{pmatrix},
    \begin{pmatrix}
        - \frac{1}{4}N_{i4} \left(\tilde{\underline{\mu}}_{4}(\underline{\beta}_{Z2}) - \tilde{\underline{\mu}}_{4}(0)\right) \\
        \frac{3}{4}N_{i4} \left(\tilde{\underline{\mu}}_{4}(\underline{\beta}_{Z2}) - \tilde{\underline{\mu}}_{4}(0)\right) \\
        - \frac{1}{4}N_{i4} \left(\tilde{\underline{\mu}}_{4}(\underline{\beta}_{Z2}) - \tilde{\underline{\mu}}_{4}(0)\right) \\
        0 \\
    \end{pmatrix},
    \begin{pmatrix}
        - \frac{1}{4}N_{i5} \left(\tilde{\underline{\mu}}_{5}(\underline{\beta}_{Z3}) - \tilde{\underline{\mu}}_{5}(0)\right) \\
        - \frac{1}{4}N_{i5} \left(\tilde{\underline{\mu}}_{5}(\underline{\beta}_{Z3}) - \tilde{\underline{\mu}}_{5}(0)\right) \\
        \frac{3}{4}N_{i5} \left(\tilde{\underline{\mu}}_{5}(\underline{\beta}_{Z3}) - \tilde{\underline{\mu}}_{5}(0)\right) \\
        - \frac{1}{4}N_{i5} \left(\tilde{\underline{\mu}}_{5}(\underline{\beta}_{Z3}) - \tilde{\underline{\mu}}_{5}(0)\right) \\
    \end{pmatrix},
    \begin{pmatrix}
        0 \\
        0 \\
        0 \\
        0 \\
    \end{pmatrix}
    \right\}
}^\top \,,\]
for clusters $i$ when $Z_i=4$:
\[\Scale[0.8]{
    \left\{
    \begin{pmatrix}
        \frac{3}{4}N_{i4} \left(\tilde{\underline{\mu}}_{4}(\underline{\beta}_{Z1}) - \tilde{\underline{\mu}}_{4}(0)\right) \\
        - \frac{1}{4}N_{i4} \left(\tilde{\underline{\mu}}_{4}(\underline{\beta}_{Z1}) - \tilde{\underline{\mu}}_{4}(0)\right) \\
        - \frac{1}{4}N_{i4} \left(\tilde{\underline{\mu}}_{4}(\underline{\beta}_{Z1}) - \tilde{\underline{\mu}}_{4}(0)\right) \\
        0 \\
    \end{pmatrix},
    \begin{pmatrix}
        - \frac{1}{4}N_{i5} \left(\tilde{\underline{\mu}}_{5}(\underline{\beta}_{Z2}) - \tilde{\underline{\mu}}_{5}(0)\right) \\
        \frac{3}{4}N_{i5} \left(\tilde{\underline{\mu}}_{5}(\underline{\beta}_{Z2}) - \tilde{\underline{\mu}}_{5}(0)\right) \\
        - \frac{1}{4}N_{i5} \left(\tilde{\underline{\mu}}_{5}(\underline{\beta}_{Z2}) - \tilde{\underline{\mu}}_{5}(0)\right) \\
        - \frac{1}{4}N_{i5} \left(\tilde{\underline{\mu}}_{5}(\underline{\beta}_{Z2}) - \tilde{\underline{\mu}}_{5}(0)\right) \\
    \end{pmatrix},
    \begin{pmatrix}
        0 \\
        0 \\
        0 \\
        0 \\
    \end{pmatrix},
    \begin{pmatrix}
        0 \\
        0 \\
        0 \\
        0 \\
    \end{pmatrix}
    \right\} 
}^\top \,,\]
and for clusters $i$ when $Z_i=5$:
\[\Scale[0.8]{
    \left\{
    \begin{pmatrix}
        \frac{3}{4}N_{i5} \left(\tilde{\underline{\mu}}_{5}(\underline{\beta}_{Z1}) - \tilde{\underline{\mu}}_{5}(0)\right) \\
        - \frac{1}{4}N_{i5} \left(\tilde{\underline{\mu}}_{5}(\underline{\beta}_{Z1}) - \tilde{\underline{\mu}}_{5}(0)\right) \\
        - \frac{1}{4}N_{i5} \left(\tilde{\underline{\mu}}_{5}(\underline{\beta}_{Z1}) - \tilde{\underline{\mu}}_{5}(0)\right) \\
        - \frac{1}{4}N_{i5} \left(\tilde{\underline{\mu}}_{5}(\underline{\beta}_{Z1}) - \tilde{\underline{\mu}}_{5}(0)\right) \\
    \end{pmatrix},
    \begin{pmatrix}
        0 \\
        0 \\
        0 \\
        0 \\
    \end{pmatrix},
    \begin{pmatrix}
        0 \\
        0 \\
        0 \\
        0 \\
    \end{pmatrix},
    \begin{pmatrix}
        0 \\
        0 \\
        0 \\
        0 \\
    \end{pmatrix}
    \right\} 
}^\top \,.\]
Additionally,
\[
\begin{split}
    & \sum_{d=1}^{J-1} \bm{\lambda}_i(d) \textbf{A}_{d} =\\
    & \frac{1}{4}
    \Scale[0.60]{
    \begin{pmatrix}
        3(I\{Z_i=2\}N_{i2}+I\{Z_i=3\}N_{i3}+I\{Z_i=4\}N_{i4}+I\{Z_i=5\}N_{i5}) & -I\{Z_i=2\}N_{i3}-I\{Z_i=3\}N_{i4}-I\{Z_i=4\}N_{i5} & -I\{Z_i=2\}N_{i4}-I\{Z_i=3\}N_{i5} & -I\{Z_i=2\}N_{i5} \\
        -I\{Z_i=3\}N_{i3}-I\{Z_i=4\}N_{i4}-I\{Z_i=5\}N_{i5} & 3(I\{Z_i=2\}N_{i3}+I\{Z_i=3\}N_{i4}+I\{Z_i=4\}N_{i5}) & -I\{Z_i=2\}N_{i4}-I\{Z_i=3\}N_{i5} & -I\{Z_i=2\}N_{i5} \\
        -I\{Z_i=4\}N_{i4}-I\{Z_i=5\}N_{i5} & -I\{Z_i=3\}N_{i4}-I\{Z_i=4\}N_{i5} & 3(I\{Z_i=2\}N_{i4}+I\{Z_i=3\}N_{i5}) & -I\{Z_i=2\}N_{i5} \\
        -I\{Z_i=5\}N_{i5} & -I\{Z_i=4\}N_{i5} & -I\{Z_i=3\}N_{i5} & 3I\{Z_i=2\}N_{i5}
    \end{pmatrix}
    }
\end{split} \,.
\]

\subsubsection{Period-specific Treatment Effect}
\label{app.sect:Theorem_SW_G_Proof_ps}

We next prove the consistency result under the period-specific treatment effect setting. Since we need to drop the data from period $J$ to avoid over-parameterization, we again use the superscript * to denote all subsequent changes in notation, as in the proof of Theorem \ref{app.Theorem:SW}.
To reiterate from Theorem \ref{app.Theorem:SW}, $J^* = J-1, \bm{Y}_i^* = (\bm{Y}_{i1},...,\bm{Y}_{i,J-1})^{\top}, \bm{\beta}_0^* = (\beta_{01},...,\beta_{0,J-1})^{\top}$, and generically, $\textbf{M}^* \in \mathbb{R}^{J^* \times J^*}$ is a matrix consisting of the first $J-1$ columns and rows of a matrix $\textbf{M} \in \mathbb{R}^{J \times J}$.

The corresponding estimating equations are then
\begin{equation}
\label{app.eq:GEE_ps}
    \bm{\psi}(\bm{O}_i^*;\bm{\theta}) =  \left(
        \begin{gathered}
           \bm{\Delta}_{GEE-g}^P - \{(\mu_{2}^*(\beta_{2Z}) - \mu_{2}^*(0)), ..., (\mu_{J^*}^*(\beta_{J^*Z}) - \mu_{J^*}^*(0)) \}^\top \\
            \mathcal{N}_{i}^* \mu_j^*(\beta_{j'Z}) - \bm{S}_i^{*\top} \bm{h}_{ij}^*(\beta_{j'Z}), j=1,...,J, j'=j \text{ or } 0 \\
            \bm{Q}_i^{*\top} \bm{D}_i^* \bm{D}_i^{*\top} (\bm{Y}_i^* - \bm{\mu}_i^*) 
        \end{gathered}
        \right)
\end{equation}
with $\bm{Q}^* = (\textbf{I}_{J^*} \otimes \bm{1}_N, \bm{\Delta}_Z^* \otimes \bm{1}_N, \bm{1}_{J^*} \otimes \bm{X}).$
Importantly, $\beta_{1Z}=0$ since period $j=1$ is the baseline period where no clusters receive treatment, hence, $\mu_{1}^*(\beta_{1Z}) = \mu_{1}^*(0)$.
Then $E[\bm{\psi}(\bm{O}^*, \underline{\bm{\theta}})] = 0$ implies that
\[
\begin{split}
    E[(\textbf{I}_{J^*} \otimes \bm{1}_N)^\top \bm{D}^* \bm{D}^{*\top} (\bm{Y}^* - \underline{\bm{\mu}}^*)] &= 0 \\
    E[(\bm{\Delta}_Z^* \otimes \bm{1}_N)^\top \bm{D}^* \bm{D}^{*\top} (\bm{Y}^* - \underline{\bm{\mu}}^*)] &= 0 \,.
\end{split}
\]
Left multiplying the first equation by $-E[\bm{\Delta}_Z^*]^\top$ and adding it to the second equation, we get $E[\{(\bm{\Delta}_Z^* - E[\bm{\Delta}_Z^*]) \otimes \bm{1}_N\}^\top \bm{D}^* \bm{D}^{*\top} (\bm{Y}^* - \underline{\bm{\mu}}^*)] = 0$.
Under the setting of period-specific treatment effects, we have 
\[
    \underline{\bm{\mu}}^* = (\bm{\Delta}_{Z}^* \otimes \textbf{I}_{N})\{
    \underline{\bm{\mu}}^{*}(1)
    - \underline{\bm{\mu}}^{*}(0)\} + (\textbf{I}_{J^*} \otimes \textbf{I}_{N})\underline{\bm{\mu}}^{*}(0) \,,
\]
where specifically for period-specific treatment effects, we denote $\bm{Y}^*(1)$ and $\underline{\bm{\mu}}^{*}(1) = \left(g^{-1}(\underline{\beta}_{0j} + \underline{\beta}_{jZ} + \underline{\bm{\beta}}_X^\top \bm{X}_{.k})\right)_{k=1,....,N, j=1,...,J^*}$ to denote the potential outcome and expected mean model under treatment, respectively.
Using formula (\ref{app.eq:PO}) with $J$ replaced by $J^*$, we get
\[
    E[\{(\bm{\Delta}_Z^* - E[\bm{\Delta}_Z^*]) \otimes \bm{1}_N\}^\top \bm{D}^* \bm{D}^{*\top} (\bm{\Delta}_Z^* \otimes \textbf{I}_N)(\{
    \bm{Y}^*(1) 
    - \bm{Y}^*(0)\} - \{
    \underline{\bm{\mu}}^{*}(1)
    - \underline{\bm{\mu}}^*(0)\})] = 0 \,.
\]
Hence
\[
\begin{split}
    E&[\{(\bm{\Delta}_Z^* - E[\bm{\Delta}_Z^*]) \otimes \bm{1}_N\}^\top \bm{D}^* \bm{D}^{*\top} (\bm{\Delta}_Z^* \otimes \textbf{I}_N)E[\{
    \bm{Y}^*(1) 
    - \bm{Y}^*(0)\} - \{
    \underline{\bm{\mu}}^{*}(1) 
    - \underline{\bm{\mu}}^*(0)\}]]  \\
    &= E[\{(\bm{\Delta}_Z^* - E[\bm{\Delta}_Z^*]) \otimes \bm{1}_N\}^\top \bm{D}^* \bm{D}^{*\top} (\bm{\Delta}_Z^* \otimes \textbf{I}_N) E[\{
    \bm{Y}^*(1) 
    - \bm{Y}^*(0)\} - \{
    \underline{\tilde{\bm{\mu}}}^{*}(1) 
    - \underline{\tilde{\bm{\mu}}}^*(0)\} \otimes \bm{1}_N]] \\
    &= 0 \,.
\end{split}
\]
where $E[\bm{Y}^*(1) - \bm{Y}^*(0)] = (0, \bm{\Delta}^{P\top})^\top 
\otimes \bm{1}_N$ with $\bm{\Delta}^P = (\Delta_2,...,\Delta_j,...,\Delta_{J^*})^\top$.
Again, recall that $\tilde{\underline{\bm{\mu}}}(b) = \left(E[g^{-1}(\underline{\beta}_{0j} + b + \underline{\bm{\beta}}_X^\top \bm{X} + \underline{\alpha})] \right)_{j=1,...,J} \in \mathbb{R}^J$.
Accordingly, we demonstrate 
$\bm{\Delta}^P = E[\underline{\tilde{\bm{\mu}}}^{*}(1) - \underline{\tilde{\bm{\mu}}}^*(0)] \in \mathbb{R}^{J^*-1}$
(with $\underline{\beta}_{1Z} = 0$).

Recall the first two entries of estimating equations (\ref{app.eq:GEE_ps}) imply the g-computation estimators $E[\underline{\mu}_j^*(\underline{\beta}_{jZ})] = \underline{\mu}_j^*(\underline{\beta}_{jZ})= E[g^{-1}(\underline{\beta}_{0j} + \underline{\beta}_{jZ} + \underline{\bm{\beta}}_X^\top \bm{X} + \underline{\alpha})]$
and 
$E[\underline{\mu}_j^*(0)] = \underline{\mu}_j^*(0) = E[g^{-1}(\underline{\beta}_{0j} + \underline{\bm{\beta}}_X^\top \bm{X} + \underline{\alpha})]$.
Finally,
$\underline{\bm{\Delta}}_{GEE-g}^P = \{ (\mu_{2}^*(\beta_{2Z}) - \mu_{2}^*(0)), ..., (\mu_{J^*}^*(\beta_{J^*Z}) - \mu_{J^*}^*(0)) \}^\top = \bm{\Delta}^P$, which completes the proof for consistency.

\subsubsection{Saturated Treatment Effect}
\label{app.sect:Theorem_SW_G_Proof_sat}

Finally, we prove the results under the saturated treatment effect setting. Like the period-specific treatment effect setting, we drop period $J$ to avoid over-parameterization, and we thus use superscript * to notate this change.
The corresponding estimating equations are
\begin{equation}
\label{app.eq:GEE_s}
    \bm{\psi}(\bm{O}_i^*;\bm{\theta}) =  \left(
        \begin{gathered}
           \bm{\Delta}_{GEE-g}^S - \{ (\mu_{2}^*(\beta_{2Z1}) - \mu_{2}^*(0)), ..., (\mu_{J^*}^*(\beta_{J^*ZJ^*}) - \mu_{J^*}^*(0)) \}^\top \\
            \mathcal{N}_{i}^* \mu_j^*(\beta_{j'Zd}) - \bm{S}_i^{*\top} \bm{h}_{ij}^*(\beta_{j'Zd}), 1 \leq d < j \leq J^*, j'=j \text{ or } 0 \\
            \bm{Q}_i^{*\top} \bm{D}_i^* \bm{D}_i^{*\top} (\bm{Y}_i^* - \bm{\mu}_i^*)
        \end{gathered}
        \right)
\end{equation}
with $\bm{Q}^* = (\textbf{I}_{J^*} \otimes \bm{1}_N, \tilde{\bm{\Lambda}}_Z^{1*} \otimes \bm{1}_N, ..., \tilde{\bm{\Lambda}}_Z^{J^**} \otimes \bm{1}_N, \bm{1}_{J^*} \otimes \bm{X})$.
Recall, $\tilde{\bm{\Lambda}}_{Z_i}^{d*} \in \mathbb{R}^{J^* \times (J^* - d)}$ contains column $d+1$ to column $J^*$ of $\bm{\Lambda}_{Z_i}^{d*}$ (Section \ref{app.sect:saturated}).
Importantly, $\beta_{1Zd}=0 \,\forall\, d$ since period $j=1$ is the baseline period where no clusters receive treatment, hence, $\mu_{1}^*(\beta_{1Zd})-\mu_{1}^*(0)=0 \,\forall\, d$.

Following a similar proof to that seen in the previous sections, we have for each exposure time $d=1,...,J$
\[
\begin{split}
    E[(\textbf{I}_{J^*} \otimes \bm{1}_N)^\top \bm{D}^* \bm{D}^{*\top} (\bm{Y}^* - \underline{\bm{\mu}}^*)] &= 0 \\
    E[(\tilde{\bm{\Lambda}}_Z^{d*} \otimes \bm{1}_N)^\top \bm{D}^* \bm{D}^{*\top} (\bm{Y}^* - \underline{\bm{\mu}}^*)] &= 0 \,.
\end{split}
\]
which as in the linear setting (Section \ref{app.sect:fe_s_rand}), implies for $d=1,...,J^*-1$
\[
\begin{split}
    \sum_{d'=1}^{J^*-1} &E\left[ \{ (\tilde{\bm{\Lambda}}_Z^{d*} - E[\tilde{\bm{\Lambda}}_Z^{d*}]) \otimes \bm{1}_N \}^\top \bm{D}^* \bm{D}^{*\top} \left\{(\bm{\Lambda}_Z^{d'*} \otimes \textbf{I}_N) E[\bm{Y}^*(d')-\bm{Y}^*(0)] - \left( \tilde{\bm{\Lambda}}_{Z}^{d'*} \otimes \bm{1}_{N} \right) \left(
    \underline{\tilde{\bm{\mu}}}^{*}(d')
    - \underline{\tilde{\bm{\mu}}}^*(0)\right) \right\} \right] \\
    &= \sum_{d'=1}^{J^*-1} E\left[ \{ (\tilde{\bm{\Lambda}}_Z^{d*} - E[\tilde{\bm{\Lambda}}_Z^{d*}]) \otimes \bm{1}_N \}^\top \bm{D}^* \bm{D}^{*\top} \left( \tilde{\bm{\Lambda}}_{Z}^{d'*} \otimes \bm{1}_{N} \right) \left\{ \bm{\Delta}_d^S -\left(
    \underline{\tilde{\bm{\mu}}}^{*}(d')
    - \underline{\tilde{\bm{\mu}}}^*(0)\right) \right\} \right] \\
    &= 0
\end{split}
\]
Again, recall that $\tilde{\underline{\bm{\mu}}}(b) = \left(E[g^{-1}(\underline{\beta}_{0j} + b + \underline{\bm{\beta}}_X^\top \bm{X} + \underline{\alpha})] \right)_{j=1,...,J} \in \mathbb{R}^J$.
Then, specifically for saturated treatment effects, we denote $\underline{\tilde{\bm{\mu}}}^{*}(d) 
= \left(E[g^{-1}(\underline{\beta}_{0j} + \underline{\beta}_{jZd} + \underline{\bm{\beta}}_X^\top \bm{X} + \underline{\alpha})]\right)_{j=d+1,...,J^*} \in \mathbb{R}^{J^*-d}$
and $\underline{\tilde{\bm{\mu}}}^{*}(0) = \left(E[g^{-1}(\underline{\beta}_{0j} + \underline{\bm{\beta}}_X^\top \bm{X} + \underline{\alpha})] \right)_{j=d+1,...,J^*} \in \mathbb{R}^{J^*-d}$.
Assuming a saturated treatment effect, $(\bm{\Lambda}_Z^{d*} \otimes \textbf{I}_N)E[\bm{Y}^*(d)-\bm{Y}^*(0)] = ( \tilde{\bm{\Lambda}}_{Z}^{d*} \otimes \bm{1}_{N})\bm{\Delta}_d^S \in \mathbb{R}^{NJ^*}$ where $\bm{\Delta}_d^S = (\Delta_{d}(d), ..., \Delta_{J^*-1}(d))^\top \in \mathbb{R}^{J^*-d}$.
By the regularity condition that the solution is unique, we have for $1 \leq d < j \leq J^*$
\[
    E[g^{-1}(\underline{\beta}_{0j} + \underline{\beta}_{jZd} + \underline{\bm{\beta}}_X^\top \bm{X}_{.k} + \underline{\alpha})] -  E[g^{-1}(\underline{\beta}_{0j} + \underline{\bm{\beta}}_X^\top \bm{X}_{.k} + \underline{\alpha})] = E[Y_{.jk}(d) - Y_{.jk}(0)] = \Delta_j(d) \,.
\]

Recall that the second entry of the estimating equations (\ref{app.eq:GEE_s}) imply that the g-computation estimators $\underline{\mu}_j^*(\underline{\beta}_{j'Zd}) = E[\mathcal{N}^*]^{-1} E[\bm{S}^{*\top} \underline{\bm{h}}_{.j}^*(\underline{\beta}_{j'Zd})] = E[g^{-1}(\underline{\beta}_{0j} + \underline{\beta}_{jZd} + \underline{\bm{\beta}}_X^\top \bm{X} + \underline{\alpha})]$
and similarly,
$\underline{\mu}_j^*(0) = E[g^{-1}(\underline{\beta}_{0j} + \underline{\bm{\beta}}_X^\top \bm{X} + \underline{\alpha})]$.
Therefore, the first entry of the estimating equations (\ref{app.eq:GEE_s}) imply $\underline{\bm{\Delta}}_{GEE-g}^S = \bm{\Delta}^S$.
In fact, the regularity condition (1) is unnecessary here since the design matrix $E[\{(\tilde{\bm{\Lambda}}_Z^{*} - E[\tilde{\bm{\Lambda}}_Z^{*}] )\otimes \bm{1}_N\}^\top \bm{D}^* \bm{D}^{*\top} (\tilde{\bm{\Lambda}}_Z^{*} \otimes \bm{1}_N)]$ is already invertible (which has a form of a covariance matrix and does not degenerate due to Assumption A3), where $\tilde{\bm{\Lambda}}_Z^{*} = (\tilde{\Lambda}_Z^{1*}, \tilde{\Lambda}_Z^{2*}, ..., \tilde{\Lambda}_Z^{J^*-1*})$.

If the mean model $E[Y_{.jk} | Z, \bm{X}_{.k}] = g^{-1}(\beta_{0j} + TE_{.j} + \bm{\beta}_X^\top \bm{X}_{.k} + \alpha)$ is correctly specified for some parameters $\bm{\beta}^*$, the classical GEE theory \citep{liang_longitudinal_1986} implies that the probability limit of $\hat{\bm{\beta}}$, $\underline{\bm{\beta}}$, is equal to $\bm{\beta}^*$.
Therefore, the probability limit of $\hat{\mu}_j(\hat{b})$, $E[g^{-1}(\underline{\beta}_{0j} + \underline{b} + \underline{\bm{\beta}}_X^\top \bm{X} + \underline{\alpha})]$, is the corresponding expectation of potential outcomes, and the desired results are implied.

\subsubsection{General proof of working constant treatment effect consistency for the P-ATO}

As in Section \ref{app.sect:SW_const=P-ATO}, we can generally prove that the constant treatment effect GEE g-computation difference-in-means estimator targets the period-average treatment effect for the overlap population (P-ATO) estimand.

Building on the proof in Section \ref{app.sect:Theorem_SW_G_Proof_constant}, where there are true underlying period-specific treatment effects  $E[\bm{Y}(1)-\bm{Y}(0)] = \bm{\Delta}^P \otimes \bm{1}_N$ (where we broadly define $\bm{\Delta}^P=(\Delta_j)_{j=1,...,J}$), then with a specified working constant treatment effect in the GEE, we have the P-ATO
\[
\begin{split}
    \left(\sum_{j=1}^J \lambda_{j}\right)^{-1}  \sum_{j=1}^J \lambda_{j} \Delta_j
    & = \left(\sum_{j=1}^J \lambda_{j}\right)^{-1} \sum_{j=1}^J \lambda_{j} E\left[ g^{-1}(\underline{\beta}_{0j} + \underline{\beta}_Z + \underline{\bm{\beta}}_X^\top X\bm{ }+ \underline{\alpha}
    ) - g^{-1}(\underline{\beta}_{0j} + \underline{\bm{\beta}}_X^\top \bm{X} + \underline{\alpha}
    ) \right] \\
    & = \Delta^{P-ATO}
\end{split}
\]

Subsequently, with the G-computation estimator being specified as
\[
    \underline{\Delta}_{GEE-g} = \left(\sum_{j=1}^J \lambda_{j}\right)^{-1} \sum_{j=1}^J \lambda_{j} \left(\mu_j(\underline{\beta}_Z) - \mu_j(0) \right)
\]
The second entry in the estimating equations (\ref{app.eq:GEE_c}) indicates that $E[\underline{\mu}_j(b)] = \underline{\mu}_j(b) = E\left[\mathcal{N} \right]^{-1}E[\bm{S}^\top \bm{h}_{.j}(b)] = E[g^{-1}(\underline{\beta}_{0j} + b + \underline{\bm{\beta}}_X^\top \bm{X} + \underline{\alpha})] \in \mathbb{R}^1$ using Assumption A2 (non-informative enrollment).
Altogether, this completes the proof of consistency 
\[
    \underline{\Delta}_{GEE-g} = \Delta^{P-ATO} \,. 
\]

\noindent $\square$

\subsection{Extension of Theorem \ref{app.Theorem:G} to longitudinal CRT designs}

\begin{appremark}
\label{app.remark:g_PB_XO}
    The consistency results outlined in Theorem \ref{app.Theorem:G} can be similarly demonstrated in other longitudinal CRT designs, including the PB-CRT and CRXO designs. 
\end{appremark}

Remark \ref{app.remark:g_PB_XO} extends directly from the above proof for Theorem \ref{app.Theorem:G}. 
Recall that treatment effect structures can coincide in different longitudinal CRT designs. 
For example, in PB-CRTs, the duration-specific, period-specific, and saturated treatment effect structures coincide.
Whereas in CRXOs, the period-specific and saturated treatment effects coincide as a more general form of a duration-specific treatment effect structure.

The constant treatment effect results hold for other longitudinal CRT designs such as PB-CRTs and CRXOs, 
as in Theorem \ref{app.Theorem:G} and Section \ref{app.sect:Theorem_SW_G_Proof_constant},
by setting the treatment effect indicators $\bm{\Delta}_{Z_i}$ to be $\bm{\Delta}_{Z_i,PB}$ or $\bm{\Delta}_{Z_i,XO}$, respectively. 
The results for different longitudinal CRTs with duration-specific treatment effects hold, as in  Theorem \ref{app.Theorem:G} and Section \ref{app.sect:Theorem_SW_G_Proof_ds}, by setting $\bf{H}_{Z_i}$ and $\bm{\Lambda}_{Z_i}^d$ to match the corresponding CRT design.
The results for different longitudinal CRTs with period-specific treatment effects hold (when they're identifiable), as in  Theorem \ref{app.Theorem:G} and Section \ref{app.sect:Theorem_SW_G_Proof_ps}, by setting $\bm{\Delta}_{Z_i}$ and $\bm{\Lambda}_{Z_i}^d$ to match the corresponding CRT design.
Finally, the results for different longitudinal CRTs with saturated treatment effects hold, as in  Theorem \ref{app.Theorem:G} and Section \ref{app.sect:Theorem_SW_G_Proof_sat}, by setting $\tilde{\bm{\Lambda}}_{Z_i}^{d}$ to match the corresponding CRT design.

\subsection{Model-robust inference in non-randomized quasi-experimental trials}
\label{app.sect:G_CQTs}

In the absence of randomization (Assumption A3) in longitudinal CQTs, consistency of the g-computation estimator can still be demonstrated by
assuming the model (Equation \ref{app.eq:GEE_mean_model}) is correctly specified, as in supplemental condition (iii.) of Theorem \ref{app.Theorem:G}.
Consistency then results from the standard model-based strict exogeneity assumption, which states $E[\bm{Y}(0)-\underline{\bm{\mu}}(0) | Z] = 0$ and
\[
\begin{split}
    \Delta &= E[g^{-1}(\underline{\beta}_{0j} + \underline{\beta}_{Z} + \underline{\bm{\beta}}_X^\top \bm{X}_{.k} + \underline{\alpha}) - g^{-1}(\underline{\beta}_{0j} + \underline{\bm{\beta}}_X^\top \bm{X}_{.k} + \underline{\alpha})] \,, \\
    \Delta(d) &= E[g^{-1}(\underline{\beta}_{0j} + \underline{\beta}_{Zd} + \underline{\bm{\beta}}_X^\top \bm{X}_{.k} + \underline{\alpha}) -  g^{-1}(\underline{\beta}_{0j} + \underline{\bm{\beta}}_X^\top \bm{X}_{.k} + \underline{\alpha})] \,, \\
    \Delta_j &= E[g^{-1}(\underline{\beta}_{0j} + \underline{\beta}_{jZ} + \underline{\bm{\beta}}_X^\top \bm{X}_{.k} + \underline{\alpha}) - g^{-1}(\underline{\beta}_{0j} + \underline{\bm{\beta}}_X^\top \bm{X}_{.k} + \underline{\alpha})] \,, \\
    \Delta_j(d)  &= E[g^{-1}(\underline{\beta}_{0j} + \underline{\beta}_{jZd} + \underline{\bm{\beta}}_X^\top \bm{X}_{.k} + \underline{\alpha}) - g^{-1}(\underline{\beta}_{0j} + \underline{\bm{\beta}}_X^\top \bm{X}_{.k} + \underline{\alpha})]
\end{split}
\]
where recall $\Delta_j(d) = E[Y_{.jk}(j-d+1)] - E[Y_{.jk}(0)] \in \mathbb{R}^{1}$ \citep{wooldridge_econometric_2010}.

By extension, with the g-computation estimator $\mu_j(b)$
\[
\begin{split}
    \Delta &= E[\underline{\mu}_j(\underline{\beta}_Z) - \underline{\mu}_j(0)] \,, \\
    \Delta(d) &= E[\underline{\mu}_j(\underline{\beta}_{Zd}) - \underline{\mu}_j(0)]  \,, \\
    \Delta_j &= E[\underline{\mu}_j(\underline{\beta}_{jZ}) - \underline{\mu}_j(0)]  \,, \\
    \Delta_j(d) &= E[\underline{\mu}_j(\underline{\beta}_{jZd}) - \underline{\mu}_j(0)]  \,.
\end{split}
\]
Notably, when the model is correctly specified, the weights defined for the estimators in the proof of Theorem \ref{app.Theorem:G} ($\lambda_{j}$ for constant treatment and $\bm{\lambda}(d)\textbf{A}_d$ for duration-specific treatment) are all $\propto 1$.

While assuming the mean model to be correctly specified is seemingly contradictory to our ``model-robustness'' aims, such a log-link fixed-effects model can still robustly account for all time-invariant confounding by relying on the previously described assumptions.
By design, $e^{\underline{\alpha}_i}$ adjusts for all cluster-level time-invariant confounding in $\underline{\mu}_{ijk} = e^{\underline{\beta}_{0j} + \underline{TE}_{ij} + \underline{\bm{\beta}}_X^\top \bm{X}_{ik} + \underline{\alpha}_i}$.
Furthermore, as described for the linear fixed-effects model, the described log-link fixed-effects can also adjust for all individual-level time-invariant confounding: with Assumptions A1 (Super-population sampling) and A2 (Non-informative enrollment)
\[
\begin{split}
    E\left[\underline{\mu}_{.jk}\right] &= E\left[e^{\underline{\beta}_{0j} + \underline{TE}_{.j} + \underline{\bm{\beta}}_X^\top \bm{X}_{.k} + \underline{\alpha}}\right] \\
    &= E\left[
    E\left[e^{\underline{\beta}_{0j} + \underline{TE}_{.j}}|\alpha\right]
    E\left[e^{\underline{\bm{\beta}}_X^\top \bm{X}_{.k}}|\alpha\right]
    e^{\underline{\alpha}}
    \right] \,.
\end{split}
\]
Where the second equality results from assuming the mean model is correctly specified and additionally assuming $X_{ik} \perp Z_i|\alpha_i$ (corresponding with the DAG in Section \ref{app.sect:assumptions}).
Finally, Assumption A2 assumes away selection bias and alongside Assumption A1, assumes that $E\left[e^{\underline{\bm{\beta}}_X^\top \bm{X}_{.k}}|\alpha\right] = c$ where $c$ denotes a time-invariant cluster-level variable that is constant within cluster. 

Therefore, the effects of measured and unmeasured time-invariant individual-level confounding can also be automatically adjusted for by the described log-link fixed-effects model.
In contrast, a similar marginal GEE \citep{wang_how_2024} would additionally require the absence of cluster-level time-invariant confounding.

\newpage
\section{Corollary \ref{app.Corollary:nonlinear_estimands}}
\label{app.sect:nonlinear_estimands}

The following corollary largely extends from Wang et al. \citep{wang_how_2024}.
While we have mainly addressed difference estimands, ratio estimands are common for binary and count outcomes.
Next, we extend our methods to accommodate alternative effect measures, focusing on the saturated treatment effect structure (see Remark \ref{app.remark:ratio}).
Generalizing Equation \ref{app.eq:PO_estimands_SW}, we define the saturated treatment effect estimands as
\[
    \Phi_j(d) = f(E[Y_{ijk}(j-d+1)], E[Y_{ijk}(0)])
\]
for $1 \leq d < j < J-1$, for a user-defined function $f$.
For example, $f(x,y) = x/y$ defines the causal risk ratio and $f(x,y) = ln\{x/(1-x)\} - ln\{y/(1-y)\}$ defines the log causal odds ratio.
The vector of estimands is defined as $\bm{\Phi}^S = (\Phi_1(1), ..., \Phi_{J-1}(J-1))^{\top}$.
We can use the GEE (Equation \ref{app.eq:GEE}) with $TE_{ij} = \sum_d^j \beta_{jZd} I\{Z_i = j-d+1\}$ to estimate $\bm{\Phi}^S$, and we define the set of estimators as $\hat{\bm{\Phi}}^S = (f\{\hat{\mu}_{2}(\hat{\beta}_{2Z1}), \hat{\mu}_{2}(0)\} ,..., f\{\hat{\mu}_{J-1}(\hat{\beta}_{J-1,Z,J-1}), \hat{\mu}_{J-1}(0)\})^{\top}$.
Again, observed data from period $J$ are dropped during model-fitting.
In other words, $\Phi_j(d)$ is estimated by $f\{\hat{\mu}_{j}(\hat{\beta}_{jZd}), \hat{\mu}_{j}(0)\}$.
Denoting the sandwich variance estimator for $\hat{\bm{\Phi}}^S$ as $\hat{\bm{V}}_{\Phi}^S$, as we obtain the following result that $\hat{\bm{\Phi}}^S$ is robust to working model misspecification.

\begin{appcorollary}
\label{app.Corollary:nonlinear_estimands}
    Under the conditions in Theorem \ref{app.Theorem:G} for a SW-CRT, $(\hat{\bm{V}}^S_{\Phi})^{-1/2} \left(\hat{\bm{\Phi}}^S - \bm{\Phi}^S\right) \xrightarrow{d} N(0, \textbf{I}_{(J-2)(J-1)/2})$.
\end{appcorollary}

Corollary \ref{app.Corollary:nonlinear_estimands} also implies that linear mixed models can be used to formalize a model-robust estimator for $\Phi^S$ in a similar way.
The asymptotic result can be considered as an application of Corollary \ref{app.Corollary:nonlinear_estimands} under supplemental condition (I) of Theorem \ref{app.Theorem:G}, where we choose the identity-link and a constant variance function.

\noindent \textit{Proof of Corollary \ref{app.Corollary:nonlinear_estimands}.} 

In this proof, we inherit all notation defined in the proof of Theorem \ref{app.Theorem:G}.
For estimating the treatment effect estimands on other scales, the estimating equations become
\begin{equation}
\label{app.eq:GEE_nonlinear}
    \bm{\psi}(\bm{O}_i^*;\bm{\theta}) =  \left(
        \begin{gathered}
           \bm{\Phi}^S_{GEE-g} - \{ f(\mu_{2}^*(\beta_{2Z1}), \mu_{2}^*(0)), ..., f(\mu_{J^*}^*(\beta_{J^*ZJ^*}), \mu_{J^*}^*(0)) \}^\top \\
            \mathcal{N}_{i}^* \mu_j^*(\beta_{j'Zd}) - \bm{S}_i^{*\top} \bm{h}_{ij}^*(\beta_{j'Zd}), 1 \leq d < j \leq J^*, j'=j \text{ or } 0 \\
            \bm{Q}_i^{*\top} \bm{D}_i^* \bm{D}_i^{*\top} (\bm{Y}_i^* - \bm{\mu}_i^*)
        \end{gathered}
        \right)
\end{equation}
which is equal to estimating equations in Equation \ref{app.eq:GEE_s} if $f(x,y)=x-y$ and the difference is only on the first entry of the estimating equations. Therefore we only need to prove consistency.

The proof of consistency then extends exactly as seen in Section \ref{app.sect:Theorem_SW_G_Proof_sat}.
That is except, based on the first entry of the estimating equations (Equation \ref{app.eq:GEE_nonlinear}), we instead have $\Phi^S_{GEE-g,j}(d) = f\{\underline{\mu}_j^*(\underline{\beta}_{jZd}), \underline{\mu}_j^*(0)\} = f(E[Y_{ijk}(d)],E[Y_{ijk}(0)]) = \Phi_j(d)$, which completes the proof for consistency.
\begin{appremark}
    \label{app.remark:ratio}
    Period-specific treatment effects with a general summary measure $f$ can be defined similarly and estimated following a similar procedure. However, there may be conceptual challenges associated with developing model-robust estimators for constant or duration-specific treatment effects defined with a ratio summary measure by simply matching the treatment effect structure.
    This is because a working model typically includes secular trend parameters on the link function scale, which may be incompatible with the constant or duration-specific treatment effect estimands defined under a nonlinear summary measure $f$.
    In other words, the compatibility between the assumed and true treatment effect structures can be distorted by the specification of summary measure in the estimands definition.
    Therefore, to address a general summary measure, we recommend using the saturated treatment effect structure and taking a weighted average of $\Phi_j(d)$ to target the constant or duration-specific treatment effect estimands, and model-robust inference can be based on Corollary \ref{app.Corollary:nonlinear_estimands}.
\end{appremark}

\newpage
\section{Additional Simulation Scenarios}
\label{app.sect:additional_sim_scenarios}

\subsection{Simulation Scenario 3}

In scenario 3, we simulate a CRXO with continuous outcomes and a true saturated/period-specific treatment effect structure with $m=6$ clusters to demonstrate the robustness of linear fixed-effects models under misspecification.

We generate potential outcomes for $b=0,1$ as
\[
\begin{split}
    Y_{ijk}(b) = \mu_0 &+ \beta_{0j} + \beta_{jZd}
    \left(1 + \delta_i + \left(\frac{X_{2,i}-E[X_{2,i}]}{E[X_{2,i}]}\right) \right)b \\
    &+ \beta_{X1} sin\left(jX_{1,ijk}\right) + \beta_{X2.1}\sqrt{X_{2,ijk}} + \beta_{X2.2}cos(X_{2,ijk}) \\
    &+ \alpha_i + \gamma_{ij} + \epsilon_{ijk} 
\end{split}
\]
where $\mu_0 = 1.5$ and $(\beta_{0j})_{j=1,...,J} = (0.2j)^\top_{j=1,...,J}$.
The saturated/period-specific time-varying treatment effects in the CRXO are simulated as $\beta_{jZd}  = \beta_{jZ} = \beta_Z\left(1+0.6\left(j-\frac{\sum_{l=1}^{J} l}{J}\right)\right) \, \forall \, j$, with the saturated/period-average treatment effect being $\Delta^{S-avg} = \beta_Z = 0.7$.
Furthermore, the treatment effect randomly varies between clusters $\delta_i \sim N(0, \beta_Z^2/100)$ and by cluster-level covariate $X_{2,i}$.

An individual-level binary covariate $X_{1,ijk} \sim Bernoulli(X_{1,i})$ is generated where on the cluster-level $X_{1,i} \sim Beta(6,4)$, and $\beta_{X1} = 1.5$.
An individual-level count covariate $X_{2,ijk} \sim Poisson(X_{2,i})$ is generated where on the cluster-level $X_{2,i} \sim Gamma(0.5, 200)$, and $\beta_{X2.1} = 2$, $\beta_{X2.2}=7$. 
Due to randomization, these covariates are independent of the treatment assignment.
Finally, random intercepts are simulated with $\alpha_i \sim N(0, \tau^2)$ where $\tau^2=0.05/(1-0.05)$.
Furthermore,
$(\gamma_{ij})_{j=1,...,6} \sim MVN(0,\Sigma)$
with the correlation between periods $j$ and $l$ being $\Sigma_{jl} = \kappa^2 e^{-\lambda|j-l|}$ given scale parameter $\kappa=\tau/10$ and decay rate parameter $\lambda=0.5$, to create within-cluster correlation that decays over time.
Finally, $\epsilon_{ijk} \sim N(0,1)$.

The simulation results are illustrated in Figure \ref{fig:scenario_CRXO}. The linear fixed-effects model constant treatment effect estimator $\hat{\beta}_Z$ is minimally biased for the saturated/period-average treatment effect $\Delta^{S-avg}=\Delta^{P-avg}$ in the $m=6$ CRXO design. 
The sandwich and jackknife variance estimators slightly underestimate and overestimate the empirical variances, respectively. Accordingly, the sandwich variance estimator with normal approximation has slight under-coverage of the 95\% confidence intervals, whereas the jackknife variance estimator with $t(m-2)$ has closer to nominal coverage, even with such a complicated unadjusted covariate structure and a small number of clusters.

\begin{figure}[H]
    \centering
    \includegraphics[width=\linewidth]{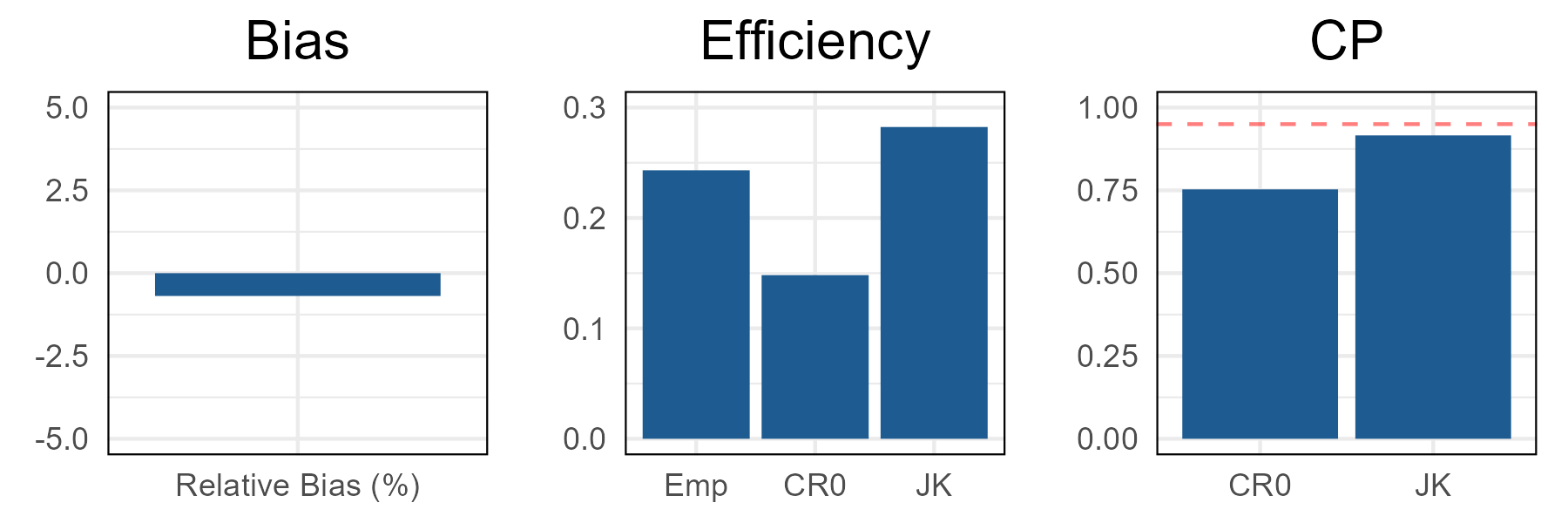}
    \caption{
        Analysis results from simulation scenario 3 of a CRXO with continuous outcomes and a true saturated/period-specific treatment effect structure, using the linear fixed-effects model with a constant treatment effect estimator. 
        Results are reported for (i.) relative bias (\%), (ii.) efficiency in terms of the empirical variance (``Emp'') in comparison to average sandwich (``CR0'') and jackknife (``JK'') variance estimates, and (iii.) coverage probabilities using the sandwich variance with normal approximation (``CR0'') and the jackknife variance with $t(m-2)$ (``JK'').
        The dashed red line denotes a nominal coverage probability of 95\%.
    }
    \label{fig:scenario_CRXO}
\end{figure}

\subsection{Simulation Scenarios 1-3 with $m=100$ clusters}
\label{app.sect:sim_scenarios100}

Results as described for simulation scenarios 1-3, but with $m=100$ clusters are included below.

In scenario 1 with a $m=100$ cluster SW-CRT and $J^*=6$ periods, the period-specific treatment effect estimands for periods $j=2,...,5$ are then computed via numerical integration
\[
\begin{split}
    (v_j(1)-v_j(0))_{j=2,...,5} &\approx (0.0729, 0.0630, 0.0540, 0.0459)^\top \,, \\ 
    \Delta^{P-avg} = \sum_{j=2}^{5} \left(v_j(1)-v_j(0)\right)/4 &\approx 0.0589 \,,
\end{split}
\]
which can be targeted by specifying analyses with a period-specific treatment effect structure.
We can additionally define the P-ATO estimand
\[
    \Delta^{P-ATO}  = \frac{ \sum_{j=1}^J \lambda_{j}[v_j(1)-v_j(0)]}{\sum_{j=1}^J \lambda_{j}} \approx 0.0588 \,, 
\]
which can then be targeted by specifying analyses with a constant treatment effect structure.

Overall, results largely match the results from the corresponding simulation scenarios with fewer clusters.
With more clusters, both the sandwich and jackknife variance estimators more closely approximate the empirical variance.
Alongside the jackknife variance with $t(m-2)$, the sandwich variance with normal approximation now also yields mostly nominal coverage probabilities across the different simulation scenarios with $m=100$ clusters. This corresponds with the consistency results reported in Lemma \ref{app.lemma:variance}.

\begin{figure}[H]
    \centering
    \includegraphics[width=\linewidth]{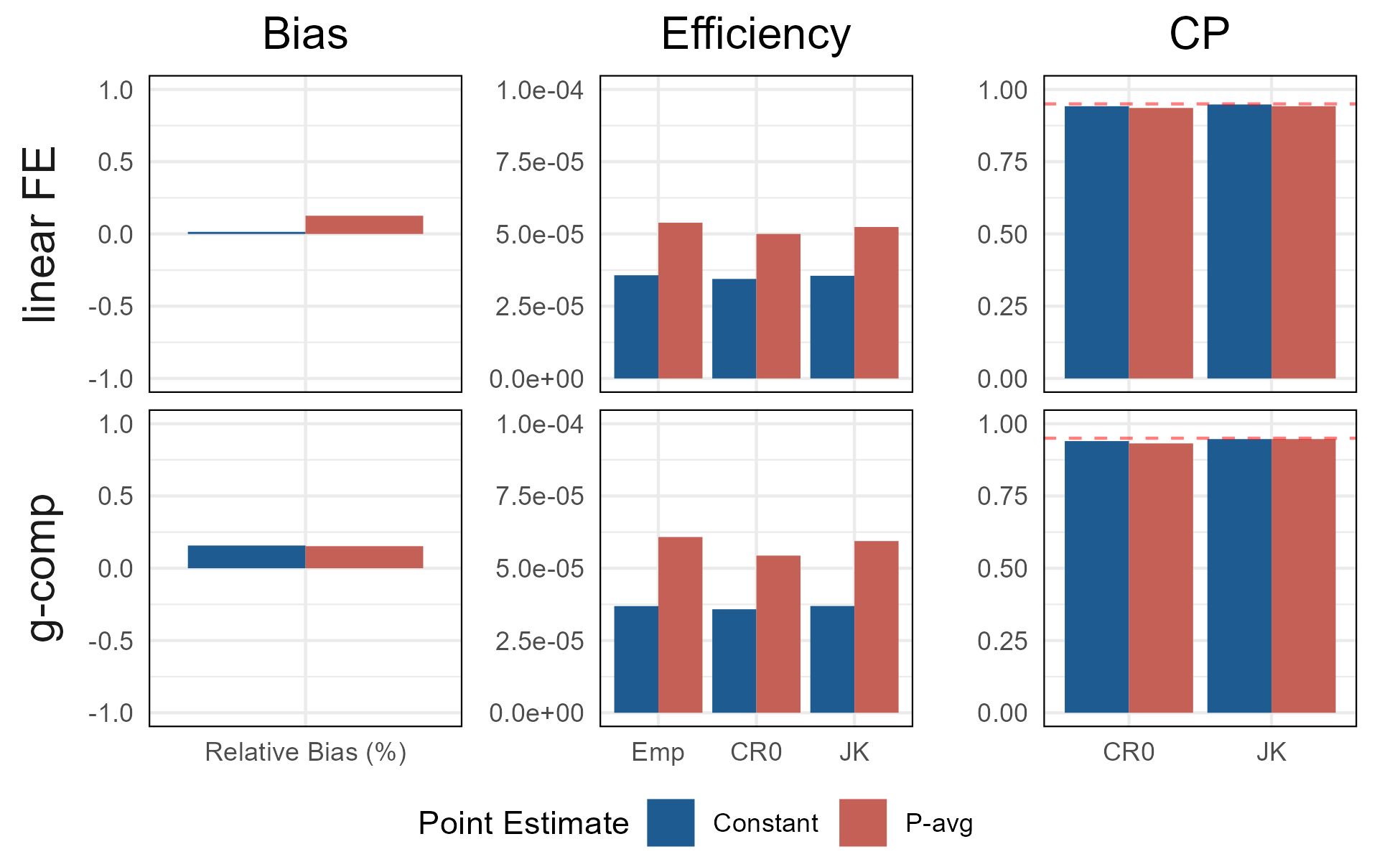}
    \caption{
        Analysis results from simulation scenario 1 of a SW-CRT with $m=100$ clusters, binary outcomes, and a true period-specific treatment effect structure, using the linear fixed-effects model (``linear FE'') and g-computation log-link fixed-effects model (``g-comp'') with constant and period-average (``P-avg'') treatment effect estimators. 
        Results are reported for (i.) relative bias (\%), (ii.) efficiency in terms of the empirical variance (``Emp'') in comparison to average sandwich (``CR0'') and jackknife (``JK'') variance estimates, and (iii.) coverage probabilities using the sandwich variance with normal approximation (``CR0'') and the jackknife variance with $t(m-2)$ (``JK'').
        The dashed red line denotes a nominal coverage probability of 95\%.
    }
    \label{fig:scenario_SWCRT_bin_100}
\end{figure}

\begin{figure}[H]
    \centering
    \includegraphics[width=\linewidth]{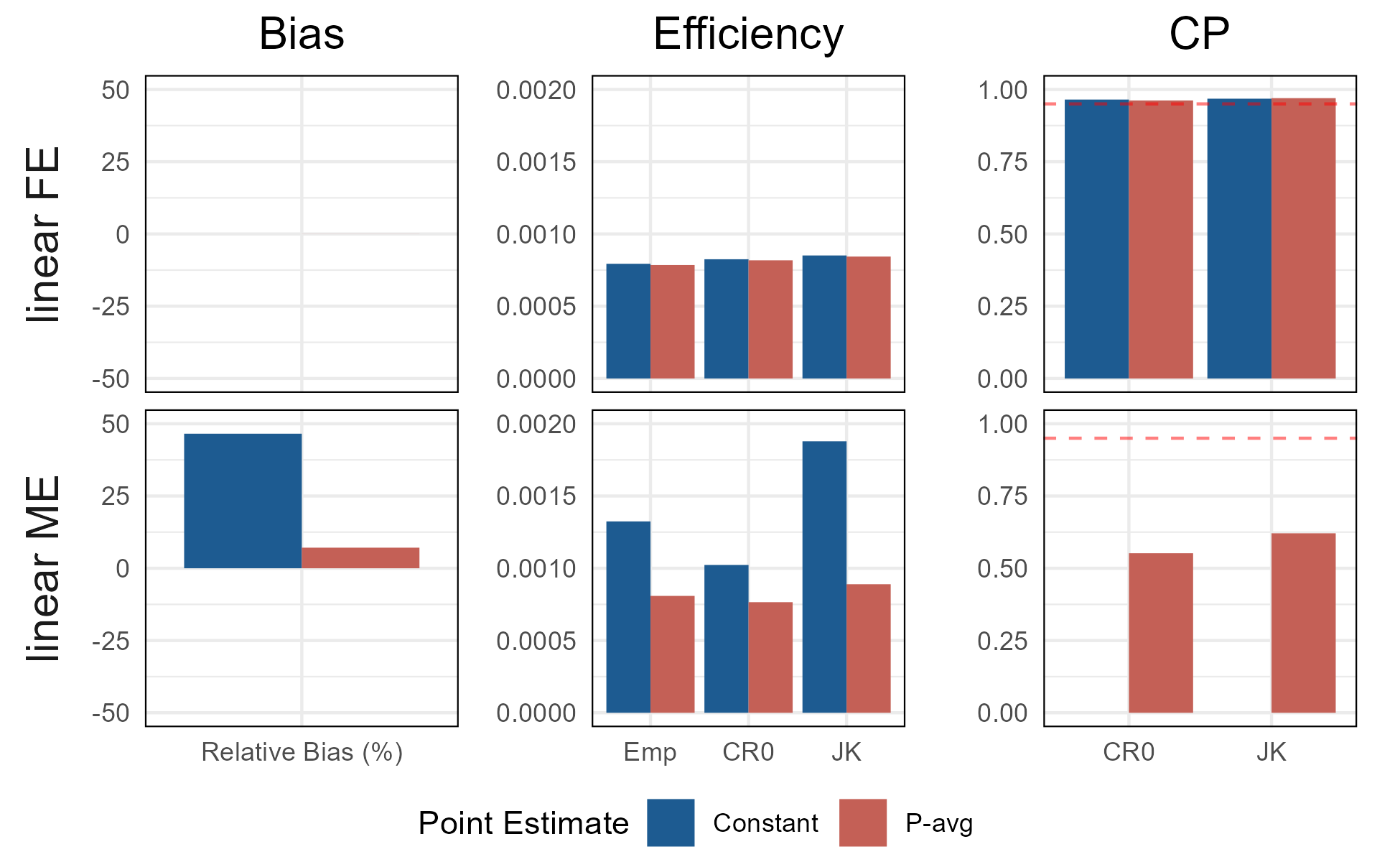}
    \caption{
        Analysis results from simulation scenario 2 of a PB-CQT with $m=100$ clusters, continuous outcomes, and a true time-varying treatment effect structure, using the linear fixed-effects model (``linear FE'') and mixed-effects model (``linear ME'') with constant and period-average (``P-avg'') treatment effect estimators. 
        Results are reported for (i.) relative bias (\%), (ii.) efficiency in terms of the empirical variance (``Emp'') in comparison to average sandwich (``CR0'') and jackknife (``JK'') variance estimates, and (iii.) coverage probabilities using the sandwich variance with normal approximation (``CR0'') and the jackknife variance with $t(m-2)$ (``JK'').
        The dashed red line denotes a nominal coverage probability of 95\%.
    }
    \label{fig:scenario_PBCQT_100}
\end{figure}

\begin{figure}[H]
    \centering
    \includegraphics[width=\linewidth]{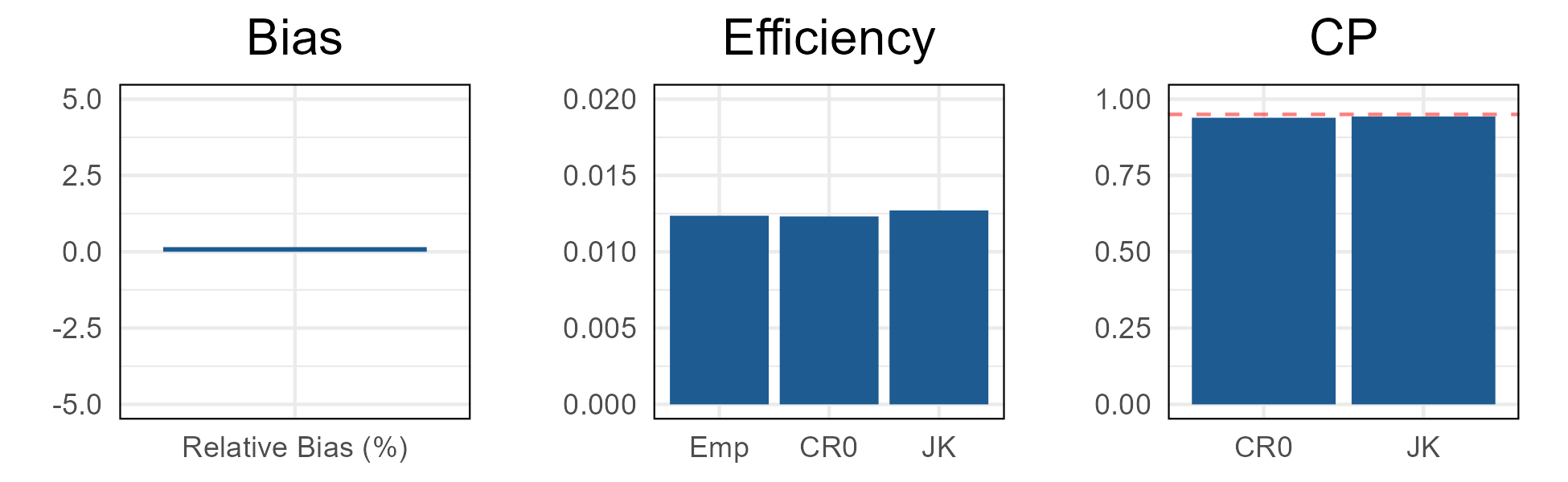}
    \caption{
        Analysis results from simulation scenario 3 of a CRXO with $m=100$ clusters, continuous outcomes, and a true saturated/period-specific treatment effect structure, using the linear fixed-effects model with a constant treatment effect estimator. 
        Results are reported for (i.) relative bias (\%), (ii.) efficiency in terms of the empirical variance (``Emp'') in comparison to average sandwich (``CR0'') and jackknife (``JK'') variance estimates, and (iii.) coverage probabilities using the sandwich variance with normal approximation (``CR0'') and the jackknife variance with $t(m-2)$ (``JK'').
        The dashed red line denotes a nominal coverage probability of 95\%.
    }
    \label{fig:scenario_CRXO_100}
\end{figure}

\subsection{Simulation Scenario 4}
\label{app.sect:sim_scenario4}

We additionally simulate a SW-CRT with count outcomes and a true constant treatment effect structure on the log-scale to demonstrate the robustness of the g-computation log-link fixed-effects model estimator under misspecification.
With the treatment effect being constant on a log scale, the resulting risk difference treatment effect estimand is then period-specific.

In this extra scenario, we generate potential outcomes for $b=0,1$ (to denote receiving control or treatment) as
\[
    ln(E[Y_{ijk}(b) | \delta_i, X_{1,ijk}, \alpha_i, \gamma_{ij}]) = \mu_0 + \beta_{0j} + \beta_{Z}(1+\delta_i)b 
    + \beta_{X1} X_{1,ijk} 
    + \alpha_i + \gamma_{ij}
\]
with $Y_{ijk}(b) \sim NB\left(r=50, p=\frac{r}{E[Y_{ijk}(b) | \delta_i, X_{1,ijk}, \alpha_i, \gamma_{ij}] + r}\right)$ following a negative-binomial distribution.
Again, $\mu_0$, $\beta_Z$ $\delta_i$, $\beta_{X1}$, $X_{1,ijk}$, and $\alpha_i$ are generated as in simulation scenario 2.
However, we now specify $(\beta_{0j})_{j=1,...,6} = (1,2,3,4,5,6)^\top$, $\tau^2 = 0.176$, and more simply generate $\gamma_{ij} \sim N(0,\tau^2/4)$.

With such a data-generating process and Assumption A1 satisfied, we can define the marginal period-specific expected potential outcome $E[Y_{ijk}(b)]=v_j(b) \, \forall \, k$ and
\[
\begin{split}
    v_j(b)
    & = E\left[e^{\mu_0 + \beta_{0j}} e^{\beta_{Z}(1+\delta_i)b} e^{\beta_{X1} X_{1,ijk}} e^{\alpha_i} e^{\gamma_{ij}}\right] \\
    &= e^{\mu_0 + \beta_{0j}} E\left[e^{\beta_{Z}(1+\delta_i)b}\right] E\left[e^{\beta_{X1} X_{1,ijk}}\right] E\left[e^{\alpha_i}\right] E\left[e^{\gamma_{ij}}\right] \\
    &= e^{\mu_0 + \beta_{0j}}e^{\beta_{Z}(1+E[\delta_i] + 0.5Var[\delta_i])b} E\left[(1 - X_{1,i}) + X_{1,i}e^{\beta_{X1}}\right]  e^{E[\alpha_i] + 0.5Var[\alpha_i]} e^{E[\gamma_{ij}] + 0.5Var[\gamma_{ij}]} \\
    &= e^{\mu_0 + \beta_{0j}}e^{\beta_{Z}(1 + 0.5Var[\delta_i])b} \left(1 - E[X_{1,i}] + E[X_{1,i}]e^{\beta_{X1}}\right)  e^{0.5Var[\alpha_i]} e^{0.5Var[\gamma_{ij}]} 
\end{split}
\]
where $E\left[e^{\beta_{X1} X_{1,ijk}}\right] = E\left[\left[e^{\beta_{X1} X_{1,ijk}}|X_{1,i}\right]\right] = E\left[(1 - X_{1,i}) + X_{1,i}e^{\beta_{X1}}\right] = 1 - E[X_{1,i}] + E[X_{1,i}]e^{\beta_{X1}}$, with $X_{1,i} \sim Beta(6,4)$ and $E[X_{1,i}]=0.6$.

The period-specific treatment effect estimands are then 
\[
\begin{split}
    (v_j(1)-v_j(0))_{j=2,...,5} &\approx (116, 316, 859, 2337)^\top \,, \\ 
    \Delta^{P-avg} &\approx 907 \,,
\end{split}
\]
which can be targeted by specifying analyses with a period-specific treatment effect structure.
We can additionally define the P-ATO estimand
\[
    \Delta^{P-ATO}  = \frac{ \sum_{j=1}^J \lambda_{j}[v_j(1)-v_j(0)]}{\sum_{j=1}^J \lambda_{j}} \approx 843 \,,
\]
which can then be targeted by specifying analyses with a constant treatment effect structure.

\begin{figure}[H]
    \centering
    \includegraphics[width=\linewidth]{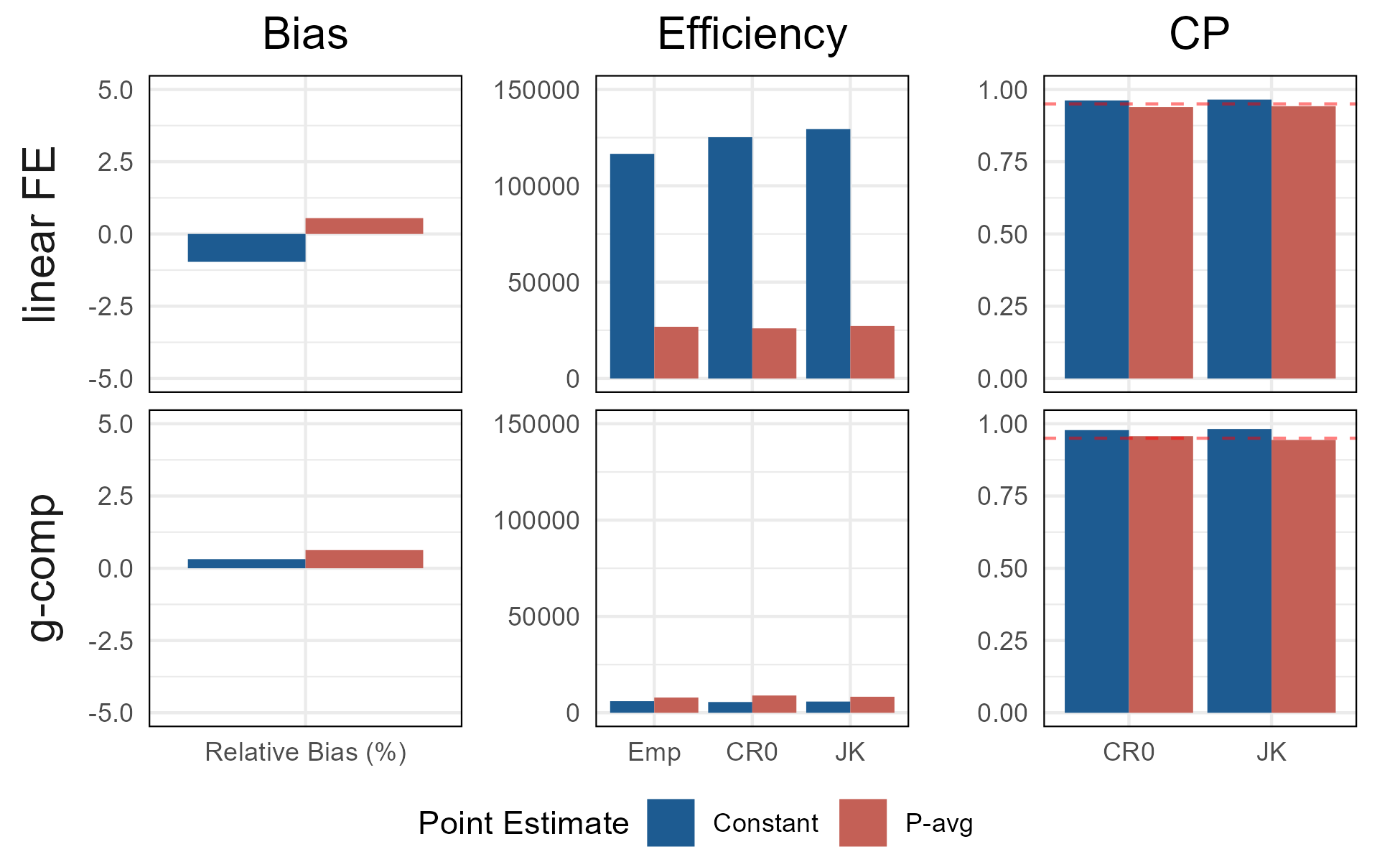}
    \caption{
        Analysis results from simulation scenario 4 of a SW-CRT with $m=100$ clusters, count outcomes, and a true period-specific treatment effect structure, using the linear fixed-effects model (``linear FE'') and g-computation log-link fixed-effects model (``g-comp'') with constant and period-average (``P-avg'') treatment effect estimators. 
        Results are reported for (i.) relative bias (\%), (ii.) efficiency in terms of the empirical variance (``Emp'') in comparison to average sandwich (``CR0'') and jackknife (``JK'') variance estimates, and (iii.) coverage probabilities using the sandwich variance with normal approximation (``CR0'') and the jackknife variance with $t(m-2)$ (``JK'').
        The dashed red line denotes a nominal coverage probability of 95\%.
    }
    \label{fig:scenario_SWCRT_count_100}
\end{figure}

\subsection{Tables of complete simulation results}
\label{app.sect:sim_table}

\begin{table}[htbp]
\centering
\caption{Simulation results ($m=6$ clusters) by scenario, model, estimand, and point estimator (``Point Est.''). Results include \% relative bias (``\% Rel. Bias''), empirical variance (``Emp.Var''), average variance with CR0, average variance with jackknife (``JK''), and corresponding coverage probabilities.}
\label{tab:sim_results}
\resizebox{\textwidth}{!}{%
\begin{tabular}{cllllrrrrr}
\toprule
Scenario & Model & Estimand & Point Est. & \% Rel.\ Bias & Emp.\ Var & CR0 & JK & CP (CR0) & CP (JK-$t$) \\
\midrule
 1 & linear FE     & P-ATO & Constant & $\phantom{-}$0.850 & 0.00101  & 0.000669 & 0.00128 & 0.829 & 0.963 \\
 & linear FE     & P-avg & P-avg    & $\phantom{-}$1.47  & 0.00166  & 0.000959 & 0.00222 & 0.796 & 0.963 \\
 & g-comp        & P-ATO & Constant & $\phantom{-}$1.26  & 0.00104  & 0.000689 & 0.00132 & 0.828 & 0.962 \\
 & g-comp        & P-avg & P-avg    & $\phantom{-}$1.96  & 0.00192  & 0.00112  & 0.00262 & 0.805 & 0.966 \\
 & linear ME EX  & P-ATO & Constant & $\phantom{-}$0.189 & 0.000933 & 0.000586 & 0.00124 & 0.816 & 0.966 \\
 & linear ME EX  & P-avg & P-avg    & $\phantom{-}$0.252 & 0.00127  & 0.000675 & 0.00188 & 0.788 & 0.968 \\
 & linear ME NEX & P-ATO & Constant & $-$0.340          & 0.000930 & 0.000559 & 0.00129 & 0.812 & 0.970 \\
 & linear ME NEX & P-avg & P-avg    & $-$1.11           & 0.00114  & 0.000542 & 0.00194 & 0.766 & 0.964 \\
\midrule
 2 & linear FE     & P-avg & Constant & $-$1.08            & 0.0152 & 0.0103  & 0.0198 & 0.819 & 0.971 \\
 & linear FE     & P-avg & P-avg    & $-$1.09            & 0.0150 & 0.0103  & 0.0197 & 0.816 & 0.971 \\
 & linear ME EX  & P-avg & Constant & $\phantom{-}$5.43  & 0.0157 & 0.00960 & 0.0210 & 0.764 & 0.967 \\
 & linear ME EX  & P-avg & P-avg    & $\phantom{-}$5.32  & 0.0155 & 0.00959 & 0.0209 & 0.777 & 0.972 \\
 & linear ME NEX & P-avg & Constant & $\phantom{-}$23.5  & 0.0224 & 0.00951 & 0.0355 & 0.525 & 0.924 \\
 & linear ME NEX & P-avg & P-avg    & $\phantom{-}$7.14  & 0.0160 & 0.00937 & 0.0227 & 0.755 & 0.973 \\
\midrule
 3 & linear FE     & P-avg & Constant & $-$0.688 & 0.243 & 0.148 & 0.282 & 0.753 & 0.916 \\
 & linear ME EX  & P-avg & Constant & $-$0.687 & 0.243 & 0.149 & 0.282 & 0.756 & 0.916 \\
 & linear ME NEX & P-avg & Constant & $-$0.898 & 0.240 & 0.149 & 0.279 & 0.754 & 0.915 \\
\bottomrule
\end{tabular}%
}
\end{table}

\begin{table}[htbp]
\centering
\caption{Simulation results ($m=100$ clusters) by scenario, model, estimand, and point estimator (``Point Est.''). Results include \% relative bias (``\% Rel.\ Bias''), empirical variance (``Emp.\ Var''), average variance with CR0, average variance with jackknife (``JK''), and corresponding coverage probabilities.}
\label{tab:sim_results}
\resizebox{\textwidth}{!}{%
\begin{tabular}{cllllrrrrr}
\toprule
Scenario & Model & Estimand & Point Est. & \% Rel.\ Bias & Emp.\ Var & CR0 & JK & CP (CR0) & CP (JK-$t$) \\
\midrule
 1 & linear FE     & P-ATO & Constant & $\phantom{-}$0.0144  & $3.57\times10^{-5}$ & $3.44\times10^{-5}$ & $3.55\times10^{-5}$ & 0.942 & 0.948 \\
 & linear FE     & P-avg & P-avg    & $\phantom{-}$0.126   & $5.39\times10^{-5}$ & $5.00\times10^{-5}$ & $5.24\times10^{-5}$ & 0.936 & 0.942 \\
 & g-comp        & P-ATO & Constant & $\phantom{-}$0.157   & $3.69\times10^{-5}$ & $3.58\times10^{-5}$ & $3.69\times10^{-5}$ & 0.940 & 0.947 \\
 & g-comp        & P-avg & P-avg    & $\phantom{-}$0.153   & $6.08\times10^{-5}$ & $5.44\times10^{-5}$ & $5.94\times10^{-5}$ & 0.932 & 0.947 \\
 & linear ME EX  & P-ATO & Constant & $-$0.0176            & $3.15\times10^{-5}$ & $3.06\times10^{-5}$ & $3.17\times10^{-5}$ & 0.941 & 0.949 \\
 & linear ME EX  & P-avg & P-avg    & $\phantom{-}$0.0731  & $4.13\times10^{-5}$ & $3.81\times10^{-5}$ & $4.02\times10^{-5}$ & 0.932 & 0.942 \\
 & linear ME NEX & P-ATO & Constant & $-$0.0578            & $2.95\times10^{-5}$ & $2.91\times10^{-5}$ & $3.02\times10^{-5}$ & 0.944 & 0.952 \\
 & linear ME NEX & P-avg & P-avg    & $\phantom{-}$0.0143  & $3.51\times10^{-5}$ & $3.28\times10^{-5}$ & $3.49\times10^{-5}$ & 0.933 & 0.948 \\
\midrule
 2 & linear FE     & P-avg & Constant & $-$0.000555          & $7.94\times10^{-4}$ & $8.25\times10^{-4}$ & $8.51\times10^{-4}$ & 0.965 & 0.968 \\
 & linear FE     & P-avg & P-avg    & $-$0.0220            & $7.85\times10^{-4}$ & $8.18\times10^{-4}$ & $8.44\times10^{-4}$ & 0.962 & 0.970 \\
 & linear ME EX  & P-avg & Constant & $\phantom{-}$6.38    & $8.11\times10^{-4}$ & $7.77\times10^{-4}$ & $8.81\times10^{-4}$ & 0.630 & 0.680 \\
 & linear ME EX  & P-avg & P-avg    & $\phantom{-}$6.06    & $8.00\times10^{-4}$ & $7.71\times10^{-4}$ & $8.70\times10^{-4}$ & 0.665 & 0.704 \\
 & linear ME NEX & P-avg & Constant & $\phantom{-}$46.6    & $1.32\times10^{-3}$ & $1.02\times10^{-3}$ & $1.88\times10^{-3}$ & 0.000 & 0.000 \\
 & linear ME NEX & P-avg & P-avg    & $\phantom{-}$7.12    & $8.09\times10^{-4}$ & $7.65\times10^{-4}$ & $8.90\times10^{-4}$ & 0.552 & 0.621 \\
\midrule
 3 & linear FE     & P-avg & Constant & $\phantom{-}$0.152   & $1.24\times10^{-2}$ & $1.23\times10^{-2}$ & $1.27\times10^{-2}$ & 0.939 & 0.943 \\
 & linear ME EX  & P-avg & Constant & $\phantom{-}$0.152   & $1.24\times10^{-2}$ & $1.22\times10^{-2}$ & $1.27\times10^{-2}$ & 0.939 & 0.943 \\
 & linear ME NEX & P-avg & Constant & $\phantom{-}$0.171   & $1.22\times10^{-2}$ & $1.22\times10^{-2}$ & $1.25\times10^{-2}$ & 0.934 & 0.940 \\
\midrule
 4 & linear FE     & P-ATO & Constant & $-$0.965             & $1.17\times10^{5}$  & $1.25\times10^{5}$  & $1.29\times10^{5}$  & 0.962 & 0.965 \\
 & linear FE     & P-avg & P-avg    & $\phantom{-}$0.548   & $2.68\times10^{4}$  & $2.60\times10^{4}$  & $2.72\times10^{4}$  & 0.939 & 0.942 \\
 & g-comp        & P-ATO & Constant & $\phantom{-}$0.316   & $6.02\times10^{3}$  & $5.57\times10^{3}$  & $5.82\times10^{3}$  & 0.978 & 0.982 \\
 & g-comp        & P-avg & P-avg    & $\phantom{-}$0.628   & $7.90\times10^{3}$  & $8.95\times10^{3}$  & $8.31\times10^{3}$  & 0.957 & 0.944 \\
  & linear ME EX  & P-ATO & Constant & $-$0.940             & $1.11\times10^{5}$  & $1.20\times10^{5}$  & $1.24\times10^{5}$  & 0.955 & 0.964 \\
 & linear ME EX  & P-avg & P-avg    & $\phantom{-}$0.548   & $2.61\times10^{4}$  & $2.53\times10^{4}$  & $2.65\times10^{4}$  & 0.939 & 0.945 \\
 & linear ME NEX & P-ATO & Constant & $-$0.103             & $1.46\times10^{4}$  & $1.82\times10^{4}$  & $2.06\times10^{4}$  & 0.928 & 0.941 \\
 & linear ME NEX & P-avg & P-avg    & $\phantom{-}$0.560   & $1.53\times10^{4}$  & $1.52\times10^{4}$  & $1.60\times10^{4}$  & 0.947 & 0.951 \\
\bottomrule
\end{tabular}%
}
\end{table}

\newpage
\section{Case study re-analysis with multiple imputation}
\label{app.sect:case_study}

\begin{figure}[H]
    \centering
    \includegraphics[width=0.8\linewidth]{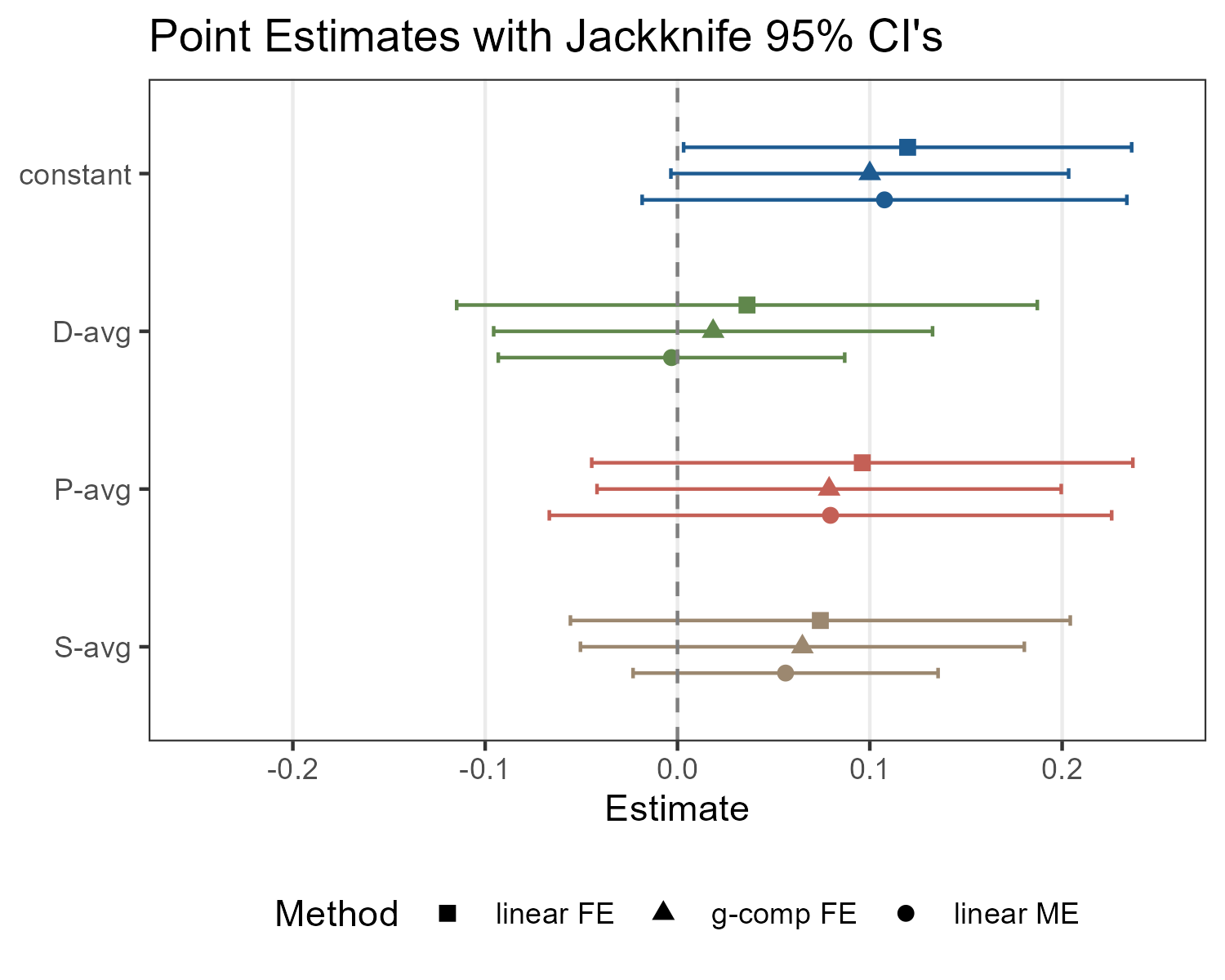}
    \caption{
        Results from the re-analysis of binary outcomes using the linear fixed-effects model (``linear FE''), g-computation log-link fixed-effects model (``g-comp''), and linear mixed-effects model (``linear ME'') with constant, duration-average (``D-avg''), period-average (``P-avg'') and saturated-average (``S-avg'') treatment effect estimators. 
        Missing data were imputed using multiple imputation (MI) for binary outcomes.
        Confidence intervals are formed using the jackknife variance estimator with 
        $t(m-2)$ and following Rubin's rules (RR) for MI.
    }
    \label{fig:case_study_MI}
\end{figure}

\newpage
\putbib
\end{bibunit}

\end{document}